\documentclass[11pt]{article}
\usepackage{yaonan}

\newcommand{\yj}[1]{{\bf \color{blue} {(Yaonan: #1)}}}

\newcommand{\FPA}{{\sf FPA}}
\newcommand{\OPT}{{\sf OPT}}
\newcommand{\PoA}{{\sf PoA}}

\newcommand{\SocialWelfare}{\textsf{Social Welfare}}
\newcommand{\SocialWelfares}{\textsf{Social Welfares}}
\newcommand{\FirstPriceAuction}{\textsf{First Price Auction}}
\newcommand{\FirstPriceAuctions}{\textsf{First Price Auctions}}
\newcommand{\PriceofAnarchy}{\textsf{Price of Anarchy}}

\newcommand{\BayesNashEquilibrium}{\textsf{Bayesian Nash Equilibrium}}
\newcommand{\BayesNashEquilibria}{\textsf{Bayesian Nash Equilibria}}

\newcommand{\rename}{\text{\tt Rename}}
\newcommand{\discretize}{\text{\tt Discretize}}
\newcommand{\translate}{\text{\tt Translate}}
\newcommand{\layer}{\text{\tt Layer}}
\newcommand{\polarize}{\text{\tt Polarize}}
\newcommand{\slice}{\text{\tt Slice}}
\newcommand{\halve}{\text{\tt Halve}}

\newcommand{\AD}{\text{\tt Ascend-Descend}}
\newcommand{\main}{\text{\tt Main}}
\newcommand{\collapse}{\text{\tt Collapse}}

\newcommand{\alloc}{\mathrm{x}}         % allocation
\newcommand{\pay}{\rho}                 % payment
\newcommand{\pays}{\boldsymbol{\rho}}   % payment vector
\newcommand{\bbFPA}{\mathbb{FPA}}       % space of FPA
\newcommand{\bbBNE}{\mathbb{BNE}}       % space of BNE
\newcommand{\bbV}{\mathbb{V}}           % space of value distributions
\newcommand{\Bvalid}{\mathbb{B}_{\sf valid}}

\title{First Price Auction is $1 - 1 / e^{2}$ Efficient}

\author{
Yaonan Jin\thanks{Huawei TCS Lab. Work done while the author was a PhD student at Columbia University.} \\
jinyaonan@huawei.com
\and
Pinyan Lu\thanks{Shanghai University of Finance and Economics \& Huawei TCS Lab.} \\
lu.pinyan@mail.shufe.edu.cn
}
\date{}

\begin{document}

\maketitle
\begin{abstract}
We prove that the {\sf PoA} of {\sf First Price Auctions} is $1 - 1 / e^{2} \approx 0.8647$, closing the gap between the best known bounds $[0.7430,\, 0.8689]$.
% We prove that the {\sf Bayesian PoA} of {\sf First Price Auctions} with independent bidders is $1 - 1 / e^{2}$.
% Sell an item to $n \geq 1$ independent bidders via {\sf First Price Auction}. Considering all {\sf Bayes Nash Equilibria}, it is shown that the tight {\sf Price of Anarchy} is $1 - 1 / e^{2} \approx 0.8647$.
% \yj{(i)~modify the valid instance theorem's proof.}
\end{abstract}
\thispagestyle{empty}

% \newpage
\thispagestyle{empty}
{\hypersetup{linkcolor=black}\tableofcontents}
\thispagestyle{empty}
\newpage

\setcounter{page}{1}

\section{Introduction}
\label{sec:intro}

%In 1950, Nash \cite{N50auction} initiated the study of auction theory. At the center of his study was the solution concept of Nash equilibrium \cite{N50equilibrium} for auctions as non-cooperative games.
%In 1961, following Nash's original spirit, Vickrey \cite{V61} further applied the tools of game theory to explain traditional auctions and the dynamics therein.

In 1961, Vickrey \cite{V61}, a Nobel laureate in Economics,
initiated auction theory. At the center of his work was the solution concept of Nash equilibrium \cite{N50} for auctions as non-cooperative games. This game-theoretical approach has shaped modern auction theory and has a tremendous influence also on other areas in mathematical economics \cite{H04}. In particular, Vickrey systematically investigated the first-price auction (or its equivalent, the Dutch auction),
the most common auction format in real business.

% which had been widely used in the real business far before Vickrey's systematic study.
% a central auction format was

In the first-price auction, the auctioneer (seller) sells an indivisible item to $n$ potential bidders (buyers).
The rule is as simple as it can be: All bidders simultaneously submit bids to the auctioneer (each of which is unknown to the other bidders); the highest bidder wins the item, paying his/her own bid.
% Although the rule is very simple,
Simple as the rule is, the bidders' optimal bidding strategies can be sophisticated.
From one bidder's own perspective, a higher bid means a higher payment on winning but also a better chance of winning.
Accordingly, this bidder's optimal bidding strategy depends on the competitive environment, which in turn is determined by the other bidders’ bidding strategies. This situation is perfectly a non-cooperative game and the standard solution concept is Bayesian Nash equilibrium.
To get a better sense, let us show a warm-up example from Vickrey's original work.

\begin{example}[{\cite{V61}}]
\label{exp:intro:1}
Consider two bidders: Alice and Bob have (independent) uniform random values $v_{1},\, v_{2} \sim U[0,\, 1]$ and respectively bid $b_{1} = \frac{v_{1}}{2}$ and $b_{2} = \frac{v_{2}}{2}$. The value distributions and bidding strategies determine the (independent) bid distributions of bidders. In this example, they are (independent) uniform random bids $b_{1},\, b_{2} \sim U[0,\, \frac{1}{2}]$, whose CDFs are $B_{1}(b) = B_{2}(b) = \min(2b,\, 1)$.
By bidding $b \geq 0$, Alice wins with probability $\min(2b,\, 1)$ and gains a utility $(v_{1} - b)$ conditioned on winning. Her expected utility is $(v_{1} - b) \cdot \min(2b,\, 1)$, which is maximized when $b = \frac{v_{1}}{2}$.
Thus, her current strategy $b_{1} = \frac{v_{1}}{2}$ is optimal. By symmetry, Bob's current strategy $b_{2} = \frac{v_{2}}{2}$ is also optimal.
In sum, the above strategy profile $(b_{1},\, b_{2}) = (\frac{v_{1}}{2},\, \frac{v_{2}}{2})$ is an equilibrium, in a sense that no bidder can gain a better utility by deviating from his/her current strategy.
\end{example}

For clarity, let us formalize the model. Each bidder $i \in [n]$ independently draws his/her value from a distribution $v_{i} \sim V_{i}$.
Only with this knowledge and depending on his/her own strategy $s_{i}$, bidder $i$ submits a possibly random bid $b_{i} = s_{i}(v_{i})$.
% Under value $v$ and bid $b$, denote by $\mathrm{x}_{i}(b)$ this bidder's winning probability, over the randomness of the others' values $\bv_{-i} = (v_{k})_{k \neq i}$ and strategies $\bs_{-i} = \{s_{k}\}_{k \neq i}$, and by $u_{i}(v,\, b) = (v - b) \cdot \mathrm{x}_{i}(b)$ this bidder's expected utility.
Then over the randomness of the other bidders' values $\bv_{-i} = (v_{k})_{k \neq i}$ and strategies $\bs_{-i} = \{s_{k}\}_{k \neq i}$, bidder $i$ wins with probability $\mathrm{x}_{i}(b_{i}) \in [0,\, 1]$ and gains an expected utility $u_{i}(v_{i},\, b_{i}) = (v_{i} - b_{i}) \cdot \mathrm{x}_{i}(b_{i})$.

% \blue{Now let us formalize the model. Each bidder $i \in [n]$ draws his/her value $v_i$ independently from a distribution $V_{i}$. After knowing his own value but without knowing other's values, bidder $i$ submits a (possibly random) bid $b_i=s_i(v_i)$  based on his value $v_i$ and strategy $s_i$. We use $u_i(v,b)=(v-b) \cdot \mathrm{x}_i(b)$ to denote bidder $i$'s expected utility with value $v$ and bid $b$, where the winning probability $\mathrm{x}_i(b)$ is taking over the randomness of other bidders' valuations and strategies.}

\begin{definition}[Equilibria]
\label{def:bne}
A strategy profile $\bs = \{s_{i}\}_{i \in [n]}$ is a {\BayesNashEquilibrium} when: For each bidder $i \in [n]$ and any possible value $v \in \supp(V_{i})$, the current strategy $s_{i}(v)$ is optimal, namely $\Ex_{s_{i}} \big[ u_{i}(v,\, s_{i}(v)) \big] \geq u_{i}(v,\, b)$ for any deviation bid $b \geq 0$.
\end{definition}

\Cref{exp:intro:1} is special in that bidders have identically distributed values. This {\em symmetric} setting is well understood: The first-price auction has a {\em unique} Bayesian Nash equilibrium \cite{CH13}, which is {\em fully efficient} -- The bidder with the highest value always wins the item.
Instead, the current trend focuses more on the {\em asymmetric} setting, where bidders' values are distinguished by their distributions.
Again, let us get a better sense through two concrete examples.

% \blue{Let us see two more examples of {\BayesNashEquilibrium} (BNE) for the first-price auction.}

\begin{example}
\label{exp:intro:2}
Consider two bidders: Alice has a fixed value $v_{1} \equiv 2$ and always bids $s_{1}(v_{1}) \equiv 1$. Bob has a uniform random value $v_{2} \sim U[0,\, 1]$ and {\em truthfully} bids his value $s_{2}(v_{2}) = v_{2}$, namely the distribution $B_{2}(b) = \min(b,\, 1)$.
By bidding $b \geq 0$, Alice gains an expected utility $(v_{1} - b) \cdot \min(b,\, 1)$, for which her current strategy $s_{1}(v_{1}) \equiv 1$ is optimal.
Bob cannot gain a positive utility because his value $v_{2} \sim U[0,\, 1]$ is at most Alice's bid $s_{1}(v_{1}) \equiv 1$; thus his current strategy $s_{2}(v_{2}) = v_{2}$ is also optimal.
In sum, this strategy profile $(s_{1}, s_{2})$ is an equilibrium.

Unlike the symmetric setting, this auction game has {\em multiple} equilibria.
E.g., it is easy to verify that the same strategy $s_{1}(v_{1}) \equiv 1$ for Alice and a different strategy $\tilde{s}_{2}(v_{2}) = \max(\frac{2v_{2} - 1}{v_{2}},\, 0)$ for Bob (namely the distribution $\tilde{B}_{2}(b) = \frac{1}{2 - b}$ for $b \in [0,\, 1]$) are another equilibrium.
\end{example}

% We first verify that there is no better strategy for Bob. Given the competitive bid of deterministic $1$, Bob with value at most $1$ cannot have any positive utility. His truthful bidding gets a zero utility and thus is indeed already the best possible.
% For Alice, there is no point for her to increase her bid since she already gets the item with probability $1$ by bidding $1$. By bidding a lower bid $b\leq 1$, her winning probability is $b$ (since the competitive bid is uniform from $U[0,1]$) and thus her expected utility is $b(2-b)$. It is easy to see that the optimal bidding is $1$.

Both equilibria $(s_{1},\, s_{2})$ and $(s_{1},\, \tilde{s}_{2})$ given in \Cref{exp:intro:2} have two features:
(i)~The strategies $s_{1}(v_{1})$, $s_{2}(v_{2})$ and $\tilde{s}_{2}(v_{2})$ are {\em pure strategies} -- Each of them has no randomness, just the function of a value.
(ii)~Both equilibria are {\em fully efficient} akin to the symmetric setting -- Alice always has the highest value $\equiv 2$ and always wins the item.
However, these are not always the case in the asymmetric setting, such as in the following example, which only slightly modifies \cref{exp:intro:2} by changing the fixed value of Alice from $2$ to $1$.

\begin{example}
\label{exp:intro:3}
Consider two bidders: Alice has a fixed value $v_{1} \equiv 1$ and bids $s_{1}(v_{1}) \sim B_{1}$ following the distribution $B_1(b) = \frac{1}{4b - 2}\exp(\frac{4b - 3}{2b - 1})$ for $b \in [\frac{1}{2},\, \frac{3}{4}]$. Bob has a uniform random value $v_{2} \sim U[0,\, 1]$ and bids $s_{2}(v_{2}) = \max(\frac{4v_{2} - 1}{4v_{2}},\, 0)$, namely the distribution $B_2(b) = \frac{1}{4 - 4b}$ for $b \in [0,\, \frac{3}{4}]$.
By bidding $b \geq 0$, Alice gains an expected utility $(v_{1} - b) \cdot B_{2}(b)$, which is maximized $= \frac{1}{4}$ anywhere between $b \in [0,\, \frac{3}{4}]$, so her current strategy $s_{1}(v_{1}) \sim B_{1}$ is optimal. Using elementary algebra, we can check that Bob's current strategy $s_{2}(v_{2}) = \max(\frac{4v_{2} - 1}{4v_{2}},\, 0)$ also is optimal. In sum, this strategy profile $(s_{1}, s_{2})$ is an equilibrium.
\end{example}

In \Cref{exp:intro:3}: Alice has a {\em mixed strategy} -- A fixed value $v_{1} \equiv 1$ but a random bid $s_{1}(v_{1}) \sim B_{1}$.
Furthermore, the equilibrium $(s_{1},\, s_{2})$ is {\em not} fully efficient. E.g., with a value $v_{2} = \frac{3}{4}$, although not the highest value $< v_{1} \equiv 1$, Bob bids $s_{2}(v_{2}) = \frac{2}{3}$ and wins with probability $B_{1}(\frac{2}{3}) = \frac{3}{2e} \approx 0.5518$. Indeed, \Cref{exp:intro:3} has infinite equilibria, among which the given one $(s_{1}, s_{2})$ has the relatively ``simplest'' format. But none of those equilibria is a pure equilibrium or is fully efficient, despite that \cref{exp:intro:3} is only a minor modification of  \cref{exp:intro:2}.
% \red{Although the instance is only a minor modification of \cref{exp:intro:2}, there is no pure or efficient equilibrium for \cref{exp:intro:3}. }

% \blue{This is called the asymmetric setting while \Cref{exp:intro:1} belongs to symmetric settings. The following \Cref{exp:intro:3} is more complicated in the sense that the strategy in the equilibrium is a mixing one: the bid for a given value is a distribution rather than a fixed bid. In general, mixing strategies are necessary for equilibrium.}

% The bid of the first bidder with fixed value $1$ is a distribution defined by the following CDF on the support :
% The bidding strategy for the second bidder is .
% The second bidder's strategy transforms his uniform value distribution to a bid distribution defined by CDF function $B_2(b) = \frac{1}{4(1-b)}$ for $b\in [\frac{1}{4},\, \frac{3}{4}]$.

From the above examples, we observe that Bayesian Nash equilibria
% in the first-price auction
can be very complicated and sensitive to the value distributions, despite the intrinsic simple rule of the first-price auction.  After an extensive study for more than 60 years, the first-price auction and its equilibria remain the centerpiece of modern auction theory and have promoted a rich literature; see \cite{SZ90,L06,L96,L99,MR00a,MR00b,JSSZ02,MR03,HKMN11,CH13,ST13,S14,HHT14,FLN16,HTW18,WSZ20,FGHLP21} etc. These efforts are justifiable since the study of the first-price auction and its equilibria is both theoretically challenging and practically important.

Among various aspects of the first-price auction, {\em efficiency} at equilibria is of primary interest. In economics, efficiency measures to what extent a resource can be allocated to the persons who value it the most, thus maximizing the {\em social welfare}, particularly in a competitive environment.
% \red{Hence, the first-price auction or generally all auctions are perfect scenarios.}
As shown above (\Cref{exp:intro:3}), the first-price auction generally is not fully efficient at an equilibrium: The winner has the highest {\em bid} but possibly not the highest {\em value}; this crucially depends on both (i)~the instance $\bV$ itself and (ii)~which Bayesian Nash equilibrium $\bs \in \bbBNE(\bV)$ it falls into.
Earlier works in economics focus on (generalizing) the conditions for the value distribution $\bV = \{V_{i}\}_{i \in [n]}$
that guarantee the full efficiency.

However, the quality of (in)efficiency should not be all-or-nothing.
E.g., when the highest value is $1$, a value-$0.99$ bidder versus a value-$0.01$ bidder is very different, although neither of them is fully efficient. Towards a quantitative analysis,
Koutsoupias and Papadimitriou introduced a new measure on the efficiency degradation under selfish behaviors, the Price of Anarchy \cite{KP99} (which is an analog of the ``approximation ratio'' in theoretical computer science).
For the first-price auction, denote by $\OPT(\bV)$ the expected optimal social welfare from an instance $\bV$, and by $\FPA(\bV,\, \bs)$ the expected social welfare at an equilibrium $\bs \in \bbBNE(\bV)$, then the Price of Anarchy is defined to be the minimum possible ratio, as follows.

\begin{definition}[{\PriceofAnarchy}]
\label{def:poa}
Regarding {\FirstPriceAuctions}, the {\PriceofAnarchy} is given by
\[
    \PoA ~\eqdef~ \inf \bigg\{\, \frac{\FPA(\bV,\, \bs)}{\OPT(\bV)} \,\biggmid\, \bs \in \mathbb{BNE}(\bV) ~\text{and}~ \OPT(\bV) < +\infty \,\bigg\}.
\]
\end{definition}

The Price of Anarchy is bounded between $[0,\, 1]$. Namely, a larger ratio means a higher efficiency and the $= 1$ ratio means the full efficiency.
For the first-price auction, Syrgkanis and Tardos first proved that the {\PoA} is at least  $1 - 1 / e \approx 0.6321$ \cite{ST13}.
Later, Hoy, Taggart and Wang derived an improved lower bound of $\approx 0.7430$ \cite{HTW18}.
On the other hand, Hartline, Hoy and Taggart gave a concrete instance of ratio $\approx 0.8689$ \cite{HHT14}, which remains the best known upper bound.

% When an instance $\bV$ admits multiple equilibria $|\bbBNE(\bV)| > 1$, the Price of Anarchy considers the worst case, i.e., a {\em pessimistic} measure.
% The {\em optimistic} counterpart, called the Price of Stability \cite{ADKTWR08}, considers the best case and is also of fundamental interest.

% \begin{definition}[{\PriceofStability}]
% \label{def:pos}
% Regarding {\FirstPriceAuctions}, the {\PriceofStability} is given by
% \[
%     \PoS ~\eqdef~ \inf \bigg\{\, \sup \bigg\{\, \frac{\FPA(\bV,\, \bs)}{\OPT(\bV)} \,\biggmid\, \bs \in \mathbb{BNE}(\bV) \,\bigg\} \,\biggmid\, \OPT(\bV) < +\infty \,\bigg\}.
% \]
% \end{definition}

% The Price of Stability is bounded between $[\PoA,\, 1]$. Especially, if this is {\em equal} to the tight {\PoA}, then regarding the worst-case instance $\bV^{*}$, either one equilibrium is unique $|\bbBNE(\bV^{*})| = 1$ or any two equilibria $\bs \neq \tilde{\bs} \in \bbBNE(\bV^{*})$ are equally efficient $\FPA(\bV^{*},\, \bs) = \FPA(\bV^{*},\, \tilde{\bs})$.
% Conceivably, the {\em minimax} optimization for {\PoS} is technically more difficult than the {\em minimization} for {\PoA}. Indeed, the best known {\PoS} lower bound of $\approx 0.7430$ \cite{HTW18} just follows from the same {\PoA} lower bound by implication, and no nontrivial upper bound $< 1$ has been formally proved.

Despite the prevalence of the first-price auction and much effort in studying its efficiency, there has been a persistent gap in the state of the art. In the current paper, we will completely solve this long-standing open problem.

\begin{theorem}[Tight {\PoA}]
\label{thm:main}
The {\PriceofAnarchy} in {\FirstPriceAuctions} is $1 - 1 / e^{2} \approx 0.8647$.
\end{theorem}

Remarkably, neither the best known lower bound $\approx 0.7430$ nor the upper bound $\approx 0.8689$ is tight; we close the gap by improving both of them. Our tight bound $1 - 1 / e^{2} \approx 0.8647$ not only is a mathematically elegant result but has further implications in real business since it is fairly close to $1$. Namely, at any equilibrium,  the efficiency degradation in the first-price auction is small, no worse than 13.53\%, which might be acceptable given other merits of the first-price auction.

En route to the tight {\PoA}, we obtain many new important perspectives, characterizations, and properties of equilibria in the first-price auction. These might be of independent interest and find applications in the future study of other aspects of the first-price auction, e.g., the complexity of computing equilibria. Beyond the first-price auction, our overall approach is general enough and might help determine the Price of Anarchy in other auctions.

\subsection{Technical overview}
\label{sec:overview}

% This subsection sketches our high-level proof ideas. To ease the readability, we would elaborate on the tight {\PoA}. (The upper bound analysis can be slightly modified to accommodate the tight {\PoS}, accomplishing the whole proof.)
% Below, some descriptions are roughly accurate but not perfectly, and most technical details are deferred to \Cref{sec:structure,sec:preprocessing,sec:reduction,sec:UB,sec:LB}.
% Our approach is very different from the prior works \cite{ST13,HTW18}, which mainly adopt the smoothness techniques or the extension.

This subsection sketches our high-level proof ideas. To ease the readability, some descriptions are roughly accurate but not perfectly, and most technical details are deferred to \Cref{sec:structure,sec:preprocessing,sec:reduction,sec:UB,sec:LB}.
Our approach is very different from the prior works \cite{ST13,HTW18}, which mainly adopt the smoothness techniques or an extension.
More concretely, we employ a first principle approach that directly characterizes the {\em worst-case} instance and the {\em worst-case} Bayesian Nash equilibrium regarding the definition of the Price of Anarchy. To this end, we step by step narrow down the search space of the worst case by proving more and more necessary conditions it must satisfy.
Finally, we capture the exact worst case and thus derive the tight {\PoA} bound.

\subsection*{Changing the viewpoints (\Cref{sec:structure})}

% \yj{to continue}

Regarding the {\PoA} definition, we need to prove that
the auction social welfare $\FPA(\bV,\, \bs)$ is within a certain ratio of the optimal social welfare $\OPT(\bV)$, given any value distribution $\bV = \{V_{i}\}_{i \in [n]}$ and any equilibrium thereof $\bs = \{s_{i}\}_{i \in [n]}$.

Two difficulties emerge immediately: (i)~One value distribution $\bV \in \bbV$ generally has {\em multiple} or even {\em infinite} equilibria. (ii)~One equilibrium $\bs \in \bbBNE(\bV)$ generally has {\em no} analytic
% closed-form
solution; even an efficient algorithm for computing (approximate) equilibria from value distributions is unknown.

To the rescue, instead of the original value-strategy representation $(\bV,\, \bs) \in \bbV \times \bbBNE$ of equilibria,
we use another representation, the {\em bid distributions} $\bB(\bV,\, \bs)$ resulted from equilibria.
These two representations and their relationship are demonstrated in \Cref{fig:mappings}.
Given one equilibrium bid distribution $\bB(\bV,\, \bs) = \{B_{i}\}_{i \in [n]}$: (i)~One equilibrium $(\bV,\, \bs) \in \bbV \times \bbBNE$ is {\em uniquely} determined.
(ii)~The reconstruction of the $(\bV,\, \bs)$ essentially {\em has} an analytic
% a closed-form
solution, through the {\em bid-to-value mappings} $\varphi_{i}(b) \eqdef b + (\sum_{k \neq i} B'_{k}(b) / B_{k}(b))^{-1}$ that are almost the inverse functions of the strategies $\bs = \{s_{i}\}_{i \in [n]}$.\footnote{Although the mapping formulas $\varphi_{i}(b) = b + (\sum_{k \neq i} B'_{k}(b) / B_{k}(b))^{-1}$ are previously known, NO prior work takes ANY further step. Even this structure result is unknown: ``One (valid) bid distribution $\bB \in \Bvalid$ backward uniquely determines the underlying value-strategy tuple $(\bV,\, \bs) \in \bbV \times \bbBNE$''; let alone techniques towards the tight $\PoA = 1 - 1 / e^{2}$.}
These mitigate the above two difficulties.

In sum, there is a bijection between the space of equilibria $\bbV \times \bbBNE \ni (\bV,\, \bs)$ and the space of equilibrium bid distributions $\{\bB(\bV,\, \bs)\}$.
Further, the new representation $\{\bB(\bV,\, \bs)\}$ is mathematically equivalent but technically easier -- this is {\em one} infinite set instead of the Cartesian product $\bbV \times \bbBNE$ of  {\em two} infinite sets -- showing an avenue towards the tight {\PoA}. (Later in \Cref{sec:structure,sec:preprocessing,sec:reduction,sec:UB,sec:LB}, we will see more advantages.)

% Thus, we use bid distribution rather than value distribution to describe an instance and analyze the {\PoA} bound of the instance from that. It is clear from the figure that the bid distribution perspective has many advantages. 

\newlength{\FigMappingsHeight}
\setlength{\FigMappingsHeight}{3.1cm}

\begin{figure}[t]
    \centering
    \subfloat[Forward direction: $\bbV \times \bbBNE \mapsto \Bvalid$
    \label{fig:mappings:1}]{
    \resizebox{.38\textwidth}{!}{
    \begin{tikzpicture}[thick, smooth, scale = 1]
        \draw[fill = red!15] (-1.35, 0) arc(-180: 180: 1.35cm and \FigMappingsHeight);
        \draw[fill = green!15] (3.65, 0) arc(-180: 180: 1.35cm and \FigMappingsHeight) -- cycle;
        \node[below] at (0, -\FigMappingsHeight) {Space $\bbV$};
        \node[below] at (5, -\FigMappingsHeight) {Space $\Bvalid$};
        
        \draw[fill = black] (0, 0.7) circle (3pt);
        \node[left] at (0, 0.7) (V1) {$\bV$};
        \draw[fill = black] (5, 2.1) circle (3pt);
        \node[right] at (5, 2.1) (B1) {$\bB$};
        \draw[fill = black] (5, 0.7) circle (3pt);
        \node[right] at (5, 0.7) (B2) {$\bar{\bB}$};
        \draw[fill = black] (5, -0.7) circle (3pt);
        \node[right] at (5, -0.7) (B3) {$\hat{\bB}$};
        
        \draw[fill = black] (0, -2.1) circle (3pt);
        \node[left] at (0, -2.1) (V4) {$\tilde{\bV}$};
        \draw[fill = black] (5, -2.1) circle (3pt);
        \node[right] at (5, -2.1) (B4) {$\tilde{\bB}$};
        
        \draw[shorten <=0.2cm, shorten >=0.2cm, ->, >=triangle 45] (V1) to node[above, rotate = 15, yshift = -.1cm] {\small $\bs \in \bbBNE(\bV)$} (B1);
        \draw[shorten <=0.2cm, shorten >=0.2cm, ->, >=triangle 45] (V1) to node[above, yshift = -.1cm] {\small $\bar{\bs} \in \bbBNE(\bV)$} (B2);
        \draw[shorten <=0.2cm, shorten >=0.2cm, ->, >=triangle 45] (V1) to node[above, rotate = -13, yshift = -.1cm] {\small $\hat{\bs} \in \bbBNE(\bV)$} (B3);
        
        \draw[shorten <=0.2cm, shorten >=0.2cm, ->, >=triangle 45] (V4) to node[above, yshift = -.1cm] {\small $\tilde{\bs} \in \bbBNE(\bV)$} (B4);
    \end{tikzpicture}}}
    \hfill
    \subfloat[Inverse direction: $\Bvalid \mapsto \bbV \times \bbBNE$
    \label{fig:mappings:2}]{
    \resizebox{.38\textwidth}{!}{
    \begin{tikzpicture}[thick, smooth, scale = 1]
        \draw[fill = orange!15] (8.65, 0) arc(-180: 180: 1.35cm and \FigMappingsHeight);
        \draw[fill = green!15] (3.65, 0) arc(-180: 180: 1.35cm and \FigMappingsHeight) -- cycle;
        \node[below] at (10, -\FigMappingsHeight) {Space $\bbV \times \bbBNE$};
        \node[below] at (5, -\FigMappingsHeight) {Space $\Bvalid$};
        
        \draw[fill = black] (10, 2.1) circle (3pt);
        \node[right] at (10, 2.1) (V1) {};
        \node[below] at (10, 2.1) {$(\bV,\, \bs)$};
        \draw[fill = black] (5, 2.1) circle (3pt);
        \node[left] at (5, 2.1) (B1) {$\bB$};
        \draw[fill = black] (10, 0.7) circle (3pt);
        \node[right] at (10, 0.7) (V2) {};
        \node[below] at (10, 0.7) {$(\bV,\, \bar{\bs})$};
        \draw[fill = black] (5, 0.7) circle (3pt);
        \node[left] at (5, 0.7) (B2) {$\bar{\bB}$};
        \draw[fill = black] (10, -0.7) circle (3pt);
        \node[right] at (10, -0.7) (V3) {};
        \node[below] at (10, -0.7) {$(\bV,\, \hat{\bs})$};
        \draw[fill = black] (5, -0.7) circle (3pt);
        \node[left] at (5, -0.7) (B3) {$\hat{\bB}$};
        
        \draw[fill = black] (10, -2.1) circle (3pt);
        \node[right] at (10, -2.1) (V4) {};
        \node[below] at (10, -2.1) {$(\tilde{\bV},\, \tilde{\bs})$};
        \draw[fill = black] (5, -2.1) circle (3pt);
        \node[left] at (5, -2.1) (B4) {$\tilde{\bB}$};
        
        \draw[shorten <=0.2cm, shorten >=0.2cm, ->, >=triangle 45] (B1) to (V1);
        \draw[shorten <=0.2cm, shorten >=0.2cm, ->, >=triangle 45] (B2) to (V2);
        \draw[shorten <=0.2cm, shorten >=0.2cm, ->, >=triangle 45] (B3) to (V3);
        
        \draw[shorten <=0.2cm, shorten >=0.2cm, ->, >=triangle 45] (B4) to (V4);
    \end{tikzpicture}}}
    \caption{Diagram of the two representations of equilibrium, (i)~the value-strategy representation $(\bV,\, \bs) \in \bbV \times \bbBNE$ versus (ii)~the bid distribution representation $\bB = \bB(\bV,\, \bs) \in \Bvalid$.
    One value distribution $\bV$ has multiple or even infinite equilibria $\bs \in \bbBNE(\bV)$; each equilibrium induces one valid bid distribution $\bB = \bB(\bV,\, \bs) \in \Bvalid$. Backward, the $\bB$ determines the underlying tuple $(\bV,\, \bs) \in \bbV \times \bbBNE$ via the bid-to-value mappings $\bvarphi$. Hence, there is a bijection between the two spaces $\bbV \times \bbBNE$ and $\Bvalid$.
    \label{fig:mappings}}
\end{figure}

% \blue{Mappings between value distributions $\bV \in \bbV$ and bid distributions $\bB \in \mathbb{B}$. \\
% \Cref{fig:mappings:1}: A value distribution $\bV \in \bbV$ can be mapped to multiple (or even uncountably many) bid distributions $\bB \in \Bvalid \subsetneq \mathbb{B}$,
% each of which is determined by one exact {\BayesNashEquilibrium} $\bs \in \bbBNE(\bV,\, \mathrm{x}) \neq \emptyset$.
% This ``$\bV \mapsto \bB$'' mapping is {\em non-surjective}, since the image $\Bvalid$ (i.e., the space of the $\bB$'s that can be constructed from some value distribution $\bV \in \bbV$ and some {\BayesNashEquilibrium} $\bs \in \bbBNE(\bV)$) is a proper subset of the codomain $\mathbb{B}$ (i.e., the space of all $\bB$'s). \\
% \Cref{fig:mappings:2}: A bid distribution $\bB \in \mathbb{B}$ determines either a {\em well-defined} value distribution $\bV \in \bbV$ or a {\em fake} value distribution $\tilde{\bV} \notin \bbV$, through the bid-to-value mappings $\{\varphi_{i}\}_{i \in [n]}$ (\Cref{def:mapping}) and \Cref{cor:high_bid}. The membership of the result (either $\bV$ or $\tilde{\bV}$) in space $\bbV$ depends on the validity of the considered bid distribution $\bB \in \mathbb{B}$, namely whether $\bB \in \Bvalid$ or $\bB \in (\mathbb{B} \setminus \Bvalid)$. Moreover, this ``$(\bB \in \Bvalid)$-to-$(\bV \in \bbV)$'' mapping is {\em surjective} but {\em non-injective}.}

Nonetheless, there is no free lunch. The new representation has two drawbacks. (i)~Unlike that any value distribution $\bV \in \bbV$ must be feasible since the existence of an equilibrium $\bs \in \bbBNE(\bV) $ is promised \cite{L96},\footnote{This existence result requires suitable {\em tie-breaking} rules in the first-price auction; see \Cref{sec:structure} for more details.}
% \footnote{This should be combined with suitable tie-breaking rules; see \Cref{sec:structure} for details.}
a generic bid distribution $\bB$ not necessarily induces an equilibrium.
To remedy this issue, we introduce the notion of {\em validity}, which almost refers to {\em monotonicity} of the bid-to-value mappings $\bvarphi = \{\varphi_{i}\}_{i \in [n]}$.
% To capture this point, We introduce the notion of validity, which mainly says that the bid-to-value mappings $\bvarphi = \{\varphi_{i}(b)\}_{i \in [n]}$ must be monotonic.
In addition, (ii)~a mapping $\varphi_{i}$ is undefined at a singular point where the equilibrium bid distribution $B_{i}$ of that bidder $i \in [n]$ has a probability mass, which must be handled separately.
Luckily, such probability masses are possible only at the {\em left endpoint} of the distribution's support.
So, we further consider the {\em conditional value distribution} $P$ given one's bid being at the left endpoint.
% \red{left endpoint}.
% this probability mass.
Indeed, the concept of validity refers to a tuple $(\bB,\, P)$ of equilibrium bid distributions and a conditional value distribution.

More rigorously, our new representation considers all the {\em valid} tuples/instances $(\bB,\, P) \in \Bvalid$.
Here, the {\em valid} distributions $\bB$ can be replaced with the {\em increasing} bid-to-value mappings $\bvarphi$: One determines the other and vice versa (\Cref{lem:pseudo_distribution}).
More importantly, monotonicity of the mappings $\bvarphi$ gives a strong geometric intuition, making the third representation $(\bvarphi,\, P)$ of equilibrium more convenient for our later use.
In particular, in what follows, the figures of the mappings $\bvarphi = \{\varphi_{i}\}_{i \in [n]}$ \textbf{\em play the role of visual demonstrations}, i.e., \textbf{\em horizontal bid axes} and \textbf{\em vertical value axes}, where each mapping $\varphi_{i}$ denotes one bidder $i \in [n]$.

\subsection*{The worst-case instance (\Cref{sec:LB})}

Our overall approach is to reduce any given instance to the {\em worst-case} instance step by step. Thus, it is helpful to describe the worst case in advance, which explains why we design those reductions since our target is very specific.

\begin{wrapfigure}{r}{0.48\textwidth}
    \centering
    \includegraphics[width = 0.8\linewidth]
    {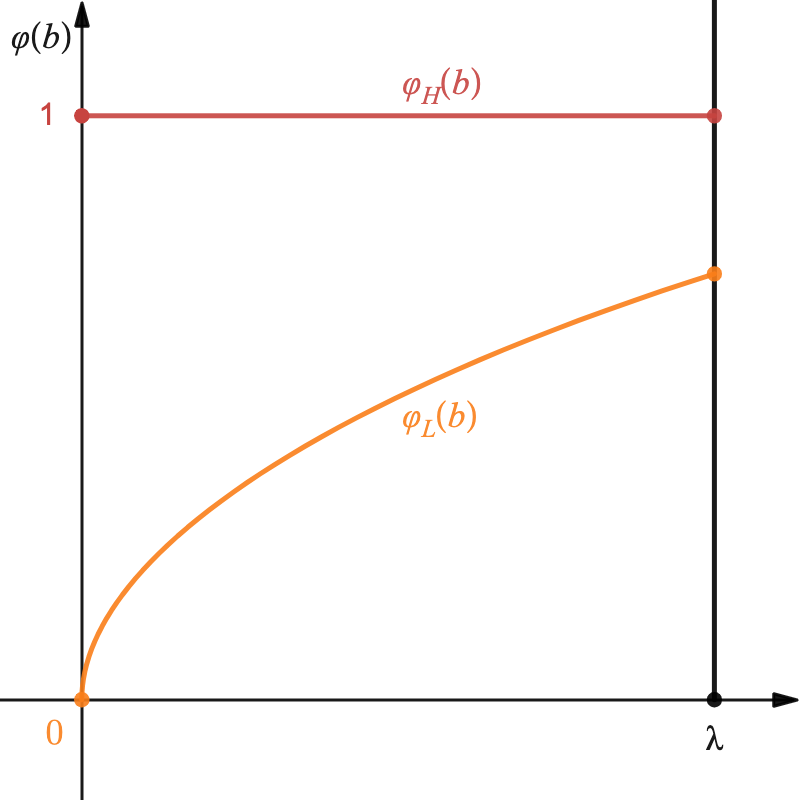}
    \caption{\label{fig:intro:worst_case}
    The worst case instance $H \otimes \{\sqrt[n]{L}\}^{\otimes n}$. Bidder $H$ has a fixed high value $\varphi_{H}(b) \equiv 1$ for $b \in [0,\, 1 - 4 / e^{2}]$ and each bidder $L$ has the (parametric) bid-to-value mapping $\varphi_{L}(1 - t^{2} \cdot e^{2 - 2t}) = 1 - t \cdot e^{2 - 2t}$ for $t \in [1,\, 2]$.}
\end{wrapfigure}

The worst case has $n + 1$ bidders, one bidder $H$ with a {\em fixed} high value $v_{H} \equiv 1$ and {\em sufficiently many} $n \gg 1$ identical low-value bidders $\{\sqrt[n]{L}\}^{\otimes n}$.
Among the low-value bidders, the highest value $v_{L} \sim V_{L}$ is distributed following the parametric equation $V_{L}(1 - t \cdot e^{2 - 2t}) = 4 / t^{2} \cdot e^{2t - 4}$ for $t \in [1,\, 2]$, over the value support $v_{L} \in \supp(V_{L}) = [0,\, 1 - 2 / e^{2}]$.
See \Cref{fig:intro:worst_case} for a visual aid, based on the bid-to-value mappings $\varphi_{H} \otimes \{\varphi_{L}\}^{\otimes n}$.

The equilibrium strategies $s_{H} \otimes \{s_{L}\}^{\otimes n}$ and the bid distributions $H \otimes \{\sqrt[n]{L}\}^{\otimes n}$ are less important for understanding our approach;\footnote{\label{footnote:worst_case}\cite{HHT14} ``just'' considered a reasonable instance $H(b) = \sqrt{b / \lambda}$ and $L(b) = \frac{1 - \lambda}{1 - b}$ for $\lambda = 0.57$, but neither gave argument/evidence for  ``worstness'' of the instance family $H \otimes \{\sqrt[n]{L}\}^{\otimes n}$, nor searched the worst case $(H^{*},\,\lambda^{*})$.

In contrast, our contributions are twofold: (primary) ``worstness'' of this family; (secondary) the nontrivial worst case $b = 1 - {4H^{*}} \cdot {e^{2 - 4\sqrt{H^{*}}}}$ and $L^{*}(b) = \frac{1 - \lambda^{*}}{1 - b}$ for $\lambda^{*} = 1 - 4 / e^{2}$. In this family, the {\PoA} is less sensitive to different instances, hence two numerically close bounds $0.8689$ versus $0.8647$.}
%\footnote{The CDF $H(b)$ is given by the implicit equation $b = 1 - 4H \cdot \exp(2 - 4\sqrt{H})$ for $0 \leq b \leq 1 - 4 / e^{2}$ and $1 / 4 \leq H \leq 1$, and the CDF $L(b)$ is given by $L(b) = \frac{4 / e^{2}}{1 - b}$ for $0 \leq b \leq 1 - 4 / e^{2}$.}
we will show in \Cref{sec:LB} how to deduce them and the tight {\PoA} $= 1 - 1 / e^{2}$.
% we defer a description and the proof of the tight {\PoA} $= 1 - 1 / e^{2}$ to \Cref{sec:LB}.
% (There, we will prove that this equilibrium gives the tight {\PoA} $= 1 - 1 / e^{2}$.)
The parametric equation $V_{L}$ also is less important, being included just for completeness.

% A detail description of the instance is given in 
% We do not describe an equilibrium or the resulting bid distributions $\bB$ here, since these details are less important to understand our proof overview. (See  for more details,)

In contrast, the point is the underlying structure: The bidder $H$ always contributes the optimal social welfare $\equiv 1$.
The low-value bidders $\{\sqrt[n]{L}\}^{\otimes n}$ {\em individually} have negligible effects (the winning probabilities etc)
but {\em collectively} make the auction game less efficient.

We are inspired by the instance due to Hartline, Hoy and Taggart \cite{HHT14}, which has the same $H \otimes \{\sqrt[n]{L}\}^{\otimes n}$ structure. Our construction differs from theirs in the concrete distributions, hence a slightly worse {\PoA} of $1 - 1 / e^{2} \approx 0.8647$ in comparison with their ratio of $\approx 0.8689$ (\Cref{footnote:worst_case}).

% We use the bid-to-value mapping figures to demonstrate our worst case instance in  \Cref{fig:intro:worst_case}. 

%\blue{We often use the figures of these mappings to depict an instance.}

\subsection*{The pseudo instances (\Cref{subsec:pseudo}) and \blackref{alg:collapse} (\Cref{subsec:collapse})}

Rigorously, the worst case $H \otimes \{\sqrt[n]{L}\}^{\otimes n}$ for $n \gg 1$ described above is not a specific instance, but a sequence of instances with the {\em limit} $n \to +\infty$ being the worst case.
This incurs some notational inconvenience.
% This incurs some inconvenience especially in notation.
When such issues occur in real/complex analysis, the standard method is to add the infinite point(s) to the space of the real/complex numbers, making the space more complete.
%axis $\mathbb{R}_{\geq 0} = [0,\, +\infty) \cup \{+\infty\}$ or the two-dimensional space $\R_{\geq 0}^{2} = [0,\, +\infty]^{2}$ etc, making a closed space ``more closed''.
% As in mathematics, we often add $\infty$ in to real number $\mathbb{R}$ or the infinite points to the 2D Euclidean plane to make the space more complete.
In the same spirit, we introduce the notion of {\em pseudo bidders} and {\em pseudo instances} (\Cref{def:pseudo}), thus including the above limit to our search space.

\begin{figure}[t]
    \centering
    \subfloat[\label{fig:intro:collapse_input}
    The input {\em non collapsed} mappings $\bvarphi$]{
    \includegraphics[width = .4\textwidth]
    {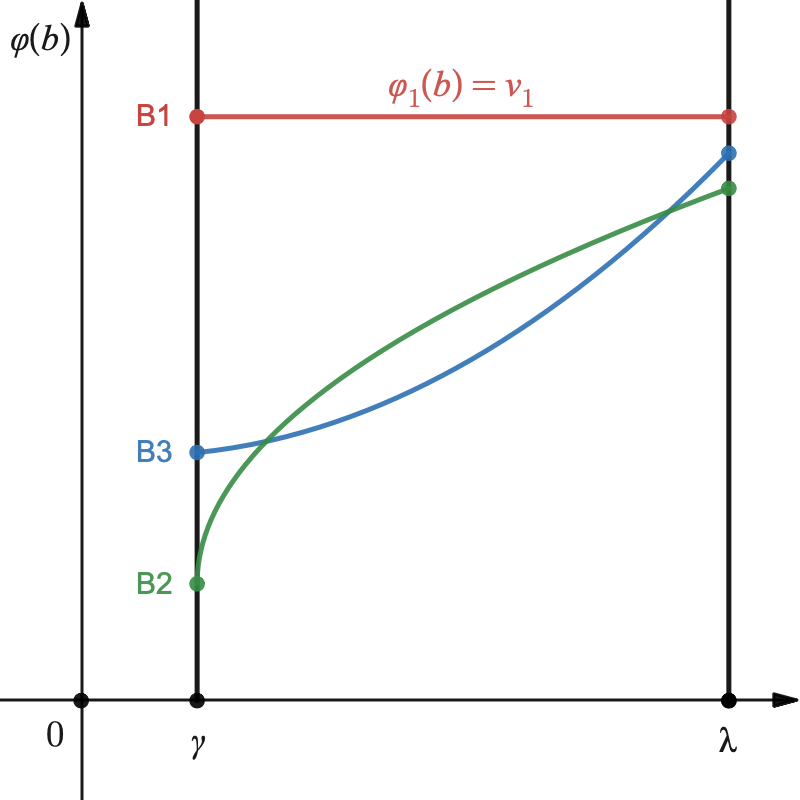}}
    \hfill
    \subfloat[\label{fig:intro:collapse_output}
    The output {\em collapsed} mappings $\tilde{\bvarphi}$]{
    \includegraphics[width = .4\textwidth]
    {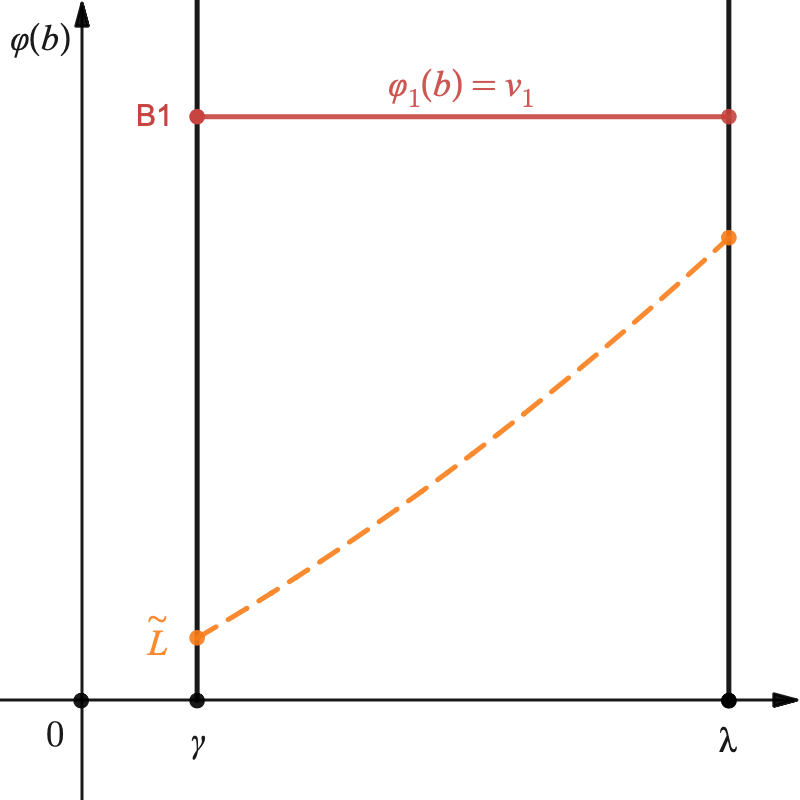}}
    \caption{The {\collapse} reduction.
    When a specific bidder $i \in [n]$ has a fixed high value $\varphi_{i}(b) \equiv v_{i} \geq \max\{ \varphi_{k}(b): k \neq i \}$ over the bid support $b \in [\gamma,\, \lambda]$, the {\collapse} reduction transforms this instance $\bB = B_{i} \otimes \bB_{-i}$ to a two-bidder pseudo instance $B_{i} \otimes \tilde{L}$ that yields a worse-or-equal {\PoA}.
    \label{fig:intro:twin}}
\end{figure}

In our extended language, the worst case has a succinct format: one high-value bidder $H$ and one pseudo bidder $L$,
%$L = \lim_{n \to +\infty} (\sqrt[n]{L})^{n}$, 
namely a two-bidder pseudo instance (\Cref{fig:intro:worst_case}).
Beyond notational brevity, the notion of pseudo bidders/instances also simplifies our proof in several places.

% , which is the combination of all these infinite many infinitesimal bidders. The detailed definition can be found in (\Cref{subsec:pseudo}).

% Then the above worst case can be described in our extended language as a two-bidder Pseudo instance: one high-value bidder and one Pseudo bidder, which is the combination of all these infinite many infinitesimal bidders. The detailed definition can be found in (\Cref{subsec:pseudo}). The convenience is not only for the description of the worst-case instance, it is also beneficial in many other places of our proof.

This extension does not hurt our proof, given that pseudo instances are only considered in the {\em lower-bound} analysis. Precisely, we show that even the worst case in the extended search space has the {\PoA} $\geq 1 - 1 / e^{2}$, implying a lower bound of $1 - 1 / e^{2}$ for the original problem.
In addition, the {\em upper-bound} analysis leverages the above instance sequence $H \otimes \{\sqrt[n]{L}\}^{\otimes n}$.
Precisely, we show that when the $n \gg 1$ is sufficiently large, the {\PoA} can be arbitrarily close to $1 - 1 / e^{2}$. As a combination, the tight {\PoA} in \Cref{thm:main} gets established.

% We also use the above worse case instance sequence, each of which is a real instance for fixed $n$, to prove that there are real instances whose PoA are arbitrarily close to $1 - 1 / e^{2}$. As a combination, we get the tight {\PoA}.

% We remark that this extension of the search space does not hurt our proof. Pseudo instances are only consider in the lower-bound analysis, namely the PoA of pseudo instance is at least $1 - 1 / e^{2}$. Thus, the PoA of any real instance (as a subset of pseudo instance space) is at least $1 - 1 / e^{2}$. We also use the above worse case instance sequence, each of which is a real instance for fixed $n$, to prove that there are real instances whose PoA are arbitrarily close to $1 - 1 / e^{2}$. To combine these two, we get a tight PoA.

Another related thing is the \blackref{alg:collapse} reduction. As \Cref{fig:intro:collapse_input,fig:intro:collapse_output} show,
when a specific bidder $i \in [n]$ has a fixed high value $\varphi_{i}(b) \equiv v_{i} \geq \max\{ \varphi_{k}(b): k \neq i \}$ over the bid support $b \in [\gamma,\, \lambda]$,
% when there is one bidder $i \in [n]$ whose value is fixed $\varphi_{i}(b) \equiv v_{i}$ and is always higher than the other bidders,
the \blackref{alg:collapse} reduction
can replace all the other bidders $\{B_{k}\}_{k \neq i}$ by one pseudo bidder $\tilde{L}$, resulting in a two-bidder pseudo instance $B_{i} \otimes \tilde{L}$ that has a worse-or-equal {\PoA}.
(I.e, whether before or after the \blackref{alg:collapse} reduction, the optimal social welfare is the bidder $i$'s fixed high value $\equiv v_{i}$. Moreover, it turns out that the auction social welfare can only decrease.)
Such a two-bidder pseudo instance $B_{i} \otimes \tilde{L}$ {\em will} be easy to handle since it has the same shape as the worst case $H \otimes L$.
% (\Cref{fig:intro:twin}) and.

In sum, the remaining task is to transform every valid instance $\in \Bvalid$ into a specific instance that the {\collapse} reduction can work on, i.e., an instance that has one {\em fixed-high-value} bidder.

% in the reduction is to convert the instance to a new one with a single high-value bidder with a constant valuation.

% ``split'' the other bidders $k \neq i$ into ``infinitesimal'' bidders and then ``merge'' all of them into one pseudo bidder to get a pseudo instance with a worse-or-equal {\PoA}.
% After the reduction, it is of the shape of the worst case instance (See \Cref{fig:intro:collapse} and \Cref{fig:intro:worst_case}). As a result, the key remaining task in the reduction is to convert the instance to a new one with a single high-value bidder with a constant valuation.

\subsection*{\blackref{alg:layer} (\Cref{subsec:layer}): Hierarchizing the bidders}

\begin{figure}[t]
    \centering
    \subfloat[\label{fig:intro:layer:old}
    The input {\em increasing} mappings $\bvarphi$]{
    \includegraphics[width = .4\textwidth]
    {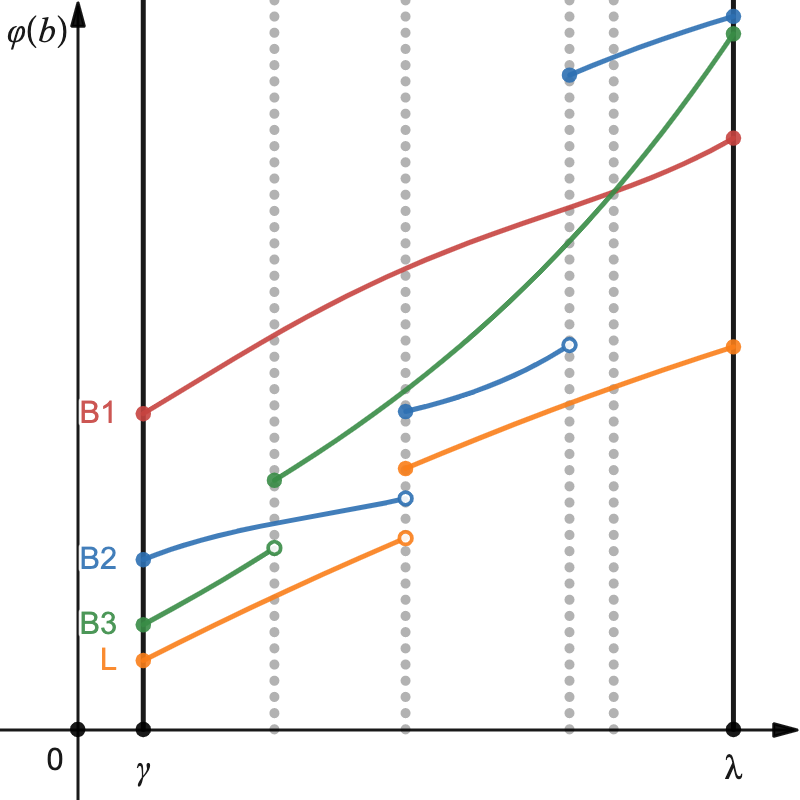}}
    \hfill
    \subfloat[\label{fig:intro:layer:new}
    The output {\em layered} mappings $\tilde{\bvarphi}$]{
    \includegraphics[width = .4\textwidth]
    {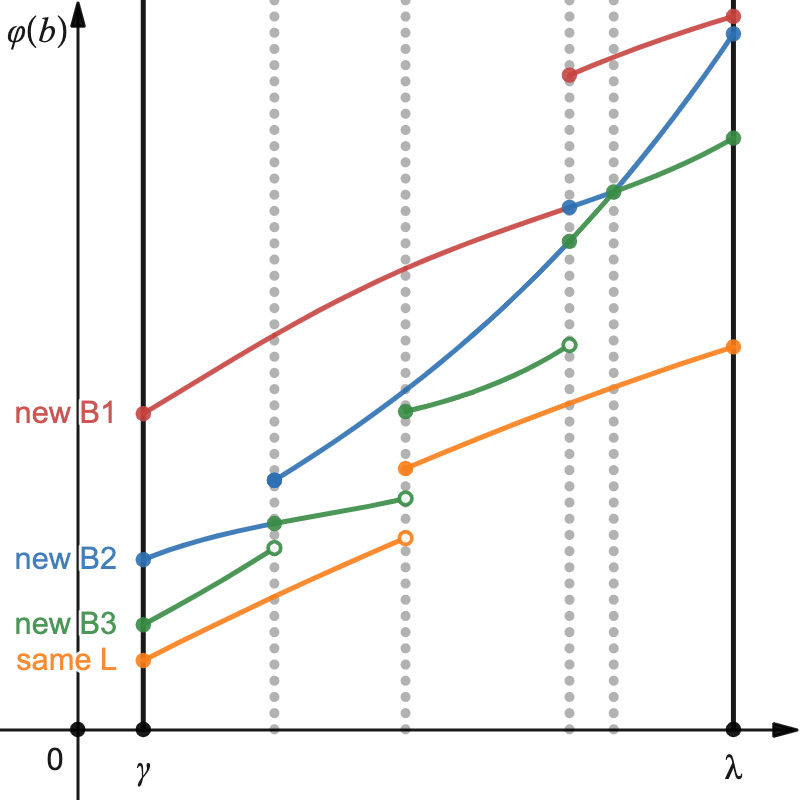}}
    \caption{The {\layer} reduction. Each color corresponds to a bidder. After the reduction, the bidders are ordered. For example, the red bidder is the highest bidder. \label{fig:intro:layer}}
\end{figure}

To transform a given instance into the worst case, we shall identify one special high-value bidder $H \in [n]$ and ``derandomize'' his/her value, i.e., making the bid-to-value mapping a constant function $\varphi_{H}(b) \equiv v_{H}$. (If so, upon scaling, we can normalize this fixed value $v_{H} = 1$.)
But it is even unclear which bidder $i \in [n]$ should be the candidate: The highest-value bidder $\argmax\{ \varphi_{i}(b): i \in [n]\}$ at different bids $b \in [\gamma,\, \lambda]$ can be different.
More concretely, consider the example in \Cref{fig:intro:layer:old},
the highest-value bidder is the red bidder $B_{1}$ initially for small bids,
but changes to the blue bidder $B_{2}$ later for large bids.

We can segment the bid-to-value mappings into pieces and rearrange them, in a sense of \Cref{fig:intro:layer:new}, making them {\em ordered point-wise} $\tilde{\varphi}_{1}(b) \geq \dots \geq \tilde{\varphi}_{n}(b)$.
Then, the candidate high-value bidder clearly is the index-$1$ bidder $\tilde{B}_{1}$ (i.e., the red one in \Cref{fig:intro:layer:new}).
We call this reduction \blackref{alg:layer}.

Under the \blackref{alg:layer} reduction, the auction/optimal social welfares {\FPA} and {\OPT} each remain the same, so the {\PoA} is invariant.
% (\Cref{lem:layer}).
% We prove that the expected optimal social welfare and expected welfare obtained from first price auction both remain unchanged before and after this layer reduction.
% Therefore, we can apply this reduction without affect the PoA of the instance.
This proof is not that technically involved,
% about \blackref{alg:layer},
once we formalize the reduction.
However, without switching our viewpoint to the bid-to-value mappings $\bvarphi = \{\varphi_{i}\}_{i \in [n]}$, even describing this reduction seems difficult.

% as you can see such reduction is not easy to identify if one only take the value distribution viewpoint, or even difficult to describe this reduction.

In sum, we narrow down the search space to the ``layered'' instances and the remaining task to address the candidate high-value bidder $B_{1}$, particularly derandomizing his/her value $\varphi_{1}(b) \equiv v_{1}$.
% \red{and make it a constant}.

\subsection*{\blackref{alg:discretize} and \blackref{alg:translate} (\Cref{subsec:discretize,subsec:translate}): Two simplifications}

\begin{figure}[t]
    \centering
    \subfloat[\label{fig:infro:discretize:old}
    The input {\em increasing} mappings $\bvarphi$]{
    \includegraphics[width = .425\textwidth]
    {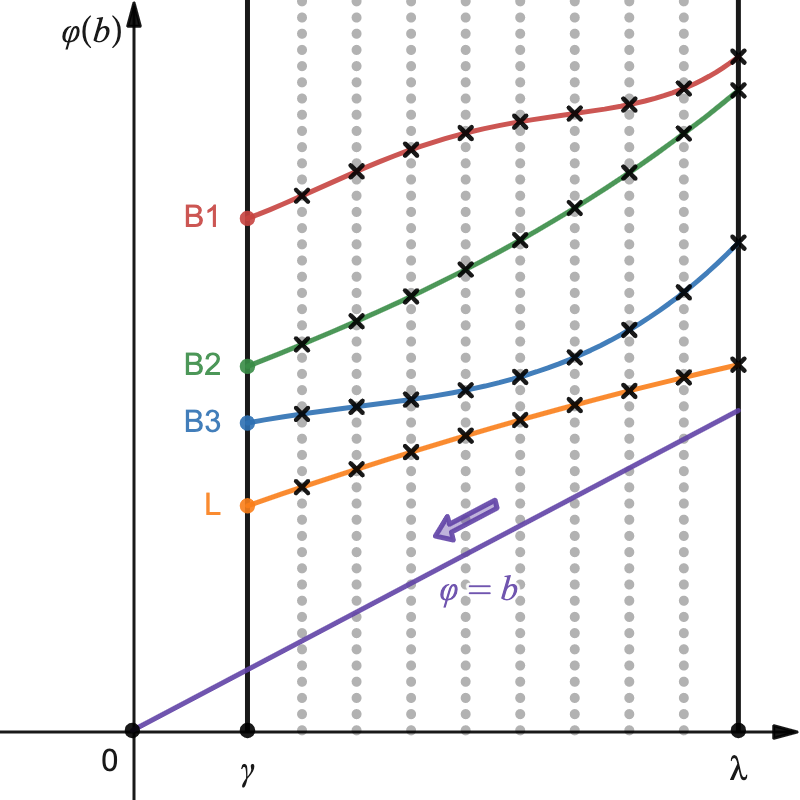}}
    \hfill
    \subfloat[\label{fig:infro:discretize:new}
    The output {\em discretized}/{\em translated} mappings $\tilde{\bvarphi}$]{
    \includegraphics[width = .425\textwidth]
    {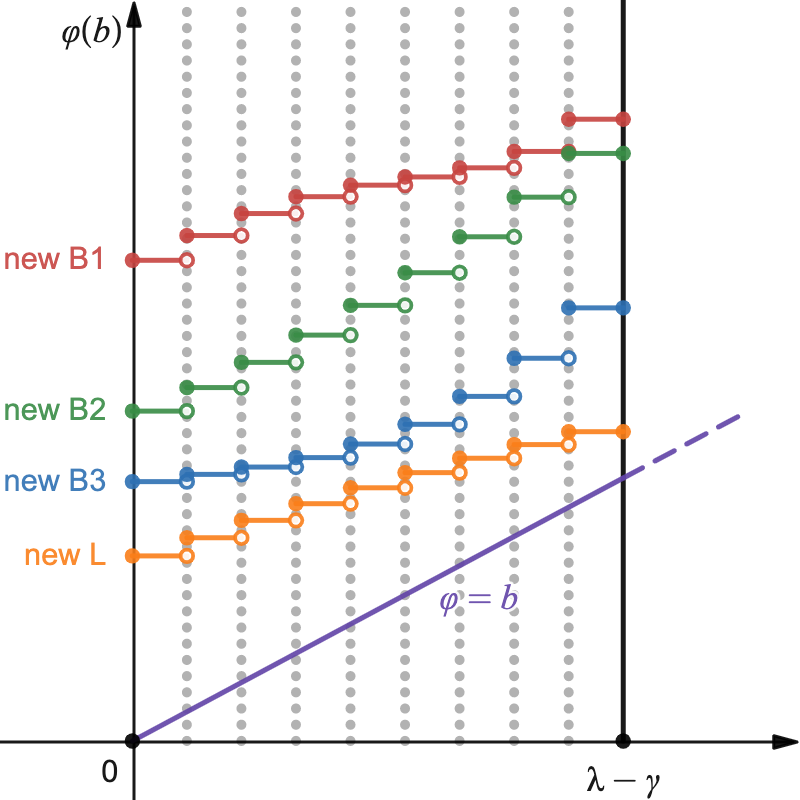}}
    \caption{The {\discretize} and {\translate} reductions
    \label{fig:intro:discretize}}
\end{figure}

Our main reduction is an {\em iterative} algorithm; thus it is convenient to work with {\em bounded discrete} instances, for which the main reduction will terminate in {\em finite} iterations.
Such an idea appears in many prior works \cite{CGL14,AHNPY19,JLTX20,JLQTX19a,JJLZ21}, and the proof plan is twofold: (i)~For any bounded discrete instance, the {\PoA} is at least $1 - 1 / e^{2}$. (ii)~For any valid instance, there exists a bounded discrete instance approximating the {\PoA}, up to an {\em arbitrary but fixed} error $\epsilon > 0$.
As a combination, the {\PoA} for any valid instance is at least $1 - 1 / e^{2}$.

However, our discretization scheme is subtle. First, we cannot discretize the valid bid distributions $\bB$ given that they must be continuous everywhere except the left endpoint of the bid support.
Second, we cannot discretize the value distributions and the strategies $(\bV,\, \bs)$.
Although bounded discrete value distributions $\tilde{\bV}$ that well approximate the given ones $\bV$ are trivial, it is unclear how to obtain desirable strategies $\tilde{\bs} \in \bbBNE(\tilde{\bV})$:
No algorithm for computing equilibria from value distributions is available.
Also, with modified value distributions, even the existence of $\tilde{\bs} \in \bbBNE(\tilde{\bV})$ that is ``close enough'' to the given equilibrium $\bs \in \bbBNE(\bV)$ is doubtful.

% Although it is easy to find some discretized value distributions $\tilde{\bV}$ that well approximate the given ones $\bV$, it is unclear how to accordingly get a modified equilibrium $\tilde{\bs} \in \bbBNE(\tilde{\bV})$:
% As mentioned, no algorithm for computing equilibria from value distributions is known.
% Indeed, even the existence of such an equilibrium $\tilde{\bs} \in \bbBNE(\tilde{\bV})$ that is ``close enough'' to the given one $\bs \in \bbBNE(\bV)$ is doubtful.
% It is fine to discretize the value distributions $\tilde{\bV}$ so that it is a good approximation to $\bV$.  But it is unclear how to get the corresponding BNE $\tilde{\bs}$ after the discretization since it is not easy to compute BNE from value distributions. It is even not clear if there exists a $\tilde{\bs}$ in the neighborhood of $\bs$ such that $\tilde{\bs}$ is a BNE of $\tilde{\bV}$.

% \color{blue}

We circumvent these issues by discretizing the bid-to-value mappings $\bvarphi$ instead (\Cref{fig:intro:discretize}), i.e., approximating them
% the given mappings $\bvarphi$
through {\em piecewise constant} functions $\tilde{\bvarphi}$.\footnote{In real analysis, for various purposes, piecewise constant functions (a.k.a.\ simple functions) are widely used as approximations to general functions.}
The benefits are threefold:
(i)~The valid bid distributions $\tilde{\bB}$, the value distributions $\tilde{\bV}$, and the equilibrium strategies $\tilde{\bs} \in \bbBNE(\tilde{\bV})$ can be reconstructed from the new mappings $\tilde{\bvarphi}$, through analytic expressions.
(ii)~The piecewise constancy of the new mappings $\tilde{\bvarphi}$ makes the value distribution $\tilde{\bV}$ {\em bounded} and {\em discrete}.
% , although different from those by directly discretizing the $\bV$ in naive ways.
(iii)~More importantly, we are able to obtain a {\em sequence} of piecewise constant mappings $\{\tilde{\bvarphi}\}$ that {\em pointwise converge} to the given ones $\bvarphi$.
The corresponding $\{ \tilde{\bB} \}$, $\{ \tilde{\bV}\}$, and $\{ \tilde{\bs} \}$ can be arbitrarily close to the given $\bB$,\, $\bV$, and $\bs$, approximating the {\PoA} up to any fixed error $\epsilon > 0$.
This exactly implements the second part of our proof plan.

Afterward, the \blackref{alg:translate} reduction further vanishes the lowest bid, through shifting both the value space and the bid space by the same distance of $-\gamma$ (\cref{fig:intro:discretize}).
As a consequence, both the auction/optimal social welfares decrease by an amount of $\gamma$, giving a worse-or-equal {\PoA}.

In sum, we can focus on {\em piecewise constant} and {\em increasing} mappings $\bvarphi$ over a {\em bounded} support $b \in [0,\, \lambda]$, together with the conditional value $P$ given the nil bid $b = 0$ (i.e., the left endpoint).

\subsection*{\blackref{alg:polarize} and \blackref{alg:slice} (\Cref{subsec:polarize,subsec:slice}): Dealing with the probability masses}

Now we address the probability masses and the conditional value $P$. (This is the only information that cannot be reconstructed from the bid-to-value mappings $\bvarphi$ and must be handled separately.)
Our main observation is (\Cref{lem:monopolist}) that {\em at most one} bidder can have a nontrivial conditional value $P$, who will be called the monopolist $H$.
(After the \blackref{alg:layer} reduction, this monopolist $H$ will be the {\em high-value} bidder; cf.\ \Cref{fig:intro:layer}.)
That is why we can use a single conditional value $P$, instead of $n$ values akin to the distributions $\bV$ or $\bB$. Moreover, following the monotonicity, this value $P$ is supported on $[0,\varphi_{H}(0)]$.

The \blackref{alg:polarize} reduction derandomizes the conditional value by moving all the probabilities to either $P \equiv 0$ or $P \equiv \varphi_{H}(0)$.
Namely, this is a {\em win-win} reduction, in a sense that at least one of the two possibilities $P \equiv 0$ and $P \equiv \varphi_{H}(0)$ induces a worse-or-equal {\PoA} than the given instance. Later, such win-win reductions also appear in several other places.

% We first deal with the probability mass and the conditional value distribution $P$. This part of information is not included in the bid-to-value mappings and  need to be handled separately. Our most important observation here is that there is at most one bidder with nontrivial conditional distribution $P$ (\Cref{lem:monopolist}), and we call this bidder the monopolist $H$. (After the \blackref{alg:layer}, the monopolist $H$ is also the high-value bidder.) That is the reason why we can use a single distribution $P$ rather than $n$ distributions like $\bV$ and $\bB$.  Due to the monotonicity, we know that this distribution $P$ is supported on $[0,\varphi_{H}(0)]$. In the \blackref{alg:polarize} reduction, we prove that we can move all the probability mass to either $0$ or  $\varphi_{H}(0)$. This is a win-win reduction, namely we construct two candidate instances from the given instance and prove that at least one candidate yields a worse-or-equal {\PoA}. We use such win-win type reductions in several places.

Further, the \blackref{alg:slice} reduction totally determines the conditional value $P \equiv \varphi_{H}(0)$ by eliminating the other possibility.
This reduction modifies both the conditional value $P$ and the mappings $\bvarphi$, thus far more complicated than the \blackref{alg:polarize} reduction, which only modifies the $P$.\footnote{For very technical reasons, we {\em should not} unify \blackref{alg:polarize} and \blackref{alg:slice} into a single reduction (cf.\ \Cref{sec:preprocessing,sec:reduction}).}

In conclusion, we can drop the notation $P$ hereafter: Just the {\em piecewise constant} and {\em increasing} mappings $\bvarphi$ are enough to describe the entire instance and the reductions on it.

% Later, we further use the \blackref{alg:slice} reduction to show that we can assume that this distribution is deterministic on $\varphi_{H}(0)$.  \blackref{alg:slice} reduction is significantly more complicated  since it also modifies the bid-to-value mappings while  \blackref{alg:polarize} only changes the distribution $P$ itself. After that, we can drop the distribution $P$ and only use $\bB$ (or $\bvarphi$ the bid-to-value mappings) to describe an instance and the reductions on it.

\subsection*{\blackref{alg:halve} and \blackref{alg:AD} (\Cref{subsec:halve,subsec:AD}): The main reductions}

\afterpage{
\begin{figure}
\centering
\begin{minipage}{.49\textwidth}
    \centering
    \subfloat[\label{fig:intro:halve:input}
    {The input with a {\bf pseudo jump}}]{\includegraphics[width = .79\linewidth]{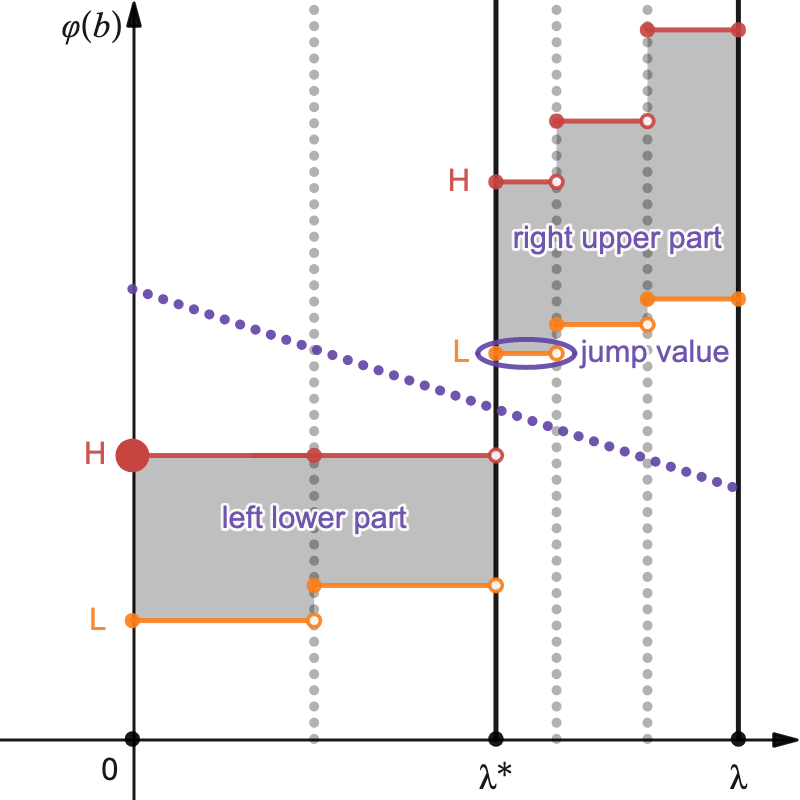}} \\
    \subfloat[\label{fig:intro:halve:left}
    The {\em left} candidate output]{\includegraphics[width = .79\linewidth]{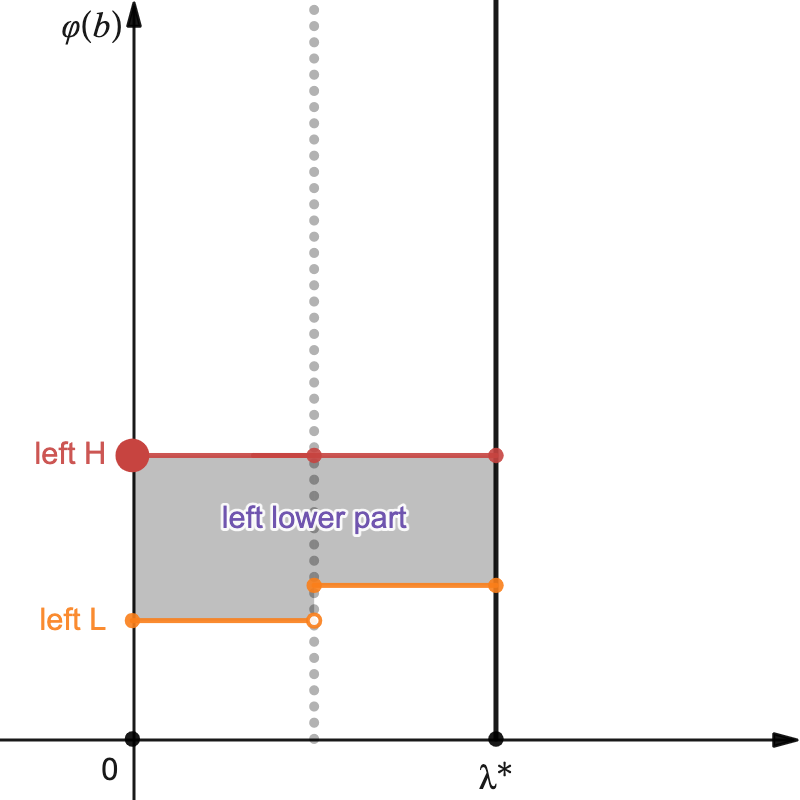}} \\
    \subfloat[\label{fig:intro:halve:right}
    The {\em right} candidate output]{\includegraphics[width = .79\linewidth]{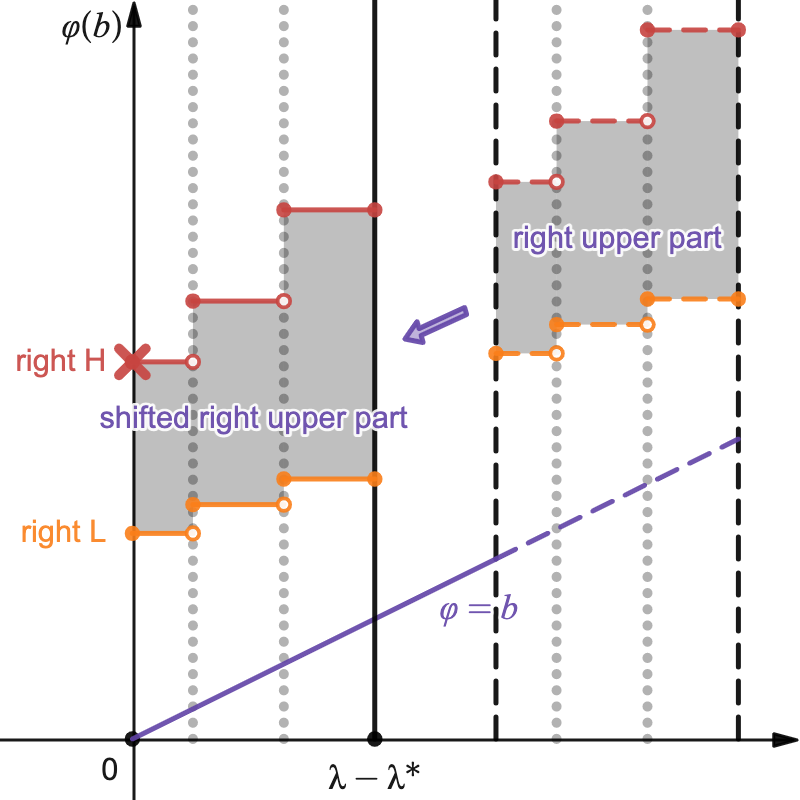}}
    \captionof{figure}{The {\halve} reduction \\
    (A baby version of \Cref{fig:halve})}
    \label{fig:intro:halve}
\end{minipage}
\hfill
\begin{minipage}{.49\textwidth}
    \centering
    \subfloat[\label{fig:intro:AD:input}
    {The input with a {\bf real jump}}]{\includegraphics[width = .79\linewidth]{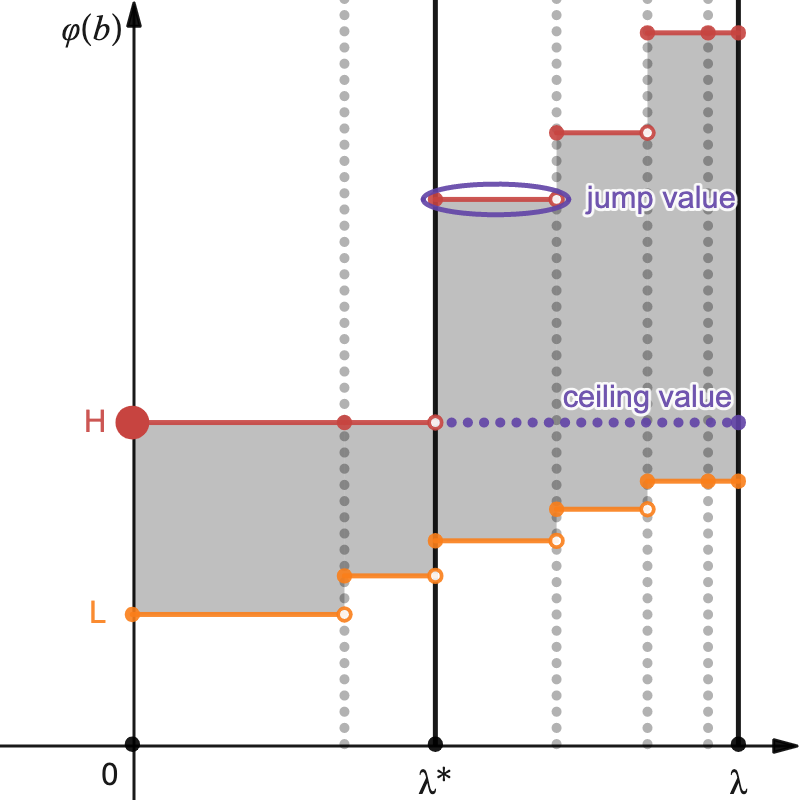}} \\
    \subfloat[\label{fig:intro:AD:ascend}
    The {\em ascended} candidate output]{\includegraphics[width = .79\linewidth]{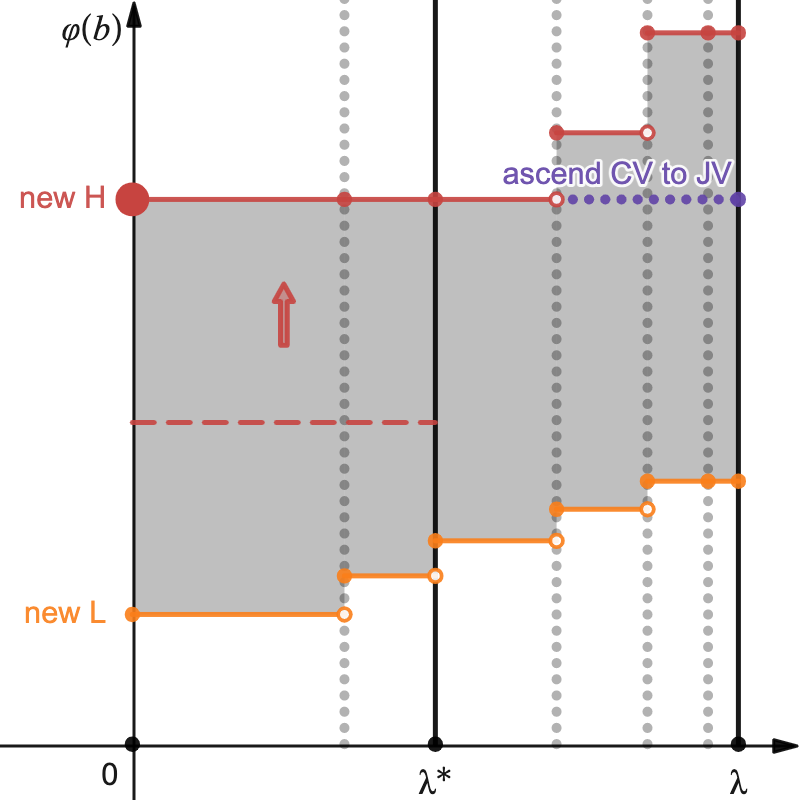}} \\
    \subfloat[\label{fig:intro:AD:descend}
    The {\em descended} candidate output]{\includegraphics[width = .79\linewidth]{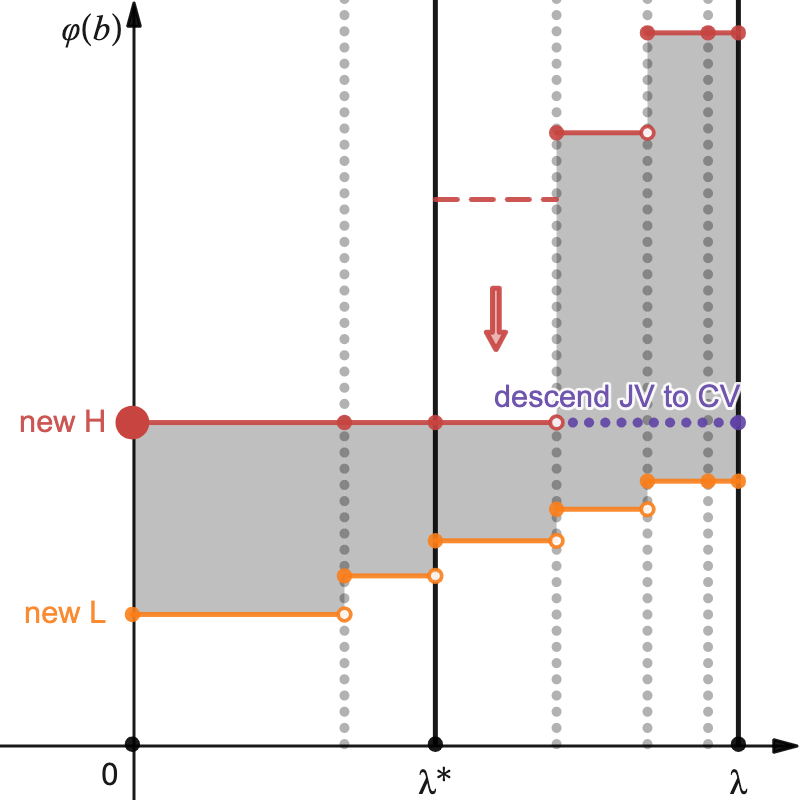}}
    \captionof{figure}{The {\AD} reduction \\
    (A baby version of \Cref{fig:AD})}
    \label{fig:intro:AD}
\end{minipage}
\end{figure}
\clearpage}

Our main reduction aims to reduce the number of pieces of the high-value bidder $H$'s mapping $\varphi_{H}$.
(After the \blackref{alg:discretize} reduction, this mapping has {\em finite} many pieces. In particular, the one in the worst-case instance has one {\em single} piece (\Cref{fig:intro:worst_case}), namely a constant mapping.)
To this end, we look at the first {\em jump point} of this mapping $\varphi_{H}$.
Furthermore, we distinguish two types of jumps, {\bf pseudo jumps} versus {\bf real jumps}, and adopt two different reductions respectively.

% (\Cref{fig:intro:halve,fig:intro:AD}).
% The main part of the reduction is to reduce the number of pieces for the high-value bidder's bid-to-value mapping. 
% Since it is a bounded piecewise constant function, it only has finite many pieces. In the target worst-case instance, the high-value bidder has a constant value, namely its bid-to-value mapping only has one single piece(see \Cref{fig:intro:worst_case}). 
% To this end, we look at the first jump point of the high-value bidder's bid-to-value mapping. 
%Thanks to the discretization reduction, we only have finite many of them.
% We distinguish two types of jumps: pseudo jump and real jump as depicted in \Cref{fig:intro:halve,fig:intro:AD}. We use different reduction schemes for different types.

For a \blackref{halve_jump} (\Cref{fig:intro:halve:input}),
the after-jump parts of all mappings $\bvarphi$ are {\em universally} higher than the before-jump parts, so the entire instance naturally divides into the {\em left-lower} part versus the {\em right-upper} part (\Cref{fig:intro:halve:left,fig:intro:halve:right}).
We prove that at least one part induces a worse-or-equal {\PoA} than the given instance.\footnote{In particular, for the right-upper part, we reapply the \blackref{alg:translate} reduction to make the bid space start from $0$.}
This is another win-win reduction and will be called \blackref{alg:halve}.

% For a pseudo jump (\Cref{fig:intro:halve}), the entire mappings after the jump is above that before the jump. Thus, the instance is naturally divided into the left-lower part versus the right-upper part. We prove that at least one of them  has a worse-or-equal {\PoA} than the original one. We call this \blackref{alg:halve} reduction, which is another win-win type reduction. For the right-upper part, we apply a \blackref{alg:translate} reduction so that its bid space also starts from $0$.

For a \blackref{AD_jump} (\Cref{fig:intro:AD:input}), the after-jump parts (of some mappings) still interlace with the before-jump parts.
We would modify the first two pieces of the high-value mapping $\varphi_{H}$ to reduce the number of pieces by one.
I.e., we either {\em ascend} the first piece to the value of the second piece (\Cref{fig:intro:AD:ascend}) or {\em descend} the second piece to the value of the first piece (\Cref{fig:intro:AD:descend}).
We show that at least one modification gives a worse-or-equal {\PoA} than the given instance. This is again a win-win reduction and will be called \blackref{alg:AD}.

% For a real jump (\Cref{fig:intro:AD}), the mappings still interlace with each other before and after the jump. In this case, we simply make the first two pieces of the high-value bidder's bid-to-value mapping equal to each other so that we reduce the number of pieces by one.
% We introduce two reductions: Ascend (to increase the first piece to the value of the second one, see \Cref{fig:intro:AD:ascend}) and Descend (to decrease the second piece to the value of the first one, see \Cref{fig:intro:AD:descend}). We also prove that at least one of these two instances has a worse-or-equal {\PoA} than the original one.

To conclude, whether a \blackref{halve_jump} or a \blackref{AD_jump}, we obtain a new instance that (i)~has a worse-or-equal {\PoA} and (ii)~the high-value mapping $\varphi_{H}$ has one fewer piece.
Clearly, we can {\em iterate} the entire process until a constant mapping $\varphi_{H}(b) \equiv v_{H}$, i.e., one {\em single} piece.
Then, the \blackref{alg:collapse} reduction can replace all the other bidders $\{B_{k}\}_{k \neq H}$ by one pseudo bidder $L$, hence a two-bidder pseudo instance $H \otimes L$ that has the same shape as the worst case (\Cref{fig:intro:twin}), as desired.

% In both cases, we transform a given instance to a new one with a worse-or-equal {\PoA} and the number of pieces is decreased by at least one. Then we can apply the reductions to the new instance and further reduce the number of pieces. Finally, there is only one piece and the high-value bidder's value becomes a constant. Then further apply the collapse reduction and all the other bidders become  a single pseudo bidder. We reach the shape of worst case instance as in \Cref{fig:intro:worst_case}.

We stress again that the above description is greatly simplified and thus just roughly accurate -- The high-level idea is delivered, but most technical details are hidden.

\subsection*{Solving the functional optimization (\Cref{sec:UB})}

After the above reductions, we get a nice-looking instance: (i)~one bidder $H$ with a fixed high value $\varphi_{H}(b) \equiv v_{H}$ and (ii)~infinite  $n \to +\infty$ identical low-value bidders $\{\sqrt[n]{L}\}^{\otimes n}$, or one pseudo bidder $L$ in our extended language.
Towards the worst case, it remains to determine the pseudo mapping $\varphi_{L}(b)$, or essentially, the high-value bidder $H$'s bid distribution.

This is accomplished by standard tools from the calculus of variations.
Via the Euler-Lagrange equation \cite{GS20}, we formulate the {\em worst-case} bid distribution $H$ as the solution to an ordinary differential equation (ODE).
% Then further use standard calculus to determine the remaining parameters to get the optimal value.
Luckily, this ODE admits a closed-form solution, and ``coincidentally'' the tight bound $1 - 1 / e^{2}$ is pretty nice-looking.
% \red{If one only looks for a numerical solution, nothing special there.}

% \red{It is surprising that as a centuries-old theory, Calculus of Variations has not been widely used in algorithmic game theory. We hope to see more applications in future, for which the techniques developed in this work (\Cref{sec:UB}) may be helpful.}

\subsection*{Comparison with previous techniques}

On the Price of Anarchy in auctions, the canonical approach is the {\em smoothness framework} proposed by Roughgarden \cite{R15} and developed by Syrgkanis and Tardos \cite{ST13}.
The past decade has seen an abundance of its applications and extensions (cf.\ the survey \cite{RST17}).

However, the smoothness framework has an intrinsic restriction -- It focuses on the structure of an auction game BUT ignores the {\em independence}, i.e., both the independence of value distributions $\bV = \{V_{i}\}_{i \in [n]}$ and the independence of strategies $\bs = \{s_{i}\}_{i \in [n]}$.
Thus, although giving an arsenal of tools for proving {\em lower bounds}, this approach seems hard to access {\em tight bounds} for some auctions. In particular for the first-price auction, the smoothness-based bound of $1 - 1 / e \approx 0.6321$ obtained by Syrgkanis and Tardos \cite{ST13} is tight when {\em correlated} distributions are allowed. Hence, this bound cannot be improved without taking the independence into account. 

The work by Hoy, Taggart, and Wang \cite{HTW18} is one of the very few follow-ups that transcend the smoothness framework.\footnote{To the best of our knowledge, \cite{HTW18} is the only work for Bayesian {\em single-item} auctions that transcends the smoothness framework. Other such works for {\em multi-item} auctions are thoroughly discussed in \cite[Section~8.3]{RST17}.}
For the first-price auction, they show an improved lower bound of $\approx 0.7430$, but this bound is still not tight.
Technically, the improvement stems from an extension of the smoothness framework that leverages the independence.

In sum, towards {\em tight bounds} for the first-price auction and the beyond, the primary consideration is:
What is the consequence of the {\em independence} of values $\bv = (v_{i})_{i \in [n]} \sim \bV$ and strategies $\bs$?
Our answer is, the bids $\bs(\bv) = (s_{i}(v_{i}))_{i \in [n]}$ are also independent and follow a product distribution $\bB = \{B_{i}\}_{i \in [n]}$. This is the starting point of our overall approach.

In the literature, ALL prior works (within/beyond the smoothness framework) follow the proof paradigm: ``searching over value distributions $\bV \in \bbV$ and equilibrium strategies $\bs \in \bbBNE(\bV)$ plus deviation strategies thereof $\bs' \neq \bs$'' \cite[Section~8.3]{RST17}. Instead, we present the FIRST different proof paradigm: ``searching over (valid) bid distributions $\bB = \bB(\bV,\, \bs) \in \Bvalid$ that are resulted from certain value-strategy tuples $(\bV,\, \bs) \in \bbV \times \bbBNE$''.

In our main proof, we step by step narrow down the search space of the worst-case instance, by showing stronger and stronger conditions it must satisfy, and completely capture the worst-case instance eventually.
En route, an abundance of tools is developed, which may find applications in the future.
% which may be of independent interest.
It is conceivable that our approach can be adapted to other auctions.
So, we hope that this would develop into a complement to the smoothness framework, especially in the setting of independent distributions.
%We hope that this would develop into an equally powerful one as the smoothness framework \cite{R15,ST13}, each of which complements the other depending on whether or not the independence is imposed.
% ary approach in addition

% Our proof is a first principle approach, namely directly characterizing the worst-case instance and equilibrium by the definition of Price of Anarchy. To this end, 

% \subsection*{Discussion}
% overbid
% \Cref{sec:beyond}

\newpage

\section{Structural Results}
\label{sec:structure}

This section shows a number of structural results on {\FirstPriceAuction} and {\BayesNashEquilibrium}. Some results or their analogs may already appear in the literature \cite{SZ90,L96,MR00a,MR00b,JSSZ02,MR03,HKMN11,CH13}. However, we still formalize and reprove them for completeness.
Provided these structural results, we can reformulate the {\PriceofAnarchy} problem in terms of bid distributions rather than value distributions.

To sell a single indivisible item to one of $n \geq 1$ bidders, each bidder $i \in [n]$ is asked to submit a non-negative bid $b_{i} \geq 0$. A generic auction $\calA = (\alloc,\, \pays)$ is given by its allocation rule $\alloc(\bb): \R^{n} \mapsto [n]$ and payment rule $\pays(\bb) = (\pay_{i}(\bb))_{i \in [n]}: \R^{n} \to \R^{n}$, both of which can be {\em randomized}.
Concretely, the item is allocated to {\em the} bidder $\alloc(\bb) \in [n]$, and $\pay_{i}(\bb)$'s are the payments of bidders $i \in [n]$.

Rigorously, {\FirstPriceAuction} is not a single auction but a family of auctions, each of which has a specific allocation rule that obeys the \blackref{pro:allocation}/\blackref{pro:payment} principles.

\begin{definition}[{\FirstPriceAuctions}]
\label{def:fpa}
An $n$-bidder single-item auction $\calA = (\alloc,\, \pays)$ is called a {\FirstPriceAuction} when its allocation rule $\alloc(\bb)$ and payment rule $\pays(\bb) = (\pay_{i}(\bb))_{i \in [n]}$ satisfy that:
\begin{itemize}
    \item \term[{\bf first price allocation}]{pro:allocation}{\bf:}
    Let $X(\bb) \eqdef \argmax \big\{b_{i}: i \in [n]\big\}$ be the set of first-order bidders. When such bidders are not unique $|X(\bb)| \geq 2$, allocate the item to one of them $\alloc(\bb) \in X(\bb)$, through a (randomized) {\em tie-breaking} rule given by the auction $\calA$ itself on this bid profile $\bb$. When such a bidder is unique $|X(\bb)| = 1$, allocate the item to him/her $\alloc(\bb) \equiv X(\bb)$.

    \item \term[{\bf first price payment}]{pro:payment}{\bf:} The allocated bidder $\alloc(\bb)$ pays his/her own bid, while non-allocated bidders $[n] \setminus \{\alloc(\bb)\}$ each pay nothing. Formally, $\pay_{i}(\bb) = b_{i} \cdot \indicator\big(i = \alloc(\bb)\big)$ for $i \in [n]$.
\end{itemize}
Without ambiguity, the allocation rule $\alloc(\bb)$ itself will be called the considered {\FirstPriceAuction}, since it controls the payment rule $\pays(\bb)$.
Denote by $\bbFPA \ni \alloc$ the space of all {\FirstPriceAuctions}.
\end{definition}

The only undetermined part of a {\FirstPriceAuction} is the tie-breaking rule, which is subtle and plays an important role in our analysis.
To clarify it, let us rigorously define the interim allocation/ utility formulas (\Cref{def:interim_utility}) and {\BayesNashEquilibrium} (\Cref{def:bne_formal}).

% \violet{Thus we define it explicitly here.} Based on \Cref{def:fpa}, let us rigorously define the interim allocation/utility formulas (\Cref{def:interim_utility}) and {\BayesNashEquilibrium} (\Cref{def:bne_formal}).

\begin{definition}[Interim allocations/utilities]
\label{def:interim_utility}
Given a {\FirstPriceAuction} $\alloc \in \bbFPA$, value distributions $\bV = \{V_{i}\}_{i \in [n]}$,  and a strategy profile $\bs = \{s_{i}\}_{i \in [n]}$:
\begin{itemize}
    \item Each interim allocation formula $\alloc_{i}(b) \eqdef \Pr_{\alloc,\, \bs_{-i}(\bv_{-i})} \big[ i = \alloc\big(b,\, \bs_{-i}(\bv_{-i})\big) \big]$ for $b \geq 0$, over the randomness of the others' bids $\bs_{-i}(\bv_{-i})$ and the considered {\FirstPriceAuction} $\alloc \in \bbFPA$ itself.

    \item Each interim utility formula $u_{i}(v,\, b) \eqdef (v - b) \cdot \alloc_{i}(b)$ for $v \in \supp(V_{i})$ and $b \geq 0$.
\end{itemize}
\end{definition}

\begin{definition}[{\BayesNashEquilibria}]
\label{def:bne_formal}
Following \Cref{def:interim_utility}, the strategy profile $\bs = \{s_{i}\}_{i \in [n]}$ reaches a {\BayesNashEquilibrium} for this distribution-auction tuple, namely $\bs \in \bbBNE(\bV,\, \alloc)$, when: For each bidder $i \in [n]$, any possible value $v \in \supp(V_{i})$, and any deviation bid $b \geq 0$,
\begin{align*}
    \Ex_{s_{i}}\big[\, u_{i}(v,\, s_{i}(v)) \,\big] ~\geq~ u_{i}(v,\, b).
\end{align*}
\end{definition}

\Cref{def:bne_formal} means for any value $v \in \supp(V_{i})$, nearly all the equilibrium bids $b \in \supp(s_{i}(v))$ EACH maximize the interim utility formula $u_{i}(v,\, b)$, except a {\em zero-measure} set.
By excluding these zero-measure sets from the bid supports $\supp(s_{i}(v))$, for $i\in [n]$, the modified strategy profile still satisfies the definition of {\BayesNashEquilibrium}. Hereafter, without loss of generality, we assume that EVERY equilibrium bid $b \in \supp(s_{i}(v))$ maximizes the formula $u_{i}(v,\, b)$.

The following existence result can be concluded from \cite{L96}.

\begin{theorem}[{{\BayesNashEquilibrium}~\cite{L96}}]
\label{thm:exist_bne}
Given any value distribution $\bV = \{V_{i}\}_{i \in [n]}$, there exists at least one {\FirstPriceAuction} $\alloc \in \bbFPA$ that admits an equilibrium $\bbBNE(\bV,\, \alloc) \neq \emptyset$.
\end{theorem}

\begin{comment}

\yj{organization of this section.}

\blue{This existence result is the basics of our later discussions

As mentioned in \red{(overview)}, ???. In particular, we will leverage  to prove our universal approximation result.

En route, we gradually characterize the given exact equilibrium $\bs \in \bbBNE(\bV,\, \alloc)$ and conclude with a {\em bid-based} equivalence condition for {\BayesNashEquilibrium} (\Cref{thm:bne}), rather than the original {\em value-/utility-based} statement (\Cref{def:bne_formal}).
Using this equivalent condition, it will be trivial how to derive the universal ``$\epsilon$-approximate'' strategy $\bs^{*} = \{s_{i}^{*}\}_{i \in [n]}$ by perturbing the $\alloc$-exclusive strategy $\bs = \{s_{i}\}_{i \in [n]}$.}

First of all, \Cref{lem:optimal_utility} restates \Cref{def:bne_formal}, but we include it for completeness.

\begin{lemma}[Utility optimality]
\label{lem:optimal_utility}
For $i \in [n]$ and any possible value $v \in \supp(V_{i})$, the equilibrium bid $s_{i}(v)$ optimizes the interim utility formula $u_{i}(v,\, s_{i}(v)) = \sup\big\{ u_{i}(v,\, b): b \geq 0\big\}$, almost surely.
\end{lemma}

\begin{proof}
By the condition (\Cref{def:bne_formal}) for making $\bs \in \bbBNE(\bV,\, \alloc) \neq \emptyset$ an exact equilibrium.
\end{proof}

\end{comment}

\subsection{The forward direction: {\BayesNashEquilibria}}
\label{subsec:bne}

This subsection presents a bunch of lemmas that characterize an equilibrium $\bs \in \bbBNE(\bV,\, \alloc) \neq \emptyset$.
To this end, let us introduce some helpful notations, which will be adopted throughout this paper.
%\begin{mdframed}
%\begin{minipage}{\textwidth}
\begin{itemize}%[leftmargin = 1.55em]
    \item $\bB = \{B_{i}\}_{i \in [n]}$ denotes the equilibrium bid distributions $\bs(\bv) = (s_{i}(v_{i}))_{i \in [n]} \sim \bB$.
    % 's product distribution.
    % , for $\bv = (v_{i})_{i \in [n]} \sim \bV$.
    % According to \Cref{lem:bid_independence}, this IS a product distribution $\bB = \{B_{i}\}_{i \in [n]}$.

    % \item Denote by $\bB = \{B_{i}\}_{i \in [n]}$ the equilibrium bid profile $\bs(\bv)$'s distribution, for $\bv = (v_{i})_{i \in [n]} \sim \bV$.
    % According to \Cref{lem:bid_independence}, this IS a product distribution $\bB = \{B_{i}\}_{i \in [n]}$.

    \item $\calB(b) = \prod_{i \in [n]} B_{i}(b)$ denotes the first-order bid distribution $\max(\bs(\bv)) \sim \calB$.

    \item $\calB_{-i}(b) = \prod_{k \in [n] \setminus \{i\}} B_{k}(b)$ denotes the competing bid distribution of each bidder $i \in [n]$.

    \item $\gamma \eqdef \inf(\supp(\calB))$ and $\lambda \eqdef \sup(\supp(\calB))$ denote the ``infimum''/``supremum'' first-order bids, respectively.
    % Clearly we always have $0 \leq \gamma \leq \max(\bs(\bv)) \leq \lambda \leq +\infty$.
    Without ambiguity, we call $v,\, b < \gamma$ the low values/bids, $v,\, b = \gamma$ the boundary values/bids, and $v,\, b > \gamma$ the normal values/bids. As their names suggest:
    (i)~low bids $b < \gamma$ each gives a zero winning probability, so these bids are less important;
    (ii)~normal bids $b > \gamma$ are the most common ones and we will show that they behave nicely; and
    (iii)~boundary bids $b = \gamma$ are tricky and we will deal with them separately.
    
    \item $\calU_{i}(v) \eqdef \max \big\{u_{i}(v,\, b): b \geq 0\big\}$ denotes the optimal utility formula of each bidder $i \in [n]$. This is well defined on the value support $v \in \supp(V_{i})$, because (\Cref{def:bne_formal}) every equilibrium bid $s_{i}(v)$ does optimize the interim utility formula $u_{i}(v,\, s_{i}(v)) = \calU_{i}(v)$.
\end{itemize}
%\end{minipage}
%\end{mdframed}

% Without loss of generality, we would assume $n \geq 2$. (The other case $n = 1$ is trivial, where the unique bidder always bids zero in any {\FirstPriceAuction} $\alloc^{*} \in \bbFPA$.)

Each bidder's interim allocation $\alloc_{i}(b)$ relies on his/her competing bid distribution $\calB_{-i}(b)$, plus the tie-breaking rule of the considered auction $\alloc \in \bbFPA$.
When $\calB_{-i}(b)$ is a continuous distribution,
they are identical $\alloc_{i}(b) \equiv \calB_{-i}(b)$,
because the probability of being ONE of the first-order bidders $\Prx_{\bb} [X(\bb) \ni i]$ is exactly the probability of being the ONLY one $\Prx_{\bb} [X(\bb) = \{i\}]$.
% since it is the probability that he is the first price bidder.
In general, we can at least obtain (\Cref{lem:allocation}) the following properties for an interim allocation $\alloc_{i}(b)$.

\begin{lemma}[Interim allocations]
\label{lem:allocation}
For each interim allocation formula $\alloc_{i}(b)$:
\begin{enumerate}[font = {\em\bfseries}]
    \item\label{lem:allocation:1}
    $\alloc_{i}(b)$ is weakly increasing for $b \geq 0$.

    \item\label{lem:allocation:2}
    It is point-wise lower/upper bounded as $\sup \big\{ \calB_{-i}(t): t < b\big\} \leq \alloc_{i}(b) \leq \calB_{-i}(b)$ for $b \geq 0$.

    \item\label{lem:allocation:3}
    $\alloc_{i}(b) = 0$ for any low bid $b < \gamma$, AND $\alloc_{i}(b) > 0$ for any normal bid $b > \gamma$.
    % , where $\gamma = \inf(\supp(\calB))$.
\end{enumerate}
\end{lemma}

\begin{proof}
{\bf \Cref{lem:allocation:1,lem:allocation:2}} directly follows from the \blackref{pro:allocation} principle (\Cref{def:fpa}). Especially for {\bf \Cref{lem:allocation:2}}, suppose that (with a nonzero probability) bidder $B_{i}$ ties with someone else $[n] \setminus \{i\}$ for this first-order bid $s_{i}(v_{i}) = \max(\bs_{-i}(\bv_{-i})) = b$, then the auction always favors this bidder $B_{i}$ when ``$\alloc_{i}(b) = \calB_{-i}(b)$'', or always disfavor him/her when ``$\alloc_{i}(b) = \sup \big\{ \calB_{-i}(t): t < b\big\}$''.
{\bf \Cref{lem:allocation:3}} holds because $\gamma = \inf(\supp(\calB))$ is the infimum first-order bid $\max(\bs(\bv)) \sim \calB$.
\end{proof}

\Cref{lem:dichotomy} shows that both of the value space $v \in \supp(V_{i})$ and the bid space $b \geq 0$ of each bidder $i \in [n]$, are almost separated by the infimum first-order bid $\gamma = \inf(\supp(\calB))$.

\begin{lemma}[Bidding dichotomy]
\label{lem:dichotomy}
For each bidder $i \in [n]$:
\begin{flushleft}
\begin{enumerate}[font = {\em\bfseries}]
    \item\label{lem:dichotomy:1}
    For a low/boundary value $v \in \supp_{\leq \gamma}(V_{i})$, the optimal utility is zero $\calU_{i}(v) = 0$ and, almost surely over random strategy $s_{i}$, the equilibrium bid is low/boundary $s_{i}(v) \leq \gamma$.

    \item\label{lem:dichotomy:2}
    For a normal value $v \in \supp_{> \gamma}(V_{i})$, the optimal utility is nonzero $\calU_{i}(v) > 0$ and, almost surely over random strategy $s_{i}$, the equilibrium bid is boundary/normal $\gamma \leq s_{i}(v) < v$.
\end{enumerate}
\end{flushleft}
\end{lemma}

\begin{proof}
Recall \Cref{def:interim_utility} for the interim utility formula $u_{i}(v,\, b) = (v - b) \cdot \alloc_{i}(b)$.

\vspace{.1in}
\noindent
{\bf \Cref{lem:dichotomy:1}.}
A low/boundary value $v \in \supp_{\leq \gamma}(V_{i})$ induces EITHER a low (under)bid $b < v \leq \gamma$ and a zero interim allocation/utility $\alloc_{i}(b) = u_{i}(v,\, b) = 0$, OR a truthful/over bid $b \geq v$ and a negative interim utility $u_{i}(v,\, b) \leq 0$. Especially, a normal bid $b > \gamma \geq v$ induces a nonzero interim allocation $\alloc_{i}(b) > 0$ and a strictly negative interim utility $u_{i}(v,\, b) < 0$. Therefore, the optimal utility is zero $\calU_{i}(v) = 0$ and the equilibrium bid must be low/boundary $s_{i}(v) \leq \gamma$.
{\bf \Cref{lem:dichotomy:1}} follows then.

\vspace{.1in}
\noindent
{\bf \Cref{lem:dichotomy:2}.}
A normal value $v \in \supp_{> \gamma}(V_{i})$ is able to gain a positive utility $u_{i}(v,\, b) > 0$; for example, underbidding $b^{*} = (\gamma + v) / 2 \in (\gamma,\, v)$ gives a nonzero allocation/utility $\alloc_{i}(b^{*}),\, u_{i}(v,\, b^{*}) > 0$.
So, the equilibrium bid $s_{i}(v)$ neither can be truthful/over $b \geq v$ (hence a negative utility $u_{i}(v,\, b) \leq 0$), nor can be low $b < \gamma$ (hence a zero allocation/utility $\alloc_{i}(b) = u_{i}(v,\, b) = 0$). That is, the optimal utility is nonzero $\calU_{i}(v) > 0$ and the equilibrium (under)bid must be boundary/normal $\gamma \leq s_{i}(v) < v$.
{\bf \Cref{lem:dichotomy:2}} follows then.
\end{proof}

\Cref{lem:bid_monotonicity} shows, in the normal value regime $v \in \supp_{> \gamma}(V_{i})$, the optimal utility formula $\calU_{i}(v)$ and the equilibrium strategy $s_{i}(v)$ are monotonic (cf.\ \cite[Lemma~3.9]{CH13} for a similar result).

\begin{lemma}[Bidding monotonicity]
\label{lem:bid_monotonicity}
For any two normal values $v,\, w \in \supp_{> \gamma}(V_{i})$ that $v > w$ of a bidder $i \in [n]$:
\begin{enumerate}[font = {\em\bfseries}]
    \item\label{lem:bid_monotonicity:1}
    The two optimal utilities are strictly monotonic $\calU_{i}(v) > \calU_{i}(w) > 0$.

    \item\label{lem:bid_monotonicity:2}
    The two random bids are weakly monotonic $s_{i}(v) \geq s_{i}(w) \geq \gamma$, almost surely over the random strategy $s_{i}$.
\end{enumerate}
\end{lemma}

\begin{proof}
\Cref{lem:dichotomy} already shows that the utilities are nonzero $\calU_{i}(v),\, \calU_{i}(w) > 0$ and the equilibrium bids are boundary/normal $s_{i}(v),\, s_{i}(w) \geq \gamma$. It remains to verify the utility/bidding monotonicity.

\vspace{.1in}
\noindent
{\bf \Cref{lem:bid_monotonicity:1}.}
The utility monotonicity $\calU_{i}(v) > \calU_{i}(w)$. We deduce that
\begin{align*}
    \calU_{i}(v)
    ~\geq~ u_{i}(v,\, s_{i}(w))
    ~=~ \frac{v - s_{i}(w)}{w - s_{i}(w)} \cdot u_{i}(w,\, s_{i}(w))
    ~=~ \frac{v - s_{i}(w)}{w - s_{i}(w)} \cdot \calU_{i}(w)
    ~>~ \calU_{i}(w),
\end{align*}
The first step: Under the value $v$, the optimal utility $\calU_{i}(v)$ is at least the interim utility $u_{i}(v,\, s_{i}(w))$ resulted from bidding $s_{i}(w)$, namely the equilibrium bid at the (lower) value $w$. \\
The second step: Apply the interim utility formula to $u_{i}(v,\, s_{i}(w))$ and $u_{i}(w,\, s_{i}(w))$. \\
The third step: The random bid $s_{i}(w)$ optimizes the interim utility formula $u_{i}(w,\, s_{i}(w)) = \calU_{i}(w)$ at the value $w$. \\
The fourth step: (\Cref{lem:dichotomy}) $v > w > s_{i}(w)$ and $\calU_{i}(w) > 0$.
{\bf \Cref{lem:bid_monotonicity:1}} follows then.

\vspace{.1in}
\noindent
{\bf \Cref{lem:bid_monotonicity:2}.}
The bidding monotonicity $s_{i}(v) \geq s_{i}(w)$. This result also appears in \cite[Lemma~3.9]{CH13}; we would recap their proof for completeness.
Assume to the opposite that $s_{i}(v) < s_{i}(w)$.
Following the interim utility formula (\Cref{def:interim_utility}),
\begin{align*}
    u_{i}(v,\, s_{i}(v)) & ~=~ u_{i}(w,\, s_{i}(v)) ~+~ (v - w) \cdot \alloc_{i}(s_{i}(v)), \phantom{\big.} \\
    u_{i}(v,\, s_{i}(w)) & ~=~ u_{i}(w,\, s_{i}(w)) ~+~ (v - w) \cdot \alloc_{i}(s_{i}(w)). \phantom{\big.}
\end{align*}
We have \term[{\bf (i)}]{lem:bid_monotonicity:2:i}~$u_{i}(v,\, s_{i}(v)) \geq u_{i}(v,\, s_{i}(w))$, since the equilibrium bid $s_{i}(v)$ optimizes the interim utility formula $u_{i}(v,\, b)$ for this value $v$. Similarly, we have \term[{\bf (ii)}]{lem:bid_monotonicity:2:ii}~$u_{i}(w,\, s_{i}(w)) \geq u_{i}(w,\, s_{i}(v))$.
Also, under the assumption $s_{i}(v) < s_{i}(w)$, we know from \Cref{lem:allocation:1} of \Cref{lem:allocation} that \term[{\bf (iii)}]{lem:bid_monotonicity:2:iii}~$\alloc_{i}(s_{i}(v)) \leq \alloc_{i}(s_{i}(w))$.
Given the above two equations (note that $v > w$), each of \blackref{lem:bid_monotonicity:2:i}, \blackref{lem:bid_monotonicity:2:ii}, and \blackref{lem:bid_monotonicity:2:iii}
must achieve the equality. However, this means\footnote{In particular, the equality of \blackref{lem:bid_monotonicity:2:i} gives $(v - s_{i}(v)) \cdot \alloc(s_{i}(v)) = (v - s_{i}(w)) \cdot \alloc(s_{i}(w))$; notice that $s_{i}(v) < s_{i}(w)$. And the equality of \blackref{lem:bid_monotonicity:2:iii} gives $\alloc(s_{i}(v)) = \alloc(s_{i}(w))$. For these reasons, we must have $\alloc(s_{i}(v)) = \alloc(s_{i}(w)) = 0$.} all interim allocations/utilities are zero $\alloc_{i}(s_{i}(v)) = \alloc_{i}(s_{i}(w)) = u_{i}(v,\, s_{i}(v)) = u_{i}(w,\, s_{i}(w)) = 0$, which contradicts {\bf \Cref{lem:bid_monotonicity:1}}.
Rejecting our assumption gives {\bf \Cref{lem:bid_monotonicity:2}}.
\end{proof}

\begin{remark}[Quantiles]
\label{quantiles}
Due to \Cref{lem:bid_monotonicity}, the normal value space $v_{i} > \gamma$ and normal/boundary bid space $s_{i}(v_{i}) \geq \gamma$ each identify the other in terms of quantiles. From this perspective, the stochastic process of the auction game runs as follows: Draw a uniform random quantile $q_{i} \sim U[0,1]$, and then realize the value/bid $V_i^{-1}(q_{i})$ and $B_i^{-1}(q_{i})$ accordingly. The two stochastic processes are equivalent conditioned on the realized value being normal $V_i^{-1}(q_{i}) > \gamma$.
\end{remark}

\Cref{lem:bid_distribution} shows that, on the whole interval $b \in [\gamma = \inf(\supp(\calB)),\, \lambda = \sup(\supp(\calB))]$, all of the equilibrium/competing/first-order bid distributions $B_{i}(b)$, $\calB_{-i}(b)$, and $\calB(b)$ are well structured.
The earlier works \cite{MR00a,MR00b,MR03} derive similar results (under additional assumptions).

\begin{lemma}[Bid distributions]
\label{lem:bid_distribution}
Each of the following holds:
%\begin{flushleft}
\begin{enumerate}[font = {\em\bfseries}]
    \item\label{lem:bid_distribution:monotonicity}
    The competing/first-order bid distributions $\calB_{-i}(b)$ for $i \in [n]$ and $\calB(b)$ each have probability densities almost everywhere on $b \in [\gamma,\, \lambda]$, therefore being strictly increasing CDF's on the CLOSED interval $b \in [\gamma,\, \lambda]$.
    
    \item\label{lem:bid_distribution:continuity}
    The equilibrium/competing/first-order bid distributions $B_{i}(b)$ for $i \in [n]$, $\calB_{-i}(b)$ for $i \in [n]$ and $\calB(b)$, each have no probability mass on $b \in (\gamma,\, \lambda]$, excluding the boundary $\gamma = \inf(\supp(\calB))$, therefore being continuous CDF's on the CLOSED interval $b \in [\gamma,\, \lambda]$.
\end{enumerate}
%\end{flushleft}
\end{lemma}

\begin{proof}
We safely assume $n \geq 2$. (Otherwise, \Cref{lem:bid_distribution} is vacuously true since the unique bidder always bids zero $[\gamma,\, \lambda] = \{0\}$ in any {\FirstPriceAuction}.)
% It suffices to show that the first-order bid distribution $\calB(b)$ is strictly increasing and continuous on $b \in [\gamma,\, \lambda]$, given that $\gamma = \inf(\supp(\calB))$, that $\calB(b) = B_{i}(b) \cdot \calB_{-i}(b)$, and \orange{(\Cref{lem:support_inclusion})} that $\supp_{> \gamma}(\calB_{-i}) = \supp_{> \gamma}(\calB)$.
The proof relies on {\bf \Cref{fact:bid_distribution}}.

\setcounter{fact}{0}

\begin{fact}
\label{fact:bid_distribution}
Any two equilibrium bid distributions $B_{i}$ and $B_{k}$ for $i \neq k \in [n]$ cannot both have probability masses at a normal bid $b > \gamma$.
\end{fact}

\begin{proof}
Assume the opposite. Then both bidders $B_{i}$ and $B_{k}$ (and possibly someone else) tie for this first-order bid $s_{i}(v_{i}) = s_{k}(v_{k}) = \max(\bs(\bv)) = b > \gamma$ with a nonzero probability $> 0$. Conditioned on this, with a nonzero probability $> 0$, at least one bidder between $B_{i}$ and $B_{k}$ cannot be allocated. Hence, either $B_{i}$ or $B_{k}$ or both can gain a nonzero extra allocation by raising his/her (equilibrium) bid $s_{i}(v_{i}) = s_{k}(v_{k}) = b$. More formally, for some nonzero probabilities $\delta_{i},\, \delta_{k} > 0$, we have EITHER $\alloc_{i}(b^{*}) \geq \alloc_{i}(b) + \delta_{i}$ for any $b^{*} > b$, OR $\alloc_{k}(b^{*}) \geq \alloc_{k}(b) + \delta_{k}$ for any $b^{*} > b$, OR both.

Without loss of generality, let us consider the case that $\alloc_{i}(b^{*}) \geq \alloc_{i}(b) + \delta_{i}$ for any $b^{*} > b$. However, this means some deviation bid, such as $b^{*} = b + \frac{\delta_{i} / 2}{\alloc_{i}(b) + \delta_{i}} \cdot (v_{i} - b) $, can benefit. Namely, because this bidder $B_{i}$ has a normal equilibrium bid $s_{i}(v_{i}) = b > \gamma$, (\Cref{lem:dichotomy:2} of \Cref{lem:dichotomy}) he/she must have a higher value $v_{i} > s_{i}(v_{i}) = b$. This implies that (by construction) $b < b^{*} < v_{i}$ and thus, that the deviated utility $u_{i}(v_{i},\, b^{*}) = (v_{i} - b^{*}) \cdot \alloc_{i}(b^{*})$ is lower bounded as
\begin{align*}
    u_{i}(v_{i},\, b^{*})
    & ~\geq~ (v_{i} - b^{*}) \cdot \big(\alloc_{i}(b) + \delta_{i}\big)
    && \text{the ``WLOG''}
    \phantom{\big.} \\
    % & ~=~ (v_{i} - b) \cdot \Big(1 - \frac{\delta_{i}}{2(\alloc_{i}(b) + \delta_{i})}\Big) \cdot \big(\alloc_{i}(b) + \delta_{i}\big) \\
    & ~=~ (v_{i} - b) \cdot \big(\alloc_{i}(b) + \delta_{i} / 2\big)
    && \text{construction of $b^{*}$}
    \phantom{\big.} \\
    & ~>~ (v_{i} - b) \cdot \alloc_{i}(b)
    && \text{$v_{i} > b$ and $\delta_{i} > 0$}
    \phantom{\big.} \\
    & ~=~ u_{i}(v_{i},\, b)
    ~=~ u_{i}(v_{i},\, s_{i}(v_{i})).
    && \text{$s_{i}(v_{i}) = b$}
    \phantom{\big.}
\end{align*}
To conclude, the deviation bid $b^{*}$ strictly surpasses the equilibrium bid $s_{i}(v_{i})$, which contradicts (\Cref{def:bne_formal}) the optimality of $s_{i}(v_{i})$. Rejecting our assumption gives {\bf \Cref{fact:bid_distribution}}.
\end{proof}

% \yj{temporarily stop here}

\begin{remark}
In the proof of {\bf \Cref{fact:bid_distribution}}, the concrete construction of the deviation bid $b^{*}$ is less important.
Instead, the point is that
we are able to find two bids $b^{*} > b = s_{i}(v_{i})$ that can be arbitrarily close $(b^{*} - b) \searrow 0$ BUT
admit a {\em nonzero} interim allocation gap $\alloc_{i}(b^{*})-\alloc_{i}(b) \geq \delta_{i} > 0$.
Suppose so, bidder $i$ strictly benefits $u_{i}(v,\, b^{*}) > u_{i}(v,\, b)$ from the deviation bid $b^{*}$ when it is close enough to the equilibrium bid $b^{*} \searrow b$, hence a contradiction to the optimality of $b = s_{i}(v_{i})$.

Also, suppose there are two bids $b^{*} < b = s_{i}(v_{i})$ that yield two arbitrarily close {\em nonzero} interim allocations $(\alloc_{i}(b^{*}) - \alloc_{i}(b)) \nearrow 0$ BUT themselves are bounded away $b - b^{*} \geq \delta_{i} > 0$,
% by a fixed non-zero constant,
then bidder $i$ again benefits from the deviation bid $b^{*}$, hence a contradiction to the optimality of $b = s_{i}(v_{i})$.

We will apply such arguments in many places, WITHOUT specifying the deviation bids $b^{*}$ as the explicit constructions are less informative.
For clarity, such arguments will be called the \term[{\bf Neighborhood Deviation Arguments}]{pro:deviation}.
\end{remark}

% \vspace{.1in}
% \noindent
% \term[{\bf Neighborhood Deviation Arguments}]{pro:deviation}{\bf .}
% The concrete construction of the above deviation bid $b^{*}$ is less important.
% Instead, the point is that
% we are able to find two bids $b^{*} > b = s_{i}(v_{i})$ that can be arbitrarily close $(b^{*} - b) \searrow 0$ BUT
% admit a {\em nonzero} interim allocation gap $\alloc_{i}(b^{*})-\alloc_{i}(b) \geq \delta_{i} > 0$.
% Suppose so, bidder $i$ strictly benefits $u_{i}(v,\, b^{*}) > u_{i}(v,\, b)$ from the deviation bid $b^{*}$ when it is close enough to the equilibrium bid $b^{*} \searrow b$, hence a contradiction to the optimality of $b = s_{i}(v_{i})$.

% Also, suppose there are two bids $b^{*} < b = s_{i}(v_{i})$ that yield two arbitrarily close {\em nonzero} interim allocations $(\alloc_{i}(b^{*}) - \alloc_{i}(b)) \nearrow 0$ BUT themselves are bounded away $b - b^{*} \geq \delta_{i} > 0$,
% % by a fixed non-zero constant,
% then bidder $i$ again benefits from the deviation bid $b^{*}$, hence a contradiction to the optimality of $b = s_{i}(v_{i})$.

% We will apply such arguments in many places, WITHOUT specifying the deviation bids  $b^{*}$ as the explicit constructions are less informative.
% For clarity, such arguments will be called the \blackref{pro:deviation}.

Below let us prove {\bf \Cref{lem:bid_distribution:monotonicity}} and {\bf \Cref{lem:bid_distribution:continuity}}. We reason about the first-order bid distribution $\calB(b)$ and the competing bid distributions $\calB_{-i}(b)$ separately.

\vspace{.1in}
\noindent
{\bf \Cref{lem:bid_distribution:monotonicity}: $\calB(b)$.}
Assume the opposite to the $\calB(b)$'s strong monotonicity: The first-order bid distribution $\calB(b)$ has no probability density/mass on an OPEN interval $(\alpha,\, \beta) \subseteq [\gamma,\, \lambda]$.

Indeed, because $\gamma = \inf(\supp(\calB))$ and $\lambda = \sup(\supp(\calB))$,
we can find a {\em maximal} such interval $(\alpha,\, \beta)$ such that, the $\calB(b)$ has probability densities or probability masses within both of the left/right neighborhoods $(\alpha - \delta,\, \alpha]$ and $[\beta,\, \beta + \delta)$, for {\em whatever} $\delta > 0$. This gives $\calB(\beta) \geq \calB(\alpha) > 0$.
Let us do case analysis:
\begin{itemize}
    \item {\bf Case 1: $\calB(\beta) > \calB(\alpha) > 0$.} Then the $\calB(b)$ must have ONE probability mass at the $\beta$.

    ({\bf \Cref{fact:bid_distribution}}) Exactly one equilibrium bid distribution, $B_{k}(b)$ for a specific $k \in [n]$, has a probability mass at the $\beta$; his/her competing bid distribution $\calB_{-k}(b)$ has no probability mass there, thus being continuous at the $\beta$. This competing bid distribution $\calB_{-k}(b)$ has no probability density/mass on the open interval $(\alpha,\, \beta)$, \`{a} la the first-order bid distribution $\calB(b)$.

    Hence, (\Cref{lem:allocation:1} of \Cref{lem:allocation}) bidder $B_{k}(b)$ has a {\em constant} interim allocation $\alloc_{k}(b) = \alloc_{k}(\beta) = \calB_{-k}(\beta) > 0$ on the left-open {\em right-closed} interval $b \in (\alpha,\, \beta]$.
    But this means, conditioned on this bidder's equilibrium bid being at the probability mass $\big\{ s_{k}(v_{k}) = \beta \big\}$, any lower deviation bid $b^{*} \in (\alpha,\, \beta)$ yields a better deviated utility $u_{k}(v_{k},\, b^{*}) > u_{k}(v_{k},\, s_{k}(v_{k}))$. This contradicts (\Cref{def:bne_formal}) the optimality of this equilibrium bid $s_{k}(v_{k}) = \beta$.

    \item {\bf Case 2: $\calB(\beta) = \calB(\alpha) > 0$.} Then the $\calB(b)$ must have NO probability mass at the $\beta$.
    
    \`{A} la the $\calB(b)$, at least one bidder $B_{k}(b)$ for some $k \in [n]$, has probability densities within the right neighborhood $[\beta,\, \beta + \delta)$, for whatever $\delta > 0$. No matter how close $b \searrow \beta$, we can find an equilibrium bid $b = s_k(v)$, for some $v \in \supp(V_{k})$, that gives a positive utility $u_k(v,\, b) > 0$.
    
    \`{A} la {\bf Case~1}, this bidder has a {\em constant} interim allocation $\alloc_{k}(b^{*}) = \alloc_{k}(\beta) = \calB_{-k}(\beta) > 0$ on the left-open {\em right-closed} interval $b^{*} \in (\alpha,\, \beta]$; let us consider a particular bid $b^{*} = \frac{1}{2}(\alpha + \beta)$.
    % makes $\alloc_{k}(b^{*}) = \calB_{-k}(\beta)$.
    This deviation bid $b^{*}$ is bounded away from the above equilibrium bid $b - b^{*} \geq \frac{1}{2}(\beta - \alpha) > 0$ BUT, because there is no probability mass at $\beta$, gives an arbitrarily close interim allocation $\big(\alloc_{k}(b) - \alloc_{k}(b^{*})\big) \searrow 0$.
    Therefore, we can apply the \blackref{pro:deviation} to get a contradiction.
    
    % the above two bids $b > b^{*}$ are bounded away from each other 
    % we can choose $b$ as close to $\beta$ so that the interim allocation of $\alloc_{k}(b^{*})$ and $\alloc_{k}(b)$ can be
    % arbitrary close to each other while the two bids are constant away $b - b^{*} \geq \frac{1}{2}(\beta-\alpha )>0$. 
\end{itemize}
To conclude, we get a contradiction in either case. Reject our assumption: The first-order bid distribution $\calB(b)$ is strictly increasing on $b \in [\gamma,\, \lambda]$.

\vspace{.1in}
\noindent
{\bf \Cref{lem:bid_distribution:monotonicity}: $\calB_{-i}(b)$.}
Assume the opposite to the $\calB_{-i}(b)$'s strong monotonicity: At least one competing bid distribution $\calB_{-i}(b)$ has no probability density/mass on an OPEN interval $(\alpha,\, \beta) \subseteq [\gamma,\, \lambda]$.

In contrast, the remaining bidder $B_{i}(b)$ has probability densities/masses almost everywhere on $(\alpha,\, \beta)$, since the first-order bid distribution $\calB(b) = B_{i}(b) \cdot \calB_{-i}(b)$ is strongly increasing. Consider such an equilibrium bid $s_{i}(v_{i}) \in (\alpha,\, \beta)$. However, any lower deviation bid $b^{*} \in \big(\alpha,\, s_{i}(v_{i})\big)$ results in the same nonzero allocation $\alloc_{i}(b^{*}) = \alloc_{i}(s_{i}(v_{i})) > 0$ and thus a better deviated utility $u_{i}(v_{i},\, b^{*}) = (v_{i} - b^{*}) \cdot \alloc_{i}(b^{*}) > \big(v_{i} - s_{i}(v_{i})\big) \cdot \alloc_{i}(b^{*}) = u_{i}(v_{i},\, s_{i}(v_{i}))$.
This contradicts (\Cref{def:bne_formal}) the optimality of the equilibrium bid $s_{i}(v_{i}) \in (\alpha,\, \beta)$.
Reject our assumption: Each competing bid distribution $\calB_{-i}(b)$ is strictly increasing on $b \in [\gamma,\, \lambda]$.
{\bf \Cref{lem:bid_distribution:monotonicity}} follows then.

\vspace{.1in}
\noindent
{\bf \Cref{lem:bid_distribution:continuity}.}
We only need to verify continuity of the first-order bid distribution $\calB(b) = B_{i}(b) \cdot \calB_{-i}(b)$, which implies continuity of the equilibrium/competing bid distributions $\calB_{i}(b)$ and $\calB_{-i}(b)$.

Assume the opposite: $\calB(b)$ has a probability mass at some normal bid $\beta \in (\gamma,\, \lambda]$.
According to {\bf \Cref{fact:bid_distribution}}, exactly one bidder $B_{i}(b)$, for a specific $i \in [n]$, has a probability mass at this $\beta$. Further, following {\bf \Cref{lem:bid_distribution:monotonicity}}, his/her competing bid distribution $\calB_{-i}(b)$ has probability densities within the OPEN left neighborhood $b \in (\beta - \delta,\, \beta)$, for {\em whatever} $\delta > 0$; therefore at least one OTHER bidder $k \in [n] \setminus \{i\}$ has probability densities there.

This other bidder $B_{k}(b)$'s interim allocation formula $\alloc_{k}(b)$ is {\em discontinuous} at the $\beta \in (\gamma,\, \lambda]$,
because of $B_{i}(b)$'s probability mass at the $\beta$. Formally (\Cref{lem:allocation}), there exists a nonzero interim allocation gap $\xi^{*} \eqdef \inf \big\{\, \alloc_{k}(b^{*}) - \alloc_{k}(b): b^{*} > \beta > b \,\big\} > 0$ around the $\beta \in (\gamma,\, \lambda]$.
Hence, for some bids $b \in (\beta - \delta,\, \beta)$ and $b^{*} > \beta$ that are arbitrary close $(b^{*} - b) \searrow 0$, we can use the \blackref{pro:deviation} to get a contradiction.
Rejecting our assumption gives {\bf \Cref{lem:bid_distribution:continuity}}.

This finishes the proof of \Cref{lem:bid_distribution}.
\end{proof}

An important and direct implication of the continuity is that $\alloc_{i}(b)=\calB_{-i}(b)$ for any $b\in (\gamma, \lambda]$. This is formalized as \Cref{cor:allocation} and will be used in many places in the paper.

\begin{corollary}[Interim allocations]
\label{cor:allocation}
The interim allocation formula $\alloc_{i}(b)$ of each bidder $i \in [n]$ is identical to his/her competing bid distribution $\alloc_{i}(b) = \calB_{-i}(b)$, for any normal bid $b \in (\gamma, \lambda]$.
\end{corollary}

According to the Lebesgue differentiation theorem \cite{L04}, a monotonic function $f: (\gamma,\, \lambda) \mapsto \R$ is differentiable almost everywhere, except a set $\mathbb{D}_{f} \subseteq (0,\, \lambda)$ of a zero measure.\footnote{Rigorously, the set $\mathbb{D}_{f}$ has a zero {\em Lebesgue} measure. Yet since the bid distributions $B_{i}(b)$, $\calB_{-i}(b)$, and $\calB(b)$ are continuous on $b \in (\gamma,\, \lambda)$, we need not to distinguish the Lebesgue/probabilistic measures.}
We thus conclude \Cref{cor:bid_distribution} directly from \Cref{lem:bid_distribution}.

\begin{corollary}[Bid distributions]
\label{cor:bid_distribution}
Each of the equilibrium/competing/first-order bid distributions, $\{B_{i}\}_{i \in [n]}$, $\{\calB_{-i}\}_{i \in [n]}$, and $\calB$
is differentiable almost everywhere on the OPEN interval $b \in (\gamma,\, \lambda)$, except a set $\mathbb{D} \subseteq (\gamma,\, \lambda)$ on which this bid distribution, $B_{i}(b)$, $\calB_{-i}(b)$, or $\calB(b)$, has a zero measure.
\end{corollary}

\begin{comment}
% {\bf \Cref{cor:bid_distribution:1}} restates \Cref{lem:bid_distribution}. Also,
Both of {\bf \Cref{cor:bid_distribution:2,cor:bid_distribution:3}} directly follow from \Cref{lem:bid_distribution}, that each bid distribution, $B_{i}(b)$, $\calB_{-i}(b)$, or $\calB(b)$, is continuous on $b \in [\gamma,\, \lambda]$. In particular, ({\bf \Cref{cor:bid_distribution:3}}) we can choose $\mathbb{D} \eqdef \mathbb{D}_{\calB}$ as the (countable) set of points $\subseteq (\lambda,\, \gamma)$ on which the first-order bid distribution $\calB(b) = B_{i}(b) \cdot \calB_{-i}(b)$ is indifferentiable; this set $\mathbb{D} = \mathbb{D}_{\calB}$ must cover the indifferentiable points of $B_{i}(b)$'s and $\calB_{-i}(b)$'s.
\end{comment}

\begin{remark}[Bid distributions]
\label{rem:bid_distribution}
The continuity/monotonicity/differentiability shown in \Cref{lem:bid_distribution,cor:bid_distribution} are the basics of our later discussions; all subsequent bid distributions will satisfy them. So for brevity, we often omit a formal clarification. Also, we often simply write the derivatives $B'_{i}(b)$ etc by omitting (\Cref{cor:bid_distribution}) the zero-measure indifferentiable points $\mathbb{D} \subseteq (\gamma,\, \lambda)$;
we can use standard tools from real analysis to deal with those points $\mathbb{D} \subseteq (\gamma,\, \lambda)$ separately.
\end{remark}

% shows that conditioned on the boundary first-order bid $\big\{\max(\bs(\bv)) = \gamma\big\}$, the underlying values $\bv = (v_{i})_{i \in [n]}$ are well structured.
% (The statement of \Cref{lem:tiebreak} might be prolix. Instead, the upcoming \Cref{lem:tiebreak_bid} is more concise and is sufficient for our later uses.)
% $\gamma = \inf(\supp(\calB))$, the equilibrium/first-order bids $s_{i}(v_{i}) \sim B_{i}(b)$ for $i \in [n]$ and $\max(\bs(\bv)) \sim \calB(b)$ are well structured.

%This subsection establishes the {\em bid-based} equivalence condition (\Cref{thm:bne}) for an auction-strategy tuple $(\alloc,\, \bs)$ being an exact {\BayesNashEquilibrium} $\bs \in \bbBNE(\bV,\, \alloc)$.

% As before, we would assume $n \geq 2$. (The other case $n = 1$ is trivial: The unique bidder always bids zero in any {\FirstPriceAuction} $\alloc^{*} \in \bbFPA$.)

%Following \Cref{lem:bid_monotonicity,lem:bid_distribution}, the equilibrium bids $s_{i}(v)$ almost are monotonic and have no probability mass, so we can talk about the {\em inverses} of equilibrium strategies $s_{i}^{-1}(b)$.
%\Cref{def:inverse} formalizes this concept and \Cref{lem:inverse} shows the properties. (Regarding \Cref{lem:inverse:dichotomy} of \Cref{lem:inverse}, the inverses at the boundary $s_{i}^{-1}(\gamma)$ are tricky and will be studied in \Cref{lem:boundary_bid}.)

\subsection{The inverse direction: Bid-to-value mappings}
\label{subsec:mapping}

In this subsection, we study the {\em inverse} mappings of strategies $\bs^{-1} = \{s_{i}^{-1}\}_{i \in [n]}$ and try to reconstruct the value distributions $\bV = \{V_{i}\}_{i \in [n]}$ from (equilibrium) bid distributions.

\begin{definition}[Inverses of strategies]
\label{def:inverse}
Consider bid distributions $\bB = \{B_{i}\}_{i \in [n]}$ given by value distributions $\bv = (v_{i})_{i \in [n]} \sim \bV$ and a strategy profile $\bs = \{s_{i}\}_{i \in [n]}$.
% (not necessarily an equilibrium).
Each {\em random} inverse $s_{i}^{-1}(b)$ for $b \in \supp(B_{i})$ is defined as the {\em conditional} distribution $\{ v_i \sim V_i\,\bigmid\, s_{i}(v_{i}) = b \}$.
%\[
%    \Prx \big[s_{i}^{-1}(b) = v\big] ~=~ \Prx_{v_{i},\, s_{i}} \big[\, v_{i} = v \,\bigmid\, s_{i}(v_{i}) = b \,\big],
%    \qquad\qquad \forall v \geq 0.
%\]
So, each random inverse $s_{i}^{-1}(b_{i})$ for $b_{i} \sim B_{i}$ is identically distributed as the random value $v_{i} \sim V_{i}$.
\end{definition}

\Cref{lem:inverse} shows two basic properties of the random inverses $\bs^{-1} = \{s_{i}^{-1}\}_{i \in [n]}$.

\begin{lemma}[Inverses of strategies]
\label{lem:inverse}
Almost surely over each random inverse $s_{i}^{-1}$:
\begin{enumerate}[font = {\em\bfseries}]
    \item\label{lem:inverse:dichotomy}
    {\bf dichotomy:}
    $s_{i}^{-1}(b) \leq \gamma$ for any low bid $b \in \supp_{< \gamma}(B_{i})$. \\
    \white{\bf dichotomy:}
    $s_{i}^{-1}(b) > b > \gamma$ for any normal bid $b \in \supp_{> \gamma}(B_{i})$.

    \item\label{lem:inverse:monotone}
    {\bf monotonicity:}
    $s_{i}^{-1}(b) \geq s_{i}^{-1}(t)$ for two boundary/normal bids $b,\, t \in \supp_{\geq \gamma}(B_{i})$ that $b > t$.
    % , namely being element-wise weakly increasing on $\supp_{\geq \gamma}(B_{i})$.
    % $\inf(s_{i}^{-1}(b)) \geq \sup(s_{i}^{-1}(t))$ for any two boundary/normal bids $b,\, t \in \supp_{\geq \gamma}(B_{i})$ that $b > t$, namely being element-wise weakly increasing on $\supp_{\geq \gamma}(B_{i})$.
\end{enumerate}
\end{lemma}

\begin{proof}
Under our definition of the $s_{i}^{-1}(b)$, {\bf \Cref{lem:inverse:dichotomy}} follows directly from \Cref{lem:dichotomy}, while {\bf \Cref{lem:inverse:monotone}} follows from a combination of \Cref{lem:dichotomy} and \Cref{lem:bid_monotonicity} (\Cref{lem:bid_monotonicity:2}).
\end{proof}

Below, we characterize (\Cref{lem:high_bid}) the inverses of {\em equilibrium} strategies $s_{i}^{-1}(b)$ for {\em normal} bids $b \in (\gamma,\, \lambda]$, using (\Cref{def:mapping}) the concept of {\em bid-to-value mappings} $\varphi_{i}(b)$. We prove that the inverse $s_{i}^{-1}(b)$ at a normal bid $b \in (\gamma,\, \lambda]$ is essentially a {\em deterministic} value. (The inverses at the boundary bid $s_{i}^{-1}(\gamma)$ are tricky and will be studied later.)

\begin{definition}[Bid-to-value mappings]
\label{def:mapping}
Given equilibrium bid distributions $\bB = \{B_{i}\}_{i \in [n]}$ (cf.\ \Cref{lem:bid_distribution,cor:bid_distribution}), define each {\em bid-to-value mapping} $\varphi_{i}(b)$ for $b \in (\gamma,\, \lambda)$ as follows:
\[
    \varphi_{i}(b)
    ~\eqdef~ b + \frac{\calB_{-i}(b)}{\calB'_{-i}(b)}
    ~=~ b + \Big(\sum_{k \in [n] \setminus \{i\}} \frac{B'_{k}(b)}{B_{k}(b)}\Big)^{-1}
    ~=~ b + \Big(\frac{\calB'(b)}{\calB(b)} - \frac{B'_{i}(b)}{B_{i}(b)}\Big)^{-1}
\]
It turns out that (\Cref{lem:high_bid:monotone} of \Cref{lem:high_bid}) each $\varphi_{i}(b)$ is an increasing function, so the domain can be extended to include the both endpoints $\varphi_{i}(\gamma) \eqdef \lim_{b \searrow \gamma} \varphi_{i}(b)$ and $\varphi_{i}(\lambda) \eqdef \lim_{b \nearrow \lambda} \varphi_{i}(b)$.
% Given a product bid distribution $\bB = \{B_{i}\}_{i \in [n]}$ that:
% \begin{itemize}
%     \item The first-order bid CDF $\calB(b) = \prod_{i \in [n]} B_{i}(b)$ is continuous and strictly increasing (so the right derivative $\calB'(b) > 0$) on the closed interval/support $b \in \supp(\calB) = [\gamma,\, \lambda]$.
%
%     \item Each competing bid CDF $\calB_{-i}(b) = \prod_{k \in [n] \setminus \{i\}} B_{i}(b)$ for $i \in [n]$ is also continuous and strictly increasing (so the right derivative $\calB'_{-i}(b) > 0$) on the closed interval $b \in [\gamma,\, \lambda]$.
% \end{itemize}
% Then each {\em bid-to-value mapping} $\varphi_{i}: [\gamma,\, \lambda] \mapsto \RR$ for $i \in [n]$ is given by
% \begin{align*}
%     \varphi_{i}(b)
%     ~=~ b + \frac{\calB_{-i}(b)}{\calB'_{-i}(b)}
%     ~=~ b + \bigg(\sum_{k \in [n] \setminus \{i\}} \frac{B'_{k}(b)}{B_{k}(b)}\bigg)^{-1}.
% \end{align*}
\end{definition}

The above bid-to-value mapping is known in the literature (\cite{HHT14} etc). However, it is usually defined only on the support of $B_i$. We define it and study its properties on the entire interval  $[\gamma,\, \lambda]$, which is very important to our argument.

\begin{lemma}[Normal bids]
\label{lem:high_bid}
For each bid-to-value mapping $\varphi_{i}(b)$, the following hold:
\begin{enumerate}[font = {\em\bfseries}]
    \item\label{lem:high_bid:inverse}
    \term[{\bf invertibility}]{value_invertibility}{\bf :}
    $s_{i}^{-1}(b) = \varphi_{i}(b)$ almost surely over the random inverse $s_{i}^{-1}$, for any normal bid $b \in \supp_{> \gamma}(B_{i})$.\footnote{More rigorously, (\Cref{cor:bid_distribution,rem:bid_distribution}) we shall exclude a zero-measure set $\mathbb{D} \subseteq (\gamma,\, \lambda]$, namely the indifferentiable points of the $B_{i}(b)$'s}

    \item\label{lem:high_bid:monotone}
    \term[{\bf monotonicity}]{value_monotonicity}{\bf :}
    $\varphi_{i}(b)$ is weakly increasing on the closed interval $b \in [\gamma,\, \lambda]$.

    \item\label{lem:high_bid:rational}
    {\bf rationality:}
    $\varphi_{i}(b) > b$ on the left-open right-closed interval $b \in (\gamma,\, \lambda]$.
    % for any normal bid $b > \gamma$,
    Further, $\varphi_{i}(\gamma) \geq \gamma$.
    % on the boundary $b = \gamma$.
\end{enumerate}
\end{lemma}

\begin{proof}
We safely assume $n \geq 2$. (Otherwise, \Cref{lem:high_bid} is vacuously true since the unique bidder always bids zero $[\gamma,\, \lambda] = \{0\}$ in any {\FirstPriceAuction}.)
% We check {\bf \Cref{lem:high_bid:inverse,lem:high_bid:monotone,lem:high_bid:rational}} one by one.

\vspace{.1in}
\noindent
{\bf \Cref{lem:high_bid:inverse}.}
Consider a specific inverse/value $s_{i}^{-1}(b) = v$ resulted from a normal bid $b \in \supp_{> \gamma}(B_{i})$. Regarding a normal bid $b^{*} > \gamma$, we can rewrite (\Cref{lem:allocation} (\Cref{lem:allocation:2}) and \Cref{cor:allocation}) the interim utility formula $u_{i}(v,\, b^{*}) = (v - b^{*}) \cdot \alloc_{i}(b^{*}) = (v - b^{*}) \cdot \calB_{-i}(b^{*})$ and deduce the partial derivative
\[
    \tfrac{\partial u_{i}}{\partial b^{*}}
    ~=~ -\calB_{-i}(b^{*}) + (v - b^{*}) \cdot \calB'_{-i}(b^{*})
    ~=~ \big(v - \varphi_{i}(b^{*})\big) \cdot \calB'_{-i}(b^{*}).
\]
Hence, to make (\Cref{def:bne_formal}) the equilibrium bid $b \in \supp_{> \gamma}(B_{i})$ optimize the formula $u_{i}(v,\, b^{*})$, we must have $\varphi_{i}(b) = v = s_{i}^{-1}(b)$. Here we use the fact that $B'_{-i}(b^{*})>0$ since $B_{-i}$ is strictly increasing (\Cref{lem:bid_distribution}).
Thus, {\bf \Cref{lem:high_bid:inverse}} follows then.

\vspace{.1in}
\noindent
{\bf \Cref{lem:high_bid:monotone}.}
For brevity, we consider a {\em twice differentiable}\footnote{In general, following the Lebesgue differentiation theorem \cite{L04}, the $\calB(b)$ is twice differentiable almost everywhere on $(\gamma,\, \lambda)$, except a zero-measure set $\mathbb{D}_{\calB} \subseteq (\gamma,\, \lambda)$ that can be handled by standard tools from real analysis.} first-order bid CDF $\calB(b) = B_{i}(b) \cdot \calB_{-i}(b)$.
Thus for $i \in [n]$, each competing bid CDF $\calB_{-i}(b)$ is also twice differentiable and each bid-to-value mapping $\varphi_{i}(b) = b + \calB_{-i}(b) \big/ \calB'_{-i}(b)$ is continuous on $b \in (\gamma,\, \lambda)$.

Following a combination of {\bf \Cref{lem:high_bid:inverse}} and \Cref{lem:inverse} (\Cref{lem:inverse:monotone}), each mapping $\varphi_{i}(b)$ is increasing on this equilibrium bid distribution $B_{i}$'s normal bid support $\supp_{> \gamma}(B_{i})$, so the task is to extend this monotonicity to the whole interval $b \in (\gamma,\, \lambda)$.
The proof relies on {\bf \Cref{fact:high_bid:convex,fact:pseudo}}.

% The following lemma gives a {\em geometric} interpretation of validness, which is more convenient for our later use.

\setcounter{fact}{0}

\begin{fact}
\label{fact:high_bid:convex}
A mapping $\varphi(b) = b + B(b) \big/ B'(b)$ is increasing iff the reciprocal function $1 \big/ B(b)$ is convex.
\end{fact}

\begin{proof}
Once again, we assume that the underlying CDF $B(b)$ is twice differentiable.
% (The twice-indifferentiable points can be handled using real analysis.)
The reciprocal function $1 \big/ B(b)$ has the second derivative $\frac{\d^{2}}{\d b^{2}} \big(\frac{1}{B}\big)
= \frac{\d}{\d b} \big(-\frac{B'}{B^{2}}\big)
= -\frac{B'' \cdot B - 2 \cdot (B')^{2}}{B^{2}}$.
% \[
%     \frac{\d^{2}}{\d b^{2}} \Big(\frac{1}{B}\Big)
%     ~=~ \frac{\d}{\d b} \Big(-\frac{B'}{B^{2}}\Big)
%     ~=~ -\frac{B'' \cdot B - 2 \cdot (B')^{2}}{B^{2}}.
% \]
So we can rewrite the first derivative of the mapping $\varphi'(b) = 1 + \frac{(B')^{2} - B \cdot B''}{(B')^{2}} = \frac{B^{2}}{(B')^{2}} \cdot \frac{\d^{2}}{\d b^{2}} \big(\frac{1}{B}\big)$.
% $\varphi(b) = b + B(b) \big/ B'(b)$ as follows:
% \begin{align*}
%     \frac{\d \varphi}{\d b}
%     ~=~ 1 + \frac{(B')^{2} - B \cdot B''}{(B')^{2}}
%     ~=~ \frac{B^{2}}{(B')^{2}} \cdot \frac{\d^{2}}{\d b^{2}} \Big(\frac{1}{B}\Big).
% \end{align*}
Precisely, the mapping is increasing $\varphi'(b) \geq 0$ iff the reciprocal function is convex $\frac{\d^{2}}{\d b^{2}} (\frac{1}{B}) \geq 0$. {\bf \Cref{fact:high_bid:convex}} follows then.
% Now it is clear that $\varphi'(b) \geq 0$ iff $\frac{\d^{2}}{\d b^{2}} (\frac{1}{B}) \geq 0$. Namely, the mapping $\varphi(b)$ is increasing iff the reciprocal function $1 \big/ B(b)$ is convex. {\bf \Cref{fact:high_bid:convex}} follows then.
\end{proof}

% instead of considering the monotonicity, it suffices to prove the next fact.

\begin{fact}
\label{fact:pseudo}
The first-order mapping $\varphi_{\calB}(b) \eqdef b + \calB(b) \big/ \calB'(b)$ is increasing on $b \in (\gamma,\, \lambda)$.
\end{fact}

\begin{proof}
Based on {\bf \Cref{fact:high_bid:convex}}, it suffices to show that the first-order reciprocal function $\calR(b) \eqdef 1 \big/ \calB(b)$ is convex.
For brevity, we denote the equilibrium/competing reciprocal functions $R_{i}(b) \eqdef 1 \big/ B_{i}(b)$ and $\calR_{-i}(b) \eqdef 1 \big/ \calB_{-i}(b)$ for $i \in [n]$. Also, we may simply write $R_{i} = R_{i}(b)$ etc.

% each equilibrium bid distribution $\bB = \{B_{i}\}_{i \in [n]}$ is supported at the considered

% We first establish {\bf \Cref{fact:pseudo}} under the assumption that all bid distributions $\bB = \{B_{i}\}_{i \in [n]}$ are supported at this bid, namely $b \in \supp_{> \gamma}(B_{i})$ for $i \in [n]$.
By elementary algebra, each competing reciprocal function $\calR_{-i}(b) = \prod_{k \in [n] \setminus \{i\}} R_{k}(b)$ has the second derivative
\begin{align*}
    \calR''_{-i}
    & ~=~ \sum_{k_{1} \neq k_{2} \neq i} \Big(R'_{k_{1}} R'_{k_{2}} \prod_{k_{3} \notin \{i,\, k_{1},\, k_{2}\}} R_{k_{3}}\Big)
    + \sum_{k \neq i} \Big(R''_{k} \prod_{k_{3} \notin \{i,\, k\}} R_{k_{3}}\Big) \\
    & ~=~ \calR_{-i} \cdot \Big(\sum_{k_{1} \neq k_{2} \neq i} \tfrac{R'_{k_{1}} R'_{k_{2}}}{R_{k_{1}} R_{k_{2}}} + \sum_{k \neq i} \tfrac{R''_{k}}{R_{k}}\Big).
    \qquad \mbox{factor out $\calR_{-i} = \prod_{k \neq i} R_{k}$}
\end{align*}
% Here the second step factors out the term $\calR_{-i}(b) = \prod_{k \in [n] \setminus \{i\}} R_{k}(b)$.
We thus deduce that
\begin{align*}
    \tfrac{1}{n - 1} \sum_{i \in [n]} R_{i} \cdot \calR''_{-i}
    & ~=~ \tfrac{\calR}{n - 1} \sum_{i \in [n]} \Big(\sum_{k_{1} \neq k_{2} \neq i} \tfrac{R'_{k_{1}} R'_{k_{2}}}{R_{k_{1}} R_{k_{2}}} + \sum_{k \neq i} \tfrac{R''_{k}}{R_{k}}\Big)
    && \mbox{$R_{i} \cdot \calR_{-i} = \calR$} \\
    & ~=~ \tfrac{n - 2}{n - 1} \cdot \calR \cdot \Big(\sum_{k_{1} \neq k_{2}} \tfrac{R'_{k_{1}} R'_{k_{2}}}{R_{k_{1}} R_{k_{2}}}\Big)
    ~+~ \calR \cdot \Big(\sum_{k \in [n]} \tfrac{R''_{k}}{R_{k}}\Big).
    && \mbox{combine the like terms}
\end{align*}
% Here the second step combines the like terms.
Moreover, the first-order reciprocal function $\calR(b) = \prod_{k \in [n]} R_{k}(b)$ has the second derivative
\begin{align*}
    \calR''
    & ~=~ \calR \cdot \Big(\sum_{k_{1} \neq k_{2}} \tfrac{R'_{k_{1}} R'_{k_{2}}}{R_{k_{1}} R_{k_{2}}} + \sum_{k \in [n]} \tfrac{R''_{k}}{R_{k}}\Big)
    && \mbox{\`{a} la the $\calR''_{-i}$ formulas} \\
    & ~=~ \tfrac{\calR}{n - 1} \sum_{k_{1} \neq k_{2}} \tfrac{R'_{k_{1}} R'_{k_{2}}}{R_{k_{1}} R_{k_{2}}}
    + \tfrac{1}{n - 1} \sum_{k \in [n]} R_{k} \cdot \calR''_{-k}
    ~\geq~ 0.
    && \mbox{substitute $\sum_{i \in [n]} \calR''_{-i} \cdot R_{i}$}
\end{align*}
% Here the first step is derived \`{a} la the above $\calR''_{-i}$ formulas. And the second step applies the above $\sum_{i \in [n]} \calR''_{-i} \cdot R_{i}$ formula. For this second derivative $R''$ formula, 
The first summation: $R'_{k_{1}},\, R'_{k_{2}} \leq 0$ for $k_{1} \neq k_{2} \in [n]$, so the product $R'_{k_{1}} R'_{k_{2}} \geq 0$. \\
% because each $R_{i}$ is the reciprocal of a CDF. \\
The second summation: $\calR''_{-k} \geq 0$ for $k \in [n]$. At the considered bid $b \in \supp_{> \gamma}(B_{i})$, each mapping is increasing $\varphi'_{i}(b) \geq 0$ ({\bf \Cref{lem:high_bid:inverse}}) and thus each reciprocal is convex $\calR''_{-i}(b) \geq 0$ ({\bf \Cref{fact:high_bid:convex}}).

Hence, the first-order reciprocal function $\calR(b)$ is convex on $b \in \supp_{> \gamma}(B_{i})$.
{\bf \Cref{fact:pseudo}} follows.
\end{proof}

{\bf \Cref{lem:high_bid:monotone}} follows directly: On the normal bid support $b \in \supp_{> \gamma}(B_{i})$, the mapping $\varphi_{i}(b) = s_{i}^{-1}(b)$ must be increasing (\Cref{lem:inverse:monotone} of \Cref{lem:inverse}).
Outside the support $b \in (\gamma,\, \lambda) \setminus \supp_{> \gamma}(B_{i})$, we have $\calB_{-i}(b) \big/ \calB'_{-i}(b) = \calB(b) \big/ \calB'(b)$, so the mapping $\varphi_{i}(b) = \varphi_{\calB}(b)$ is still increasing ({\bf \Cref{fact:pseudo}}).

\vspace{.1in}
\noindent
{\bf \Cref{lem:high_bid:rational}.}
Following \Cref{lem:bid_distribution}, each competing bid distribution $\calB_{-i}(b)$ is a continuous and strictly increasing function on $b \in [\gamma,\, \lambda]$.
Hence, by definition $\varphi_{i}(b) = b + \calB_{-i}(b) \big/ \calB'_{-i}(b) \geq b$; the equality is possible only at the boundary bid $b = \gamma$. {\bf \Cref{lem:high_bid:rational}} follows then. This finishes the proof.
\end{proof}

As an implication of \Cref{lem:high_bid}, the value distributions $\bV = \{V_{i}\}_{i \in [n]}$ can partially be reconstructed from the equilibrium bid distributions $\bB = \{B_{i}\}_{i \in [n]}$. (See \Cref{quantiles}.)
% ; see \Cref{fig:monopolist} for a visual aid.

\begin{corollary}[Reconstructions]
\label{cor:high_bid}
For each $i \in [n]$, the $v \geq \varphi_{i}(\gamma)$ part of the value distribution $V_{i}(v) = \Prx_{v_{i} \sim V_{i}} \big[v_{i} \leq v\big]$ can be reconstructed from the equilibrium bid distributions $\bB$ as follows.
\[
    V_{i}(v) ~=~ \Prx_{b_{i} \sim B_{i}}\big[(b_{i} \le \gamma) \vee (\varphi_{i}(b_{i}) \leq v) \big],
    \qquad\qquad \forall v \geq \varphi_{i}(\gamma).
\]
\end{corollary}

In the statement of \Cref{cor:high_bid}, a mapping $\varphi_{i}(b_{i})$ is undefined when $b_{i} < \gamma$, but since the first condition $(b_{i} \leq \gamma)$ already holds, we think of the disjunction $(b_{i} \leq \gamma) \vee (\varphi_{i}(b_{i}) \leq v)$ as being satisfied.

\begin{figure}[t]
    \centering
    \includegraphics[width = .9\textwidth]{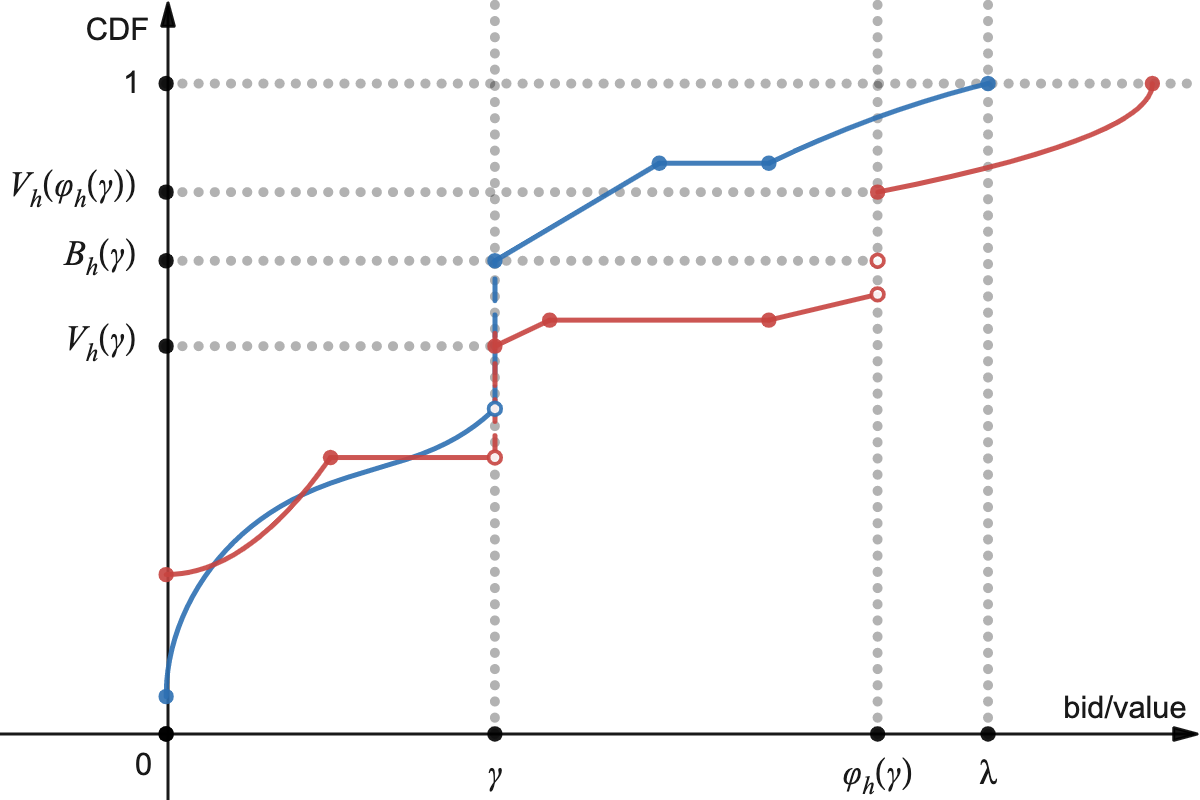}
    \caption{Diagram of the reconstruction of a monopolist $h$'s value distributions $V_{h}$ (\Cref{def:monopolist,lem:value_dist}).
    Notice that $V_{i}(\gamma) \leq B_{i}(\gamma) \leq V_{i}(\varphi_{i}(\gamma))$.
    By \Cref{lem:dichotomy,lem:bid_monotonicity,lem:high_bid}: \\
    (i)~A normal value $v_{h} > \varphi_{h}(\gamma)$ induces a normal bid $s_{h}(v_{h}) > \gamma$. \\
    (ii)~The normal value $v_{h} = \varphi_{h}(\gamma)$ induces a normal bid $s_{h}(v_{h}) > \gamma$ or the boundary bid $s_{h}(v_{h}) = \gamma$, both of which are possible especially when there is a probability mass $\Prx_{v_{h} \sim V_{h}} [v_{h} = \varphi_{h}(\gamma)] > 0$. \\
    (iii)~A normal value $v_{h} \in (\gamma,\, \varphi_{h}(\gamma))$ induces the boundary bid $s_{h}(v_{h}) = \gamma$. \\
    (iv)~The boundary value $v_{h} = \gamma$ induces a low bid $s_{h}(v_{h}) < \gamma$ or the boundary bid $s_{h}(v_{h}) = \gamma$, both of which are possible especially when there is a probability mass $\Prx [v_{h} = \gamma] > 0$. \\
    (v)~A low value $v_{h} < \gamma$ induces a bid $s_{h}(v_{h})$ that yields the optimal zero utility $u_{h}(v_{h},\, s_{h}(v_{h})) = 0$ but is otherwise arbitrary; even an overbid $s_{h}(v_{h}) > v_{h}$ is fine.
    \label{fig:monopolist}}
\end{figure}

\Cref{cor:high_bid} enables us to reconstruct the $v \geq \varphi_{i}(\gamma)$ part of each value distribution $V_{i}$.
Then how about the $v < \varphi_{i}(\gamma)$ part?
We shall consider two subparts separately.
(i)~The low value $v < \gamma$ subpart of a value distribution $V_{i}$ cannot be reconstructed from the bid distributions $\bB$. That is, an equilibrium strategy $b_{i} = s_{i}(v)$ for a low value $v < \gamma$ is arbitrary as long as it yields the optimal zero utility $u_{i}(v,\, b_{i}) = 0$. Luckily, these low value $v < \gamma$ subparts (\Cref{lem:auction_welfare,lem:optimal_welfare}) turn out to have no contribution to neither of the auction/optimal {\SocialWelfares}, so we can ignore them.
(ii)~The $\gamma \leq v < \varphi_{i}(\gamma)$ subpart of a value distribution $V_{i}$ {\em has} contributions to the auction/optimal {\SocialWelfares} but also cannot be reconstructed.
This subpart always induces to the boundary bid $b_{i} = s_{i}(v) = \gamma$, following \blackref{value_monotonicity} (\Cref{lem:dichotomy,lem:high_bid}).

% \blue{By \blackref{value_monotonicity}, we known that if bidder $h$'s value $v_h\in(\gamma,\varphi_{h}(\gamma)]$, he must bid $s_h(v_h)=\gamma$. Thus, a monopolist has a non-zero probability to bid $\gamma$. We have the following key properties regarding monopolists and the tie-breaking rule at bid $\gamma$.}

As a remedy, we introduce the concept of {\em monopolists} (\Cref{def:monopolist}); cf.\ \Cref{fig:monopolist} for a visual aid.
Here we notice that $V_{i}(\gamma) \leq B_{i}(\gamma) = \Pr_{b_{i} \sim B_{i}} [b_{i} \leq \gamma] \leq V_{i}(\varphi_{i}(\gamma))$.
I.e., the first inequality holds since (\Cref{lem:dichotomy}) a normal bid $b_{i} > \gamma$ induces a higher normal value $v_{i} = s_{i}^{-1}(b_{i}) = \varphi_{i}(b_{i}) > b_{i} > \gamma$.
The second inequality directly follows from \Cref{cor:high_bid}.

\begin{definition}[Monopolists]
\label{def:monopolist}
A bidder $h \in [n]$ is called a {\em monopolist} when the probabilities of taking normal bids/values are unequal $1 - B_{h}(\gamma) < 1 - V_{h}(\gamma)$ or equivalently, when the probability of taking the boundary bid yet a normal value is nonzero $\Pr_{b_{h},\, s_{h}^{-1}} \big[(b_{h} = \gamma) \wedge (s_{h}^{-1}(b_{h}) > \gamma)\big] > 0$.
\end{definition}

It is easier to understand this definition from the perspective of quantiles (\Cref{quantiles}): A bidder $h \in [n]$ is a monopolist when there are some quantiles $q$ such that $V_h^{-1}(q)>\gamma$ but $B_h^{-1}(q)=\gamma$.

% the probabilities of taking normal bids $\Prx[b_{h} > \gamma] = 1 - B_{h}(\gamma)$ and of taking normal values $\Prx[v_{h} > \gamma] = 1 - V_{h}(\gamma)$ are {\em unequal}, namely $1 - B_{h}(\gamma) < 1 - V_{h}(\gamma)$.
% A bidder $h \in [n]$ is called a {\em monopolist} when the has non-trivial probability density on $(\gamma,\varphi_{h}(\gamma)]$. \yj{$\Prx[b_{h} > \gamma] = 1 - B_{h}(\gamma) < 1 - V_{h}(\gamma) = \Prx[v_{h} > \gamma]$.}

\Cref{lem:monopolist} presents the properties about the monopolists.
(From the statement of \Cref{lem:monopolist}, we can infer that there is NO monopolist when the first-order bid $\max(\bb) \sim \calB$ has no probability mass at the boundary bid $\calB(\gamma) = 0$.)

\begin{lemma}[Monopolists]
\label{lem:monopolist}
There exists at most one monopolist $h$. If existential:
\begin{itemize}% [font = {\em\bfseries}]
    \item \term[{\bf monopoly}]{mono_monopoly}{\bf :}
    The probability of a boundary first-order bid $\big\{ \max(\bb) = \gamma \big\}$ is nonzero $\calB(\gamma) > 0$.
    Conditioned on the tiebreaker $\big\{ b_{h} = \max(\bb) = \gamma \big\}$, the monopolist wins $\alloc(\bb) = h$ almost surely over the random allocation $\alloc$.
    
    \item \term[{\bf boundedness}]{mono_boundedness}{\bf :}
    A boundary bid $b_{h} = \gamma$ induces a bounded random value $s_{h}^{-1}(\gamma) \in [\gamma,\, \varphi_{h}(\gamma)]$.
\end{itemize}
\end{lemma}

\begin{proof}
Suppose there is a monopolist $h \in [n]$.
We first verify \blackref{mono_monopoly} and \blackref{mono_boundedness}, and then prove that there is no other monopolist.

For some normal value $v_{h} = s_{h}^{-1}(b_{h}) > \gamma$, (\Cref{def:monopolist,lem:dichotomy}) the boundary bid $b_{h} = \gamma$ yields the optimal utility $u_{h}(v_{h},\, \gamma) = (v_{h} - \gamma) \cdot \alloc_{h}(\gamma) = \calU_{h}(v_{h}) > 0$, so the interim allocation is nonzero $\alloc_{h}(\gamma) > 0$.
The boundary first-order bid $\big\{ \max(\bb) = \gamma \big\}$ occurs with probability $\calB(\gamma) = B_{h}(\gamma) \cdot \calB_{-h}(\gamma) \geq \Prx_{b_{h}} [b_{h} = \gamma] \cdot \alloc_{h}(\gamma) > 0$.
Assume for contradiction that this monopolist $h \in [n]$ loses the tiebreaker $\big\{ b_{h} = \max(\bb) = \gamma \big\}$ with a nonzero probability $> 0$.
Based on the \blackref{pro:deviation}, some (infinitesimally) higher deviation bid $b_{h}^{*}> b_{h} = \gamma$ gives a strictly better utility $u_{h}(v_{h},\, b_{h}^{*}) > u_{h}(v_{h},\, \gamma) > 0$, which contradicts (\Cref{def:bne_formal}) the optimality of this equilibrium bid $b_{h} = \gamma$.
Reject our assumption: This monopolist $h \in [n]$ always wins the tiebreaker $\big\{ b_{h} = \max(\bb) = \gamma \big\}$ (\blackref{mono_monopoly}).

A boundary bid $b_{h} = \gamma$ induces a bounded random value $\gamma \leq s_{h}^{-1}(\gamma) \leq \varphi_{h}(\gamma)$ (\blackref{mono_boundedness}). (i)~This monopolist $h \in [n]$ can never take a low value $s_{h}^{-1}(\gamma) < \gamma$. The tiebreaker $\big\{ b_{h} = \max(\bb) = \gamma \big\}$ occurs with a nonzero probability $> 0$ and, suppose so, the monopolist wins $\alloc(\bb) = h$ almost surely over the random allocation $\alloc$.
Hence, a low value $s_{h}^{-1}(\gamma) < \gamma$ together with a boundary bid $b_{h} = \gamma$ yields a negative utility $< 0$, which is impossible.
(ii)~The value is also upper bounded $s_{h}^{-1}(\gamma) \leq \lim_{b \searrow \gamma} s_{h}^{-1}(b) = \lim_{b \searrow \gamma} \varphi_{h}(b) = \varphi_{h}(\gamma)$, given bidding monotonicity (\Cref{lem:bid_monotonicity:2} of \Cref{lem:bid_monotonicity}) and \blackref{value_invertibility} (\Cref{lem:high_bid:inverse} of \Cref{lem:high_bid}).

This monopolist $h \in [n]$ is the unique one. Otherwise, (\Cref{def:monopolist} and \blackref{mono_monopoly}) at least two monopolists $h \neq k \in [n]$ tie for the boundary first-order bid $\big\{ b_{h} = b_{k} = \max(\bb) = \gamma \big\}$ with a nonzero probability $> 0$ and, suppose so, BOTH win $\alloc(\bb) = \{h,\, k\}$ almost surely over the random allocation $\alloc$.
However, this is impossible because we are auctioning ONE item.
\Cref{lem:monopolist} follows then.
\end{proof}

Conceivably, we shall reconstruct the $\gamma \leq v < \varphi_{h}(\gamma)$ subpart just for the value distribution $V_{h}$ of the unique monopolist $h$ (if existential) -- We will show this later in \Cref{lem:value_dist}. To reconstruct the $V_{h}$, we introduce (\Cref{def:conditional_value}) the concept of {\em conditional value distributions}.

\begin{definition}[Conditional value distributions]
\label{def:conditional_value}
Regarding the monopolist $h$'s truncated random value $\max(s_{h}^{-1}(b_{h}),\, \gamma)$ for $b_{h} \sim B_{h}$, define the {\em conditional value distribution} $P(v)$ as follows.
\[
    P(v) ~\eqdef~ \Prx_{b_{h} \sim B_{h},\, s_{h}^{-1}} \big[\, \max(s_{h}^{-1}(b_{h}),\, \gamma) \leq v \,\bigmid\, b_{h} \leq \gamma \,\big],\qquad\qquad \forall v \geq 0.
\]
% Let $\tilde{v}_{h} = \max(v_h,\, \gamma)$ be the truncated random value of the monopolist $h$. Define $P$ to be the conditional value distribution for $\tilde{v}_h$ under the condition that $\tilde{v}_h\leq \varphi_{h}(\gamma)$.
In the case of no monopolist $h = \emptyset$, the $P(v) \eqdef \indicator(v \geq \gamma)$ always takes the boundary value $\gamma$.
\end{definition}

\begin{remark}[Conditional value distributions]
\label{rem:conditional_value}
The random value $s_{h}^{-1}(b_{h})$ for $b_{h} \sim B_{h}$ exactly follows the distribution $V_{h}$.
With the help of \Cref{fig:monopolist}, we can see that the CDF $P(v)$ is given by
\begin{align*}
    P(v) ~=~
    \begin{cases}
    0 & \forall v \in [0,\, \gamma) \\
    V_{h}(v) / B_{h}(\gamma) & \forall v \in [\gamma,\, \varphi_{h}(\gamma)) \\
    1 & \forall v \in [\varphi_{h}(\gamma),\, +\infty]
    \end{cases}.
\end{align*}
\end{remark}

In general, the $P(v)$ can be any $[\gamma,\, \varphi_{h}(\gamma)]$-supported distribution.
Given this extra information, we can reconstruct the value distributions $\bV = \{V_{i}\}_{i \in [n]}$ (\Cref{lem:value_dist}) except the unimportant low value $v < \gamma$ parts.
% that have no effect on the auction/optimal {\SocialWelfares}.

\begin{lemma}[Reconstructions]
\label{lem:value_dist}
For each $i \in [n]$, the $v \geq \gamma$ part of the value distribution $V_{i}(v) = \Prx_{v_{i} \sim V_{i}} \big[v_{i} \leq v\big]$ can be reconstructed from the equilibrium bid distributions $\bB$ plus the conditional value distribution $P$ as follows.
\begin{itemize}
    \item $V_{i}(v) = \Prx_{b_{i} \sim B_{i}} \big[(b_{i} \leq \gamma) \vee (\varphi_{i}(b_{i}) \leq v) \big]$ in the case of a non-monopoly bidder $i \in [n] \setminus \{h\}$.

    \item $V_h(v) = P(v) \cdot \Prx_{b_h \sim B_{h}} \big[(b_h \leq \gamma) \vee (\varphi_h(b_h) \leq v) \big]$ in the case of the monopolist $h$.
    % (if existential).
\end{itemize}
\end{lemma}

\begin{figure}[t]
    \centering
    \subfloat[\label{fig:value_dist:monopolist}
    The monopolist $h$]{
    \includegraphics[width = .49\textwidth]
    {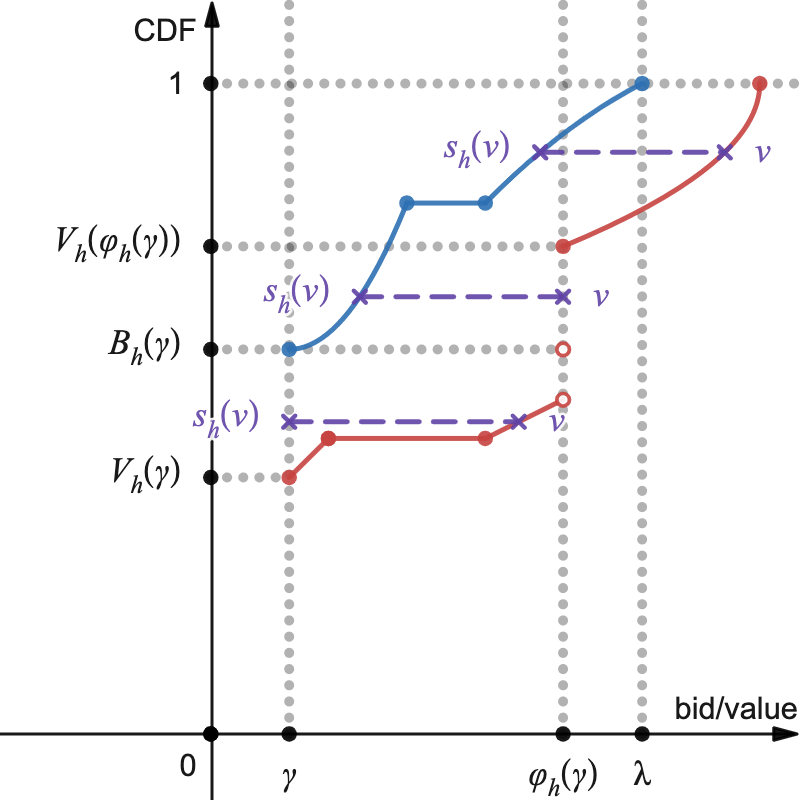}}
    \hfill
    \subfloat[\label{fig:value_dist:non}
    {A non-monopoly bidder $i \in [n] \setminus \{h\}$}]{
    \includegraphics[width = .49\textwidth]
    {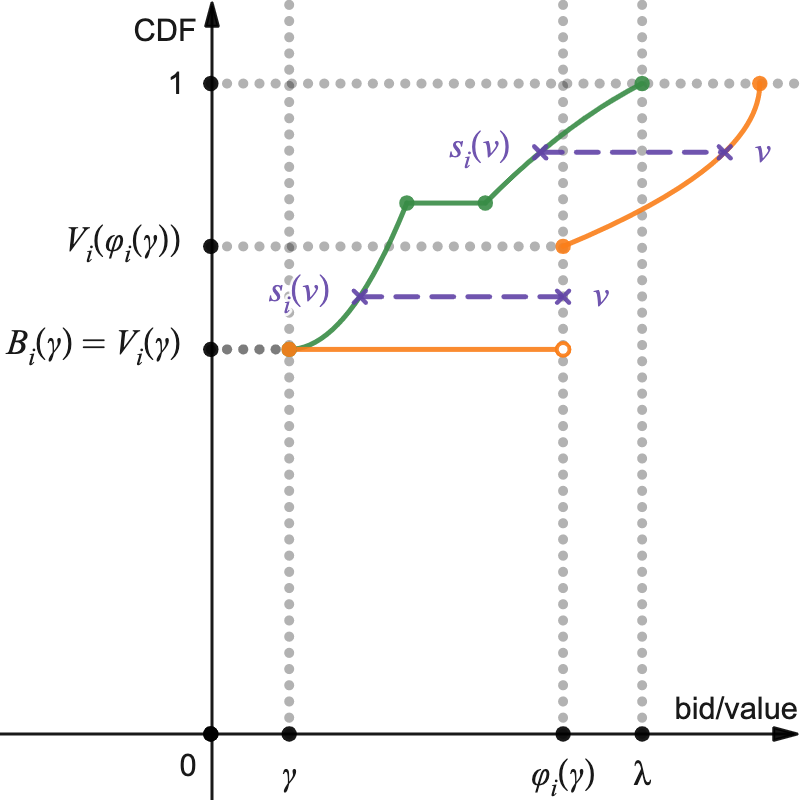}}
    \caption{Diagram of (\Cref{lem:value_dist}) the reconstruction of the value distributions $\bV$.
    \label{fig:value_dist}}
\end{figure}

\begin{proof}
See \Cref{fig:value_dist} for a visual aid.
% , namely \Cref{fig:value_dist:monopolist} for the monopolist $h$ and \Cref{fig:value_dist:non} for a non-monopoly bidder $i \in [n] \setminus \{h\}$.
% For each non-monopoly bidder $i \in [n] \setminus \{h\}$, (\Cref{lem:dichotomy,def:monopolist}) a normal bid $b_{i} > \gamma$ induces a normal value $v_{i} = s_{i}^{-1}(b_{i}) = \varphi_{i}(b_{i}) > \gamma$ and vice versa.
% has a zero probability measure on the left-open right-closed interval $(\gamma,\varphi_{i}(\gamma)]$
% We derive the claimed ``non-monopoly'' formula $V_{i}(v)$ 
% , the WHOLE normal bid $> \gamma$ part of the value distribution $V_{i}$ can be reconstructed by naively generalizing \Cref{cor:high_bid}:
Regarding a non-monopoly bidder $i \in [n] \setminus \{h\}$, (\Cref{lem:dichotomy,def:monopolist}; cf.\ \Cref{fig:value_dist:non}) a normal bid $> \gamma$ induces a normal value $> \gamma$ and vice versa.
By \blackref{value_monotonicity}, such a normal value $v_{i} = \varphi_{i}(b_{i})$ for $b_{i} > \gamma$ is at least $v_{i} \geq \varphi_{i}(\gamma)$. Hence, on the OPEN interval $v \in (\gamma,\, \varphi_{i}(\gamma))$, the value distribution $V_{i}(v)$ has no density and keeps a constant function $V_{i}(v) = B_{i}(\gamma)$. Given these, the reconstruction in \Cref{cor:high_bid} can be generalized naively: For $v \in [\gamma,\, \varphi_{i}(\gamma))$, we have $V_{i}(v) = B_{i}(\gamma) = \Prx_{b_{i} \sim B_{i}}\big[ b_{i} \le \gamma \big] = \Prx_{b_{i} \sim B_{i}}\big[(b_{i} \le \gamma) \vee (\varphi_{i}(b_{i}) \leq v \big]$.

% \begin{align*}
%     V_{i}(v)
%     ~=~ B_{i}(\gamma)
%     ~=~ \Prx_{b_{i} \sim B_{i}}\big[ b_{i} \le \gamma \big]
%     ~=~ \Prx_{b_{i} \sim B_{i}}\big[(b_{i} \le \gamma) \vee (\varphi_{i}(b_{i}) \leq v \big].
% \end{align*}
% \begin{align*}
%     V_{i}(v)
%     ~=~ V_{i}(\varphi_{i}(\gamma))
%     ~=~ \Prx_{b_{i} \sim B_{i}}\big[(b_{i} \le \gamma) \vee (\varphi_{i}(b_{i}) \leq \varphi_{i}(\gamma)) \big]
%     ~=~ \Prx_{b_{i} \sim B_{i}}\big[(b_{i} \le \gamma) \vee (\varphi_{i}(b_{i}) \leq v \big].
% \end{align*}
% This finishes the proof for the non-monopoly bidders $i \in [n] \setminus \{h\}$.
% \[
%     P(v) \cdot \Prx_{b_h \sim B_{h}} \big[(b_h \leq \gamma) \vee (\varphi_h(b_h) \leq v) \big]
%     ~=~ P(v) \cdot \Prx_{b_h \sim B_{h}} \big[ b_h \leq \gamma \big]
%     ~=~ P(v) \cdot B_{h}(\gamma)
%     ~=~ V_{h}(v).
% \]

Further, we can check the claimed ``monopolist'' formula $V_{h}(v)$ for $v \in [\gamma,\, \varphi_{i}(\gamma))$ (cf.\ \Cref{fig:value_dist:monopolist}):
$P(v) \cdot \Prx_{b_h \sim B_{h}} \big[(b_h \leq \gamma) \vee (\varphi_h(b_h) \leq v) \big]
= P(v) \cdot \Prx_{b_h \sim B_{h}} \big[ b_h \leq \gamma \big]
= P(v) \cdot B_{h}(\gamma) = V_{h}(v)$, where the first step again leverages \blackref{value_monotonicity}.
This finishes the proof.
% For the monopolist $h$, since its conditional value distribution on $(\gamma,\varphi_{i}(\gamma)]$ is $P$. We get the formula by simply multiplying the CDF $P$ to the same extension as non-monopoly bidders.
\end{proof}

\begin{comment}

Regarding the value distribution $V_{i}$, the $v \geq \varphi_{i}(\gamma)$ part is 

Together with the monotonicity

satisfies $V_{i}(\gamma) = B_{i}(\gamma) = V_{i}(\varphi_{i}(\gamma))$.

For non-monopoly bidder $i$, since the value $v_{i}$ has no probability density on $(\gamma,\varphi_{i}(\gamma)]$. The distribution for $v>\varphi_{i}(\gamma)$ in \Cref{cor:high_bid} can be directly extended to the range $v\geq \gamma$, namely the CDF value at $V_{i}(v)$ for any $v\in [\gamma,\varphi_{i}(\gamma))$ should be equal to
\[
    V_{i}(\varphi_{i}(\gamma))= \Prx_{b_{i}}\big[b_{i} \leq \gamma \vee \varphi_{i}(b_{i})\leq \varphi_{i}(\gamma) \big] = \Prx_{b_{i}}\big[b_{i} \leq \gamma \big].
\]
Interestingly, the same formula as is already satisfies this extension since for any  $v\in [\gamma,\varphi_{i}(\gamma))$,
\[
    V_{i}(v)= \Prx_{b_{i}}\big[b_{i} \leq \gamma \vee \varphi_{i}(b_{i})\leq v \big] = \Prx_{b_{i}}\big[b_{i} \leq \gamma \big].
\]
This finishes the proof for non-monopoly bidders.

\end{comment}

% Conditioned on the boundary first-order bid $\big\{ \max(\bb) = \gamma \big\}$,
% the conditional value distribution $P$ captures
% both of the allocated bidder's value $s_{\alloc(\bb)}^{-1}(\gamma)$ and the optimal {\SocialWelfare} $\max(\bs^{-1}(\bb))$ follow the conditional value distribution $P$.
% , because of the tie-breaking rule for the monopolist $h$.

\Cref{lem:conditional_value} will be useful for formulating the expected auction/optimal {\SocialWelfares}.

\begin{lemma}[{\SocialWelfares}]
\label{lem:conditional_value}
Conditioned on the boundary first-order bid $\big\{ \max(\bb) = \gamma \big\}$, each of the following exactly follows the conditional value distribution $P$:
\begin{itemize}
    \item The conditional auction {\SocialWelfare} $\big\{ s_{\alloc(\bb)}^{-1}(\gamma) \bigmid \max(\bb) = \gamma \big\}$.
    
    \item The conditional optimal {\SocialWelfare} $\big\{ \max(\bs^{-1}(\bb)) \bigmid \max(\bb) = \gamma \big\}$.
    
    % \blue{iff one of the three equivalent conditions holds: (i)~It satisfies \blackref{re:twin_ceiling}. (ii)~Its potential is zero $\Psi(H^{\uparrow} \otimes \bB \otimes L) = 0$. (iii)~}
\end{itemize}
\end{lemma}

\begin{proof}
The random events $\big\{ \max(\bb) = \gamma \big\}$ and $\big\{ \max(\bb) \leq \gamma \big\}$ are identical, regarding the infimum first-order bid $\gamma = \inf(\supp(\calB))$.
Conditioned on this, (i)~a non-monopoly bidder $i \in [n] \setminus \{h\}$ has a low/boundary bid $b_{i} \leq \gamma$ and (\Cref{lem:dichotomy,def:monopolist}) a low/boundary value $s_{i}^{-1}(b_{i}) \leq \gamma$; while
(ii)~the monopolist $h$ has EITHER a low bid $b_{h} < \gamma$ and a low/boundary value $s_{h}^{-1}(b_{h}) \leq \gamma$, OR the boundary bid $b_{h} = \gamma$ and (\Cref{lem:monopolist}) a boundary/normal value $s_{h}^{-1}(b_{h}) \in [\gamma,\, \varphi_{h}(\gamma)]$.

The allocated bidder $\alloc(\bb) \in [n]$ always take the boundary bid $b_{\alloc(\bb)} = \gamma$ and, to make the utility nonnegative, a boundary/normal value $s_{\alloc(\bb)}^{-1}(b_{\alloc(\bb)}) = s_{\alloc(\bb)}^{-1}(\gamma) \geq \gamma$; the case of a normal value $> \gamma$ occurs only if the monopolist is allocated $\alloc(\bb) = h$.
Hence, the conditional optimal {\SocialWelfare} is identically distributed as each of the following:
\begin{align*}
    \big\{ \max(\bs^{-1}(\bb)) \bigmid \max(\bb) = \gamma \big\}
    & ~\overset{\tt d}{=}~
    \big\{ \max(\bs^{-1}(\bb)) \bigmid \max(\bb) \leq \gamma \big\}
    && \gamma = \inf(\supp(\calB)) \\
    & ~\overset{\tt d}{=}~
    \big\{ \max(\bs^{-1}(\bb),\, \gamma) \bigmid \max(\bb) \leq \gamma \big\}
    && s_{\alloc(\bb)}^{-1}(b_{\alloc(\bb)}) \geq \gamma \\
    & ~\overset{\tt d}{=}~
    \big\{ \max(s_{h}^{-1}(b_{h}),\, \gamma) \bigmid \max(\bb) \leq \gamma \big\}
    && \mbox{$s_{i}^{-1}(b_{i}) \leq \gamma$ for $i \neq h$} \\
    & ~\overset{\tt d}{=}~
    \big\{ \max(s_{h}^{-1}(b_{h}),\, \gamma) \bigmid b_{h} \leq \gamma \big\}.
    && \mbox{independence}
\end{align*}
That is, the conditional optimal {\SocialWelfare} follows (\Cref{def:conditional_value}) the distribution $P$.
Further, since the monopolist $h$ (as the only possible bidder that has a normal value $> \gamma$) always wins the tiebreaker $\big\{ b_{h} = \max(\bb) = \gamma \big\}$ (\Cref{lem:monopolist}), the conditional auction/optimal {\SocialWelfares} are identically distributed $\big\{ s_{\alloc(\bb)}^{-1}(\gamma) \bigmid \max(\bb) = \gamma \big\} \overset{\tt d}{=} \big\{ \max(\bs^{-1}(\bb)) \bigmid \max(\bb) = \gamma \big\}$, which again follow the distribution $P$.
This finishes the proof.
% \blue{Conditional on $\big\{ \max(\bb) = \gamma \big\}$, the value of a non-monopoly bidder $i$ is at most $\gamma$ since it has no probability density on $(\gamma,\varphi_{i}(\gamma)]$ and the bid would be larger than $\gamma$ if the value is larger than $\varphi_{i}(\gamma)$. On the other hand, the allocated bidder $\alloc(\bb)$'s value $v_{\alloc(\bb)} $ and as a result the first order value $\max(\bv))$ is at least  $\gamma$ because otherwise the allocated bidder $\alloc(\bb)$ will have negative utility as the payment is $\gamma$.
% We also know that when $v_h\in (\gamma,\varphi_{i}(\gamma)]$, the monopolist bids $\gamma$. As a result, when $v_h\in (\gamma,\varphi_{i}(\gamma)]$, the monopolist is the allocated bidder $\alloc(\bb)$ due to the tie-breaking rule and also have the largest value since others' values are at most $\gamma$. For the remaining probability $v_h\leq \gamma$, the allocated bidder $\alloc(\bb)$'s value $v_{\alloc(\bb)} $ and the first order value $\max(\bv)$ is exactly $\gamma$. Therefore, this exact coincides with the definition of distribution $P$.}
\end{proof}

\subsection{Reformulation for the {\PriceofAnarchy} problem}
\label{subsec:reformulation}

To address the {\PriceofAnarchy} problem, we shall formulate the expected auction/optimal {\SocialWelfares} $\FPA(\bV,\, \bs,\, \alloc)$ and $\OPT(\bV,\, \bs,\, \alloc)$.
Indeed, these can be written in terms of the equilibrium bid distributions $\bB$ plus the conditional value distribution $P$ (\Cref{def:conditional_value}).

First, \Cref{lem:auction_welfare} formulates the expected auction {\SocialWelfare} $\FPA(\bV,\, \bs,\, \alloc)$.

\begin{lemma}[Auction {\SocialWelfare}]
\label{lem:auction_welfare}
The expected auction {\SocialWelfare} $\FPA(\bV,\, \bs,\, \alloc)$ can be formulated based on the conditional value distribution and the equilibrium bid distributions $(P,\, \bB)$:
\begin{align*}
    \FPA(\bV,\, \bs,\, \alloc)
    ~=~ \FPA(P,\, \bB)
    ~=~ \Ex[ P ] \cdot \calB(\gamma) + \sum_{i \in [n]} \bigg(\int_{\gamma}^{\lambda} \varphi_{i}(b) \cdot \frac{B'_{i}(b)}{B_{i}(b)} \cdot \calB(b) \cdot \d b\bigg).
\end{align*}
where the first-order bid distribution $\calB(b) = \prod_{i \in [n]} B_{i}(b)$ and the bid-to-value mappings $\varphi_{i}(b)$ can be computed from the bid distributions $\bB = \{B_{i}\}_{i \in [n]}$ (\Cref{def:mapping}).
\end{lemma}

\begin{proof}
In any realization, the outcome auction {\SocialWelfare} is the allocated bidder's value $v_{\alloc(\bs(\bv))}$. Over the randomness of the values $\bv = (v_{i})_{i \in [n]} \sim \bV$, the equilibrium strategies $\bs = \{s_{i}\}_{i \in [n]}$, and the allocation rule $\alloc(\bs(\bv))$, the expected auction {\SocialWelfare} $\FPA$ is given by
\begin{align}
    \FPA
    ~=~ \Ex_{\bv,\, \bs,\, \alloc} \big[\, v_{\alloc(\bs(\bv))} \,\big]
    & ~=~ \Ex_{\bv,\, \bs,\, \alloc} \big[\, v_{\alloc(\bs(\bv))} \cdot \indicator(\max(\bs(\bv)) = \gamma) \,\big]
    \label{eq:fpa_boundary}\tag{B} \\
    & \phantom{~=~} + \Ex_{\bv,\, \bs,\, \alloc} \big[\, v_{\alloc(\bs(\bv))} \cdot \indicator(\max(\bs(\bv)) > \gamma) \,\big]. \hspace{0.47cm}
    \label{eq:fpa_high}\tag{N}
\end{align}
% The disjoint events $\big\{\max(\bs(\bv)) = \gamma\big\}$ and $\big\{\max(\bs(\bv)) > \gamma\big\}$ are the only two possibilities, because $\gamma = \inf(\supp(\calB))$ is the infimum first-order bid $\max(\bs(\bv)) \sim \calB$.
Here we used the fact that the first-order bid $\max(\bs(\bv)) \sim \calB$ is at least $\gamma = \inf(\supp(\calB))$.

For Term~\eqref{eq:fpa_boundary}, we deduce that
\begin{align*}
    \text{Term~\eqref{eq:fpa_boundary}}
    & ~=~ \Ex_{\bv,\, \bs,\, \alloc} \big[\, v_{\alloc(\bs(\bv))} \cdot \indicator(\max(\bs(\bv)) = \gamma) \,\big] \\
    & ~=~ \Ex_{\bv,\, \bs, \, \alloc} \big[\, v_{\alloc(\bs(\bv))} \,\bigmid\, \max(\bs(\bv)) = \gamma \,\big] \cdot \Prx_{\bv,\, \bs} \big[\, \max(\bs(\bv)) = \gamma \,\big] \\
    & ~=~ \Ex[ P ] \cdot \calB(\gamma).
\end{align*}
The last equality uses the fact that conditioned on the boundary first-order bid $\big\{\max(\bs(\bv)) = \gamma\big\}$, the allocated bidder $\alloc(\bs(\bv))$'s value  follows the distribution $P$ (\Cref{lem:conditional_value}).

For Term~\eqref{eq:fpa_high}, we deduce that
\begin{align}
    \text{Term~\eqref{eq:fpa_high}}
    & ~=~ \Ex_{\bv,\, \bs,\, \alloc} \big[\, v_{\alloc(\bs(\bv))} \cdot \indicator(\max(\bs(\bv)) > \gamma) \,\big] \phantom{\bigg.}
    \nonumber \\
    & ~=~ \Ex_{\bb,\, \bs^{-1},\, \alloc} \big[\, s_{\alloc(\bb)}^{-1}(b_{\alloc(\bb)}) \cdot \indicator(\max(\bb) > \gamma) \,\big] \phantom{\bigg.}
    \label{eq:fpa:h1}\tag{N1} \\
    & ~=~ \sum_{i \in [n]} \Ex_{\bb,\, s_{i}^{-1},\, \alloc} \big[\, s_{i}^{-1}(b_{i}) \cdot \indicator(\alloc(\bb) = i) \cdot \indicator(b_{i} > \gamma) \,\big] \phantom{\bigg.} \hspace{1.63cm}
    \label{eq:fpa:h2}\tag{N2} \\
    % & ~=~ \sum_{i \in [n]} \Ex_{\bb,\, \alloc} \big[\, \varphi_{i}(b_{i}) \cdot \indicator(\alloc(\bb) = i) \cdot \indicator(b_{i} > \gamma) \,\big] \phantom{\bigg.} \\
    & ~=~ \sum_{i \in [n]} \Ex_{b_{i}} \big[\, \varphi_{i}(b_{i}) \cdot \alloc_{i}(b_{i}) \cdot \indicator(b_{i} > \gamma) \,\big] \phantom{\bigg.}
    \label{eq:fpa:h3}\tag{N3} \\
    & ~=~ \sum_{i \in [n]} \Ex_{b_{i}} \big[\, \varphi_{i}(b_{i}) \cdot \calB_{-i}(b_{i}) \cdot \indicator(b_{i} > \gamma) \,\big] \phantom{\bigg.}
    \label{eq:fpa:h4}\tag{N4} \\
    & ~=~ \sum_{i \in [n]} \bigg(\int_{\gamma}^{\lambda} \varphi_{i}(b) \cdot \calB_{-i}(b) \cdot B'_{i}(b) \cdot \d b\bigg)
    \nonumber \\
    & ~=~ \sum_{i \in [n]} \bigg(\int_{\gamma}^{\lambda} \varphi_{i}(b) \cdot \frac{B'_{i}(b)}{B_{i}(b)} \cdot \calB(b) \cdot \d b\bigg).
    \label{eq:fpa:h5}\tag{N5}
\end{align}
\eqref{eq:fpa:h1}: Redenote $b_{i} = s_{i}(v_{i}) \sim B_{i}$ for $i \in [n]$. Recall \Cref{def:inverse} that each random inverse $s_{i}^{-1}(b_{i})$ for $b_{i} \sim B_{i}$ is identically distributed as the random value $v_{i} \sim V_{i}$. \\
\eqref{eq:fpa:h2}: Divide the event $\big\{\max(\bb) > \gamma\big\}$ into subevents $\big\{\alloc(\bb) = i ~\text{and}~ b_{i} = \max(\bb) > \gamma \big\}$ for $i \in [n]$. \\
\eqref{eq:fpa:h3}: \Cref{lem:high_bid} (\Cref{lem:high_bid:inverse}) that $s_{i}^{-1}(b) = \varphi_{i}(b)$ for any normal bid $b \in \supp_{> \gamma}(B_{i})$
% (except a countable zero-measure set $b \in \mathbb{D} \subseteq (\lambda,\, \gamma]$).
and \Cref{def:interim_utility} that $\alloc_{i}(b) = \Prx_{\bb_{-i},\, \alloc} \big[ \alloc(\bb_{-i},\, b) = i \big] = \Ex_{\bb_{-i},\, \alloc} \big[ \indicator(\alloc(\bb_{-i},\, b) = i) \big]$. \\
\eqref{eq:fpa:h4}: \Cref{cor:allocation} that $\alloc_{i}(b) = \calB_{-i}(b)$ for any normal bid $b > \gamma$. \\
\eqref{eq:fpa:h5}: The first-order bid distribution $\calB(b) = B_{i}(b) \cdot \calB_{-i}(b)$.

\vspace{.1in}
Combining Terms~\eqref{eq:fpa_boundary} and \eqref{eq:fpa_high} together finishes the proof of \Cref{lem:auction_welfare}.
\end{proof}

Moreover, \Cref{lem:optimal_welfare} formulates the expected optimal {\SocialWelfare} $\OPT(\bV,\, \bs,\, \alloc)$.

\begin{lemma}[Optimal {\SocialWelfare}]
\label{lem:optimal_welfare}
The expected optimal {\SocialWelfare} $\OPT(\bV,\, \bs,\, \alloc)$ can be formulated based on the conditional value distribution and the equilibrium bid distributions $(P,\, \bB)$:
\begin{align*}
    \OPT(\bV,\, \bs,\, \alloc)
    ~=~ \OPT(P,\, \bB)
    ~=~ \gamma + \int_{\gamma}^{+\infty} \Big(1 - P(v) \cdot \prod_{i \in [n]} \Prx_{b_{i}}\big[b_{i} \leq \gamma \vee \varphi_{i}(b_{i})\leq v \big]\Big) \cdot \d v.
\end{align*}
%where the semi first-order value distribution $\calV(v) \eqdef \prod_{i \in [n]} \max\big(V_{i}(v),\, B_{i}(\gamma)\big)$ can be reconstructed from the $\bB = \{B_{i}\}_{i \in [n]}$ (\Cref{cor:high_bid}).
\end{lemma}

%\yj{to continue}

\begin{proof}
This proof simply follows from \Cref{lem:value_dist}, i.e., we essentially reconstruct the value distributions $\bV = \{V_{i}\}_{i \in [n]}$ from the tuple $(P,\, \bB)$.

In any realization, the outcome optimal {\SocialWelfare} is the first-order value $\max(\bv)$. Over the randomness of the values $\bv = (v_{i})_{i \in [n]} \sim \bV$, the expected optimal {\SocialWelfare} $\OPT$ is given by
\begin{align}
    \OPT
    & ~=~ \Ex_{\bv} \big[\, \max(\bv) \,\big]
    % \nonumber \\
    ~=~ \Ex_{\bv} \big[\, \max(\bv,\, \gamma) \,\big] \phantom{\bigg.}
    \label{eq:optimal_welfare:1}\tag{O1} \\
    & ~=~ \int_{0}^{+\infty} \Big(1 - \prod_{i \in [n]} V_{i}(v) \cdot \indicator(v \geq \gamma)\Big) \cdot \d v
    \label{eq:optimal_welfare:2}\tag{O2} \\
    & ~=~ \int_{0}^{+\infty} \Big(1 - \Big(1 - P(v) \cdot \prod_{i \in [n]} \Prx_{b_{i}}\big[b_{i} \leq \gamma \vee \varphi_{i}(b_{i})\leq v \big]\cdot \indicator(v \geq \gamma) \Big) \cdot \d v
    \label{eq:optimal_welfare:3}\tag{O3} \\
%    & ~=~ \int_{0}^{+\infty} \Big(1 - P(v) \cdot \calV(v) \cdot \indicator(v \geq \gamma)\Big) \cdot \d v
%    \label{eq:optimal_welfare:4}\tag{O4} \\
    & ~=~ \gamma + \int_{\gamma}^{+\infty} \Big(1 - P(v) \cdot \prod_{i \in [n]} \Prx_{b_{i}}\big[b_{i} \leq \gamma \vee \varphi_{i}(b_{i})\leq v \big]\Big) \cdot \d v.
    \nonumber
\end{align}
\eqref{eq:optimal_welfare:1}: The first-order value $\max(\bv)$ for $\bv \sim \bV$ is always boundary/normal $\geq \gamma$. \\
%\footnote{The first-order bid $\max(\bs(\bv)) \sim \calB$ is always boundary/high $\geq \gamma = \inf(\supp(\calB))$. (\Cref{lem:tiebreak:no_overbid} of \Cref{lem:tiebreak}) In the ``boundary'' case $\{\max(\bs(\bv)) = \gamma\}$, there exists at least one boundary/normal value $v_{i} \geq \gamma$. (\Cref{lem:dichotomy}) In the ``normal'' case $\{\max(\bs(\bv)) > \gamma\}$, each normal bid $s_{i}(v_{i}) > \gamma$ ensures a higher value $v_{i} > s_{i}(v_{i}) > \gamma$.} \\
\eqref{eq:optimal_welfare:2}: The expectation of a {\em nonnegative} distribution $F$ is given by $\E[F] = \int_{0}^{+\infty} \big(1 - F(v)\big) \cdot \d v$. \\
\eqref{eq:optimal_welfare:3}: Apply \Cref{lem:value_dist}. It also holds in case that no monopolist exists since, suppose so, we have $P(v)=1$ for $v \geq \gamma$.
This finishes the proof of \Cref{lem:optimal_welfare}.
\end{proof}

The above $(P,\, \bB)$-based {\SocialWelfare} formulas $\FPA(P,\, \bB)$ and $\OPT(P,\, \bB)$, together with the characterizations of $P$ and $\bB$, are the foundation of the whole paper. From all these discussions, we notice that the below-$\gamma$ parts of bid distributions $\bB$ and value distributions $\bV$ are less important. To stress this observation, we have the key definition of {\em valid instances} $(P,\, \bB) \in \Bvalid$.

\begin{definition}[Valid instances]
\label{def:valid}
A {\em valid} instance $(P,\, \bB) \in \Bvalid$, where $\bB = \{B_{i}\}_{i \in [n]}$ and $n \geq 2$, is given by $n + 1$ independent distributions and satisfies the following.
\begin{itemize}
    \item The bid distributions $\bB = \{B_{i}\}_{i \in [n]}$ have a common {\em infimum bid} $\inf(\supp(B_{i})) = \gamma \in [0,\, +\infty)$ and a universal {\em supremum bid} $\max_{i \in [n]} \sup(\supp(B_{i})) = \lambda \in [\gamma,\, +\infty]$.
    
    \item The bid CDF's $B_{i}(b)$ for $i \in [n]$ are {\em continuous} functions on $b \in [\gamma,\, \lambda]$, while the competing bid CDF's $\calB_{-i}(b) = \prod_{k \in [n] \setminus \{i\}} B_{k}(b)$ for $i \in[n]$ and the first-order bid CDF $\calB(b) = \prod_{i \in [n]} B_{i}(b)$ are {\em continuous} and {\em strictly increasing} functions on $b \in [\gamma,\, \lambda]$.
    
    \item {\bf monotonicity:}
    The bid-to-value mappings $\varphi_{i}(b) \eqdef b + \frac{\calB_{-i}(b)}{\calB'_{-i}(b)} = b + \big(\sum_{k \in [n] \setminus \{i\}} \frac{B'_{k}(b)}{B_{k}(b)}\big)^{-1}$ for $i \in [n]$ are {\em weakly increasing} functions on $b \in [\gamma,\, \lambda]$.

    \item {\bf monopoly:}
    The index-$1$ bidder $B_{1}$ without loss of generality is the unique {\em monopolist}.\footnote{This accommodates the ``no monopolist'' case of \Cref{lem:monopolist}, by taking $P$ as a deterministic value of $\gamma$.} I.e., this bidder always wins $\alloc(\bb) \equiv 1$ in the all-infimum tiebreaker $\{\bb = \gamma^{\otimes n}\}$.
    
    \item {\bf boundedness:}
    Under the infimum bid $\gamma$, the monopolist $B_{1}$'s conditional value $s_{1}^{-1}(\gamma) \sim P$ is bounded as $\supp(P) \subseteq [\gamma,\, \varphi_{1}(\gamma)]$, while the other bidders' conditional values are truthful, $s_{i}^{-1}(\gamma) \equiv \gamma$ for $i \in [2: n]$.
\end{itemize}
\end{definition}

The valid instances $(P,\, \bB) \in \Bvalid$ capture nearly all {\BayesNashEquilibria} $\big\{ \bs \in \bbBNE(\bV,\, \alloc) \big\}$. Essentially, there is a bijection between these two spaces (without considering some less important details about {\BayesNashEquilibria}), especially in a sense of the {\SocialWelfare} invariant.
This is formalized into \Cref{thm:valid}.

% . In particular, they are equivalent in terms of {\PoA} in the following strong sense.

\begin{theorem}[Valid instances]
\label{thm:valid}
(I)~For any valid instance $(P,\, \bB) \in \Bvalid$, there is some equilibrium $\bs \in \bbBNE(\bV,\, \alloc) \neq \emptyset$, for some value distribution $\bV = \{V_{i}\}_{i \in [n]}$, whose conditional value distribution and bid distribution are exactly the $(P,\, \bB)$.
(II)~For any value distribution $\bV = \{V_{i}\}_{i \in [n]}$ and any equilibrium $\bs \in \bbBNE(\bV,\, \alloc) \neq \emptyset$,
there is some valid instance $(P,\, \bB) \in \Bvalid$ that yields the same auction/optimal {\SocialWelfares} $\FPA(P,\, \bB) = \FPA(\bV,\, \bs,\, \alloc)$ and $\OPT(P,\, \bB) = \OPT(\bV,\, \bs,\, \alloc)$.
%Therefore, we have {\SocialWelfares} $\FPA(P,\, \bB) = \FPA(\bV,\, \bs,\, \alloc)$ and $\OPT(P,\, \bB) = \OPT(\bV,\, \bs,\, \alloc)$.
\end{theorem}

\begin{figure}[t]
    \centering
    \subfloat[\label{fig:valid:monopolist}
    The monopolist $h = B_{1}$]{
    \includegraphics[width = .49\textwidth]
    {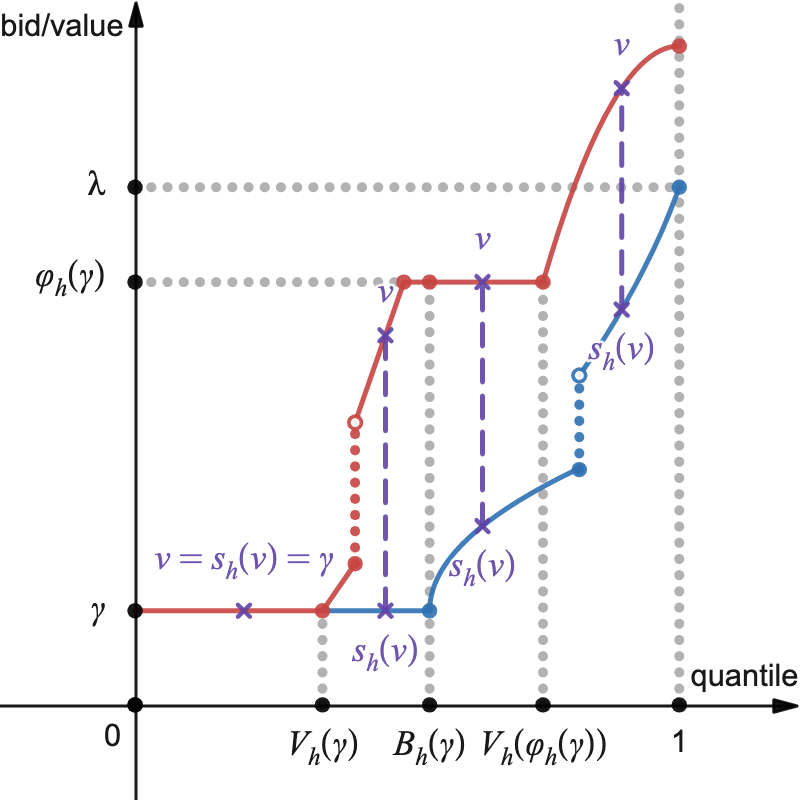}}
    \hfill
    \subfloat[\label{fig:valid:non}
    {A non-monopoly bidder $B_{i}$ for $i \in [2: n]$}]{
    \includegraphics[width = .49\textwidth]
    {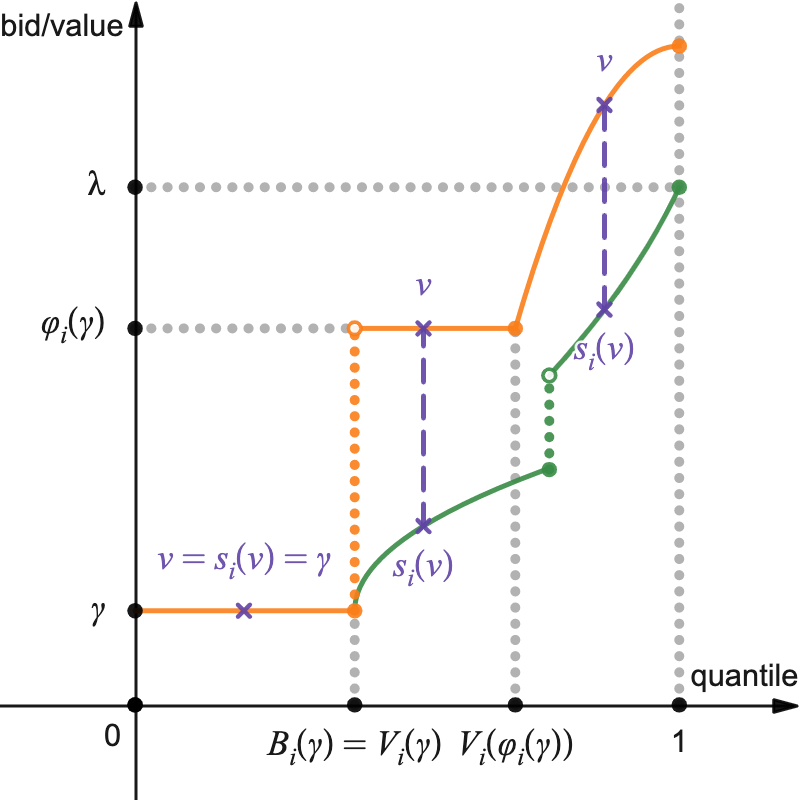}}
    \caption{Diagram of valid instances $(P,\, \bB) \in \Bvalid$ and \Cref{thm:valid}.
    \label{fig:valid}}
\end{figure}

\begin{proof}
See \Cref{fig:valid} for a visual aid. We prove Items~(I) and (II) separately.

(I)~Given a valid instance $(P,\, \bB) \in \Bvalid$, we can construct
the boundary/normal value $v \geq \gamma$ part of some value distribution $\bV = \{V_{i}\}_{i \in [n]}$, following \Cref{lem:value_dist}.
For the low value $v < \gamma$ part, we simply let $V_{i}(v) \equiv 0$ on $v < \gamma$, namely putting all the undecided probabilities to the boundary value $\gamma$. Then the value distribution $\bV = \{V_{i}\}_{i \in [n]}$ is well defined. 

Since we zero-out everything below $\gamma$, the equivalence in \Cref{quantiles} can be extend to the entire quantile space $q \in [0,1]$. Hence, we can construct the equilibrium by identifying the quantiles for value distributions and bid distributions. 
Consider the quantile bid/value functions $\{B_{i}^{-1}\}_{i \in [n]}$ and $\{V_{i}^{-1}\}_{i \in [n]}$.
As \Cref{fig:valid} shows, we construct the strategy profile $\bs = \{s_{i}\}_{i \in [n]}$ as follows: For $i \in [n]$ and $v \in \supp(V_{i})$, let $q \sim U[0,\, 1]$ be a uniform random quantile, then
\[
    \Prx_{s_{i}}\big[\, s_{i}(v) = b \,\big] ~=~ \Prx_{q \,\sim\, U[0,\, 1]} \big[\, B_{i}^{-1}(q) = b \,\bigmid\, V_{i}^{-1}(q) = v \,\big],
    \qquad\qquad \forall b \in [\gamma,\, \lambda].
\]
We conclude Item~(I) with verifying the equilibrium conditions (\Cref{def:bne_formal}).

\setcounter{fact}{0}

\begin{fact}
\label{fact:valid:equilibrium}
$u_{i}(v,\, s_{i}(v)) \geq u_{i}(v,\, b^{*})$ almost surely over the random strategy $s_{i}$, for each bidder $i \in [n]$, any value $v \in \supp(V_{i})$, and any deviation bid $b^{*} \geq 0$.
\end{fact}

\begin{proof}
% We consider the monopolist $B_{1}$ and the other non-monopoly bidders $i \in [2: n]$ respectively (\Cref{def:valid}).
First, the monopolist $B_{1}$ (\Cref{cor:allocation,def:monopolist,lem:monopolist}) has the allocation $\alloc_{1}(b) = \calB_{-1}(b)$ for a boundary/normal bid $b \geq \gamma$ and $\alloc_{1}(b) = 0$ for a low $b \in [0,\, \gamma)$. \\
{\bf (i)}~A Normal Bid $s_{1}(v) \in (\gamma,\, \lambda]${\bf .}
The underlying value $v = \varphi_{1}(s_{1}(v)) > s_{1}(v) > \gamma$, by construction and \Cref{lem:value_dist,lem:high_bid}.
Accordingly, the current allocation/utility are positive $u_{1}(v,\, s_{1}(v)) = (v - s_{1}(v)) \cdot \calB_{-1}(s_{1}(v)) > 0$.
Then a low deviation bid $b^{*} < \gamma$ is suboptimal, because the deviated allocation/utility are zero $u_{1}(v,\, b^{*}) = (v - b^{*}) \cdot \alloc_{1}(b^{*}) = 0$.
Moreover, a boundary/normal deviation bid $b^{*} \geq \gamma$ is suboptimal -- The current bid $s_{1}(v) > \gamma$ maximizes the utility formula $u_{1}(v,\, b^{*}) = (v - b^{*}) \cdot \calB_{-1}(b^{*})$, because the partial derivative $\frac{\partial u_{1}}{\partial b^{*}} = \big(v - \varphi_{1}(b^{*})\big) \cdot \calB'_{-1}(b^{*})$, the underlying value $v = \varphi_{1}(s_{1}(v))$, and the mapping $\varphi_{1}(b^{*})$ is increasing (\Cref{lem:high_bid:monotone} of \Cref{lem:high_bid}). \\
{\bf (ii)}~A Boundary Bid $s_{1}(v) = \gamma${\bf .}
The underlying value $v = s_{1}^{-1}(\gamma)$ exactly follows the conditional value distribution $P$ and ranges within $\supp(P) \subseteq [\gamma,\, \varphi_{1}(\gamma)]$, by construction and \Cref{def:conditional_value,def:valid}.
Reusing the above arguments (i.e., the utility formula $u_{1}(v,\, b^{*})$ is decreasing for $b^{*} \geq \gamma$),
we can easily see that the current bid $s_{1}(v) = \gamma$ is optimal.
In sum, the monopolist $B_{1}$ meets the equilibrium conditions.

\vspace{.1in}
Each non-monopoly bidder $i \in [2: n]$ (\Cref{cor:allocation,def:monopolist,lem:monopolist}) has the allocation $\alloc_{i}(b) = \calB_{-i}(b)$ for a normal bid $b > \gamma$ and $\alloc_{i}(b) = 0$ for a low/boundary bid $b \in [0,\, \gamma]$. \\
{\bf (i)}~A Normal Bid $s_{i}(v) \in (\gamma,\, \lambda]${\bf .}
Reusing the above arguments for the monopolist $B_{1}$,
we can easily see that the current normal bid $s_{i}(v) \in (\gamma,\, \lambda]$ is optimal. \\
{\bf (ii)}~A Boundary Bid $s_{i}(v) = \gamma${\bf .}
The underlying value must be the boundary value $v = s_{i}^{-1}(\gamma) \equiv \gamma$, by construction and \Cref{def:conditional_value,def:valid}.
Obviously, the current {\em zero} utility $u_{i}(v,\, s_{i}(v)) = 0$ and the current {\em boundary} bid $s_{i}(v) = \gamma$ are optimal.
In sum, each non-monopoly bidder $i \in [2: n]$ meets the equilibrium conditions.
{\bf \Cref{fact:valid:equilibrium}} follows then.
\end{proof}

(II)~Given an equilibrium $\bs \in \bbBNE(\bV,\, \alloc) \neq \emptyset$, as before, consider the bid distribution $\bB^{(\bs)}$, the conditional value distribution $P^{(\bs)}$, and the infimum/supremum first-order bids $\gamma = \inf(\supp(\calB^{(\bs)}))$ and $\lambda = \sup(\supp(\calB^{(\bs)}))$.
Following \Cref{lem:bid_distribution,lem:high_bid,lem:monopolist,def:conditional_value}, this instance $(P^{(\bs)},\, \bB^{(\bs)})$ is almost valid, except that
the distribution $\bB^{(\bs)} = \{B_{i}^{(\bs)}\}_{i \in [n]}$ can be supported on low bids $b < \gamma$. (Recall that the distribution $P^{(\bs)}$ is just supported on boundary/high values $v \geq \gamma$.)
By truncating the bid distribution $B_{i}(b) \equiv B_{i}^{(\bs)}(b) \cdot \indicator(b \geq \gamma)$ for $i \in [n]$ and reusing the conditional value distribution $P \equiv P^{(\bs)}$, we obtain a valid instance $(P,\, \bB) \in \Bvalid$.
In particular, the truncation $B_{i}(b) \equiv B_{i}^{(\bs)}(b) \cdot \indicator(b \geq \gamma)$ does not modify the auction/optimal {\SocialWelfares} (cf.\ \Cref{lem:auction_welfare,lem:optimal_welfare}).
% we need to vanish out the distributions of $B_{i}$ for everything $< \gamma$ and put all these probability mass to the point of $\gamma$: $\tilde{B}_{i}(b) \equiv B_{i}(b) \cdot \indicator(b \geq \gamma)$ for each $i \in [n]$.
% Then, $(P, \tilde{\bB})$ is a valid instance. This modification of $\tilde{B}_{i}(b)$ does not affect the expression in .
This finishes the proof of Item~(II).
\end{proof}

We conclude this section with (\Cref{cor:poa_identity}) an equivalent definition for {\PriceofAnarchy} in {\FirstPriceAuctions}, as an implication of \Cref{thm:valid}.

\begin{corollary}[{\PriceofAnarchy}]
\label{cor:poa_identity}
Regarding {\FirstPriceAuctions}, the {\PriceofAnarchy} is given by
\begin{align*}
    \PoA ~=~ \inf \bigg\{\, \frac{\FPA(P,\, \bB)}{\OPT(P,\, \bB)} \,\biggmid\, (P,\, \bB) \in \Bvalid ~\text{\em and}~ \OPT(P,\, \bB) < +\infty \,\bigg\}.
\end{align*}
\end{corollary}

\newpage

\newpage

\section{Preprocessing Valid Instances}
\label{sec:preprocessing}

This section presents several preparatory reductions towards the potential worst-case instances.
In \Cref{subsec:pseudo}, we introduce the concept of valid {\em pseudo} instances, which generalizes valid instances and eases the presentation.
% as a limit of real instances.
% This extended space is convenient for our uses and will be our working space for the whole paper.
% and turns out to be convenient.
In \Cref{subsec:discretize}, we {\em discretize} an instance up to any {\PoA}-error $\epsilon > 0$, simplifying the bid-to-value mappings $\bvarphi$ to a (finite-size) bid-to-value table $\bPhi$.
In \Cref{subsec:translate}, we {\em translate} the instances, making the infimum bids zero $\gamma = 0$.
In \Cref{subsec:layer}, we
% the \blackref{alg:layer} reduction
{\em layer} the bid-to-value mappings $\bvarphi$, making the table $\bPhi$ decreasing bidder-by-bidder.
% both in rows and in columns.
In \Cref{subsec:polarize}, we
% the \blackref{alg:polarize} reduction
{\em derandomize} the conditional value $P$ to either the {\em floor} value $P^{\downarrow} \equiv 0$ or the {\em ceiling} value $P^{\uparrow} \equiv \varphi_1(0)$.

To conclude (\Cref{subsec:preprocess}),
we restrict the space of the worst cases to the {\em floor}/{\em ceiling} pseudo instances, which can be represented just by the table $\bPhi$.
These materials serve as the basis for the more complicated reductions later in \Cref{sec:reduction}.

% concludes the above materials:
% After all these regularization, w
% use  to represent a potential worst-case instance.
% , dropping the $\gamma$ and the $P$.
% can be approximated by {\em discretized} instances, and later we will only work with {\em discretized} instances.
% In particular, the bid-to-value mappings degenerate to a finite bid-to-value table after discretization.

\begin{comment}

Roughly speaking, towards a lower bound on the {\PriceofAnarchy}, we can switch the search space from valid real instances to {\em valid pseudo instances}.
\begin{itemize}
    \item \Cref{subsec:pseudo} the concept of pseudo instances

    \item \Cref{subsec:discretize} {\blackref{alg:discretize}}

    \item \Cref{subsec:translate} {\blackref{alg:translate}}

    \item \Cref{subsec:layer} {\blackref{alg:layer}}

    \item \Cref{subsec:polarize} {\blackref{alg:polarize}}

    \item \Cref{subsec:preprocess} falls in the subspace $(\Bvalid^{\downarrow} \cup \Bvalid^{\uparrow})$ of {\em floor}/{\em ceiling} pseudo instances.
\end{itemize}

\begin{reminder*}
\blue{{\bf Yaonan:} For \Cref{sec:preprocessing}, we are left with
\begin{itemize}
  %  \item organization of \Cref{sec:preprocessing}

  %  \item \Cref{rem:pseudo_instance} about pseudo bidder, namely a number of small bidders

    \item \Cref{fact:discretize:function_property} in the proof of \Cref{lem:discretize}

    \item \Cref{lem:discretize}; {\bf The General Case}.
\end{itemize}
}
\end{reminder*}

\newpage

\end{comment}

\subsection{The concept of valid pseudo instances}
\label{subsec:pseudo}

This subsection introduces the concept of valid {\em pseudo} instances $(P,\, \bB \otimes L)$, a natural extension of valid {\em real} instances $(P,\, \bB)$ from \Cref{def:valid}.
Given a valid real instance $(P,\, \bB)$, we can always (\Cref{lem:pseudo_instance}) construct a valid pseudo instance $(P,\, \bB \otimes L)$ that yields the same auction/optimal {\SocialWelfares}. Thus towards a lower bound on the {\PriceofAnarchy}, we can consider valid pseudo instances instead (\Cref{cor:pseudo_instance}).

Roughly speaking, a pseudo instance $(P,\, \bB \otimes L)$ includes one more {\em pseudo bidder} $L$ to a real instance $(P,\, \bB)$. To clarify this difference, we often use the letter $i \in [n]$ and its variants for real bidders, and the Greek letter $\sigma \in [n] \cup \{L\}$ and its variants for real or pseudo bidders. Without ambiguity, we redenote by $\Bvalid$ the space of valid pseudo instances. The formal definition of {\em valid} pseudo instances is given below; cf.\ \Cref{def:valid} for comparison.

\begin{definition}[Valid pseudo instances]
\label{def:pseudo}
A {\em valid} instance $(P,\, \bB \otimes L) \in \Bvalid$, where $\bB = \{B_{i}\}_{i \in [n]}$ and $n \geq 1$, is given by $n + 2$ independent distributions and satisfies the following.
\begin{itemize}
    \item The real bidders $B_{i}$ for $i \in [n]$ each compete with other bidders $\bB_{-i} \otimes L$, while the pseudo bidder $L$ competes with all bidders $\bB \otimes L$, including him/herself $L$.
    
    \item The bid distributions $\bB = \{B_{\sigma}\}_{\sigma \in [n] \cup \{L\}}$ have a common {\em infimum bid} $\inf(\supp(B_{\sigma})) = \gamma \in [0,\, +\infty)$ and a universal {\em supremum bid} $\max_{\sigma \in [n] \cup \{L\}} \sup(\supp(B_{\sigma})) = \lambda \in [\gamma,\, +\infty]$.
    
    \item The bid CDF's $B_{\sigma}(b)$ for $\sigma \in [n] \cup \{L\}$ are {\em continuous} functions on $b \in [\gamma,\, \lambda]$, while the competing bid CDF's $\calB_{-i}(b) = \prod_{k \in [n] \setminus \{i\}} B_{k}(b) \cdot L(b)$ for $i \in [n]$ and the first-order bid CDF $\calB(b) = \prod_{\sigma \in [n] \cup \{L\}} B_{\sigma}(b)$ are {\em continuous} and {\em strictly increasing} functions on $b \in [\gamma,\, \lambda]$. \\
    (The $\calB(b)$ is exactly the pseudo bidder $L$'s competing bid CDF.)
    
    \item \term[{\bf monotonicity}]{monotonicity}{\bf :}
    The bid-to-value mappings $\varphi_{\sigma}(b)
    \eqdef b + \big(\frac{\calB'(b)}{\calB(b)} - \frac{B'_{\sigma}(b)}{B_{\sigma}(b)} \cdot \indicator(\sigma \neq L)\big)^{-1} = \\ b + \big(\sum_{k \in [n] \setminus \{\sigma\}} \frac{B'_{k}(b)}{B_{k}(b)} + \frac{L'(b)}{L(b)}\big)^{-1}$ for $\sigma \in [n] \cup \{L\}$ are {\em weakly increasing} functions on $b \in [\gamma,\, \lambda]$.

    \item \term[{\bf monopoly}]{monopoly}{\bf :}
    The index-$1$ bidder $B_{1}$ without loss of generality is the unique {\em monopolist}. I.e., this bidder always wins $\alloc(\bb) \equiv 1$ in the all-infimum tiebreaker $\{\bb = \gamma^{\otimes n + 1}\}$.
    
    \item \term[{\bf boundedness}]{boundedness}{\bf :}
    Under the infimum bid $\gamma$, the monopolist $B_{1}$'s conditional value $s_{1}^{-1}(\gamma) \sim P$ is bounded as $\supp(P) \subseteq [\gamma,\, \varphi_{1}(\gamma)]$, while the other bidders' conditional values are truthful, $s_{\sigma}^{-1}(\gamma) \equiv \gamma$ for $\sigma \in [2: n] \cup \{L\}$.
\end{itemize}
\end{definition}

\begin{remark}[Pseudo instances]
\label{rem:pseudo_instance}
The concept of pseudo instances greatly eases the presentation and similar ideas also appear in the previous works \cite{AHNPY19,JLTX20,JLQTX19a,JJLZ21}.

Essentially, a pseudo instance $\bB \otimes L$ can be viewed as the limit of a sequence of real instances. I.e., the pseudo bidder $L$ can be viewed as, in the limit $k \to +\infty$, a collection of {\em low-impact} bidders $L_{1}(b) = \dots = L_{k}(b) \eqdef (L(b))^{1 / k}$.
All these low-impact bidders have the {\em common} competing bid distribution $\mathfrak{B}(b) = \prod_{i \in [n]} B_{i}(b) \cdot (L(b))^{1 - 1 / k}$, so the highest bid among them leads to the highest value and we just need to keep track of the collective information $\prod_{j \in [k]} L_{j} = L(b)$.

On the other hand, in the limit $k \to +\infty$, the common competing bid distribution $\mathfrak{B}(b)$ pointwise converges to the first-order bid distribution $\calB(b) = \prod_{i \in [n]} B_{i}(b) \cdot L(b)$.
This accounts for the pseudo mapping $\varphi_{L}(b) = b + \calB(b) / \calB'(b)$.
\end{remark}

Because valid pseudo instances $(P,\, \bB \otimes L)$ differ from the valid real instances ONLY in the definition of the pseudo mapping $\varphi_{L}(b)$, most results in \Cref{sec:structure} still hold. In particular,
\Cref{lem:pseudo_welfare} formulates their {\SocialWelfares} (which restates \Cref{lem:auction_welfare,lem:optimal_welfare} basically).

\begin{lemma}[{\SocialWelfares}]
\label{lem:pseudo_welfare}
For a valid pseudo instance $(P,\, \bB \otimes L) \in \Bvalid$, the expected auction/ optimal {\SocialWelfares} $\FPA(P,\, \bB \otimes L)$ and $\OPT(P,\, \bB \otimes L)$ are given by
\begin{align*}
    \FPA(P,\, \bB \otimes L)\;
    & ~=~ \Ex[P] \cdot \calB(\gamma)
    ~+~ \sum_{\sigma \in [n] \cup \{L\}} \Big(\int_{\gamma}^{\lambda} \varphi_{\sigma}(b) \cdot \frac{B'_{\sigma}(b)}{B_{\sigma}(b)} \cdot \calB(b) \cdot \d b\Big), \\
    \OPT(P,\, \bB \otimes L)
    % & = \Ex_{p,\, \bb,\, \ell} \Big[ \varphi_{\OPT}(p,\, \bb,\, \ell) \Big]\,
    & ~=~ \gamma ~+~ \int_{\gamma}^{+\infty} \Big(1 - P(v) \cdot \calV(v)\Big) \cdot \d v,
\end{align*}
where the first-order bid distribution $\calB(b) \eqdef \prod_{\sigma \in [n] \cup \{L\}} B_{\sigma}(b)$ and the quasi first-order value distribution $\calV(v) \eqdef \prod_{\sigma \in [n] \cup \{L\}} \Prx_{b_{\sigma} \sim B_{\sigma}}\big[\, (b_{\sigma} \leq \gamma) \vee (\varphi_{\sigma}(b_{\sigma}) \leq v) \,\big]$ can be reconstructed from $\bB \otimes L$.
\end{lemma}

\Cref{lem:pseudo_mapping} shows that given a pseudo instance $(P,\, \bB \otimes L)$, the pseudo bidder $L$ always has a {\em dominated} bid-to-value mapping $\varphi_{L}(b) \leq \varphi_{i}(b)$, compared with other real bidders $i \in [n]$.

\begin{lemma}[Bid-to-value mappings]
\label{lem:pseudo_mapping}
For a valid pseudo instance $(P,\, \bB \otimes L) \in \Bvalid$, the bid-to-value mappings $\bvarphi = \{\varphi_{\sigma}\}_{\sigma \in [n] \cup \{L\}}$ satisfy the following:
\begin{enumerate}[font = {\em\bfseries}]
    \item\label{lem:pseudo_mapping:dominance}
    $\min(\bvarphi(b^{\otimes n + 1})) = \varphi_{L}(b) > b$ for $b \in (\gamma,\, \lambda]$. Further, $\min(\bvarphi(\gamma^{\otimes n + 1})) = \varphi_{L}(\gamma) \geq \gamma$.
    
    \item\label{lem:pseudo_mapping:inequality}
    $\sum_{i \in [n]} \frac{\varphi_{L}(b) - b}{\varphi_{i}(b) - b} \geq n - 1$ for $b \in (\gamma,\, \lambda]$.
\end{enumerate}
\end{lemma}

\begin{proof}
{\bf \Cref{lem:pseudo_mapping:dominance}} is an immediate consequence of (\Cref{def:pseudo}) the bid-to-value mappings 
$\varphi_{\sigma}(b) = b + \big(\frac{\calB'(b)}{\calB(b)} - \frac{B'_{\sigma}(b)}{B_{\sigma}(b)} \cdot \indicator(\sigma \neq L)\big)^{-1}$ for $\sigma \in [n] \cup \{L\}$.
% , where the first-order bid distribution $\calB(b) = \prod_{i \in [n]} B_{i}(b) \cdot L(b)$.
In particular, in the proof of \Cref{lem:high_bid} (Fact~\ref{lem:high_bid:rational}), we already show that the pseudo mapping $\varphi_{L}(b) > b$ for $b \in (\gamma,\, \lambda]$ and $\varphi_{L}(\gamma) \geq \gamma$. Moreover, we have
$\sum_{i \in [n]} \big(\varphi_{i}(b) - b\big)^{-1}
= n \cdot \frac{\calB'(b)}{\calB(b)} - \sum_{i \in [n]} \frac{B'_{i}(b)}{B_{i}(b)}
= (n - 1) \cdot \frac{\calB'(b)}{\calB(b)} + \frac{L'(b)}{L(b)}
\geq (n - 1) \cdot \big(\varphi_{L}(b) - b\big)^{-1}$.
Rearranging this equation gives {\bf \Cref{lem:pseudo_mapping:inequality}}.
\end{proof}

\Cref{lem:pseudo_distribution} shows that the bid distributions $\bB \otimes L = \{B_{i}\}_{i \in [n]} \otimes L$ can be reconstructed from the bid-to-value mappings $\bvarphi = \{\varphi_{\sigma}\}_{\sigma \in [n] \cup \{L\}}$.
This lemma is important in that many subsequent reductions will construct new bid distributions $\tilde{\bB} \otimes \tilde{L}$ in terms of the mappings $\tilde{\bvarphi}$.
Therefore, the correctness of such construction reduces to checking \Cref{lem:pseudo_mapping:dominance,lem:pseudo_mapping:inequality} of \Cref{lem:pseudo_mapping} and monotonicity of the mappings $\tilde{\bvarphi}$.
% the modified mappings are increasing and

\begin{lemma}[Bid distributions]
\label{lem:pseudo_distribution}
For $(n + 1)$ functions $\bvarphi = \{\varphi_{\sigma}\}_{\sigma \in [n] \cup \{L\}}$ that are increasing on $b \in [\gamma,\, \lambda]$ and satisfy \Cref{lem:pseudo_mapping:dominance,lem:pseudo_mapping:inequality} of \Cref{lem:pseudo_mapping},
the functions $\bB \otimes L$ given below are well-defined $[\gamma,\, \lambda]$-supported bid distributions and their bid-to-value mappings are exactly $\bvarphi = \{\varphi_{\sigma}\}_{\sigma \in [n] \cup \{L\}}$.
\begin{align*}
    B_{i}(b) & ~\equiv~ \exp\Big(-\int_{b}^{\lambda}
    \Big(\big(\varphi_{L}(b) - b\big)^{-1} - \big(\varphi_{i}(b) - b\big)^{-1}\Big) \cdot \d b\Big) \cdot \indicator(b \geq \gamma),
    \qquad\qquad \forall i \in [n]. \\
    L(b) & ~\equiv~ \exp\Big(-\int_{b}^{\lambda}
    \Big(\sum_{i \in [n]} \big(\varphi_{i}(b) - b\big)^{-1} - (n - 1) \cdot \big(\varphi_{L}(b) - b\big)^{-1}\Big) \cdot \d b\Big) \cdot \indicator(b \geq \gamma).
\end{align*}
\end{lemma}

\begin{proof}
The construction of $\bB \otimes L$ intrinsically ensures that $B_{\sigma}(b) = 0$ below the infimum bid $b < \gamma$ and $B_{\sigma}(\lambda) = 1$ at the supremum bid $\lambda$, for $\sigma \in [n] \cup \{L\}$.
Moreover, \Cref{lem:pseudo_mapping:dominance,lem:pseudo_mapping:inequality} of \Cref{lem:pseudo_mapping} together make each $B_{\sigma}(b)$ an increasing function on $b \in [\gamma,\, \lambda]$, hence a well-defined $[\gamma,\, \lambda]$-supported bid distribution. It remains to verify $\varphi_{\sigma}(b) = b + \big(\sum_{k \in [n] \setminus \{\sigma\}} \frac{B'_{k}(b)}{B_{k}(b)} + \frac{L'(b)}{L(b)}\big)^{-1}$ over the bid support $b \in [\gamma,\, \lambda]$, for $\sigma \in [n] \cup \{L\}$, namely the functions $\bvarphi = \{\varphi_{\sigma}\}_{\sigma \in [n] \cup \{L\}}$ are exactly (\Cref{def:pseudo}) the bid-to-value mappings stemmed from the $\bB \otimes L$.
These identities for $\sigma \in [n] \cup \{L\}$ can be easily seen via elementary algebra; we omit the details for brevity.
\Cref{lem:pseudo_distribution} follows then.
\end{proof}

\Cref{lem:pseudo_instance} suggests that any valid {\em real} instance $(P,\, \bB)$ can be reinterpreted as a valid {\em pseudo} instance $(P,\, \bB \otimes L_{\gamma})$, by employing a specific pseudo bidder $L_{\gamma}$.
As a consequence (\Cref{cor:pseudo_instance}): Towards a lower bound on the {\PriceofAnarchy} in {\FirstPriceAuctions}, we can consider valid {\em pseudo} instances instead; cf.\ \Cref{cor:poa_identity} for comparison.

\begin{lemma}[Pseudo instances]
\label{lem:pseudo_instance}
For a valid instance $(P,\, \bB) \in \Bvalid$, the following hold:
\begin{enumerate}[font = {\em\bfseries}]
    \item\label{lem:pseudo_instance:property}
    The pseudo bidder $L_{\gamma}(b) \equiv \indicator(b \geq \gamma)$ induces a valid pseudo instance $(P,\, \bB \otimes L_{\gamma}) \in \Bvalid$.

    \item\label{lem:pseudo_instance:poa}
     $\OPT(P,\, \bB \otimes L_{\gamma}) = \OPT(P,\, \bB)$ and $\FPA(P,\, \bB \otimes L_{\gamma}) = \FPA(P,\, \bB)$.
\end{enumerate}
\end{lemma}

\begin{proof}
As mentioned in \Cref{rem:pseudo_instance}, to see the validity $(P,\, \bB \otimes L_{\gamma}) \in \Bvalid$, we only need to verify \blackref{monotonicity} of mappings $\bvarphi = \{\varphi_{\sigma}\}_{\sigma \in [n] \cup \{L\}}$ and \blackref{boundedness} of the conditional value $P$.

\vspace{.1in}
\noindent
{\bf \Cref{lem:pseudo_instance:property}.}
We have $L_{\gamma}(b) = 1$ and $L'_{\gamma}(b) \big/ L_{\gamma}(b) = 0$ for $b \geq \gamma$. Thus, each real bidder $i \in [n]$ keeps the {\em same} mapping $\varphi_{i}(b)$ after including the pseudo bidder $L_{\gamma}$ (\Cref{def:valid,def:pseudo}), preserving \blackref{monotonicity}.
Also, the pseudo mapping can be written as $\varphi_{L_{\gamma}}(b) = b + \calB(b) \big/ \calB'(b)$ for $b \in [\gamma,\, \lambda]$, based on the first-order bid distribution $\calB(b) = \prod_{i \in [n]} B_{i}(b) \cdot L_{\gamma}(b) = \prod_{i \in [n]} B_{i}(b)$.
In the proof of \Cref{lem:high_bid} (Fact~\ref{fact:pseudo}), we already verify \blackref{monotonicity} of this mapping.
Moreover, the {\em unmodified} conditional value $P$ must preserve \blackref{boundedness}.

\vspace{.1in}
\noindent
{\bf \Cref{lem:pseudo_instance:poa}.}
This can be easily inferred from (\Cref{lem:auction_welfare,lem:optimal_welfare} versus \Cref{lem:pseudo_welfare}) the {\SocialWelfare} formulas for valid {\em real} versus  {\em pseudo} instances; thus we omit details for brevity. Roughly speaking, the point is that the pseudo bidder $L_{\gamma}(b) = 1$ for $b \geq \gamma$ has no effect on {\SocialWelfares}.
\end{proof}

Due to the above lemma, we know that the {\PoA} of pseudo instances is a lower bound of the real {\PoA}. (On the other hand, we know that a pseudo instance can be viewed as the limit of a sequence of real instances. As a result, any {\PoA} ratio obtained by a pseudo instance can be approximated by  real instances arbitrarily well.) Thus, we have the following corollary.

\begin{corollary}[Lower bound]
\label{cor:pseudo_instance}
Regarding {\FirstPriceAuctions}, the {\PriceofAnarchy} is at least
\begin{align*}
    \PoA ~\geq~ \inf \bigg\{\, \frac{\FPA(P,\, \bB \otimes L)}{\OPT(P,\, \bB \otimes L)} \,\biggmid\, (P,\, \bB \otimes L) \in \Bvalid ~\text{\em and}~ \OPT(P,\, \bB \otimes L) < +\infty \,\bigg\}.
\end{align*}
\end{corollary}

\subsection{{\discretize}: Towards (approximate) discrete pseudo instances}
\label{subsec:discretize}

In this subsection, we introduce the concept of {\em discretized} pseudo instances $(P,\, \bB \otimes L)$; see \Cref{def:discretize,fig:discretize} for its formal definition and a visual aid.
% Recall that the quasi value distributions $\max\big(V_{\sigma}(v),\, B_{\sigma}(\gamma)\big)$ for $\sigma \in [n] \cup \{L\}$ (\Cref{cor:reconstruction}) can be reconstructed from the bid distributions $\bB \otimes L$.
%Roughly speaking, these $\max\big(V_{\sigma}(v),\, B_{\sigma}(\gamma)\big)$ are {\em discrete} distributions when the pseudo instance $(P,\, \bB \otimes L)$ is {\em discretized}.

%The formal definition of {\em discretized} pseudo instances $(P,\, \bB \otimes L)$ is given below.

\begin{definition}[Discretized pseudo instances]
\label{def:discretize}
A {\em valid} pseudo instance $(P,\, \bB \otimes L)$ from  \Cref{def:pseudo} is further called {\em discretized} when it satisfies \blackref{finiteness} and \blackref{piecewise}.
\begin{itemize}
    \item \term[{\bf finiteness}]{finiteness}{\bf :}
    The supremum bid/value are {\em bounded away} from infinite $\lambda \leq \max(\bvarphi(\blambda)) < +\infty$.

    \item \term[{\bf piecewise constancy}]{piecewise}{\bf :}
    Given some {\em bounded} integer $0 \leq m < +\infty$ and some $(m + 1)$-piece partition of the bid support $\bLambda \eqdef [\gamma = \lambda_{0},\, \lambda_{1}) \cup
    [\lambda_{1},\, \lambda_{2}) \cup \dots \cup
    [\lambda_{m},\, \lambda_{m + 1} = \lambda]$, each bid-to-value mapping $\varphi_{\sigma}(b)$ for $\sigma \in [n] \cup \{L\}$ is a {\em piecewise constant} function under this partition.
\end{itemize}
Notice that the partition $\bLambda$ relies on the bid distributions $\bB \otimes L$ but is irrelevant to the monopolist $B_{1}$'s conditional value $P$.
% \blue{The integer $m$ will also be called the {\em partition number} $N_{\sf partition}(\bB \otimes L)$; notice that this irrelevant to the monopolist $B_{1}$'s conditional value $P$.}
\end{definition}

\begin{remark}[Discretized pseudo instances]
\label{rem:discretize}
The value distributions $V_{\sigma}(v)$ for $\sigma \in [n] \cup \{L\}$ reconstructed from a {\em discretized} pseudo instance $(P,\, \bB \otimes L)$ are almost discrete, except the conditional value $P$ of the monopolist $B_{1}$ (which ranges between $\supp(P) \subseteq [\gamma,\, \varphi_{1}(\gamma)]$ but otherwise is arbitrary).
Namely, each index $j \in [0: m]$ piece of a bid-to-value mapping $\varphi_{\sigma}(b)$ induces a probability mass $\big(B_{\sigma}(\lambda_{j + 1}) - B_{\sigma}(\lambda_{j})\big)$ of the value distribution $V_{\sigma}(v)$.
%However, at the moment, such ``partially'' {\em discretized} pseudo instances $(P,\, \bB \otimes L)$ are sufficient for our purpose.
%(In later subsections, we will further discretize the conditional values $P$.)
\end{remark}

\Cref{lem:discretize} shows that the {\PoA} bound of a (generic) valid pseudo instance $(P,\, \bB \otimes L)$ can be approximated {\em arbitrarily close} by a {\em discretized} pseudo instance $(\tilde{P},\, \tilde{\bB} \otimes \tilde{L})$.

\begin{intuition*}
For the {\PriceofAnarchy} problem, we only consider (\Cref{cor:pseudo_instance}) the pseudo instances that have {\em bounded} {\SocialWelfares} $\FPA(P,\, \bB \otimes L) \leq \OPT(P,\, \bB \otimes L) < +\infty$.
In the sense of Lebesgue integrability, we can carefully truncate\footnote{We will use a special truncation scheme. The usual truncation scheme often induces probability masses at the truncation threshold, but we require that the bid CDF's $\bB \otimes L$ are {\em continuous} over the bid support $b \in [\gamma,\, \lambda]$.} the bid distributions $\bB \otimes L$ to meet the \blackref{finiteness} condition, by which the auction/optimal {\SocialWelfares} change by negligible amounts.

Then, we replace the bid-to-value mappings $\bvarphi = \{\varphi_{\sigma}\}_{\sigma \in [n] \cup \{L\}}$ by their {\em interpolation functions} $\tilde{\bvarphi} = \{\tilde{\varphi}_{\sigma}\}_{\sigma \in [n] \cup \{L\}}$ (\blackref{piecewise}).
The mappings $\bvarphi$ determine the bid distributions $\bB \otimes L$ (\Cref{lem:pseudo_distribution}) and the auction/optimal {\SocialWelfares} (\Cref{lem:pseudo_welfare}) in terms of {\em integral formulas}, which are insensitive to the (local) changes of the interpolation functions $\tilde{\bvarphi}$. Thus as usual, when the interpolation scheme is accurate enough, everything (including the {\PoA}-bound) can be approximated arbitrarily well.
To make the arguments rigorous, we use the continuous mapping theorem \cite{B13} in measure theory.
\end{intuition*}

% each bid-to-value mapping can be approximated by a piecewise constant function arbitrarily good.
% the reconstruction from $\bvarphi = \{\varphi_{\sigma}\}_{\sigma \in [n] \cup \{L\}}$ to the bid distributions and the function of $\OPT$ and $\FPA$ in terms of these distributions and mappings are all continuous. Therefore, the approximation error for $\OPT$ and $\FPA$ and thus for PoA can be arbitrarily small.

% \blue{The finiteness is due to the fact that  $\OPT(P,\, \bB \otimes L) < +\infty$ and thus we can always choose a large enough upper bound on values (and thus bids as well) and truncated the value distribution and bid distribution. This can effect  $\OPT$ and $\FPA$ as small as possible since we choose large enough upper bound. After this truncation, each bid-to-value mapping can be approximated by a piecewise constant function arbitrarily good. We also notice that the reconstruction from $\bvarphi = \{\varphi_{\sigma}\}_{\sigma \in [n] \cup \{L\}}$ to the bid distributions and the function of $\OPT$ and $\FPA$ in terms of these distributions and mappings are all continuous. Therefore, the approximation error for $\OPT$ and $\FPA$ and thus for PoA can be arbitrarily small. To make the arguments rigorous, we use the continuous mapping theorem \cite{B13} in measure theory.}
% This is quite intuitive. 

\begin{theorem}[{The continuous mapping theorem \cite{B13}}]
\label{thm:continuous_mapping}
Let $X$ and $\{X_{m}\}_{m \geq 0}$ be random elements defined on a metric space $\mathbb{S}$. Regarding a function $g: \mathbb{S} \mapsto \mathbb{S}'$ (where $\mathbb{S}'$ is another metric space), suppose its set of discontinuity points $\mathbb{D}_{g}$ has a zero measure $\Pr[X \in \mathbb{D}_{g}] = 0$, then this function $g$ perseveres the property
% (``{\tt d}'')
of convergence in distribution. Formally,
\[
    \big\{ X_{m} \big\} ~\overset{{\tt dist}}{\longrightarrow}~ X
    \qquad\Longrightarrow\qquad
    \big\{ g(X_{m}) \big\} ~\overset{{\tt dist}}{\longrightarrow}~ g(X).
\]
\end{theorem}

\begin{lemma}[Discretization]
\label{lem:discretize}
Given a valid pseudo instance $(P,\, \bB \otimes L)$ that has bounded expected auction/optimal {\SocialWelfares} $\FPA(P,\, \bB \otimes L) \leq \OPT(P,\, \bB \otimes L) < +\infty$, for any error $\epsilon \in (0,\, 1)$, there is a discretized pseudo instance $(\tilde{P},\, \tilde{\bB} \otimes \tilde{L})$ such that $|\PoA(\tilde{P},\, \tilde{\bB} \otimes \tilde{L}) - \PoA(P,\, \bB \otimes L)| \leq \epsilon$.
\end{lemma}

\afterpage{
\begin{figure}[t]
    \centering
    \begin{mdframed}
    Approximate Reduction $\term[\discretize]{alg:discretize}(P,\, \bB \otimes L,\, \bLambda)$

    \begin{flushleft}
    {\bf Input:} A (generic) valid pseudo instance $(P,\, \bB \otimes L)$;
    \hfill \Cref{def:pseudo} \\
    \white{\bf Input:} A partition $\bLambda = [\gamma = \lambda_{0},\, \lambda_{1}) \cup
    [\lambda_{1},\, \lambda_{2}) \cup \dots \cup
    [\lambda_{m},\, \lambda_{m + 1} = \lambda]$.
    \hfill \Cref{def:discretize}
    
    \vspace{.05in}
    {\bf Output:} A {\em discretized} pseudo instance $(P,\, \tilde{\bB} \otimes \tilde{L})$.

    \begin{enumerate}
        \item\label{alg:discretize:mapping}
        Define $\tilde{\varphi}_{\sigma}(b) \eqdef \sum_{j \in [0:\, m]} \varphi_{\sigma}(\lambda_{j + 1}) \cdot \indicator(\lambda_{j} \leq b < \lambda_{j + 1})$ on $b \in [\gamma,\, \lambda)$, for $\sigma \in [n] \cup \{L\}$.

        \item[] \OliveGreen{$\triangleright$ Define $\tilde{\varphi}_{\sigma}(\lambda_{m + 1}) \eqdef \varphi_{\sigma}(\lambda_{m + 1})$ to make a closed domain.}

        \item\label{alg:discretize:distribution}
        {\bf Return} $(P,\, \tilde{\bB} \otimes \tilde{L})$, namely only the bid distributions $\tilde{\bB} \otimes \tilde{L} = \{\tilde{B}_{\sigma}\}_{\sigma \in [n] \cup \{L\}}$ are modified and (\Cref{lem:pseudo_distribution}) are reconstructed from the mappings $\tilde{\bvarphi} = \{\tilde{\varphi}_{\sigma}\}_{\sigma \in [n] \cup \{L\}}$.
    \end{enumerate}
    \end{flushleft}
    \end{mdframed}
    \caption{The {\discretize} approximate reduction.
    \label{fig:alg:discretize}}
\end{figure}
\begin{figure}
    \centering
    \subfloat[\label{fig:discretize:old}
    The input {\em increasing} mappings $\bvarphi = \{\varphi_{\sigma}\}_{\sigma \in [n] \cup \{L\}}$.]{
    \includegraphics[width = .49\textwidth]
    {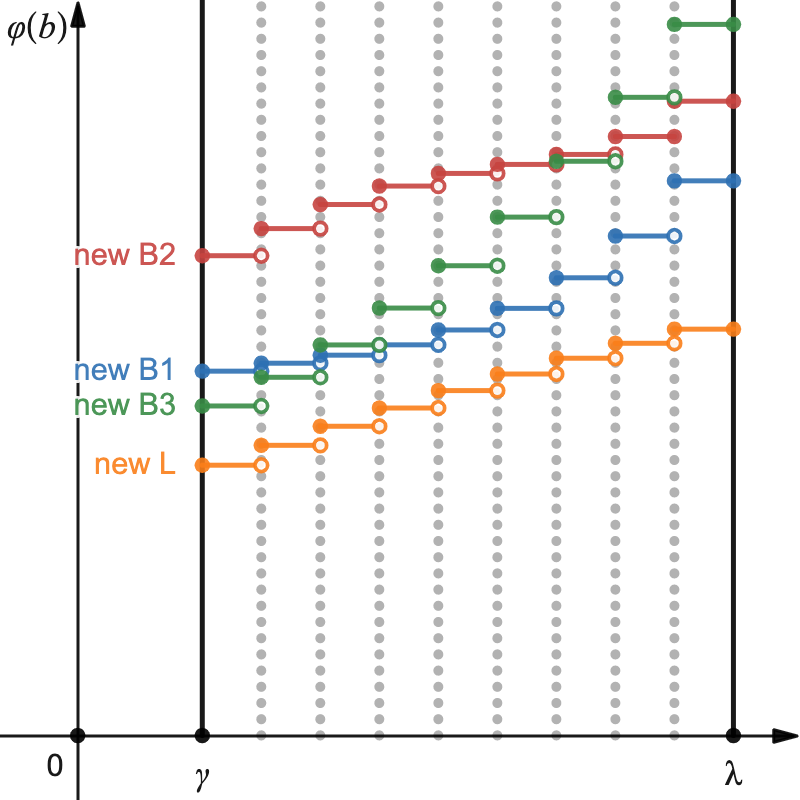}}
    \hfill
    \subfloat[\label{fig:discretize:new}
    The output {\em discretized} mappings $\tilde{\bvarphi} = \{\tilde{\varphi}_{\sigma}\}_{\sigma \in [n] \cup \{L\}}$.]{
    \includegraphics[width = .49\textwidth]
    {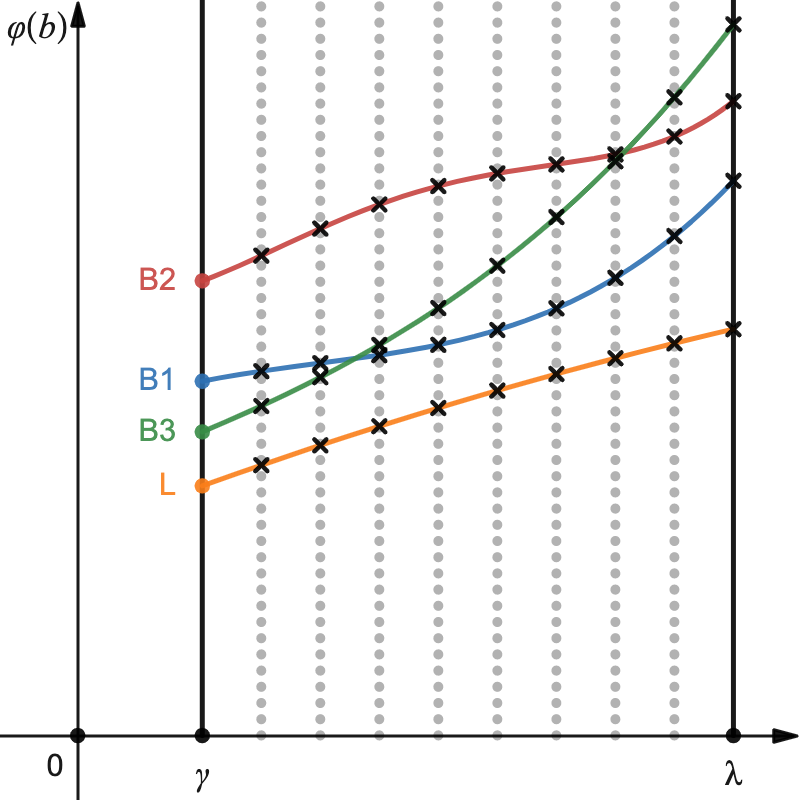}}
    \caption{Diagram of the {\discretize} approximate reduction.
    \label{fig:discretize}}
\end{figure}
\clearpage}

\begin{proof}
For ease of presentation, we first assume that the given pseudo instance $(P,\, \bB \otimes L)$ itself satisfies \blackref{finiteness} $\lambda \leq \max(\bvarphi(\blambda)) < +\infty$; this assumption will be removed later.

% \vspace{.1in}
% \noindent
% {\bf The ``Finiteness'' Case.}
Given an integer $m \geq 0$, we would divide the {\em bounded} bid support into $(m + 1)$ uniform pieces $\bLambda = [\gamma = \lambda_{0},\, \lambda_{1}) \cup [\lambda_{1},\, \lambda_{2}) \cup \dots \cup [\lambda_{m},\, \lambda_{m + 1} = \lambda]$, via the partition points
$\lambda_{j} \eqdef \gamma + \frac{j}{m + 1} \cdot (\lambda - \gamma)$ for $j \in [0:\, m + 1]$.
Regarding this partition $\bLambda$, the {\blackref{alg:discretize}} reduction (see \Cref{fig:alg:discretize,fig:discretize} for its description and a visual aid) transforms the given pseudo instance $(P,\, \bB \otimes L)$ into another {\em discretized} pseudo instance $(P,\, \tilde{\bB} \otimes \tilde{L})$. This is formalized into {\bf \Cref{fact:discretize:output}}.

\setcounter{fact}{0}

\begin{fact}
\label{fact:discretize:output}
Under reduction $(P,\, \tilde{\bB} \otimes \tilde{L}) \gets \discretize(P,\, \bB \otimes L,\, \bLambda)$, the output $(P,\, \tilde{\bB} \otimes \tilde{L})$ is a discretized pseudo instance; the conditional value $P$ is unmodified.
\end{fact}

\begin{proof}
The {\discretize} reduction (Line~\ref{alg:discretize:distribution}) reconstructs the output bid distributions $\tilde{\bB} \otimes \tilde{L}$ from the output mappings $\tilde{\bvarphi} = \{\tilde{\varphi}_{\sigma}\}_{\sigma \in [n] \cup \{L\}}$; (Line~\ref{alg:discretize:mapping}) each output mapping $\tilde{\varphi}_{\sigma}(b)$ is the $(m + 1)$-piece {\em interpolation} of the input mapping $\varphi_{\sigma}(b)$ under the partition $\bLambda$.
In particular (see \Cref{fig:discretize} for a visual aid), the output mappings $\tilde{\varphi}_{\sigma}(b) = \varphi_{\sigma}(\lambda_{j + 1}) \geq \varphi_{\sigma}(b)$ interpolate the right endpoint of each index-$j$ piece $b \in [\lambda_{j},\, \lambda_{j + 1})$ of the partition $\bLambda$.

Obviously, each interpolation $\tilde{\varphi}_{\sigma}(b)$ is a piecewise constant function on the same bounded domain $b \in [\gamma,\, \lambda]$ (\blackref{finiteness} and \blackref{piecewise}) and is an increasing function \`{a} la the input mapping $\varphi_{\sigma}(b)$ (\blackref{monotonicity}).
Also, the extended range $[\gamma,\, \tilde{\varphi}_{1}(\gamma)] = [\gamma,\, \varphi_{1}(\lambda_{1})] \supseteq [\gamma, \varphi_{1}(\gamma)] \supseteq \supp(P)$ restricts the unmodified conditional value $P$ (\blackref{boundedness}).

It remains to show that the reconstructed bid distributions $\tilde{\bB} \otimes \tilde{L}$ are well defined, namely that 
(\Cref{lem:pseudo_distribution}) the interpolations $\tilde{\bvarphi} = \{\tilde{\varphi}_{\sigma}\}_{\sigma \in [n] \cup \{L\}}$ satisfy the conditions in \Cref{lem:pseudo_mapping}.
The input mappings $\bvarphi = \{\varphi_{\sigma}\}_{\sigma \in [n] \cup \{L\}}$ themselves must satisfy those conditions:
(\Cref{lem:pseudo_mapping:dominance} of \Cref{lem:pseudo_mapping}) that $\min(\bvarphi(b^{\otimes n + 1})) = \varphi_{L}(b) > b$ for $b \in (\gamma,\, \lambda]$ and $\min(\bvarphi(\gamma^{\otimes n + 1})) = \varphi_{L}(\gamma) \geq \gamma$; (\Cref{lem:pseudo_mapping:inequality} of \Cref{lem:pseudo_mapping}) that $\sum_{i \in [n]} \frac{\varphi_{L}(b) - b}{\varphi_{i}(b) - b} \geq n - 1$ for $b \in (\gamma,\, \lambda]$.

By interpolating the right endpoint of each index-$j$ piece $b \in [\lambda_{j},\, \lambda_{j + 1})$ of the partition $\bLambda$,
the output mappings satisfy $\tilde{\varphi}_{\sigma}(b) = \varphi_{\sigma}(\lambda_{j + 1}) \geq \tilde{\varphi}_{L}(b) = \varphi_{L}(\lambda_{j + 1}) > \lambda_{j + 1} > b$ (\Cref{lem:pseudo_mapping:dominance} of \Cref{lem:pseudo_mapping})
and thus $\sum_{i \in [n]} \frac{\tilde{\varphi}_{L}(b) - b}{\tilde{\varphi}_{i}(b) - b} = \sum_{i \in [n]} \frac{\varphi_{L}(\lambda_{j + 1}) - b}{\varphi_{i}(\lambda_{j + 1}) - b} \geq \sum_{i \in [n]} \frac{\varphi_{L}(\lambda_{j + 1}) - \lambda_{j + 1}}{\varphi_{i}(\lambda_{j + 1}) - \lambda_{j + 1}} \geq n - 1$ (\Cref{lem:pseudo_mapping:inequality} of \Cref{lem:pseudo_mapping}).
This finishes the proof of {\bf \Cref{fact:discretize:output}}.
\end{proof}

We thus obtain a sequence of {\em discretized} pseudo instances $(P,\, \bB^{(m)} \otimes L^{(m)})$, for $m \geq 0$.
Below we study the bound $\PoA(P,\, \bB^{(m)} \otimes L^{(m)})$ for each of them.
To this end, consider a {\em specific} profile $(p,\, \bb,\, \ell) \in
% \supp(P,\, \bB^{(m)} \otimes L^{(m)}) \subseteq
[\gamma,\, \varphi_{1}(\gamma)] \times [\gamma,\, \lambda]^{n + 1}$.
Following \Cref{def:pseudo}, depending on whether the bids $(\bb,\, \ell) = (b_{\sigma})_{\sigma \in [n] \cup \{L\}}$ each take the infimum $= \gamma$ or not $> \gamma$:
\begin{itemize}
    \item The monopolist $B_{1}$ takes a value $v_{1}^{(m)} =  p \cdot \indicator(b_{1} = \gamma) + \varphi_{1}^{(m)}(b_{1}) \cdot \indicator(b_{1} > \gamma)$.

    \item Each other bidder $\sigma \in [2:\, n] \cup \{L\}$ takes a value $v_{\sigma}^{(m)} = \gamma \cdot \indicator(b_{1} = \gamma) + \varphi_{\sigma}^{(m)}(b_{\sigma}) \cdot \indicator(b_{\sigma} > \gamma)$.
\end{itemize}
Over the randomness of the allocation rule $\alloc(\bb,\, \ell)$, this index-$m$ pseudo instance $(P,\, \bB^{(m)} \otimes L^{(m)})$ yields an interim auction {\SocialWelfare} $\varphi_{\FPA}^{(m)}(p,\, \bb,\, \ell) \eqdef \Ex_{\alloc} \big[ v_{\alloc(\bb,\, \ell)}^{(m)} \big]$.

Likewise, under the {\em same} profile $(p,\, \bb,\, \ell)$, the input pseudo instance $(P,\, \bB \otimes L)$ yields a counterpart formula $\varphi_{\FPA}(p,\, \bb,\, \ell) \eqdef \Ex_{\alloc} \big[ v_{\alloc(\bb,\, \ell)} \big]$.

{\bf \Cref{fact:discretize:function_property,fact:discretize:function_convergence}} show certain properties of the input/index-$m$ mappings $\bvarphi$ and $\bvarphi^{(m)}$ and the input/index-$m$ formulas $\varphi_{\FPA}(p,\, \bb,\, \ell)$ and $\varphi_{\FPA}^{(m)}(p,\, \bb,\, \ell)$.

\begin{fact}
\label{fact:discretize:function_property}
On the domain $b \in [\gamma,\, \lambda]$, the mappings $\bvarphi = \{\varphi_{\sigma}\}_{\sigma \in [n] \cup \{L\}}$ and $\bvarphi^{(m)} = \{\varphi_{\sigma}^{(m)}\}_{\sigma \in [n] \cup \{L\}}$ are bounded everywhere, and continuous almost everywhere except a zero-measure set $\mathbb{D} \subseteq [\gamma,\, \lambda]$.
Also, on the domain $(p,\, \bb,\, \ell) \in [\gamma,\, \varphi_{1}(\gamma)] \times [\gamma,\, \lambda]^{n + 1}$, the formulas $\varphi_{\FPA}(p,\, \bb,\, \ell)$ and $\varphi_{\FPA}^{(m)}(p,\, \bb,\, \ell)$ are bounded everywhere, and continuous almost everywhere except a zero-measure set $\mathbb{D}_{\FPA}$.
\end{fact}

\begin{proof}
Each input mapping $\varphi_{\sigma}(b)$ as an increasing function (\blackref{monotonicity}) is continuous almost everywhere, except countably many discontinuities that in total have a zero measure. Also, each input mapping $\varphi_{\sigma}(b)$ is bounded $\varphi_{\sigma}(b) \leq \varphi_{\sigma}(\lambda) \leq \max(\bvarphi(\blambda)) < +\infty$ (\blackref{finiteness}).
Such continuity and boundedness extend to the input formula $\varphi_{\FPA}(p,\, \bb,\, \ell)$, since tiebreakers $\alloc(\bb,\, \ell) \in \argmax(\bb,\, \ell)$ occur with a zero probability.\footnote{That is, conditioned on the {\em all-infimum} bid profile $\big\{ (\bb,\, \ell) = \bgamma \big\}$, the monopolist {\em always} gets allocated $\alloc(\bgamma) = B_{1}$ (\blackref{monopoly}). Otherwise $\big\{ (\bb,\, \ell) \neq \bgamma \big\}$, tiebreakers occur with a zero probability because the underlying bid CDF's $\bB \otimes L$ are continuous functions (\Cref{def:pseudo}).}

For the same reasons, the index-$m$ mappings $\varphi_{\sigma}^{(m)}(b)$ and the index-$m$ formula $\varphi_{\FPA}^{(m)}(p,\, \bb,\, \ell)$ are also bounded and continuous almost everywhere.
{\bf \Cref{fact:discretize:function_property}} follows then.
\end{proof}

% the defining first-order bidder $\sigma_{\FPA}(\bb,\, \ell) = \argmax \big\{ b_{\sigma}: \sigma \in [n] \cup \{L\} \big\}$, by breaking ties in favor of the smallest index, is uniquely specified.

% is increasing and bounded $\varphi_{\sigma}(b) \leq \varphi_{\sigma}(\lambda) \leq \max(\bvarphi(\blambda)) < +\infty$, thus being continuous almost everywhere. Each index-$m$ mapping $\varphi_{\sigma}^{(m)}(b)$ as the $(m + 1)$-piece interpolation must also be bounded and continuous almost everywhere.

% On $b \in [\gamma,\, \lambda]$, (\blackref{monotonicity}/\blackref{finiteness})

% On $b \in [\gamma,\, \lambda]$, (\blackref{monotonicity}/\blackref{finiteness}) each input mapping $\varphi_{\sigma}(b)$ is increasing and bounded $\varphi_{\sigma}(b) \leq \varphi_{\sigma}(\lambda) \leq \max(\bvarphi(\blambda)) < +\infty$, thus being continuous almost everywhere. Each index-$m$ mapping $\varphi_{\sigma}^{(m)}(b)$ as the $(m + 1)$-piece interpolation must also be bounded and continuous almost everywhere.

% The boundedness and continuity extend to the {\FPA} formulas $\varphi_{\FPA}(p,\, \bb,\, \ell)$ and $\varphi_{\FPA}^{(m)}(p,\, \bb,\, \ell)$, since the defining first-order bidder $\sigma_{\FPA}(\bb,\, \ell) = \argmax \big\{ b_{\sigma}: \sigma \in [n] \cup \{L\} \big\}$, by breaking ties in favor of the smallest index, is uniquely specified.

\begin{fact}
\label{fact:discretize:function_convergence}
The sequence $\big\{ \varphi_{\FPA}^{(m)}(p,\, \bb,\, \ell) \big\}_{m \geq 0}$ pointwise converges to $\varphi_{\FPA}(p,\, \bb,\, \ell)$ almost everywhere.
\end{fact}

\begin{proof}
By construction, each index-$m$ {\FPA} formula $\varphi_{\FPA}^{(m)}(p,\, \bb,\, \ell)$ for $m \geq 0$ is the $(m + 1)^{n + 1}$-piece uniform interpolation of the input {\FPA} formula $\varphi_{\FPA}(p,\, \bb,\, \ell)$. Because the $\varphi_{\FPA}(p,\, \bb,\, \ell)$ itself is bounded and continuous almost everywhere ({\bf \Cref{fact:discretize:function_property}}), letting $m \to +\infty$ makes the sequence $\big\{ \varphi_{\FPA}^{(m)}(p,\, \bb,\, \ell) \big\}_{m \geq 0}$ converge to $\varphi_{\FPA}(p,\, \bb,\, \ell)$ almost everywhere.
{\bf \Cref{fact:discretize:function_convergence}} follows then.
\end{proof}

Now we would consider a {\em random} profile from the input pseudo instance $(p,\, \bb,\, \ell) \sim (P,\, \bB \otimes L)$
rather than a specific profile $\in [\gamma,\, \varphi_{1}(\gamma)] \times [\gamma,\, \lambda]^{n + 1}$.

{\bf \Cref{fact:discretize:measure_input,fact:discretize:measure_sequence}} show that the formulas $\varphi_{\FPA}(p,\, \bb,\, \ell)$ and $\varphi_{\FPA}^{(m)}(p,\, \bb,\, \ell)$ preserve the almost-everywhere continuity in the probabilistic metric space $(p,\, \bb,\, \ell) \sim (P,\, \bB \otimes L)$.

\begin{fact}
\label{fact:discretize:measure_input}
The measure $\Prx_{p,\, \bb,\, \ell} \big[\, (p,\, \bb,\, \ell) \in \mathbb{D}_{\FPA} \,\big] = 0$, where $\mathbb{D}_{\FPA}$ is the set of discontinuity points of the input {\FPA} formula $\varphi_{\FPA}(p,\, \bb,\, \ell)$.
\end{fact}

\begin{proof}
We claim that, with the $(\bb,\, \ell) \in [\gamma,\, \lambda]^{n + 1}$ held constant, the formula $\varphi_{\FPA}(p,\, \bb,\, \ell)$ is continuous with respect to $p \in [\gamma,\, \varphi_{1}(\gamma)]$, the monopolist $B_{1}$'s conditional value.
{\bf IF} the bid profile is {\em all-infimum} $\big\{ (\bb,\, \ell) = \bgamma \big\}$, the monopolist $B_{1}$ takes the conditional value $v_{1} = p$ (\Cref{def:pseudo}) and {\em will} be allocated $\alloc(\bgamma) = B_{1}$ (\blackref{monopoly}) -- the {\em identity} formula $\varphi_{\FPA}(p,\, \bgamma) = \Ex_{\alloc} [v_{\alloc(\bgamma)}] = p$.
{\bf ELSE} $\big\{ (\bb,\, \ell) \neq \bgamma \big\}$, each first-order bidder $\sigma \in \argmax(\bb,\, \ell)$ takes a {\em non-infimum} bid $b_{\sigma} = \max(\bb,\, \ell) > \gamma$ (\Cref{def:pseudo}) and thus a {\em $p$-irrelevant} value $v_{\sigma} = \varphi_{\sigma}(b_{\sigma})$; the allocated bidder must be one of them $\alloc(\bb,\, \ell) \in \argmax(\bb,\, \ell)$ -- the {\em constant} formula $\varphi_{\FPA}(p,\, \bb,\, \ell) = \Ex_{\alloc} [v_{\alloc(\bb,\, \ell)}]$.

Combining both cases {\bf IF}/{\bf ELSE} gives our claim.
As a consequence, the formula $\varphi_{\FPA}(p,\, \bb,\, \ell)$ can just have two types of discontinuity points.
{\bf FIRST}, those by {\em switches} of the allocated bidder $\alloc(\bb,\, \ell) \in \argmax(\bb,\, \ell)$. Those have a zero measure, because (\Cref{def:pseudo}) bid distributions $B_{\sigma}(b)$ for $\sigma \in [n] \cup \{L\}$ are continuous except the infimum bid $b = \gamma$ AND (\blackref{monopoly}) conditioned on the {\em all-infimum} bid profile $\big\{ (\bb,\, \ell) = \bgamma \big\}$, the monopolist {\em always} gets allocated $\alloc(\bgamma) = B_{1}$.
{\bf SECOND}, those by the (unswitched) allocated bidder's mapping $\varphi_{\alloc(\bb,\, \ell)}(b_{\alloc(\bb,\, \ell)})$ at a {\em non-infimum} first-order bid $b_{\alloc(\bb,\, \ell)} \in (\gamma,\, \lambda]$. Those again have a zero measure, because ({\bf \Cref{fact:discretize:function_property}}) each mapping $\varphi_{\sigma}(b)$ is bounded and continuous almost everywhere AND (\Cref{def:pseudo}) except the infimum bid $b \in (\gamma,\, \lambda]$, each bid distribution $B_{\sigma}(b)$ is continuous.

Both types {\bf FIRST}/{\bf SECOND} together have a zero measure. {\bf \Cref{fact:discretize:measure_input}} follows then.
\end{proof}

\begin{fact}
\label{fact:discretize:measure_sequence}
The measure $\Prx_{p,\, \bb,\, \ell} \big[\, (p,\, \bb,\, \ell) \in \mathbb{D}_{\FPA}^{(m)} \,\big] = 0$, where $\mathbb{D}_{\FPA}^{(m)}$ is the set of discontinuity points of the index-$m$ {\FPA} formula $\varphi_{\FPA}^{(m)}(p,\, \bb,\, \ell)$ for $m \geq 0$.
\end{fact}

\begin{proof}
Indeed, {\bf \Cref{fact:discretize:measure_input}} uses two ``input-specific'' arguments.
(i)~The formula $\varphi_{\FPA}(p,\, \bb,\, \ell)$ only has two kinds of discontinuity points. So is the index-$m$ formula $\varphi_{\FPA}^{(m)}(p,\, \bb,\, \ell)$, by the same analysis.
(ii)~The input mappings $\varphi_{\sigma}(b)$ are bounded and continuous almost everywhere. So are the index-$m$ mappings $\varphi_{\sigma}^{(m)}(b)$, by {\bf \Cref{fact:discretize:function_property}}.
Thus, readopting the proof of {\bf \Cref{fact:discretize:measure_input}} implies {\bf \Cref{fact:discretize:measure_sequence}}.
\end{proof}

Now we further draw {\em random} profiles $(p^{(m)},\, \bb^{(m)},\, \ell^{(m)}) \sim (P,\, \bB^{(m)} \otimes L^{(m)})$ from the {\em discretized} pseudo instances, one by one $m \geq 0$. {\bf \Cref{fact:discretize:distribution_convergence}} shows that the sequence of these profiles {\em converges in distribution} to the input {\em random} profile $(p,\, \bb,\, \ell) \sim (P,\, \bB \otimes L)$.

\begin{fact}
\label{fact:discretize:distribution_convergence}
The sequence of random profiles $\big\{ (p^{(m)},\, \bb^{(m)},\, \ell^{(m)}) \sim (P,\, \bB^{(m)} \otimes L^{(m)}) \big\}_{m \geq 0}$ converges in distribution ($\overset{{\tt dist}}{\longrightarrow}$) to the random input profile $(p,\, \bb,\, \ell) \sim (P,\, \bB \otimes L)$.
\end{fact}

\begin{proof}
The bid distributions $\bB \otimes L$ and/or $\bB^{(m)} \otimes L^{(m)}$ (\Cref{lem:pseudo_distribution}) can be reconstructed from the mappings $\bvarphi$ and/or $\bvarphi^{(m)}$ in term of integral formulas. All of them are {\em Riemann integrable}, since ({\bf \Cref{fact:discretize:function_property}}) all mappings $\bvarphi$ and $\bvarphi^{(m)}$ are bounded and continuous almost everywhere on $b \in [\gamma,\, \lambda]$.
In particular, each {\em discretized} mapping $\varphi_{\sigma}^{(m)}(b)$ as the $(m + 1)$-piece uniform interpolation of the input mapping $\varphi_{\sigma}(b)$, pointwise converges to $\varphi_{\sigma}(b)$ almost everywhere as $m \to +\infty$.
The sequence $\big\{ \bB^{(m)} \otimes L^{(m)} \big\}_{m \geq 0}$ thus pointwise converge to input bid distributions $\bB \otimes L$. This gives {\bf \Cref{fact:discretize:distribution_convergence}}.
\end{proof}

Based on {\bf \Cref{fact:discretize:distribution_convergence,fact:discretize:function_convergence,fact:discretize:measure_input,fact:discretize:measure_sequence}}, we can deduce that\footnote{The continuous mapping theorem (\Cref{thm:continuous_mapping}) considers an unmodified function $g$, rather than a modified one $\varphi_{\FPA}^{(m)}$ versus $\varphi_{\FPA}$ in our case. Thus we need to take the limit $\lim_{m \to +\infty}$ and $\lim_{t \to +\infty}$ twice.}
\begin{align}
    \lim_{\substack{m \to +\infty \\ t \to +\infty}}\, \Ex_{p^{(t)},\, \bb^{(t)},\, \ell^{(t)}} \big[\, \varphi_{\FPA}^{(m)}(p^{(t)},\, \bb^{(t)},\, \ell^{(t)}) \,\big]
    & ~=~ \lim_{m \to +\infty}\, \Ex_{p,\, \bb,\, \ell} \big[\, \varphi_{\FPA}^{(m)}(p,\, \bb,\, \ell) \,\big]
    \label{eq:discretize:1}\tag{D1} \\
    & ~=~ \Ex_{p,\, \bb,\, \ell} \big[\, \varphi_{\FPA}(p,\, \bb,\, \ell) \,\big]
    \label{eq:discretize:2}\tag{D2} \\
    & ~=~ \FPA(\text{\em input}) < +\infty. \phantom{\Big.}
    \nonumber
    \label{eq:discretize:3}\tag{D3}
\end{align}
\eqref{eq:discretize:1}: Apply the continuous mapping theorem (\Cref{thm:continuous_mapping}), as the sequence $\big\{ (p^{(t)},\, \bb^{(t)},\, \ell^{(t)}) \big\}_{t \geq 0}$ converges in distribution to the $(p,\, \bb,\, \ell) \sim (P,\, \bB \otimes L)$ ({\bf \Cref{fact:discretize:distribution_convergence}}) and each set of discontinuity points $\mathbb{D}_{\FPA}^{(m)}$ for $m \geq 0$ has a zero measure $\Prx_{p,\, \bb,\, \ell} \big[\, (p,\, \bb,\, \ell) \in \mathbb{D}_{\FPA}^{(m)} \,\big] = 0$ ({\bf \Cref{fact:discretize:measure_sequence}}). \\
\eqref{eq:discretize:2}: {\bf \Cref{fact:discretize:function_convergence}} that the sequence of {\em discretized} {\FPA} formulas $\big\{ \varphi_{\FPA}^{(m)}(p,\, \bb,\, \ell) \big\}_{m \geq 0}$ pointwise converges to the input {\FPA} formula $\varphi_{\FPA}(p,\, \bb,\, \ell)$ almost everywhere.\footnote{Rigorously, \Cref{eq:discretize:1} for $m \geq 0$ and \Cref{eq:discretize:2} are all Riemann integrable, because ({\bf \Cref{fact:discretize:function_property}}) formulas $\varphi_{\FPA}^{(m)}(p,\, \bb,\, \ell)$ for $m \geq 0$ and $\varphi_{\FPA}(p,\, \bb,\, \ell)$ each are bounded and continuous almost everywhere AND ({\bf \Cref{fact:discretize:measure_input,fact:discretize:measure_sequence}}) the sets of discontinuity points $\mathbb{D}_{\FPA}^{(m)}$ for $m \geq 0$ and $\mathbb{D}_{\FPA}$ each have a zero measure on $(p,\, \bb,\, \ell) \sim (P,\, \bB \otimes L)$.} \\
\eqref{eq:discretize:3}: $\FPA(\text{\em input}) \leq \OPT(\text{\em input}) < +\infty$, as the statement of \Cref{lem:discretize} promises.

\vspace{.1in}
It follows that $\lim_{m \to +\infty}\, \Ex_{p^{(m)},\, \bb^{(m)},\, \ell^{(m)}} \big[\, \varphi_{\FPA}^{(m)}(p^{(m)},\, \bb^{(m)},\, \ell^{(m)}) \,\big] = \FPA(\text{\em input})$. Given this and because
% Given that ``the {\LHS} of \eqref{eq:discretize:1} converges to $\eqref{eq:discretize:3} = \FPA(\text{\em input})$'' and
the error $\epsilon \in (0,\, 1)$ is a constant, for some {\em bounded} threshold $M_{\FPA} = M_{\FPA}(\epsilon) < +\infty$, every index $m \geq M_{\FPA}$ {\em discretized} pseudo instance yields a close enough auction {\SocialWelfare}
\begin{align*}
    \FPA(P,\, \bB^{(m)} \otimes L^{(m)})\;
    ~=~ \Ex_{p^{(m)},\, \bb^{(m)},\, \ell^{(m)}} \Big[\, \varphi_{\FPA}^{(m)}(p,\, \bb,\, \ell) \,\Big]\,
    ~=~ \FPA(\text{\em input}) \cdot e^{\pm \epsilon / 3}.\;
\end{align*}
Likewise,\footnote{Namely, the counterpart {\OPT} formulas $\varphi_{\OPT}^{(m)}(p,\, \bb,\, \ell)$ and $\varphi_{\OPT}(p,\, \bb,\, \ell)$ for profiles $(p,\, \bb,\, \ell) \in [\gamma,\, \varphi_{1}(\gamma)] \times [\gamma,\, \lambda]^{n + 1}$ satisfy the analogs of {\bf \Cref{fact:discretize:function_property,fact:discretize:measure_input,fact:discretize:function_convergence,fact:discretize:measure_sequence}}. These analogs together with {\bf \Cref{fact:discretize:distribution_convergence}} result in the {\OPT} version of the claim ``$\lim_{m \to +\infty}\, \Ex_{p^{(m)},\, \bb^{(m)},\, \ell^{(m)}} \big[\, \varphi_{\FPA}^{(m)}(p^{(m)},\, \bb^{(m)},\, \ell^{(m)}) \,\big] = \FPA(\text{\em input})$''.}
for some $M_{\OPT} = M_{\OPT}(\epsilon) < +\infty$, every index $m \geq M_{\OPT}(\epsilon)$ {\em discretized} pseudo instance yields a close enough optimal {\SocialWelfare}
\begin{align*}
    % \label{eq:discretize:5}\tag{D5}
    \OPT(P,\, \bB^{(m)} \otimes L^{(m)})
    ~=~ \Ex_{p^{(m)},\, \bb^{(m)},\, \ell^{(m)}} \Big[\, \varphi_{\OPT}^{(m)}(p,\, \bb,\, \ell) \,\Big]
    ~=~ \OPT(\text{\em input}) \cdot e^{\pm \epsilon / 3}.
\end{align*}
So the index $M \eqdef \max(M_{\FPA},\, M_{\OPT})$ {\em discretized} pseudo instance gives a close enough bound
% (Recall that any {\PoA} is $\leq 1$.)
\begin{align*}
    \PoA(P,\, \bB^{(M)} \otimes L^{(M)})
    ~=~ \PoA(\text{\em input}) \cdot e^{\pm 2 \cdot (\epsilon / 3)}
    ~=~ \PoA(\text{\em input}) \pm \epsilon.
    \hspace{2.05cm}
\end{align*}
This finishes the proof of \Cref{lem:discretize} assuming that the input pseudo instance $(P,\, \bB \otimes L)$ admits \blackref{finiteness} that $\lambda \leq \max(\bvarphi(\blambda)) < +\infty$. Below we move on to the general case.

\vspace{.1in}
\noindent
{\bf The General Case.}
When the input $(P,\, \bB \otimes L) \in \Bvalid$ violates \blackref{finiteness}, we instead consider the {\em truncated} bid distributions $\bB^{(t)} \otimes L^{(t)} = \{B_{\sigma}^{(t)}\}_{\sigma \in [n] \cup \{L\}}$ given by $B_{\sigma}^{(t)}(b) \equiv \min(\frac{B_{\sigma}(b)}{B_{\sigma}(t)},\, 1)$, where the parameter $t \in (\gamma,\, \lambda)$; see \Cref{fig:finiteness} for a visual aid.\footnote{The proof works for any supremum bid $\lambda \leq +\infty$, including the infinite one $\lambda = +\infty$. However, we demonstrate a bounded one $\lambda < +\infty$ in \Cref{fig:finiteness} for convenience.}

By construction, (\Cref{def:pseudo}) the mappings keep the same $\varphi_{\sigma}^{(t)}(b) = \varphi_{\sigma}(b)$ on the {\em truncated} bid support $b \in [\gamma,\, t]$, which means $\sup(\supp(\bB^{(t)} \otimes L^{(t)})) = t \leq \max(\bvarphi^{(t)}(\bt)) = \max(\bvarphi(\bt)) < +\infty$.
Precisely, together with the unmodified conditional value $P$, this {\em truncated} pseudo instance is valid $(P,\, \bB^{(t)} \otimes L^{(t)}) \in \Bvalid$ and satisfies \blackref{finiteness}.

% \begin{align*}
%     \FPA(\text{\em input})
%     ~=~ \Ex[P] \cdot \calB(\gamma)
%     ~+~ \sum_{\sigma \in [n] \cup \{L\}} \Big(\int_{\gamma}^{\lambda} \varphi_{\sigma}(b) \cdot \frac{B'_{\sigma}(b)}{B_{\sigma}(b)} \cdot \calB(b) \cdot \d b\Big)
%     ~<~ +\infty.
% \end{align*}

% We claim that the {\em truncated} pseudo instance $(P,\, \bB^{(t)} \otimes L^{(t)})$ approximates the bounded input {\SocialWelfares} $\FPA(\text{\em input}) \leq \OPT(\text{\em input}) < +\infty$ arbitrarily well, when the parameter $t \in (\gamma,\, \lambda)$ is close enough to the supremum bid $\lambda \leq +\infty$.

We claim that when the parameter $t \in (\gamma,\, \lambda)$ is close enough to the supremum bid $\lambda \leq +\infty$, the truncation has an arbitrarily small effect on the {\SocialWelfares} $\FPA(\text{\em input}) \leq \OPT(\text{\em input}) < +\infty$.

Following \Cref{lem:pseudo_welfare}, the input auction {\SocialWelfare} $\FPA(\text{\em input}) = \Ex[P] \cdot \calB(\gamma) + \int_{\gamma}^{\lambda} f(b) \cdot \d b$, for the integrand $f(b) \eqdef \sum_{\sigma} \varphi_{\sigma}(b) \cdot \frac{B'_{\sigma}(b)}{B_{\sigma}(b)} \cdot \calB(b)$.
This formula $< +\infty$ is Lebesgue integrable, hence $\lim_{t \nearrow \lambda} \int_{t}^{\lambda} f(b) \cdot \d b = 0$.
Likewise, the {\em truncated} auction {\SocialWelfare}
\begin{align*}
    \FPA(P,\, \bB^{(t)} \otimes L^{(t)})
    & ~=~ \Ex[P] \cdot \calB^{(t)}(\gamma)
    + \int_{\gamma}^{t} \Big(\sum_{\sigma} \varphi_{\sigma}^{(t)}(b) \cdot \tfrac{{B_{\sigma}^{(t)}}'(b)}{B_{\sigma}^{(t)}(b)} \cdot \calB^{(t)}(b)\Big) \cdot \d b ~~~~ \\
    & ~=~ \frac{1}{\calB(t)} \cdot \Ex[P] \cdot \calB(\gamma)
    + \frac{1}{\calB(t)} \cdot \int_{\gamma}^{t} f(b) \cdot \d b \\
    & ~=~ \frac{1}{\calB(t)} \cdot \Big(\FPA(\text{\em input}) - \int_{t}^{\lambda} f(b) \cdot \d b\Big).
\end{align*}
We have $\lim_{t \nearrow \lambda} \FPA(P,\, \bB^{(t)} \otimes L^{(t)}) = \FPA(\text{\em input})$, given continuity $\lim_{t \nearrow \lambda} \calB(t) = 1$ (\Cref{def:pseudo}).
% since the first-order bid CDF $\calB(b)$ is continuous on $b \in [\gamma,\, \lambda]$ (\Cref{lem:bid_distribution}).
Accordingly, for some threshold $T_{\FPA} = T_{\FPA}(\epsilon) \in (\gamma,\, \lambda)$, any parameter $t \in [T_{\FPA},\, \lambda)$ yields a close enough auction {\SocialWelfare} $\FPA(P,\, \bB^{(t)} \otimes L^{(t)}) = \FPA(\text{\em input}) \cdot e^{\pm \epsilon / 3}$.

% Since $\lim_{t \nearrow \lambda} \int_{t}^{\lambda} f(b) \cdot \d b = 0$ and t

\begin{figure}[t]
    \centering
    \includegraphics[width = .9\textwidth]{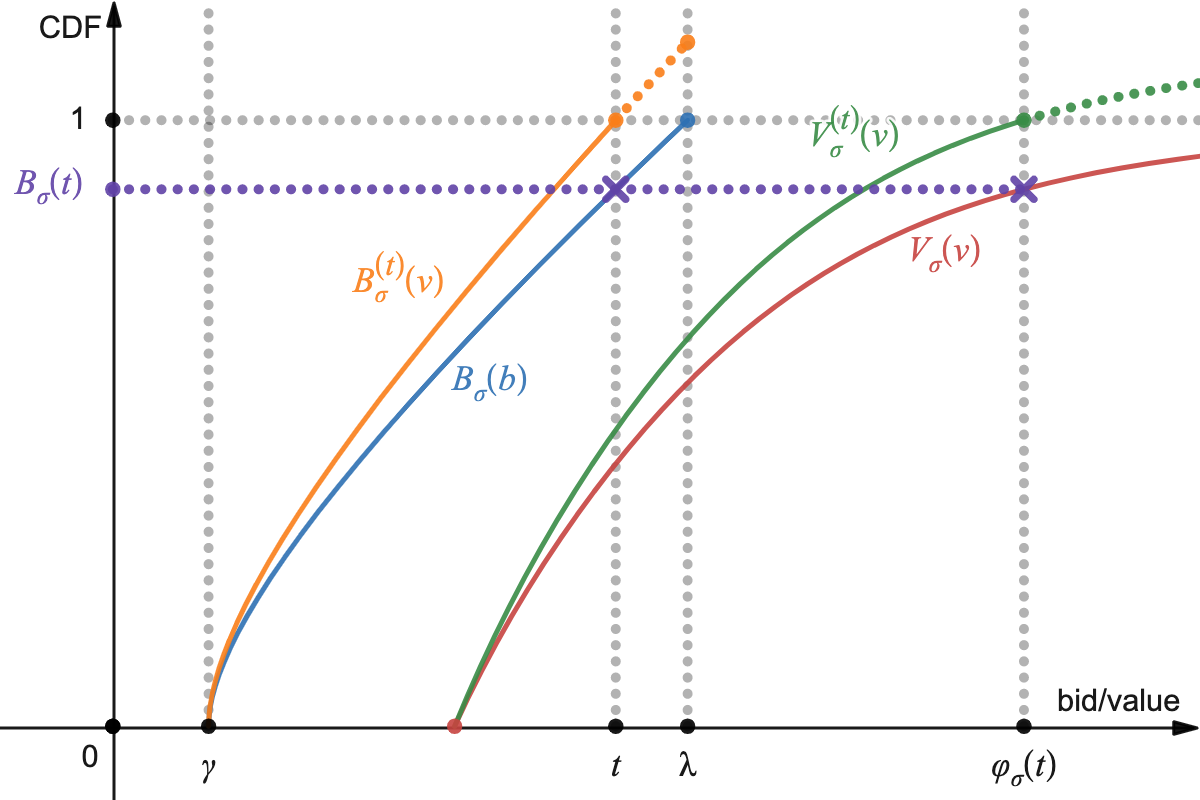}
    \caption{Diagram of the transformation from $(P,\, \bB)$ into $(P,\, \tilde{\bB})$, making a valid pseudo instance further satisfy {\bf finiteness} (\Cref{def:discretize}). For convenience, the demonstrated supremum bid is bounded $\lambda < +\infty$. But the transformation works for an arbitrary one $\lambda \leq +\infty$.}
    \label{fig:finiteness}
\end{figure}

Due to \Cref{lem:pseudo_welfare}, the input optimal {\SocialWelfare} $\OPT(\text{\em input}) = \gamma + \int_{\gamma}^{+\infty} (1 - \prod_{\sigma} V_{\sigma}(v)) \cdot \d v$.
This formula $< +\infty$ is Lebesgue integrable, hence $\int_{\gamma}^{+\infty} (1 - V_{\sigma}(v)) \cdot \d v < +\infty$.
Since the bid-to-value mappings are invariant $\varphi_{\sigma}^{(t)}(b) = \varphi_{\sigma}(b)$ over the {\em truncated} support $b \in [\gamma,\, t]$,
(\Cref{lem:value_dist}) the bid distributions $B_{\sigma}^{(t)}(b) \equiv \min(\frac{B_{\sigma}(b)}{B_{\sigma}(t)},\, 1)$ induce the value distributions $V_{\sigma}^{(t)}(v) \equiv \min(\frac{V_{\sigma}(v)}{B_{\sigma}(t)},\, 1)$.
Given continuity $\lim_{t \nearrow \lambda} B_{\sigma}(t) = 1$ (\Cref{def:pseudo}), we deduce that $\lim_{t \nearrow \lambda} \int_{\gamma}^{+\infty} \big|V_{\sigma}^{(t)}(v) - V_{\sigma}(v)\big| \cdot \d v = 0$.
Based on \Cref{lem:pseudo_welfare}, we can bound the {\em truncated} optimal {\SocialWelfare}:
\begin{align*}
    \OPT(P,\, \bB^{(t)} \otimes L^{(t)})
    & ~=~ \OPT(\text{\em input})
    - \int_{\gamma}^{+\infty} P(v) \cdot \Big(\prod_{\sigma} V_{\sigma}^{(t)}(v) - \prod_{\sigma} V_{\sigma}(v)\Big) \cdot \d v \\
    & ~=~ \OPT(\text{\em input})
    - \sum_{\sigma} \int_{\gamma}^{+\infty} \Big(P(v) \cdot \big(V_{\sigma}^{(t)}(v) - V_{\sigma}(v)\big) \cdot \\
    & \phantom{~=~ \OPT(\text{\em input})
    - \int_{\gamma}^{+\infty} \sum_{\sigma} \Big(} \prod_{k < \sigma} V_{k}^{(t)}(v) \cdot \prod_{k > \sigma} V_{k}(v)\Big) \cdot \d v \\
    & ~=~ \OPT(\text{\em input})
    \pm \sum_{\sigma} \int_{\gamma}^{+\infty} \big| V_{\sigma}^{(t)}(v) - V_{\sigma}(v) \big| \cdot \d v.
\end{align*}
Here the last step uses the relaxation $|P(v)|,\, |V_{\sigma}^{(t)}(v)|,\, |V_{\sigma}(v)| \leq 1$.

Now it is clear that $\lim_{t \nearrow \lambda} \OPT(P,\, \bB^{(t)} \otimes L^{(t)}) = \OPT(\text{\em input})$.
Accordingly, for some threshold $T_{\OPT} = T_{\OPT}(\epsilon) \in (\gamma,\, \lambda)$, any parameter $t \in [T_{\OPT},\, \lambda)$ yields a close enough optimal {\SocialWelfare} $\OPT(P,\, \bB^{(t)} \otimes L^{(t)}) = \OPT(\text{\em input}) \cdot e^{\pm \epsilon / 3}$.

% \red{Likewise, for some threshold $T_{\OPT} = T_{\OPT}(\epsilon) \in (\gamma,\, \lambda)$, every parameter $t \in [T_{\OPT},\, \lambda)$ results in a close enough optimal {\SocialWelfare} $\OPT(P,\, \bB^{(t)} \otimes L^{(t)}) = \OPT(P,\, \bB \otimes L) \cdot e^{\pm \epsilon / 3}$.}

Combining everything together, the parameter $T \eqdef \max(T_{\FPA},\, T_{\OPT})$ pseudo instance ensures a close enough bound $\PoA(P,\, \bB^{(T)} \otimes L^{(T)}) = \PoA(\text{\em input}) \cdot e^{\pm 2 \cdot (\epsilon / 3)} = \PoA(\text{\em input}) \pm \epsilon$ and satisfies \blackref{finiteness}.
Through the \blackref{alg:discretize} reduction, we can transform this {\em truncated} pseudo instance $(P,\, \bB^{(T)} \otimes L^{(T)})$ into a {\em discretized} pseudo instance $(P,\, \tilde{\bB} \otimes \tilde{L})$, as desired.
% that satisfies \blackref{finiteness} and \blackref{piecewise}.
In particular, the resulting bound $\PoA(P,\, \tilde{\bB} \otimes \tilde{L}) = \PoA(P,\, \bB^{(T)} \otimes L^{(T)}) \pm \epsilon = \PoA(\text{\em input}) \pm 2\epsilon$.
Reducing the error $\epsilon \in (0,\, 1)$ by a factor of $2$ finishes the proof of \Cref{lem:discretize}.
\end{proof}

\begin{comment}

\[
    \PoA(P,\, \tilde{\bB} \otimes \tilde{L})
    ~=~ \PoA(P,\, \bB^{(T)} \otimes L^{(T)}) \pm \epsilon
    ~=~ \PoA(\text{\em input}) \pm 2\epsilon.
\]

pseudo instance $(P,\, \bB^{(t)} \otimes L^{(t)})$ for $t \in [T_{\OPT},\, \lambda)$ yields a close enough auction {\SocialWelfare}

\red{When the supremum bid is unbounded $\lambda = +\infty$,
we can adjust our partition to $\lambda_{j} = \gamma + j / 2^{m}$ for each $j \in [0:\, 4^{m}]$, and then reuse the {\blackref{alg:discretize}} reduction.}
Lebesgue integral

\begin{align*}
    \FPA(P,\, \bB^{(t)} \otimes L^{(t)}) ~=~ \Ex_{p,\, \bb,\, \ell} \Big[\, \varphi_{\FPA}(p,\, \bb,\, \ell) \cdot \indicator((\bb,\, \ell) \leq t^{\otimes n + 1}) \,\Big] \cdot \frac{1}{\calB(t)}.
\end{align*}

\[
    \OPT(\tilde{P},\, \tilde{\bB} \otimes \tilde{L}) ~\geq~ \Ex_{p,\, \bb,\, \ell} \Big[\, \varphi_{\OPT}(p,\, \bb,\, \ell) \cdot \indicator((\bb,\, \ell) \leq t^{\otimes n + 1}) \,\Big].
\]

\[
    \OPT(\tilde{P},\, \tilde{\bB} \otimes \tilde{L}) ~\leq~ \Ex_{p,\, \bb,\, \ell} \Big[\, \varphi_{\OPT}(p,\, \bb,\, \ell) \,\Big].
\]

\end{comment}

\subsection{{\translate}: Vanishing the minimum winning bids}
\label{subsec:translate}

This subsection presents the {\blackref{alg:translate}} reduction (see \Cref{fig:alg:translate,fig:translate} for its description and a visual aid), which transforms a {\em discretized} pseudo instance (\Cref{def:discretize}) into a more special {\em translated} pseudo instance (\Cref{def:translate}).

\begin{definition}[Translated pseudo instances]
\label{def:translate}
A {\em discretized} pseudo instance $(P,\, \bB \otimes L)$ from \Cref{def:discretize} is further called {\em translated} when (\term[{\bf nil infimum bid}]{nil}) the infimum bid is nil $\gamma = 0$.
\end{definition}

\Cref{lem:translate} presents performance guarantees of the {\blackref{alg:translate}} reduction.

\begin{lemma}[{\translate}; \Cref{fig:alg:translate}]
\label{lem:translate}
Under reduction $(\tilde{P},\, \tilde{\bB} \otimes \tilde{L}) \gets \translate(P,\, \bB \otimes L)$:
\begin{enumerate}[font = {\em\bfseries}]
    \item\label{lem:translate:property}
    The output $(\tilde{P},\, \tilde{\bB} \otimes \tilde{L})$ is a translated pseudo instance.

    \item\label{lem:translate:poa}
    A (weakly) worse bound is yielded $\PoA(\tilde{P},\, \tilde{\bB} \otimes \tilde{L}) \leq \PoA(P,\, \bB \otimes L)$.
\end{enumerate}
\end{lemma}

\afterpage{
\begin{figure}[t]
    \centering
    \begin{mdframed}
    Reduction $\term[\translate]{alg:translate}(P,\, \bB \otimes L)$

    \begin{flushleft}
    {\bf Input:} A (generic) {\em discretized} pseudo instance $(P,\, \bB \otimes L)$
    \hfill \Cref{def:discretize}
    
    \vspace{.05in}
    {\bf Output:} A {\em translated} pseudo instance $(\tilde{P},\, \tilde{\bB} \otimes \tilde{L})$.
    \hfill \Cref{def:translate}

    \begin{enumerate}
        \item Define $\tilde{P}(v) \equiv P(v + \gamma)$ and $\tilde{B}_{\sigma}(b) \equiv B_{\sigma}(b + \gamma)$ for each bidder $\sigma \in [n] \cup \{L\}$.

        \item {\bf Return} $(\tilde{P}, \tilde{\bB} \otimes \tilde{L})$.
    \end{enumerate}
    \end{flushleft}
    \end{mdframed}
    \caption{The {\translate} reduction.
    \label{fig:alg:translate}}
\end{figure}
\begin{figure}
    \centering
    \includegraphics[width = .65\textwidth]{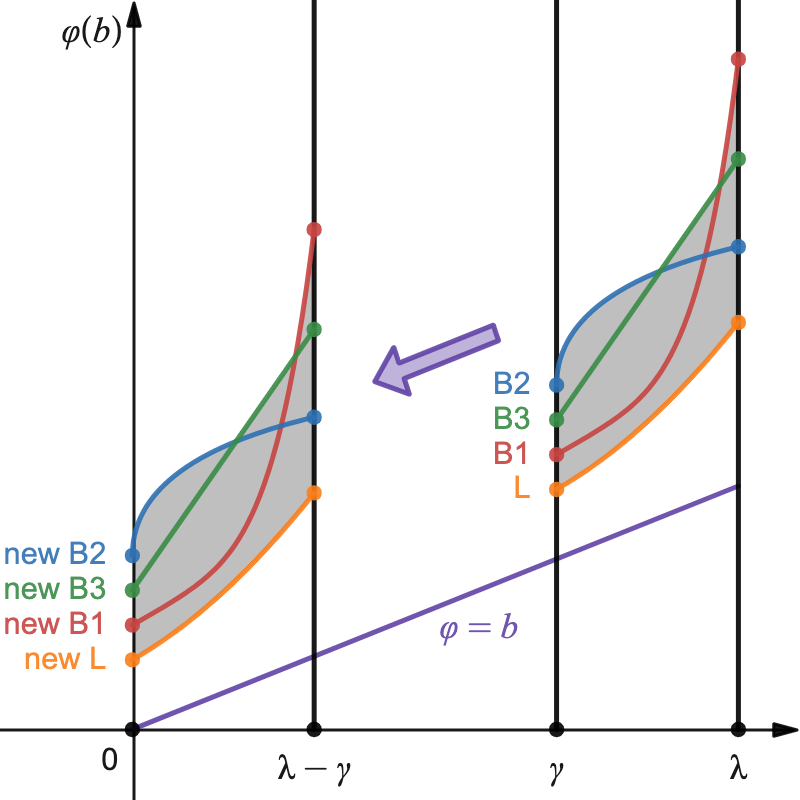}
    \caption{Diagram of the {\translate} reduction (\Cref{fig:alg:translate}), which indeed works for all VALID pseudo instances $(P,\, \bB \otimes L) \in \Bvalid$ (\Cref{def:pseudo}), not only the {\em discretized} ones (\Cref{def:discretize}).
    On the whole shifted bid support $b \in [0,\, \lambda - \gamma]$, each bidder $\sigma \in [n] \cup \{L\}$ has a shifted bid-to-value mapping $\tilde{\varphi}_{\sigma}(b) \equiv \varphi_{\sigma}(b + \gamma) - \gamma$.
    \label{fig:translate}}
\end{figure}
\clearpage}

\begin{proof}
See \Cref{fig:translate} for a visual aid.
Let us verify {\bf \Cref{lem:translate:property,lem:translate:poa}} one by one.

\vspace{.1in}
\noindent
{\bf \Cref{lem:translate:property}.}
The {\blackref{alg:translate}} reduction shifts the bid/value spaces each by a distance of $-\gamma$, inducing a nil infimum bid $\gamma - \gamma = 0$ (\blackref{nil}). Clearly, each shifted mapping $\tilde{\varphi}_{\sigma}(b) = \varphi_{\sigma}(b + \gamma) - \gamma$ for $\sigma \in [n] \cup \{L\}$ (akin to the input mapping $\varphi_{\sigma}$) is increasing on the shifted support $b \in [0,\, \lambda - \gamma]$ (\blackref{monotonicity}), and the shifted conditional value $\tilde{P} = P - \gamma$ (akin to the input one $\supp(P) \subseteq [\gamma,\, \varphi_{1}(\gamma)]$) ranges between $\supp(\tilde{P}) \subseteq [\gamma - \gamma,\, \varphi_{1}(\gamma) - \gamma] = [0,\, \tilde{\varphi}_{1}(\gamma)]$ (\blackref{boundedness}).
The shifted pseudo instance $(\tilde{P},\, \tilde{\bB} \otimes \tilde{L})$ must preserve (\Cref{def:discretize}) \blackref{piecewise} of the mappings $\tilde{\bvarphi} = \{\tilde{\varphi}_{\sigma}\}_{\sigma \in [n] \cup \{L\}}$ and \blackref{finiteness}, namely $\lambda - \gamma \leq \max(\tilde{\bvarphi}(\blambda - \bgamma)) = \max(\bvarphi(\blambda)) - \gamma < +\infty$. Thus, {\bf \Cref{lem:translate:property}} follows.

\vspace{.1in}
\noindent
{\bf \Cref{lem:translate:poa}.}
Because the bid/value spaces each are shifted by a distance of $-\gamma$, the auction/optimal {\SocialWelfares} (\Cref{lem:pseudo_welfare}) each drop by an amount of $\gamma$. We thus deduce that
\[
    \PoA(\tilde{P},\, \tilde{\bB} \otimes \tilde{L})
    % = \frac{\FPA(\tilde{P},\, \tilde{\bB} \otimes \tilde{L})\;}{\OPT(\tilde{P},\, \tilde{\bB} \otimes \tilde{L})}
    ~=~ \frac{\FPA(P,\, \bB \otimes L)\; - \gamma}{\OPT(P,\, \bB \otimes L) - \gamma}
    ~\leq~ \frac{\FPA(P,\, \bB \otimes L)\;}{\OPT(P,\, \bB \otimes L)}
    ~=~ \PoA(P,\, \bB \otimes L).
\]
Here the inequality holds, because the {\SocialWelfares} $\OPT(P,\, \bB \otimes L) \geq \FPA(P,\, \bB \otimes L)$ are at least the infimum bid $\gamma$ (cf.\ \Cref{lem:pseudo_welfare}). {\bf \Cref{lem:translate:poa}} and \Cref{lem:translate} follow then.
\end{proof}

\Cref{lem:translate_welfare} presents alternative {\SocialWelfare} formulas for {\em translated} pseudo instances, which supplement \Cref{lem:pseudo_welfare} and will be more convenient in several places.

\begin{lemma}[{\SocialWelfares}]
\label{lem:translate_welfare}
For a translated pseudo instance $(P,\, \bB \otimes L)$, the expected auction {\SocialWelfare} $\FPA(P,\, \bB \otimes L)$ is given by
\begin{align*}
    \FPA(P,\, \bB \otimes L)\;
    & ~=~ \Ex[P] \cdot \calB(0) ~+~ \int_{0}^{\lambda} \sum_{i \in [n]} \Big(\frac{\varphi_{i}(b)}{\varphi_{L}(b) - b} - \frac{\varphi_{i}(b) - \varphi_{L}(b)}{\varphi_{i}(b) - b}\Big) \cdot \calB(b) \cdot \d b~~~~~\;\; \\
    & \phantom{~=~ \Ex[P] \cdot \calB(0)} ~-~ \int_{0}^{\lambda} (n - 1) \cdot \frac{\varphi_{L}(b)}{\varphi_{L}(b) - b} \cdot \calB(b) \cdot \d b.
\end{align*}

Regarding the underlying partition $\bLambda = [0 = \lambda_{0},\, \lambda_{1}) \cup
[\lambda_{1},\, \lambda_{2}) \cup \dots \cup
[\lambda_{m},\, \lambda_{m + 1} = \lambda]$,
represent the piecewise constant bid-to-value mappings $\bvarphi = \{\varphi_{\sigma}\}_{\sigma \in [n] \cup \{L\}}$ as a bid-to-value table $\bPhi = \big[\phi_{\sigma,\, j}\big]$ for $(\sigma,\, j) \in ([n] \cup \{L\}) \times [0: m]$, i.e., $\varphi_{\sigma}(b) \equiv \phi_{\sigma,\, j}$ on each index-$j$ piece $b \in [\lambda_{j},\, \lambda_{j + 1})$.
Then the expected optimal {\SocialWelfare} $\OPT(P,\, \bB \otimes L)$ is given by
\begin{align*}
    \OPT(P,\, \bB \otimes L)
    ~=~ \int_{0}^{+\infty} \Big(1 - P(v) \cdot \prod_{(\sigma,\, j) \,\in\, \bPhi} \big(1 - \omega_{\sigma,\, j} \cdot \indicator(v < \phi_{\sigma,\, j})\big)\Big) \cdot \d v.
    \hspace{2.27cm}
\end{align*}
where the probabilities $\omega_{\sigma,\, j} \eqdef 1 - \frac{B_{\sigma}(\lambda_{j})}{B_{\sigma}(\lambda_{j + 1})}$ can be reconstructed from the partition-table tuple $(\bLambda,\, \bPhi)$ following \Cref{lem:pseudo_distribution}.
\end{lemma}

% , each bid-to-value mapping $\varphi_{\sigma}(b)$ for $\sigma \in [n] \cup \{L\}$ is a {\em piecewise constant} function under this partition.
% \begin{align*}
%     & \frac{B'_{i}(b)}{B_{i}(b)}
%     ~=~ \frac{\d}{\d b}\big(\ln B_{i}(b)\big)
%     ~=~ \big(\varphi_{L}(b) - b\big)^{-1} - \big(\varphi_{i}(b) - b\big)^{-1},
%     \qquad\qquad \forall i \in [n], \\
%     & \frac{L'(b)}{L(b)}
%     ~=~ \frac{\d}{\d b}\big(\ln L(b)\big)
%     ~=~ \sum_{i \in [n]} \big(\varphi_{i}(b) - b\big)^{-1} - (n - 1) \cdot \big(\varphi_{L}(b) - b\big)^{-1}.
% \end{align*}

\begin{proof}
We deduce from \Cref{lem:pseudo_distribution} that $\frac{B'_{i}(b)}{B_{i}(b)}
= \frac{\d}{\d b}\big(\ln B_{i}(b)\big)
= \big(\varphi_{L}(b) - b\big)^{-1} - \big(\varphi_{i}(b) - b\big)^{-1}$ for $i \in [n]$ and
$\frac{L'(b)}{L(b)}
= \frac{\d}{\d b}\big(\ln L(b)\big)
= \sum_{i \in [n]} \big(\varphi_{i}(b) - b\big)^{-1} - (n - 1) \cdot \big(\varphi_{L}(b) - b\big)^{-1}$.
Thus we can rewrite the auction {\SocialWelfare} formula from \Cref{lem:pseudo_welfare} (with \blackref{nil} $\gamma = 0$) as follows:
\begin{align*}
    \FPA(P,\, \bB \otimes L)
    & ~=~ \Ex[P] \cdot \calB(0) ~+~ \sum_{\sigma \in [n] \cup \{L\}} \Big(\int_{0}^{\lambda} \varphi_{\sigma}(b) \cdot \frac{B'_{\sigma}(b)}{B_{\sigma}(b)} \cdot \calB(b) \cdot \d b\Big) \\
    & ~=~ \Ex[P] \cdot \calB(0) ~+~ \sum_{i \in [n]} \Big(\int_{0}^{\lambda} \Big(\frac{\varphi_{i}(b)}{\varphi_{L}(b) - b} - \frac{\varphi_{i}(b)}{\varphi_{i}(b) - b}\Big) \cdot \calB(b) \cdot \d b\Big) \\
    & \phantom{~=~ \Ex[P] \cdot \calB(0)} ~+~ \int_{0}^{\lambda} \Big(\sum_{i \in [n]} \frac{\varphi_{L}(b)}{\varphi_{i}(b) - b} - (n - 1) \cdot \frac{\varphi_{L}(b)}{\varphi_{L}(b) - b}\Big) \cdot \calB(b) \cdot \d b,
    % \\
    % & ~=~ \Ex[P] \cdot \calB(0) ~+~ \sum_{i \in [n]} \Big(\int_{0}^{\lambda} \Big(\frac{\varphi_{i}(b)}{\varphi_{L}(b) - b} - \frac{\varphi_{i}(b) - \varphi_{L}(b)}{\varphi_{i}(b) - b}\Big) \cdot \calB(b) \cdot \d b\Big) \\
    % & \phantom{~=~ \Ex[P] \cdot \calB(0)} ~-~ \int_{0}^{\lambda} (n - 1) \cdot \frac{\varphi_{L}(b)}{\varphi_{L}(b) - b} \cdot \calB(b) \cdot \d b.
\end{align*}
which after being rearranged gives the formula in the statement of \Cref{lem:translate_welfare}.

Regarding the partition $\bLambda = [0 \equiv \lambda_{0},\, \lambda_{1}) \cup
[\lambda_{1},\, \lambda_{2}) \cup \dots \cup
[\lambda_{m},\, \lambda_{m + 1} \equiv \lambda]$ and the bid-to-value table $\bPhi = \big[\phi_{\sigma,\, j}\big]$, namely $\varphi_{\sigma}(b) \equiv \phi_{\sigma,\, j}$ on each piece $b \in [\lambda_{j},\, \lambda_{j + 1})$, for $(\sigma,\, j) \in ([n] \cup \{L\}) \times [0: m]$,
we can formulate the conditional probabilities $\omega_{\sigma,\, j} = 1 - \frac{B_{\sigma}(\lambda_{j})}{B_{\sigma}(\lambda_{j + 1})}$ through \Cref{lem:pseudo_distribution} as follows:
\begin{align*}
    & \omega_{i,\, j}
    % = 1 - \frac{B_{i}(\lambda_{j})}{B_{i}(\lambda_{j + 1})}
    ~=~ 1 - \exp\Big(-\int_{\lambda_{j}}^{\lambda_{j + 1}}
    \Big(\big(\phi_{L,\, j} - b\big)^{-1} - \big(\phi_{i,\, j} - b\big)^{-1}\Big) \cdot \d b\Big),
    \qquad \forall i \in [n]. \\
    & \omega_{L,\, j}
    % = 1 - \frac{L(\lambda_{j})}{L(\lambda_{j + 1})}
    ~=~ 1 - \exp\Big(-\int_{\lambda_{j}}^{\lambda_{j + 1}}
    \Big(\sum_{i \in [n]} \big(\phi_{i,\, j} - b\big)^{-1} - (n - 1) \cdot \big(\phi_{L,\, j} - b\big)^{-1}\Big) \cdot \d b\Big). \hspace{1.92cm}
\end{align*}

Following \Cref{lem:pseudo_welfare} (with \blackref{nil} $\gamma = \lambda_{0} \equiv 0$), the optimal {\SocialWelfare} is given by
$\OPT(P,\, \bB \otimes L) = \int_{0}^{+\infty} \big(1 - P(v) \cdot \prod_{\sigma \in [n] \cup \{L\}} \Prx_{b_{\sigma} \sim B_{\sigma}}\big[\, (b_{\sigma} \leq \lambda_{0}) \vee (\varphi_{\sigma}(b_{\sigma}) \leq v) \,\big]\big) \cdot \d v$.
Because we are considering piecewise constant mappings, namely $\varphi_{\sigma}(b) \equiv \phi_{\sigma,\, j}$ on each piece $b \in [\lambda_{j},\, \lambda_{j + 1})$, we can deduce that for $\sigma \in [n] \cup \{L\}$ and any nonnegative value $v \geq 0$,
\begin{align*}
    \Prx_{b_{\sigma} \sim B_{\sigma}}\big[\, (b_{\sigma} \leq \lambda_{0}) \vee (\varphi_{\sigma}(b_{\sigma}) \leq v) \,\big]
    & ~=~ B_{\sigma}(\lambda_{0}) + \sum_{j \in [0: m]} (B_{\sigma}(\lambda_{j + 1}) - B_{\sigma}(\lambda_{j})) \cdot \indicator(\phi_{\sigma,\, j} \leq v)~\, \\
    & ~=~ 1 - \sum_{j \in [0: m]} (B_{\sigma}(\lambda_{j + 1}) - B_{\sigma}(\lambda_{j})) \cdot \indicator(v < \phi_{\sigma,\, j}) \\
    & ~=~ \prod_{j \in [0: m]} \Big(1 - \big(1 - \tfrac{B_{\sigma}(\lambda_{j})}{B_{\sigma}(\lambda_{j + 1})}\big) \cdot \indicator(v < \phi_{\sigma,\, j})\Big)
\end{align*}
Here the first/third steps holds because (\blackref{monotonicity}) the mappings $\varphi_{\sigma}(b)$ are increasing over the bid support $b \in [0,\, \lambda]$, namely $\phi_{\sigma,\, 0} \leq \dots \leq \phi_{\sigma,\, j} \leq \dots \leq \phi_{\sigma,\, m}$. And the second step uses the boundary conditions $B_{\sigma}(\lambda_{m + 1}) = 1$ at the supremum bid $\lambda_{m + 1} \equiv \lambda$.

Applying the above identities (together with $\omega_{\sigma,\, j} = 1 - \frac{B_{\sigma}(\lambda_{j})}{B_{\sigma}(\lambda_{j + 1})}$) to the optimal {\SocialWelfare} formula gives the alternative formula claim in \Cref{lem:translate_welfare}.
This finishes the proof.
\end{proof}

\subsection{{\layer}: Rearranging the bid-to-value mappings}
\label{subsec:layer}

This subsection shows the {\blackref{alg:layer}} reduction (see \Cref{fig:alg:layer,fig:layer} for its description and for a visual aid), which transforms a {\em translated} pseudo instance (\Cref{def:translate}) into a more special {\em layered} pseudo instance (\Cref{def:layer}). Recall \Cref{lem:pseudo_mapping} that over the bid support $b \in [0,\, \lambda]$, the pseudo mapping is dominated $\varphi_{L}(b) \leq \varphi_{i}(b)$ by the real mappings $i \in [n]$.

\begin{definition}[Layered pseudo instances]
\label{def:layer}
A {\em translated} pseudo instance $(P,\, \bB \otimes L)$ from \Cref{def:translate} is further called {\em layered} when (\term[{\bf layeredness}]{layeredness}) the {\em real} bid-to-value mappings $\{\varphi_{i}\}_{i \in [n]}$ are ordered $\varphi_{1}(b) \geq \dots \geq \varphi_{i}(b) \geq \dots \geq \varphi_{n}(b) \geq \varphi_{L}(b)$ over the bid support $b \in [0,\, \lambda]$.
\end{definition}

\Cref{lem:layer} shows performance guarantees of the {\blackref{alg:layer}} reduction.

\begin{lemma}[{\layer}; \Cref{fig:alg:layer}]
\label{lem:layer}
Under reduction $(P,\, \tilde{\bB} \otimes \tilde{L}) \gets \layer(P,\, \bB \otimes L)$:
\begin{enumerate}[font = {\em\bfseries}]
    \item\label{lem:layer:property}
    The output $(P,\, \tilde{\bB} \otimes \tilde{L})$ is a layered pseudo instance; the conditional value $P$ is unmodified.

    \item\label{lem:layer:poa}
    The bound keeps the same $\PoA(P,\, \tilde{\bB} \otimes \tilde{L}) = \PoA(P,\, \bB \otimes L)$.
\end{enumerate}
\end{lemma}

\afterpage{
\begin{figure}[t]
    \centering
    \begin{mdframed}
    Reduction $\term[\layer]{alg:layer}(P,\, \bB \otimes L)$

    \begin{flushleft}
    {\bf Input:}
    A (generic) {\em translated} pseudo instance $(P,\, \bB \otimes L)$.
    \hfill \Cref{def:translate}
    
    \vspace{.05in}
    {\bf Output:}
    A {\em layered} pseudo instance $(P,\, \tilde{\bB} \otimes \tilde{L})$.
    \hfill \Cref{def:layer}

    \begin{enumerate}
        \item\label{alg:layer:real_mapping}
        Define the real mappings $\tilde{\varphi}_{i} \equiv \varphi_{(i)}$ as the pointwise {\em ordering} of input real mappings $\{\varphi_{i}\}_{i \in [n]}$, namely $\varphi_{(1)}(b) \geq \dots \geq \varphi_{(i)}(b) \geq \dots \geq \varphi_{(n)}(b)$, for $b \in [0,\, \lambda]$.

        \item\label{alg:layer:pseudo_mapping}
        Reuse the same pseudo mapping $\tilde{\varphi}_{L} \equiv \varphi_{L}$. (Given Line~\ref{alg:layer:real_mapping} and \Cref{lem:pseudo_mapping}, this remains the {\em dominated} mapping $\min(\tilde{\bvarphi}) \equiv \min(\tilde{\bvarphi}_{-L},\, \tilde{\varphi}_{L}) \equiv \min(\bvarphi_{-L},\, \varphi_{L}) \equiv \varphi_{L} \equiv \tilde{\varphi}_{L}$.)

        \item\label{alg:layer:distribution}
        {\bf Return} $(P,\, \tilde{\bB} \otimes \tilde{L})$, namely only the bid distributions $\tilde{\bB} \otimes \tilde{L} = \{\tilde{B}_{\sigma}\}_{\sigma \in [n] \cup \{L\}}$ are modified and (\Cref{lem:pseudo_distribution}) are reconstructed from the {\em layered} mappings $\tilde{\bvarphi}$.
    \end{enumerate}
    \end{flushleft}
    \end{mdframed}
    \caption{The {\layer} reduction.}
    \label{fig:alg:layer}
\end{figure}
\begin{figure}
    \centering
    \subfloat[\label{fig:layer:old}
    The input {\em increasing} mappings $\bvarphi = \{\varphi_{\sigma}\}_{\sigma \in [n] \cup \{L\}}$.]{
    \includegraphics[width = .49\textwidth]
    {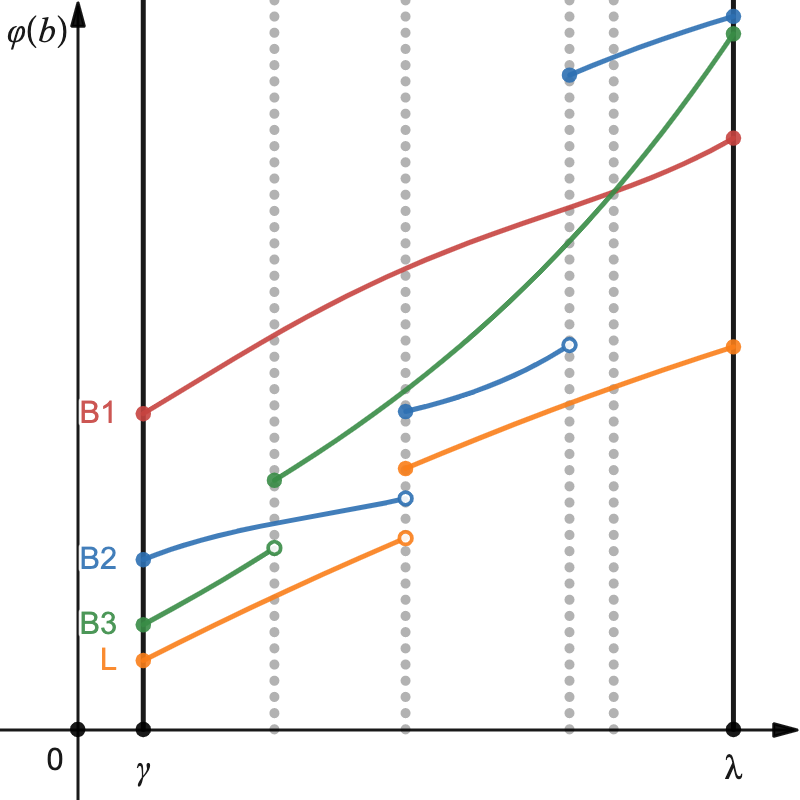}}
    \hfill
    \subfloat[\label{fig:layer:new}
    The output {\em layered} mappings $\tilde{\bvarphi} = \{\tilde{\varphi}_{\sigma}\}_{\sigma \in [n] \cup \{L\}}$.]{
    \includegraphics[width = .49\textwidth]
    {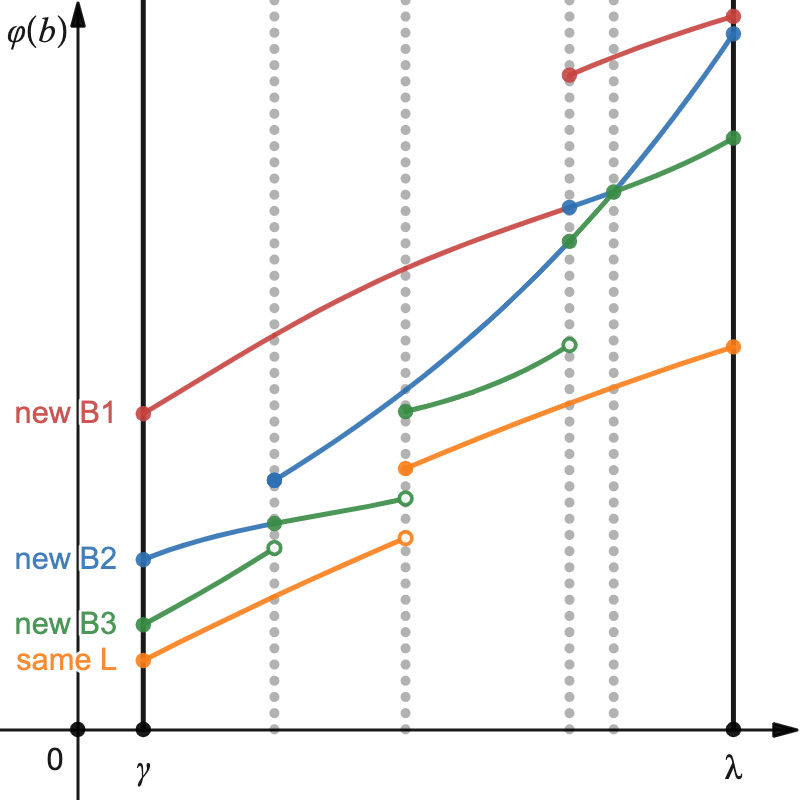}}
    \caption{Diagram of the {\layer} reduction (\Cref{fig:alg:layer}), which indeed works for all VALID pseudo instances $(P,\, \bB \otimes L) \in \Bvalid$ (\Cref{def:pseudo}), not only the {\em translated} ones (\Cref{def:translate}).
    \label{fig:layer}}
\end{figure}
\clearpage}

\begin{proof}
Let us verify {\bf \Cref{lem:layer:property,lem:layer:poa}} one by one; see \Cref{fig:layer} for a visual aid.

\vspace{.1in}
\noindent
{\bf \Cref{lem:layer:property}.}
The {\blackref{alg:layer}} reduction (Line~\ref{alg:layer:real_mapping}) pointwise/piecewise reorders $\{\tilde{\varphi}_{i} \equiv \varphi_{(i)}\}_{i \in [n]}$ all of the {\em real} mappings $\{\varphi_{i}\}_{i \in [n]}$ over the bid support $b \in [0,\, \lambda]$, namely $\varphi_{(1)}(b) \geq \dots \geq \varphi_{(i)}(b) \geq \dots \geq \varphi_{(n)}(b)$,
and (Line~\ref{alg:layer:pseudo_mapping}) preserves the {\em pseudo} mapping $\tilde{\varphi}_{L} \equiv \varphi_{L}$, which keeps being the {\em dominated} mapping $\min(\tilde{\bvarphi}) \equiv \min(\tilde{\bvarphi}_{-L},\, \tilde{\varphi}_{L}) \equiv \min(\bvarphi_{-L},\, \varphi_{L}) \equiv \varphi_{L} \equiv \tilde{\varphi}_{L}$ provided Line~\ref{alg:layer:real_mapping} and \Cref{lem:pseudo_mapping}.

Clearly, such reordered mappings $\tilde{\bvarphi} = \{\tilde{\varphi}_{\sigma}\}_{\sigma \in [n] \cup \{L\}}$ satisfy \blackref{layeredness}, \blackref{nil}, \blackref{finiteness}, and \blackref{piecewise} (\Cref{def:layer,def:translate,def:discretize})
and are increasing over the bid support $b \in [0,\, \lambda]$ (\blackref{monotonicity}; \Cref{def:pseudo})
\`{a} la the input mappings $\bvarphi$, which can be easily seen with the help of \Cref{fig:layer}.
The extended range $[0,\, \tilde{\varphi}_{1}(0)] = [0,\, \varphi_{(1)}(0)] = [0,\, \max(\bvarphi(\zeros))] \supseteq [0, \varphi_{1}(0)] \supseteq \supp(P)$ must restrict the unmodified conditional value $P$ (\blackref{boundedness}).

It remains to show that (Line~\ref{alg:layer:distribution}) the reconstructed bid distributions $\tilde{\bB} \otimes \tilde{L}$ are well defined, i.e,
(\Cref{lem:pseudo_distribution}) the reordered mappings $\tilde{\bvarphi} = \{\tilde{\varphi}_{\sigma}\}_{\sigma \in [n] \cup \{L\}}$ satisfy the conditions in \Cref{lem:pseudo_mapping}.
This is obvious in that (i)~those conditions are symmetric about the real mappings $\{\tilde{\varphi}_{i} \equiv \varphi_{(i)}\}_{i \in [n]}$ and (ii)~the pseudo mapping is unmodified $\tilde{\varphi}_{L} \equiv \varphi_{L}$.
This finishes the proof of {\bf \Cref{lem:layer:property}}.

\vspace{.1in}
\noindent
{\bf \Cref{lem:layer:poa}.}
Consider the auction/optimal {\SocialWelfare} formulas from \Cref{lem:translate_welfare}:
\begin{align*}
    & \FPA(P,\, \bB \otimes L)\;
    ~=~ \Ex[P] \cdot \calB(0) ~+~ \int_{0}^{\lambda} \sum_{i \in [n]} \Big(\frac{\varphi_{i}(b)}{\varphi_{L}(b) - b} - \frac{\varphi_{i}(b) - \varphi_{L}(b)}{\varphi_{i}(b) - b}\Big) \cdot \calB(b) \cdot \d b \\
    & \phantom{\FPA(P,\, \bB \otimes L)\; ~=~ \Ex[P] \cdot \calB(0)} ~-~ \int_{0}^{\lambda} (n - 1) \cdot \frac{\varphi_{L}(b)}{\varphi_{L}(b) - b} \cdot \calB(b) \cdot \d b, \\
    & \OPT(P,\, \bB \otimes L)
    ~=~ \int_{0}^{+\infty} \Big(1 - P(v) \cdot \prod_{(\sigma,\, j) \,\in\, \bPhi} \big(1 - \omega_{\sigma,\, j} \cdot \indicator(v < \phi_{\sigma,\, j})\big)\Big) \cdot \d v.
\end{align*}
The probabilities $\omega_{\sigma,\, j}$ are given by the partition $\bLambda = [0 = \lambda_{0},\, \lambda_{1}) \cup
[\lambda_{1},\, \lambda_{2}) \cup \dots \cup
[\lambda_{m},\, \lambda_{m + 1} = \lambda]$ and the bid-to-value table $\bPhi = \big[\phi_{\sigma,\, j}\big]$ for $(\sigma,\, j) \in ([n] \cup \{L\}) \times [0: m]$.

The auction {\SocialWelfare} keeps the same $\FPA(P,\, \tilde{\bB} \otimes \tilde{L}) = \FPA(P,\, \bB \otimes L)$. Concretely, (i)~the pseudo mapping $\varphi_{L}$, the first-order bid distribution (\Cref{lem:pseudo_distribution}) $\calB(b) = \exp\big(-\int_{b}^{\lambda} (\varphi_{L}(b) - b)^{-1} \cdot \d b\big)$, and the conditional value $P$ are invariant, while (ii)~the auction {\SocialWelfare} formula is {\em symmetric} about the reordered {\em real} mappings $\{\tilde{\varphi}_{i} \equiv \varphi_{(i)}\}_{i \in [n]}$.

The optimal {\SocialWelfare} again keeps the same $\OPT(P,\, \tilde{\bB} \otimes \tilde{L}) = \OPT(P,\, \bB \otimes L)$, following the same arguments (see the proof of \Cref{lem:translate_welfare} for more details). Basically, (i)~the {\em magnitude} of the values $\phi_{\sigma,\, j}$ and the probabilities $\omega_{\sigma,\, j}$ for $(\sigma,\, j) \in ([n] \cup \{L\}) \times [0: m]$ keeps the same, just {\em column-wise} ($j \in [0: m]$) reordering the {\em real} rows such that $\tilde{\phi}_{1,\, j} \geq \dots \tilde{\phi}_{i,\, j} \geq \dots \geq \tilde{\phi}_{n,\, j}$, while
(ii)~the optimal {\SocialWelfare} formula considers the {\em overall} effect of all entries $(\sigma,\, j) \in ([n] \cup \{L\}) \times [0: m]$ and is irrelevant to the reordering.

To conclude, the {\PoA}-bound keeps the same and {\bf \Cref{lem:layer:poa}} follows. This finishes the proof.
\end{proof}

% Obviously, the reordered {\em real} mappings $\{\tilde{\varphi}_{i} \equiv \varphi_{(i)}\}_{i \in [n]}$ together with the the unmodified pseudo mapping $\tilde{\varphi}_{L}(b) \equiv \varphi_{L}(b) = b + \calB(b) \big/ \calB'(b)$ and the unmodified first-order bid distribution $\tilde{\calB}(b) \equiv \calB(b)$

%\blue{({\bf Yaonan:} to be proved by using \Cref{lem:translate_welfare})
% \blue{This is implied directly by the expressions for $\OPT$ and $\FPA$ in \Cref{lem:translate_welfare}. We note that the {\em pseudo} mapping $\varphi_{L}$ and the first order bid distribution $\calB$ remains unchanged after the Layer reduction. For the real bidders,  these expressions for $\OPT$ and $\FPA$ in \Cref{lem:translate_welfare} are basically a summation over all the pieces without connecting different pieces of the same bidder. These pieces remains unchanged after the layer reduction as a result the expressions for the layered instance is simply a reshuffle of of the original one and the summation keeps the same.} Thus, {\bf \Cref{lem:layer:poa}} and \Cref{lem:layer} follow then.

\subsection{{\polarize}: Derandomizing the conditional values under zero bids}
\label{subsec:polarize}

This subsection shows the {\blackref{alg:polarize}} reduction (see \Cref{fig:alg:polarize} for its description), which transforms a {\em layered} pseudo instance (\Cref{def:layer}) into a more special pseudo instance (\Cref{def:ceiling_floor}): EITHER a {\em floor} pseudo instance OR a {\em ceiling} pseudo instance. Recall that (\blackref{nil} and \blackref{boundedness}) the monopolist $B_{1}$'s conditional value $P$ ranges between $\supp(P) \subseteq [0,\, \varphi_{1}(0)]$.

\begin{definition}[Floor/Ceiling pseudo instances]
\label{def:ceiling_floor}
A {\em layer} pseudo instance $(P,\, \bB \otimes L)$ from \Cref{def:layer} is further called {\em floor}/{\em ceiling} when it satisfies either \blackref{floorness} or \blackref{ceilingness}. Namely, in either case, the conditional value $P$ is determined by the bid distributions $\bB \otimes L$.
\begin{itemize}
    \item \term[{\bf floorness}]{floorness}{\bf :}
    The monopolist $B_{1}$'s conditional value always takes the {\em nil value} $P \equiv 0$. \\
    Therefore, this pseudo instance can be redenoted as $H^{\downarrow} \otimes \bB_{-1} \otimes L = (P^{\downarrow},\, \bB \otimes L)$, with the monopolist $H^{\downarrow} = (P^{\downarrow} \equiv 0,\, B_{1})$. Denote by $\Bvalid^{\downarrow}$ the space of such pseudo instances.

    \item \term[{\bf ceilingness}]{ceilingness}{\bf :}
    The monopolist $B_{1}$'s conditional value always takes the {\em ceiling value} $P \equiv \varphi_{1}(0)$. \\
    Therefore, this pseudo instance can be redenoted as $H^{\uparrow} \otimes \bB_{-1} \otimes L = (P^{\uparrow},\, \bB \otimes L)$, with the monopolist $H^{\uparrow} = (P^{\uparrow} \equiv \varphi_{1}(0),\, B_{1})$. Denote by $\Bvalid^{\uparrow}$ the space of such pseudo instances.
\end{itemize}
\end{definition}

\Cref{lem:polarize} presents performance guarantees of the {\blackref{alg:polarize}} reduction.

\begin{lemma}[{\polarize}; \Cref{fig:alg:polarize}]
\label{lem:polarize}
Under reduction $H \otimes \bB_{-1} \otimes L \gets \polarize(P,\, \bB \otimes L)$:
\begin{enumerate}[font = {\em\bfseries}]
    \item\label{lem:polarize:property}
    The output $H \otimes \bB_{-1} \otimes L$ is EITHER a floor pseudo instance OR a ceiling pseudo instance; the bid distributions $\bB \otimes L$ are unmodified.

    \item\label{lem:polarize:poa}
    A (weakly) worse bound is yielded $\PoA(H \otimes \bB_{-1} \otimes L) \leq \PoA(P,\, \bB \otimes L)$.
\end{enumerate}
\end{lemma}

\afterpage{
\begin{figure}[t]
    \centering
    \begin{mdframed}
    Reduction $\term[\polarize]{alg:polarize}(P,\, \bB \otimes L)$

    \begin{flushleft}
    {\bf Input:}
    A (generic) {\em layered} pseudo instance $(P,\, \bB \otimes L)$.
    \hfill \Cref{def:layer}
    
    \vspace{.05in}
    {\bf Output:}
    A {\em floor}/{\em ceiling} pseudo instance $H \otimes \bB_{-1} \otimes L \in (\Bvalid^{\downarrow} \cup \Bvalid^{\uparrow})$.
    \hfill \Cref{def:ceiling_floor}

    \begin{enumerate}
        \item\label{alg:polarize:floor}
        Define the \term[\text{\em floor}]{floor} candidate $H^{\downarrow} \otimes \bB_{-1} \otimes L \eqdef (P^{\downarrow} \equiv 0,\, \bB \otimes L) \in \Bvalid^{\downarrow}$.

        \item\label{alg:polarize:ceiling}
        Define the \term[\text{\em ceiling}]{ceiling} candidate $H^{\uparrow} \otimes \bB_{-1} \otimes L \eqdef (P^{\uparrow} \equiv \varphi_{1}(0),\, \bB \otimes L) \in \Bvalid^{\uparrow}$.

        \item\label{alg:polarize:return}
        {\bf Return} the {\PoA}-worse candidate $\argmin \big\{\PoA(\blackref{floor}),~ \PoA(\blackref{ceiling})\big\}$; \\
        \white{\bf Return} breaking ties in favor of the {\em ceiling} candidate $H^{\uparrow} \otimes \bB_{-1} \otimes L$.
    \end{enumerate}
    \end{flushleft}
    \end{mdframed}
    \caption{The {\polarize} reduction.}
    \label{fig:alg:polarize}
\end{figure}
\begin{figure}[t]
    \centering
    \subfloat[\label{fig:polarize:input}
    The {\em input} monopolist $(P,\, B_{1})$]{
    \includegraphics[width = .49\textwidth]
    {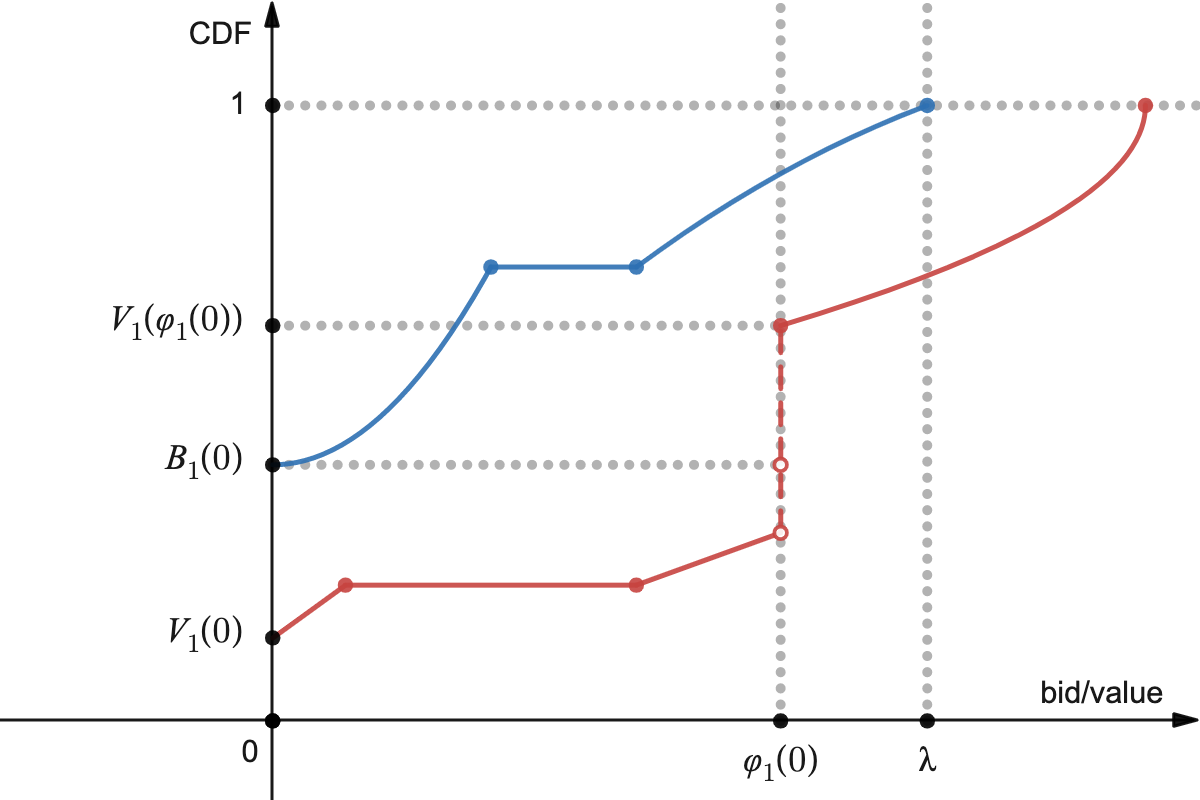}} \\
    \vspace{1cm}
    \subfloat[\label{fig:polarize:floor}
    The {\em floor} monopolist $H^{\downarrow} = (P^{\downarrow} \equiv 0,\, B_{1})$]{
    \includegraphics[width = .49\textwidth]
    {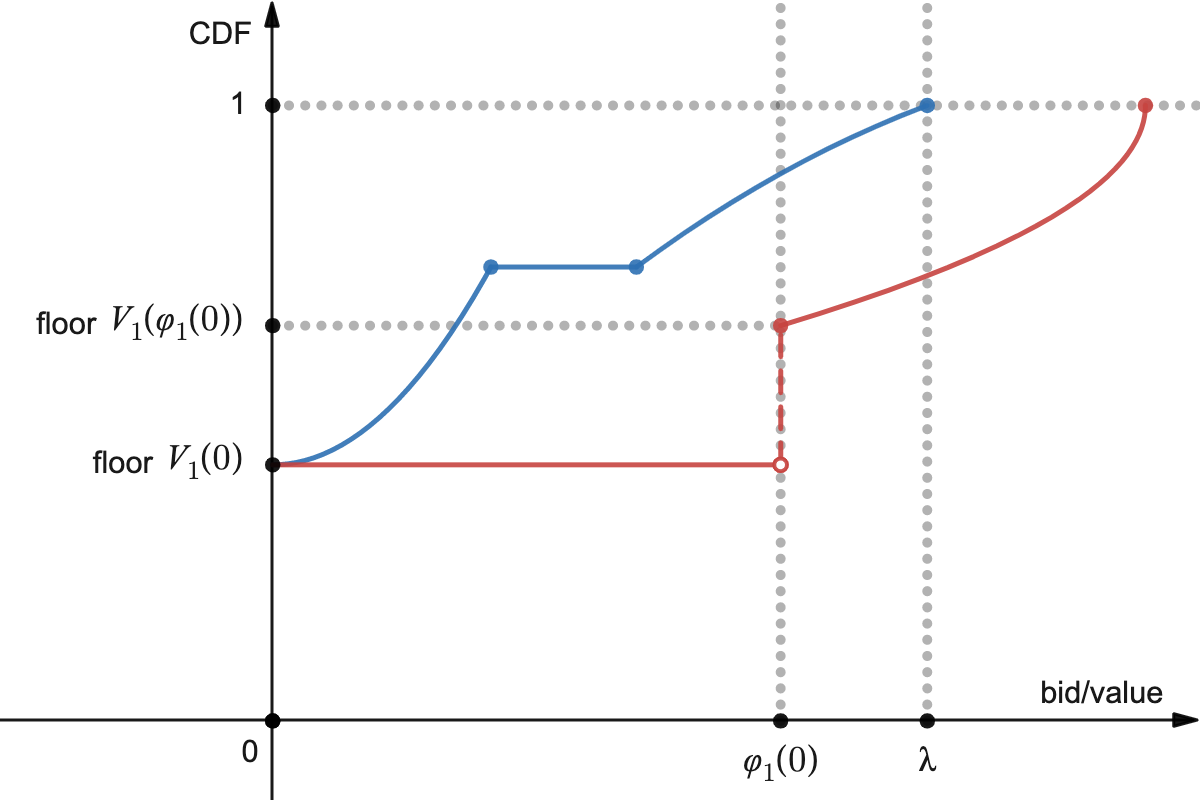}}
    \hfill
    \subfloat[\label{fig:polarize:ceiling}
    The {\em ceiling} monopolist $H^{\uparrow} = (P^{\uparrow} \equiv \varphi_{1}(0),\, B_{1})$]{
    \includegraphics[width = .49\textwidth]
    {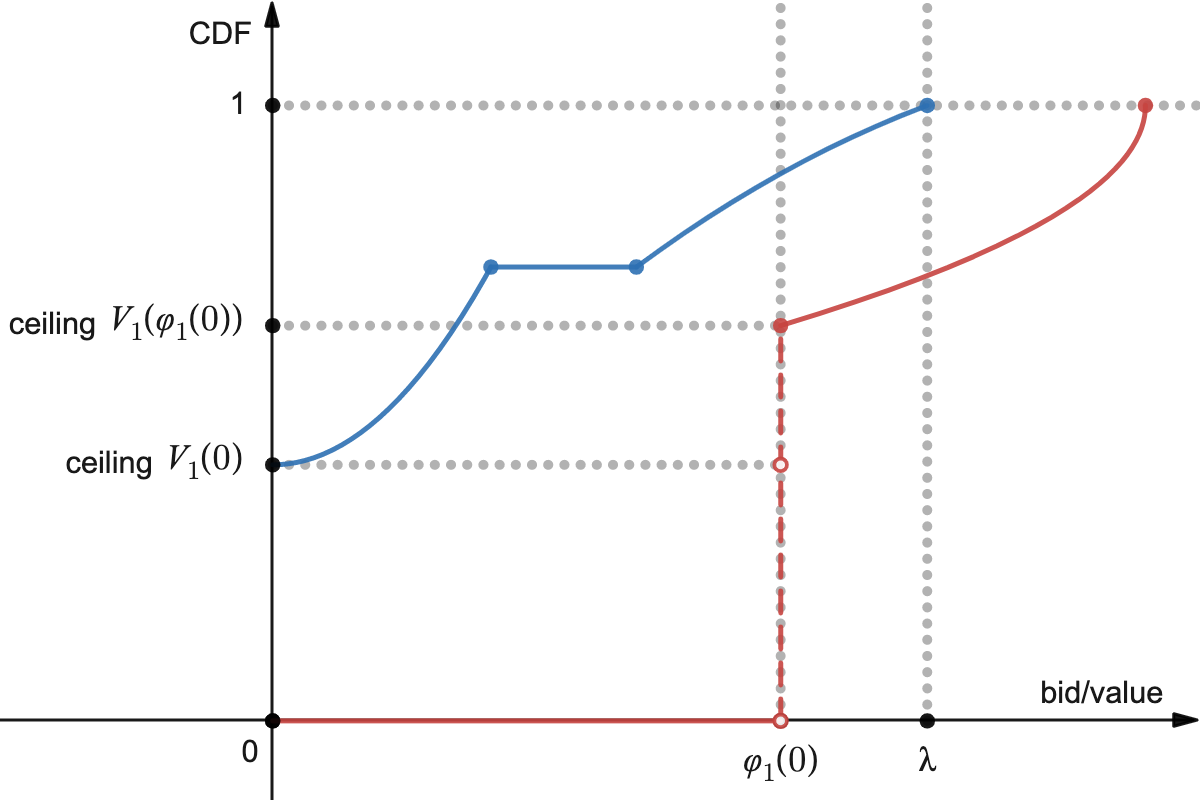}}
    \caption{Diagram of the {\polarize} reduction (\Cref{fig:alg:polarize}), where blue/red curves denote bid/value CDF's, respectively.
    \label{fig:polarize}}
\end{figure}
\clearpage}

\begin{proof}
Let us verify {\bf \Cref{lem:polarize:property,lem:polarize:poa}} one by one.

\vspace{.1in}
\noindent
{\bf \Cref{lem:polarize:property}.}
By construction (Lines~\ref{alg:polarize:floor} to \ref{alg:polarize:return}): The output $H \otimes \bB_{-1} \otimes L$ is the {\PoA}-worse one between two candidates, the \blackref{floor} candidate $H^{\downarrow} \otimes \bB_{-1} \otimes L$ versus the \blackref{ceiling} candidate $H^{\uparrow} \otimes \bB_{-1} \otimes L$. Each of them does not modify the bid distributions $\bB \otimes L$ and the mappings $\bvarphi$, thus preserving \blackref{layeredness}, \blackref{nil}, \blackref{finiteness}, \blackref{piecewise}, and \blackref{monotonicity} (\Cref{def:layer,def:translate,def:discretize,def:pseudo}).
Further, the \blackref{floor} candidate $H^{\downarrow} \otimes \bB_{-1} \otimes L$ promises an {\em invariant} nil conditional value $P^{\downarrow} \equiv 0 \in [0,\, \varphi_{1}(0)]$ (\blackref{floorness}/\blackref{boundedness}) and the \blackref{ceiling} candidate $H^{\uparrow} \otimes \bB_{-1} \otimes L$ promises an {\em invariant} ceiling  value $P^{\uparrow} \equiv \varphi_{1}(0) \in [0,\, \varphi_{1}(0)]$ (\blackref{ceilingness}/\blackref{boundedness}). So these two candidates are \blackref{floor}/\blackref{ceiling} pseudo instances, respectively. {\bf \Cref{lem:polarize:property}} follows then.

\vspace{.1in}
\noindent
{\bf \Cref{lem:polarize:poa}.}
We claim that either or both of the two candidates $H^{\downarrow} \otimes \bB_{-1} \otimes L$ and $H^{\uparrow} \otimes \bB_{-1} \otimes L$ has a (weakly) worse {\PoA} bound than the input $(P,\, \bB \otimes L)$.
Notice that each candidate only modifies the conditional value distribution $P$, to EITHER the \blackref{floor} value distribution $P^{\downarrow}(v) \equiv \indicator(v \geq 0)$
OR the \blackref{ceiling} value distribution $P^{\uparrow}(v) \equiv \indicator(v \geq \varphi_{1}(0))$.

\vspace{.1in}
\noindent
{\bf Auction {\SocialWelfares}.}
Following \Cref{lem:pseudo_welfare} (with \blackref{nil} $\gamma = 0$), the input auction {\SocialWelfare}
$\FPA(\text{\em input}) \equiv \FPA(P,\, \bB \otimes L) = \Ex[P] \cdot \calB(0) + \sum_{\sigma \in [n] \cup \{L\}} \big(\int_{0}^{\lambda} \varphi_{\sigma}(b) \cdot \frac{B'_{\sigma}(b)}{B_{\sigma}(b)} \cdot \calB(b) \cdot \d b\big)$,
where the $\calB(b) = \prod_{\sigma \in [n] \cup \{L\}} B_{\sigma}(b)$.
The summation term is {\em invariant}, since the bid distributions $\bB \otimes L$ and the mappings $\bvarphi$ are unmodified.
Therefore, we can formulate the \blackref{floor} counterpart as $\FPA(\blackref{floor}) = \FPA(\text{\em input}) - \Delta_{\FPA}^{\downarrow}$, using\footnote{Recall that the expectation of a {\em nonnegative} distribution $F$ can be written as $\E[F] = \int_{0}^{+\infty} \big(1 - F(v)\big) \cdot \d v$.}
\begin{align*}
    \Delta_{\FPA}^{\downarrow}
    = \big(\Ex[P] - \Ex[P^{\downarrow}]\big) \cdot \calB(0)
    = \calB(0) \cdot \int_{0}^{\varphi_{1}(0)} \big(1 - P(v)\big) \cdot \d v,
    \hspace{1.72cm}
    && \parbox{2.45cm}{\blackref{boundedness} \\ $0 \leq P \leq \varphi_{1}(0)$}
\end{align*}
and formulate the \blackref{ceiling} counterpart as $\FPA(\blackref{ceiling}) = \FPA(\text{\em input}) + \Delta_{\FPA}^{\uparrow}$, using
\begin{align*}
    \Delta_{\FPA}^{\uparrow}
    = \big(\Ex[P^{\uparrow}] - \Ex[P]\big) \cdot \calB(0)
    = \calB(0) \cdot \int_{0}^{\varphi_{1}(0)} P(v) \cdot \d v.
    \hspace{2.72cm}
    && \parbox{2.45cm}{\blackref{boundedness} \\ $0 \leq P \leq \varphi_{1}(0)$}
\end{align*}

\noindent
{\bf Optimal {\SocialWelfares}.}
Following \Cref{lem:translate_welfare}, the input optimal {\SocialWelfare}
$\OPT(\text{\em input}) \equiv \OPT(P,\, \bB \otimes L) = \int_{0}^{+\infty} \big(1 - P(v) \cdot \calV(v)\big) \cdot \d v$,
where the $\calV(v)$ is determined by the unmodified bid distributions $\bB \otimes L$ (\Cref{lem:pseudo_welfare}) and thus is {\em invariant}.
Accordingly, we can formulate the \blackref{floor} counterpart as $\OPT(\blackref{floor}) = \OPT(\text{\em input}) - \Delta_{\OPT}^{\downarrow}$, using
\begin{align*}
    \Delta_{\OPT}^{\downarrow}
    = \int_{0}^{+\infty} \big(P^{\downarrow}(v) - P(v)\big) \cdot \calV(v) \cdot \d v
    = \int_{0}^{\varphi_{1}(0)} \big(1 - P(v)\big) \cdot \calV(v) \cdot \d v,
    && \parbox{2.45cm}{\blackref{boundedness}}
\end{align*}
and formulate the \blackref{ceiling} counterpart as $\OPT(\blackref{ceiling}) = \OPT(\text{\em input}) + \Delta_{\OPT}^{\uparrow}$, using
\begin{align*}
    \Delta_{\OPT}^{\uparrow}
    = \int_{0}^{+\infty} \big(P(v) - P^{\uparrow}(v)\big) \cdot \calV(v) \cdot \d v
    = \int_{0}^{\varphi_{1}(0)} P(v) \cdot \calV(v) \cdot \d v.
    \hspace{1cm}
    && \parbox{2.45cm}{\blackref{boundedness}}
\end{align*}

Notice that all the four terms $\Delta_{\FPA}^{\downarrow}$ etc are nonnegative. The remaining proof relies on {\bf \Cref{fact:polarize}}. (For brevity, we ignore the ``$0 / 0$'' issue, which can happen only if EITHER $P \equiv P^{\downarrow}$ OR $P \equiv P^{\uparrow}$, making {\bf \Cref{lem:polarize:poa}} vacuously true.)

\setcounter{fact}{0}

\begin{fact}
\label{fact:polarize}
$\Delta_{\FPA}^{\downarrow} / \Delta_{\OPT}^{\downarrow} \geq \Delta_{\FPA}^{\uparrow} / \Delta_{\OPT}^{\uparrow}$ or equivalently, $\Delta_{\OPT}^{\downarrow} / \Delta_{\FPA}^{\downarrow} \leq \Delta_{\OPT}^{\uparrow} / \Delta_{\FPA}^{\uparrow}$.
\end{fact}

\begin{proof}
It is more convenient to verify the second equation: Using the (normalized) antiderivatives
$\bar{\mathfrak{P}}(v) \eqdef (1 / \Delta_{\FPA}^{\downarrow}) \cdot \calB(0) \cdot \int_{0}^{v} (1 - P(t)) \cdot \d t$ AND
$\mathfrak{P}(v) \eqdef (1 / \Delta_{\FPA}^{\uparrow}) \cdot \calB(0) \cdot \int_{0}^{v} P(t) \cdot \d t$, we can rewrite
$\LHS \equiv \Delta_{\OPT}^{\downarrow} \big/ \Delta_{\FPA}^{\downarrow}
= \int_{0}^{\varphi_{1}(0)} \bar{\mathfrak{P}}'(v) \cdot \calV(v) \cdot \d v$ AND
$\RHS \equiv \Delta_{\OPT}^{\uparrow} / \Delta_{\FPA}^{\uparrow} = \int_{0}^{\varphi_{1}(0)} \mathfrak{P}'(v) \cdot \calV(v) \cdot \d v$.

Using integration by parts, we can deduce {\bf \Cref{fact:polarize}} as follows:
\begin{align}
    \LHS
    & ~=~ \Big(\bar{\mathfrak{P}}(v) \cdot \calV(v)\Big) \Bigmid_{v = 0}^{\varphi_{1}(0)}
    ~-~ \int_{0}^{\varphi_{1}(0)} \bar{\mathfrak{P}}(v) \cdot \calV'(v) \cdot \d v
    \nonumber \\
    % \label{eq:polarize:1}\tag{P1} \\
    & ~=~ \Big(\mathfrak{P}(v) \cdot \calV(v)\Big) \Bigmid_{v = 0}^{\varphi_{1}(0)}
    ~-~ \int_{0}^{\varphi_{1}(0)} \bar{\mathfrak{P}}(v) \cdot \calV'(v) \cdot \d v
    \label{eq:polarize:2}\tag{P1} \\
    & ~\leq~ \Big(\mathfrak{P}(v) \cdot \calV(v)\Big) \Bigmid_{v = 0}^{\varphi_{1}(0)}
    ~-~ \int_{0}^{\varphi_{1}(0)} \mathfrak{P}(v) \cdot \calV'(v) \cdot \d v
    ~=~ \RHS.
    \label{eq:polarize:3}\tag{P2}
\end{align}
% \eqref{eq:polarize:1}: Integration by parts. \\
\eqref{eq:polarize:2}: For the first term, by construction $\bar{\mathfrak{P}}(v) = \mathfrak{P}(v) = 0$ and $\bar{\mathfrak{P}}(\varphi_{1}(0)) = \mathfrak{P}(\varphi_{1}(0)) = 1$. \\
\eqref{eq:polarize:3}: For the second term, by construction $\bar{\mathfrak{P}}(v) \geq \mathfrak{P}(v) \geq 0$ and the PDF $\calV'(v) \geq 0$.
In particular, the \blackref{floor} antiderivative $\bar{\mathfrak{P}}(v)$ is concave $\bar{\mathfrak{P}}''(v) = -(1 / \Delta_{\FPA}^{\downarrow}) \cdot \calB(0) \cdot P'(v) \leq 0$, whereas the \blackref{ceiling} antiderivative $\mathfrak{P}(v)$ is convex $\mathfrak{P}''(v) = (1 / \Delta_{\FPA}^{\uparrow}) \cdot \calB(0) \cdot P'(v) \geq 0$.
\end{proof}

Assume the opposite to {\bf \Cref{lem:polarize:poa}}: The \blackref{floor}/\blackref{ceiling} candidates $H^{\downarrow} \otimes \bB_{-1} \otimes L$ and $H^{\uparrow} \otimes \bB_{-1} \otimes L$ EACH yield a strictly larger bound than the input $(P,\, \bB \otimes L)$. That is,
\begin{align*}
    & \PoA(\blackref{floor}) ~~\,\;>~ \PoA(\text{\em input})
    && \iff &&
    \frac{\FPA(\text{\em input}) - \Delta_{\FPA}^{\downarrow}}{\OPT(\text{\em input}) - \Delta_{\OPT}^{\downarrow}}
    ~>~ \frac{\FPA(\text{\em input})}{\OPT(\text{\em input})}, \\
    & \PoA(\blackref{ceiling}) ~>~ \PoA(\text{\em input})
    && \iff &&
    \frac{\FPA(\text{\em input}) + \Delta_{\FPA}^{\uparrow}}{\OPT(\text{\em input}) + \Delta_{\OPT}^{\uparrow}}
    ~>~ \frac{\FPA(\text{\em input})}{\OPT(\text{\em input})}.
\end{align*}
Rearranging both equations gives $\Delta_{\FPA}^{\downarrow} / \Delta_{\OPT}^{\downarrow}
< \PoA(\text{\em input})
< \Delta_{\FPA}^{\uparrow} / \Delta_{\OPT}^{\uparrow}$, which contradicts {\bf \Cref{fact:polarize}}.
Refute our assumption: At least one candidate between $H^{\downarrow} \otimes \bB_{-1} \otimes L$ and $H^{\uparrow} \otimes \bB_{-1} \otimes L$ has a weakly worse bound. {\bf \Cref{lem:polarize:poa}} and \Cref{lem:polarize} follow then.
\end{proof}

% \begin{remark}[Floor/Ceiling pseudo instances]
% \label{rem:ceiling_floor}
% We emphasize that a pair of \blackref{floor}/\blackref{ceiling} conditional value $P^{\downarrow} \equiv 0$ and $P^{\uparrow} \equiv \varphi_{1}(0)$ are also determined by the underlying bid distributions $H \otimes \bB_{-1} \otimes L$.
% Hence, to see the membership $H^{\downarrow} \otimes \bB_{-1} \otimes L \in \Bvalid^{\downarrow}$ and/or $H^{\uparrow} \otimes \bB_{-1} \otimes L \in \Bvalid^{\uparrow}$, it suffices to check \blackref{monotonicity}
% \end{remark}
% As a consequence of \Cref{lem:polarize}, we can just consider \blackref{floor}/\blackref{ceiling} pseudo instances hereafter.

\subsection{The floor/ceiling pseudo instances}
\label{subsec:preprocess}

Below we summarize our discussions throughout \Cref{sec:preprocessing}, about how to preprocess valid pseudo instances (\Cref{def:pseudo}).
For ease of reference, we rephrase \Cref{cor:pseudo_instance,lem:discretize,lem:translate,lem:layer,lem:polarize} (after minor modifications).

\begin{restate}[\Cref{cor:pseudo_instance}]
Regarding {\FirstPriceAuctions}, the {\PriceofAnarchy} is at least
\begin{align*}
    \PoA ~\geq~ \inf \bigg\{\, \frac{\FPA(P,\, \bB \otimes L)}{\OPT(P,\, \bB \otimes L)} \,\biggmid\, (P,\, \bB \otimes L) \in \Bvalid ~\text{\em and}~ \OPT(P,\, \bB \otimes L) < +\infty \,\bigg\}.
\end{align*}
\end{restate}

\begin{restate}[\Cref{lem:discretize}]
Given a valid pseudo instance $(P,\, \bB \otimes L)$ that has bounded expected auction/optimal {\SocialWelfares} $\FPA(P,\, \bB \otimes L) \leq \OPT(P,\, \bB \otimes L) < +\infty$, for any error $\epsilon \in (0,\, 1)$, there is a discretized pseudo instance $(\tilde{P},\, \tilde{\bB} \otimes \tilde{L})$ such that $|\PoA(\tilde{P},\, \tilde{\bB} \otimes \tilde{L}) - \PoA(P,\, \bB \otimes L)| \leq \epsilon$.
\end{restate}

\begin{restate}[\Cref{lem:translate}]
Given a discretized pseudo instance $(P,\, \bB \otimes L)$, there is a translated pseudo instance $(\tilde{P},\, \tilde{\bB} \otimes \tilde{L})$ such that $\PoA(\tilde{P},\, \tilde{\bB} \otimes \tilde{L}) \leq \PoA(P,\, \bB \otimes L)$.
\end{restate}

\begin{restate}[\Cref{lem:layer}]
Given a translated pseudo instance $(P,\, \bB \otimes L)$, there is a layered pseudo instance $(\tilde{P},\, \tilde{\bB} \otimes \tilde{L})$ such that $\PoA(\tilde{P},\, \tilde{\bB} \otimes \tilde{L}) \leq \PoA(P,\, \bB \otimes L)$.
\end{restate}

\begin{restate}[\Cref{lem:polarize}]
Given a layered pseudo instance $(P,\, \bB \otimes L)$, there is a floor/ceiling pseudo instance $\tilde{H} \otimes \tilde{\bB} \otimes \tilde{L} \in (\Bvalid^{\downarrow} \cup \Bvalid^{\uparrow})$ such that $\PoA(\tilde{H} \otimes \tilde{\bB} \otimes \tilde{L}) \leq \PoA(P,\, \bB \otimes L)$.
\end{restate}

\Cref{cor:preprocess} is a refinement of \Cref{cor:pseudo_instance}: Towards a lower bound on the {\PriceofAnarchy} in {\FirstPriceAuctions}, we can restrict our attention to the subspace $(\Bvalid^{\downarrow} \cup \Bvalid^{\uparrow})$ of {\em floor}/{\em ceiling} pseudo instances (\Cref{def:ceiling_floor}).

\begin{corollary}[Lower bound]
\label{cor:preprocess}
Regarding {\FirstPriceAuctions}, the {\PriceofAnarchy} is at least
\begin{align*}
    \PoA ~\geq~ \inf \bigg\{\, \frac{\FPA(H \otimes \bB \otimes L)}{\OPT(H \otimes \bB \otimes L)} \,\biggmid\, H \otimes \bB \otimes L \in (\Bvalid^{\downarrow} \cup \Bvalid^{\uparrow}) ~\text{\em and}~ \OPT(H \otimes \bB \otimes L) < +\infty \,\bigg\}.
\end{align*}
\end{corollary}

\begin{proof}
For a valid pseudo instance $(P,\, \bB \otimes L) \in \Bvalid$ considered in \Cref{cor:pseudo_instance}, the {\SocialWelfares} are bounded $\FPA(P,\, \bB \otimes L) \leq \OPT(P,\, \bB \otimes L) < +\infty$.
For any error $\epsilon \in (0,\, 1)$, (\Cref{lem:discretize,lem:translate,lem:layer,lem:polarize}) there is a {\em floor}/{\em ceiling} pseudo instance $\tilde{H} \otimes \tilde{\bB} \otimes \tilde{L} \in (\Bvalid^{\downarrow} \cup \Bvalid^{\uparrow})$ such that
\[
    \PoA(\tilde{H} \otimes \tilde{\bB} \otimes \tilde{L})
    ~\leq~ \PoA(P,\, \bB \otimes L) + \epsilon.
\]
Since the error can be arbitrarily small $\epsilon \to 0^{+}$, it has no effect on the {\em infimum} {\PoA} bound.
\end{proof}

\newpage

% yaonan: notation in this subsection

\newcommand{\Bfloor}{\mathbb{B}_{\sf valid}^{\downarrow}}
\newcommand{\Bceiling}{\mathbb{B}_{\sf valid}^{\uparrow}}
\newcommand{\Bstrong}{\mathbb{B}_{\sf strong}^{\uparrow}}
\newcommand{\Btwin}{\mathbb{B}_{\sf twin}^{\uparrow}}

\section{Towards the Worst Case Pseudo Instances}
\label{sec:reduction}

Following \Cref{cor:preprocess}, towards a lower bound on the {\PriceofAnarchy} in {\FirstPriceAuctions}, we can concentrate on {\em floor}/{\em ceiling} pseudo instances $\in (\Bfloor \cup \Bceiling)$.
This section will characterize the worst cases in this search space. More concretely:
\begin{itemize}
    \item \Cref{subsec:tech_prelim} presents some prerequisite notions, including the concept of {\em twin ceiling} pseudo instances $\Btwin$ (\Cref{def:twin}), which form a subset of the space of {\em ceiling} pseudo instances $\Btwin \subsetneq \Bceiling$ and turn out to be the worst cases.
    
    \item Given an undesirable pseudo instance $\in (\Bfloor \cup (\Bceiling \setminus \Btwin))$, we will see that there are four types of possible modifications; see the \blackref{ref:reduction:sketch} part in \Cref{subsec:tech_prelim} for details. We give four reductions in \Cref{subsec:slice,subsec:collapse,subsec:halve,subsec:AD}, one-to-one dealing with each of these four types.
    
    \item \Cref{subsec:main} presents the \blackref{alg:main} procedure, which {\em iteratively} invokes the mentioned four reductions, transforming a given {\em floor}/{\em ceiling} pseudo instance $\in (\Bfloor \cup \Bceiling)$ to a {\em twin ceiling} pseudo instance $\in \Btwin$, as desired. Overall, we will use the {\em potential method} to upper bound the number of invocations.
\end{itemize}

\begin{comment}
\blue{As before, we adopt these terminologies and notations:
\begin{itemize}
    \item bidders $\sigma \in \{H\} \cup [n] \cup \{L\}$
    
    \item real bidders $i \in \{H\} \cup [n]$
    
    \item the monopolist $H$
    
    \item non-monopoly bidders $i \in [n]$
    
    \item the pseudo bidder $L$
\end{itemize}}

\begin{reminder*}
\blue{{\bf Yaonan:} For \Cref{sec:reduction}, we are left with
\begin{itemize}
    % \item organization of \Cref{sec:reduction} \red{\bf (Yaonan: Pinyan, if okay with the above short organization [Following ... the number of invocations], just remove this bullet. Also, remove ``these terminologies and notations'' if you think they are unnecessary.)}
    
    \item \Cref{subsec:main}; the very last part.
\end{itemize}
}
\end{reminder*}

\end{comment}

\subsection{Twin ceiling, strong ceiling, jumps, and potentials}
\label{subsec:tech_prelim}

This subsection introduces several prerequisite concepts.
But before that, for ease of reference, let us recap the definitions of {\em floor}/{\em ceiling} pseudo instances (\Cref{def:pseudo,def:discretize,def:translate,def:layer,def:ceiling_floor}).

\begin{definition}[Floor/Ceiling pseudo instances]
\label{def:ceiling_floor:restate}
For a {\em floor}/{\em ceiling} pseudo instance $H \otimes \bB \otimes L$:
\begin{itemize}
    \item The {\em monopolist} $H$ is a real bidder, who competes with OTHER bidders $\bB \otimes L$.
    
    \item The {\em non-monopoly} bidders $\bB = \{B_{i}\}_{i \in [n]}$ (if existential; $n \geq 0$) are real bidders, each of which competes with OTHER bidders $H \otimes \bB_{-i} \otimes L$.
    
    \item The pseudo bidder $L$ competes with ALL bidders $H \otimes \bB \otimes L$, {\em including him/herself $L$}.
\end{itemize}
Using the first-order bid distribution $\calB(b) \eqdef \prod_{\sigma \in \{H\} \cup [n] \cup \{L\}} B_{\sigma}(b)$, each bidder $\sigma \in \{H\} \cup [n] \cup \{L\}$ has the bid-to-value mapping
\[
    \varphi_{\sigma}(b)
    ~\eqdef~
    % b + \bigg(\sum_{k \in (\{H\} \cup [n]) \setminus \{\sigma\}} \frac{B'_{k}(b)}{B_{k}(b)} + \frac{L'(b)}{L(b)}\bigg)^{-1}
    % ~=~
    b + \bigg(\frac{\calB'(b)}{\calB(b)} - \frac{B'_{\sigma}(b)}{B_{\sigma}(b)} \cdot \indicator(\sigma \neq L)\bigg)^{-1}.
\]
This pseudo instance, no matter being {\em floor} $H^{\downarrow} \otimes \bB \otimes L \in \Bfloor$ or being {\em ceiling} $H^{\uparrow} \otimes \bB \otimes L \in \Bceiling$, satisfies \blackref{re:discretization}, \blackref{re:monotonicity}, and \blackref{re:layeredness}.
\begin{itemize}
    \item \term[\textbf{discretization}]{re:discretization}{\bf :}
    This pseudo instance has a {\em bounded} bid support $b \in [0 \equiv \gamma,\, \lambda] \subseteq [0,\, +\infty)$. Regarding some $(m + 1)$-piece partition $\bLambda \eqdef [0 \equiv \lambda_{0},\, \lambda_{1}) \cup
    [\lambda_{1},\, \lambda_{2}) \cup \dots \cup
    [\lambda_{m},\, \lambda_{m + 1} \equiv \lambda]$, for $0 \leq m < +\infty$, the bid-to-value mappings $\bvarphi = \{\varphi_{\sigma}\}_{\sigma \in \{H\} \cup [n] \cup \{L\}}$ are {\em piecewise constant} functions and thus, can be represented as an $(n + 2)$-to-$(m + 1)$ bid-to-value table $\bPhi = [\phi_{\sigma,\, j}]$ for $\sigma \in \{H\} \cup [n] \cup \{L\}$ and $j \in [0:\, m]$:
    \begin{align*}
        \bPhi ~\eqdef~
        \Big[\, \text{$\varphi_{\sigma}(b) \equiv \phi_{\sigma,\, j}$ on every piece $b \in [\lambda_{j},\, \lambda_{j + 1})$} \,\Big].
    \end{align*}
    
    \item \term[\textbf{monotonicity}]{re:monotonicity}{\bf :}
    The mappings $\bvarphi$ are increasing over the bid support $b \in [0,\, \lambda]$; thus the table $\bPhi = [\phi_{\sigma,\, j}]$ is increasing $\phi_{\sigma,\, 0} \leq \dots \leq \phi_{\sigma,\, j} \leq \dots \leq \phi_{\sigma,\, m}$ in each row $\sigma \in \{H\} \cup [n] \cup \{L\}$.
    
    \item \term[\textbf{layeredness}]{re:layeredness}{\bf :} 
    The mappings $\bvarphi$ are ordered $\varphi_{H}(b) \geq \varphi_{1}(b) \geq \dots \geq \varphi_{n}(b) \geq \varphi_{L}(b)$ over the bid support $b \in [0,\, \lambda]$; thus the table $\bPhi = [\phi_{\sigma,\, j}]$ is decreasing $\phi_{H,\, j} \geq \phi_{1,\, j} \geq \dots \geq \phi_{n,\, j} \geq \phi_{L,\, j}$ in each column $j \in [0: m]$.
\end{itemize}
Given these, a pair of {\em floor}/{\em ceiling} pseudo instances (\Cref{def:ceiling_floor}) are uniquely determined. Namely, the {\em floor} one $H^{\downarrow} \otimes \bB \otimes L \in \Bfloor$ further satisfies
(\term[\textbf{floorness}]{re:floorness}) that the monopolist $H^{\downarrow}$'s conditional value always takes the {\em nil value} $P^{\downarrow} \equiv 0$,
while the {\em ceiling} one $H^{\uparrow} \otimes \bB \otimes L \in \Bceiling$ further satisfies (\term[\textbf{ceilingness}]{re:ceilingness}) that the monopolist $H^{\uparrow}$'s conditional value always takes the {\em ceiling value} $P^{\uparrow} \equiv \phi_{H,\, 0}$.
\end{definition}

We notice that possibly there are NO non-monopoly bidders $\bB = \{B_{i}\}_{i \in [n]} = \emptyset$, or possibly all non-monopoly bidders are ``dummies'' $B_{i}(b) \equiv 1$ on $b \in [0,\, \lambda]$. Suppose so, only the monopolist $H$ and the pseudo bidder $L$ have effects; then without ambiguity, we can write $H \otimes \bB \otimes L \cong H \otimes L$.
Indeed, the {\em twin ceiling} pseudo instances $H^{\uparrow} \otimes L \in \Bceiling$ given below (\Cref{def:twin}; a visual aid is deferred to \Cref{fig:twin}) have this structure and turn out to be the WORST CASES.

\begin{definition}[Twin ceiling pseudo instances]
\label{def:twin}
A {\em ceiling} pseudo instance $H^{\uparrow} \otimes \bB \otimes L \in \Bceiling$ is further called {\em twin ceiling} when it satisfies \blackref{re:twin_ceiling} and \blackref{re:twin_collapse}.
\begin{itemize}
    \item \term[\textbf{absolute ceilingness}]{re:twin_ceiling}{\bf :}
    The monopolist $H^{\uparrow}$ takes a {\em constant} bid-to-value mapping $\varphi_{H}(b) \equiv \varphi_{H}(0)$ on $b \in [0,\, \lambda]$; thus the row-$H$ entries in the bid-to-value table $\bPhi = [\phi_{\sigma,\, j}]$ are the same $\phi_{H,\, 0} = \dots = \phi_{H,\, j} = \dots = \phi_{H,\, m}$. (As before, the {\em ceiling} conditional value $P^{\uparrow} \equiv \phi_{H,\, 0}$.)
    
    \item \term[\textbf{non-monopoly collapse}]{re:twin_collapse}{\bf :}
    All non-monopoly bidders $i \in [n]$ (if existential) each exactly take the pseudo bid-to-value mapping $\varphi_{i}(b) \equiv \varphi_{L}(b)$ on $b \in [0,\, \lambda]$; thus in each column $j \in [0: m]$ of the bid-to-value table $\bPhi = [\phi_{\sigma,\, j}]$, all the non-monopoly entries are the same as the pseudo entry $\phi_{1,\, j} = \dots = \phi_{n,\, j} = \phi_{L,\, j}$.
    That is,\footnote{Following \Cref{def:ceiling_floor:restate}, $\varphi_{i}(b) \equiv \varphi_{L}(b)$ means $B'_{i}(b) \big/ B_{i}(b) \equiv 0$, which together with the boundary condition at the supremum bid $B_{i}(\lambda) = 1$ implies that $B_{i}(b) \equiv 1$ on $b \in [0,\, \lambda]$.}
    each non-monopoly bidder $i \in [n]$ has no effect $B_{i}(b) \equiv 1$ over the bid support $b \in [0,\, \lambda]$.
\end{itemize}
Since all non-monopoly bidders have no effect, this {\em twin ceiling} pseudo instance can be written as $H^{\uparrow} \otimes \bB \otimes L \cong H^{\uparrow} \otimes L$.
Denote by $\Btwin$ the space of such pseudo instances, a subset of the space of {\em ceiling} pseudo instances $\Btwin \subsetneq \Bceiling$.
\end{definition}

To convert a {\em floor}/{\em ceiling} pseudo instance $H \otimes \bB \otimes L \in (\Bfloor \cup \Bceiling)$ to another {\em twin ceiling} pseudo instance, a very first question is how to measure its distance ``$\mathrm{dist}(H \otimes \bB \otimes L,\, \Btwin)$'' to the target space $\Btwin$. Also, a {\em floor} pseudo instance $H^{\downarrow} \otimes \bB \otimes L \in \Bfloor$ should have a larger distance than the paired {\em ceiling} pseudo instance $H^{\uparrow} \otimes \bB \otimes L \in \Bfloor$, but just slightly larger since they only differ on the monopolist's conditional value $P^{\downarrow} \equiv 0$ versus $P^{\uparrow} \equiv \phi_{H,\, 0}$.

The above intuition guides us to introduce the potential $\Psi(H \otimes \bB \otimes L)$ of a {\em floor}/{\em ceiling} pseudo instance (\Cref{def:potential,fig:potential:table}). As the name ``potential'' suggests, we will transform such a pseudo instance iteratively until getting a {\em twin ceiling} pseudo instance $\tilde{H} \otimes \tilde{L} \in \Btwin$ and use the potential method to bound the number of iterations.

\begin{definition}[Potentials]
\label{def:potential}
Given a pair of {\em floor}/{\em ceiling} pseudo instances $H^{\downarrow} \otimes \bB \otimes L \in \Bfloor$ and $H^{\uparrow} \otimes \bB \otimes L \in \Bceiling$, consider their {\em common} $(n + 2)$-to-$(m + 1)$ bid-to-value table $\bPhi = [\phi_{\sigma,\, j}]$ for $\sigma \in \{H\} \cup [n] \cup \{L\}$ and $j \in [0:\, m]$ (\Cref{def:ceiling_floor:restate}):
\begin{align*}
    \bPhi ~=~
    \Big[\, \text{$\varphi_{\sigma}(b) \equiv \phi_{\sigma,\, j}$ on every piece $b \in [\lambda_{j},\, \lambda_{j + 1})$} \,\Big].
\end{align*}
An entry $\phi_{\sigma,\, j}$ is called {\em ultra-ceiling} when it {\em strictly} exceeds the ceiling value $\phi_{\sigma,\, j} > \phi_{H,\, 0} \equiv \varphi_{H}(0)$ of the monopolist $H$. Then the {\em ceiling} pseudo instance's potential counts the ultra-ceiling entries
\[
    \Psi(H^{\uparrow} \otimes \bB \otimes L) ~\eqdef~ \big|\big\{(\sigma,\, j) \in \bPhi: \phi_{\sigma,\, j} > \phi_{H,\, 0}\big\}\big|. \hspace{0.35cm}
\]
In contrast, the {\em floor} pseudo instance's potential adds an extra ONE (which accounts for the nil conditional value $P^{\downarrow} \equiv 0$ rather than $P^{\uparrow} \equiv \phi_{H,\, 0}$):
\[
    \Psi(H^{\downarrow} \otimes \bB \otimes L) ~\eqdef~ \Psi(H^{\uparrow} \otimes \bB \otimes L) ~+~ 1.~~ \hspace{1.22cm}
\]
\end{definition}

% \begin{remark}[Potentials]
% \label{rem:potential}
% Because a bid-to-value table $\bPhi = [\phi_{\sigma,\, j}]$ is row-wise increasing and column-wise decreasing (\blackref{re:monotonicity}/\blackref{re:layeredness}), as \Cref{fig:potential} illustrates,
% (i)~the ultra-ceiling entries $\phi_{\sigma,\, j} > \phi_{H,\, 0}$ form an upper-right subregion of the table $\bPhi$; and
% (ii)~the column-$0$ entries $\phi_{\sigma,\, 0}$ for $\sigma \in \{H\} \cup [n] \cup \{L\}$, especially the ceiling value $\phi_{H,\, 0}$ itself, cannot be ultra-ceiling.
% \end{remark}

\afterpage{
\begin{figure}
    \centering
    \subfloat[{The bid-to-value table $\bPhi = \big[\phi_{\sigma,\, j}\big]$; ultra-ceiling entries $\phi_{\sigma,\, j} > \phi_{H,\, 0}$ marked in blue.}
    \label{fig:potential:table}]{
    \parbox[c]{14cm}{
    {\centering
    \includegraphics[
    height = 10cm]
    {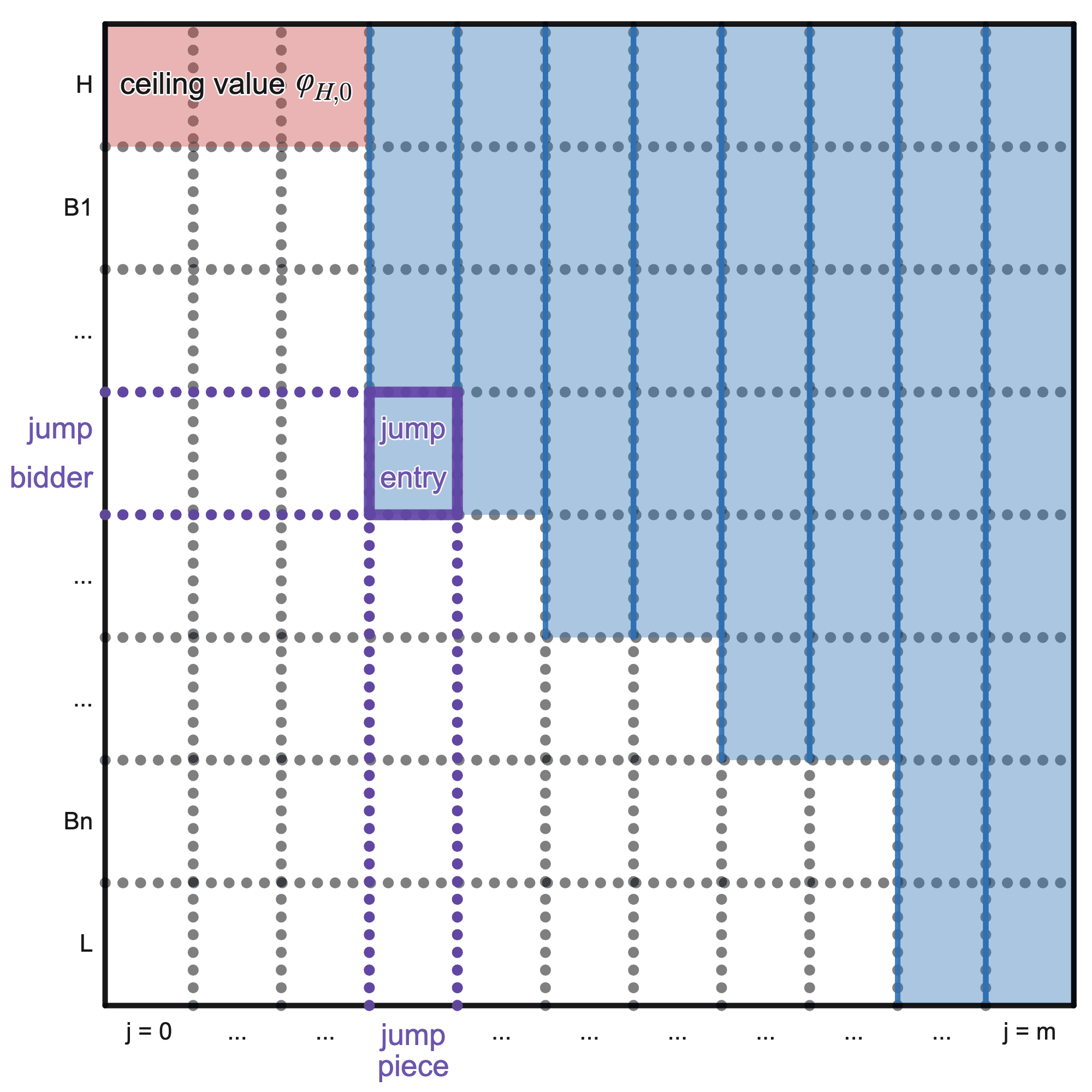}
    \par}}} \\
    \vspace{1cm}
    \subfloat[{A {\bf pseudo jump} $\sigma^{*} = L$}
    \label{fig:potential:pseudo}]{
    \includegraphics[width = .49\textwidth]
    {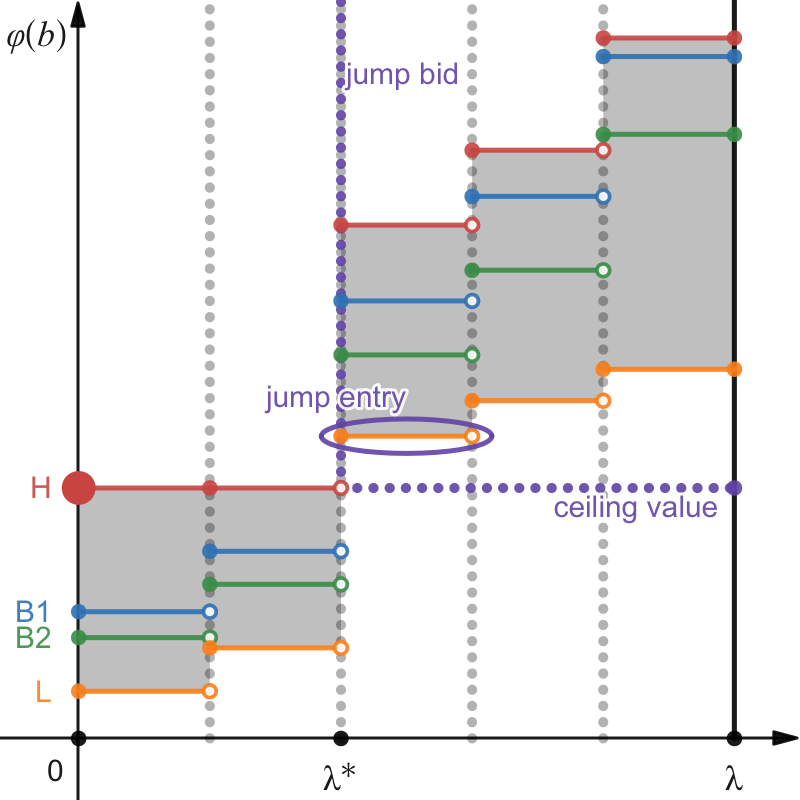}}
    \hfill
    \subfloat[{A {\bf real jump} $\sigma^{*} \in \{H\} \cup [n]$}
    \label{fig:potential:real}]{
    \includegraphics[width = .49\textwidth]
    {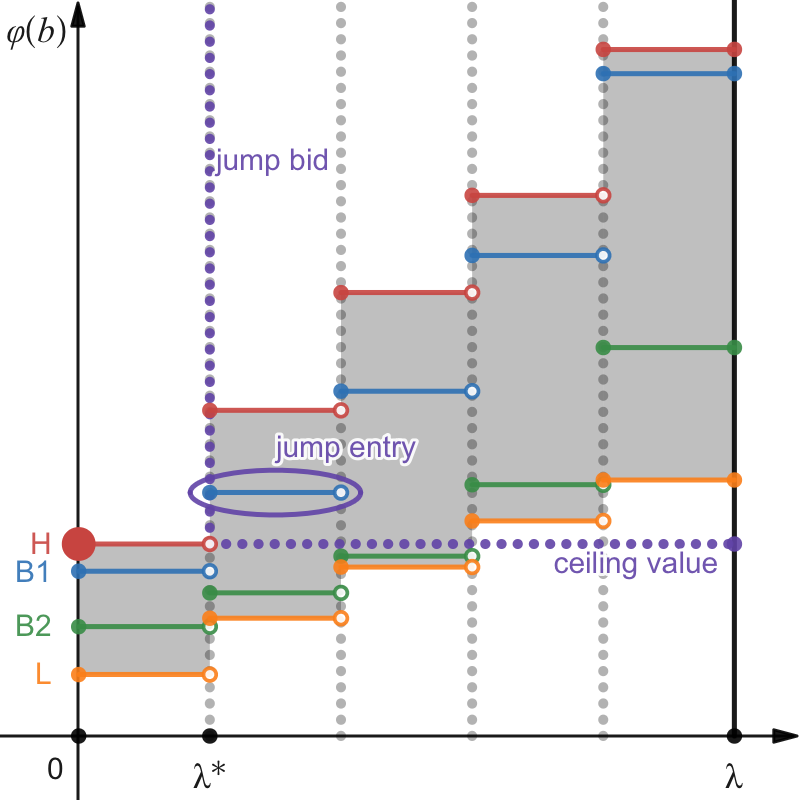}}
    \caption{Diagram of potentials and jumps of a {\em ceiling} pseudo instance $H^{\uparrow} \otimes \bB \otimes L \in \Bceiling$ (\Cref{def:potential,def:jump}).
    % The potential $\Psi(H^{\uparrow} \otimes \bB \otimes L)$ of a {\em ceiling} pseudo instance counts (blue in \Cref{fig:potential:table}) ultra-ceiling entries $\phi_{\sigma,\, j} > \phi_{H,\, 0}$ in the table $\bPhi = \big[\phi_{\sigma,\, j}\big]$.
    % The paired {\em floor} potential adds another ONE $\Psi(H^{\downarrow} \otimes \bB \otimes L) = \Psi(H^{\uparrow} \otimes \bB \otimes L) + 1$.
    }
    \label{fig:potential}
\end{figure}
\clearpage}

\Cref{lem:potential} summarizes several basic properties about the potentials.

\begin{lemma}[Potentials]
\label{lem:potential}
The following hold:
% \begin{flushleft}
\begin{enumerate}[font = {\em\bfseries}]
    \item\label{lem:potential:bound}
    A floor/ceiling pseudo instance $H \otimes \bB \otimes L \in (\Bfloor \cup \Bceiling)$ has a bounded potential \\
    $0 \leq \Psi(H \otimes \bB \otimes L) \leq |\bPhi| < +\infty$.
    
    \item\label{lem:potential:floor}
    A floor pseudo instance $H^{\downarrow} \otimes \bB \otimes L \in \Bfloor$ has a nonzero potential $\Psi(H^{\downarrow} \otimes \bB \otimes L) \geq 1$.
    
    % \item\label{lem:potential:ceiling}
    % A ceiling pseudo instance $H^{\uparrow} \otimes \bB \otimes L \in \Bceiling$ has a nonnegative potential $\Psi(H^{\uparrow} \otimes \bB \otimes L) \geq 0$.
    
    \item\label{lem:potential:ceiling}
    A ceiling pseudo instance $H^{\uparrow} \otimes \bB \otimes L \in \Bceiling$ has a zero potential $\Psi(H^{\uparrow} \otimes \bB \otimes L) = 0$ iff it satisfies \blackref{re:twin_ceiling}.
    % A twin ceiling pseudo instance $H^{\uparrow} \otimes \bB \otimes L \cong H^{\uparrow} \otimes L \in \Btwin$ thus has a zero potential $\Psi(H^{\uparrow} \otimes \bB \otimes L) = 0$.
\end{enumerate}
% \end{flushleft}
\end{lemma}

\begin{proof}
{\bf \Cref{lem:potential:bound}} follows as $\Psi(H \otimes \bB \otimes L) \leq \big(|\bPhi| - 1\big) + 1 = |\bPhi|$. Namely, the potential cannot counts the ceiling value $\phi_{H,\, 0}$ itself, but may add an extra ONE for the nil conditional value $P^{\downarrow} \equiv 0$.
{\bf \Cref{lem:potential:floor}} follows directly from \Cref{def:potential}.
{\bf \Cref{lem:potential:ceiling}} holds since (\blackref{re:monotonicity}/\blackref{re:layeredness}) the bid-to-value table $\bPhi = [\phi_{\sigma,\, j}]$ is row-wise increasing and column-wise decreasing, so the potential $\Psi(H^{\uparrow} \otimes \bB \otimes L) = \big|\big\{(\sigma,\, j) \in \bPhi: \phi_{\sigma,\, j} > \phi_{H,\, 0}\big\}\big|$ is zero iff 
(\blackref{re:twin_ceiling}) all the row-$H$ entries are the same $\phi_{H,\, 0} = \dots = \phi_{H,\, j} = \dots = \phi_{H,\, m}$.
\end{proof}

Because the table $\bPhi = [\phi_{\sigma,\, j}]$ is row-wise increasing and column-wise decreasing,
all ultra-ceiling entries $\phi_{\sigma,\, j} > \phi_{H,\, 0}$ (if existential) together form an upper-right subregion (\Cref{fig:potential}). Among those entries, the {\em leftmost-then-lowest} one $\phi_{\sigma^{*},\, j^{*}} > \phi_{H,\, 0}$ will be of particular interest and we would call it the {\em jump entry} (\Cref{def:jump}). In the sense of \Cref{fig:potential:pseudo,fig:potential:real}, the name ``jump'' describes the behavior of the bid-to-value mappings $\bvarphi = \{\varphi_{\sigma}\}_{\sigma \in \{H\} \cup [n] \{L\}}$ at that entry.

\begin{definition}[Jumps]
\label{def:jump}
For a {\em ceiling} pseudo instance $H^{\uparrow} \otimes \bB \otimes L \in \Bceiling$, define the \term[\textbf{jump entry}]{jump_entry} $(\sigma^{*},\, j^{*})$ as the {\em leftmost-then-lowest} ultra-ceiling entry $\phi_{\sigma,\, j} > \phi_{\phi_{H,\, 0}}$ (if existential; \Cref{def:potential}) of the bid-to-value table $\bPhi = [\phi_{\sigma,\, j}]$:
\begin{itemize}
    \item $j^{*} \eqdef \min \big\{ j \in [0: m]: \phi_{H,\, j} > \phi_{H,\, 0} \big\}$ denotes the {\em jump piece} of the underlying partition $\bLambda$.
    
    \item $\sigma^{*} \eqdef \max \big\{ \sigma \in \{H < 1 < \dots < n < L\}: \phi_{\sigma,\, j^{*}} > \varphi_{H,\, 0} \big\}$ denotes the {\em jump bidder}.
\end{itemize}
Technically,\footnote{That is, we will handle a \blackref{pseudo_jump} (\Cref{subsec:halve}) versus a \blackref{real_jump} (\Cref{subsec:AD}) via different approaches. But specifically for a \blackref{real_jump} $\sigma^{*} \in \{H\} \cup [n]$, there is no difference between a jump monopolist $\sigma^{*} = H$ and a jump non-monopoly bidder $\sigma^{*} \in [n]$.}
there are two types of jumps (if well-defined; see \Cref{lem:jump}):
\begin{itemize}
    \item \term[\textbf{pseudo jump}]{pseudo_jump}{\bf :}
    The jump bidder is the pseudo bidder $\sigma^{*} = L$ (\Cref{fig:potential:pseudo}). \\
    At such a jump, the bid-to-value spectrum $\big\{\, \varphi_{L}(b) \leq v \leq \varphi_{H}(b): b \in [0,\, \lambda] \,\big\}$ decomposes to two {\em disconnected} parts.
    
    \item \term[\textbf{real jump}]{real_jump}{\bf :}
    The jump bidder is one of the real bidders $\sigma^{*} \in \{H\} \cup [n]$ (\Cref{fig:potential:real}). \\
    At such a jump, the bid-to-value spectrum $\big\{\, \varphi_{L}(b) \leq v \leq \varphi_{H}(b): b \in [0,\, \lambda] \,\big\}$ is still {\em connected}.
\end{itemize}
Moreover, denote by $\lambda^{*} \eqdef \lambda_{j^{*}}$ the {\em jump bid} and by $\phi^{*} \eqdef \phi_{\sigma^{*},\, j^{*}}$ the {\em jump value}.
\end{definition}

\Cref{lem:jump} follows directly from \Cref{def:potential,def:jump}, yet we include it for completeness.

% As mentioned (\Cref{def:potential}), the ultra-ceiling entries $\phi_{\sigma,\, j} > \phi_{H,\, 0}$ of a table $\bPhi = [\phi_{\sigma,\, j}]$ form an upper-right subregion, and a {\em ceiling} pseudo instance's potential exactly counts these ultra-ceiling entries $\Psi(H^{\uparrow} \otimes \bB \otimes L) = \sum_{(\sigma,\, j) \in \bPhi} \indicator(\phi_{\sigma,\, j} > \phi_{H,\, 0})$. For these reasons, we can immediately conclude \Cref{lem:jump}.

\begin{lemma}[Jumps]
\label{lem:jump}
% \begin{flushleft}
When a ceiling pseudo instance $H^{\uparrow} \otimes \bB \otimes L \in \Bceiling$ has a nonzero potential $\Psi(H^{\uparrow} \otimes \bB \otimes L) > 0$, i.e., (\Cref{lem:potential:ceiling} of \Cref{lem:potential}) when it violates \blackref{re:twin_ceiling}:
\begin{enumerate}[font = {\em\bfseries}]
    \item The \blackref{jump_entry} $(\sigma^{*},\, j^{*}) \in (\{H\} \cup [n] \cup \{L\}) \times [m]$ exists and cannot be in column $0$.
    
    \item The jump bid as the index $j^{*} \in [m]$ partition bid $\lambda^{*} \equiv \lambda_{j^{*}} \in (0,\, \lambda)$ neither can be the nil bid $\lambda_{0} \equiv 0$ nor can be the supremum bid $\lambda_{m + 1} \equiv \lambda$.
\end{enumerate}
% \end{flushleft}
\end{lemma}

\begin{remark}[Jumps]
\label{rem:jump}
Without ambiguity, we can slightly generalize the concept of jump. Concretely, when a ceiling pseudo instance $H^{\uparrow} \otimes \bB \otimes L \in \Bceiling$ has a zero potential $\Psi(H^{\uparrow} \otimes \bB \otimes L) = 0$, i.e., (\Cref{lem:potential}) when it satisfies \blackref{re:twin_ceiling}, we can reinterpret the {\em undefined} jump bid as the supremum bid $\lambda^{*} = \lambda_{m + 1} \equiv \lambda$.
\end{remark}

The above discussions
% on potentials and jumps (\Cref{def:potential,def:jump})
investigate the first condition, \blackref{re:twin_ceiling}, for a {\em twin ceiling} pseudo instance $H^{\uparrow} \otimes \bB \otimes L \in \Btwin$ (\Cref{def:twin}).
% (Maybe we can interpret this condition as ``the jump bid $\lambda^{*}$ takes the supremum bid $\lambda_{m + 1} = \lambda$.'')
We shall further consider the second condition, \blackref{re:twin_collapse}.
Indeed, by considering this condition just before the jump bid $b \in [0,\, \lambda^{*})$ (instead of the whole bid support $b \in [0,\, \lambda]$), we will introduce the concept of {\em strong ceiling} pseudo instances (\Cref{def:strong}).

\begin{definition}[Strong ceiling pseudo instances]
\label{def:strong}
A {\em ceiling} pseudo instance $H^{\uparrow} \otimes \bB \otimes L \in \Bceiling$ from \Cref{def:ceiling_floor:restate} is further called {\em strong ceiling} when (\term[\textbf{non-monopoly collapse}]{re:collapse}) before the jump bid $b \in [0,\, \lambda^{*})$, each non-monopoly bidder $i \in [n]$ (if existential) exactly takes the pseudo bid-to-value mapping $\varphi_{i}(b) \equiv \varphi_{L}(b)$; therefore in each before-jump column $j \in [0: j^{*} - 1]$ of the bid-to-value table $\bPhi = [\phi_{\sigma,\, j}]$, all of the non-monopoly entries are the same as the pseudo entry $\phi_{1,\, j} = \dots = \phi_{n,\, j} = \phi_{L,\, j}$.
That is, before the jump bid $b \in [0,\, \lambda^{*})$, each non-monopoly bidder $i \in [n]$ has a {\em constant} bid distribution $B_{i}(b) = B_{i}(\lambda^{*})$ and thus has no effect.\footnote{Following \Cref{def:ceiling_floor:restate}, $\varphi_{i}(b) = \varphi_{L}(b)$ means $B'_{i}(b) \big/ B_{i}(b) = 0$, which together with the boundary condition at the jump bid $\lambda^{*} \equiv \lambda_{j^{*}}$ implies that $B_{i}(b) = B_{i}(\lambda^{*})$ on $b \in [0,\, \lambda^{*}]$.}

Denote by $\Bstrong$ the space of such pseudo instances, an intermediate class between the {\em ceiling} class and the {\em twin ceiling} class $\Bceiling \supsetneq \Bstrong \supsetneq \Btwin$ (\Cref{def:ceiling_floor:restate,def:twin}).
\end{definition}

\Cref{lem:strong} follows directly from \Cref{def:twin,def:strong}, \Cref{lem:potential} (\Cref{lem:potential:ceiling}) and \Cref{rem:jump}.

\begin{lemma}[Twin ceiling pseudo instances]
\label{lem:strong}
A strong ceiling pseudo instance $H^{\uparrow} \otimes \bB \otimes L \in \Bstrong$ is further a twin ceiling pseudo instance $H^{\uparrow} \otimes \bB \otimes L \cong H^{\uparrow} \otimes L \in \Btwin$ iff one of the three equivalent conditions holds: (i)~It satisfies \blackref{re:twin_ceiling}. (ii)~Its potential is zero $\Psi(H^{\uparrow} \otimes \bB \otimes L) = 0$. (iii)~Its jump bid is the supremum bid $\lambda^{*} = \lambda$.
\end{lemma}

\afterpage{
\begin{figure}
    \centering
    \includegraphics[width = .45\textwidth, height = 6.375cm]
    {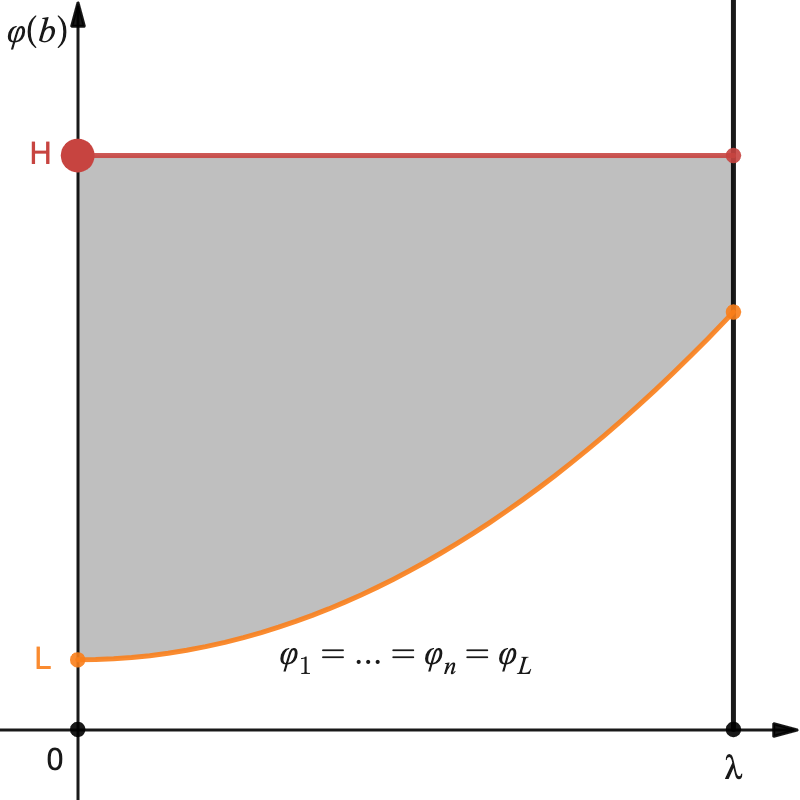}
    \caption{Example of {\em twin ceiling} pseudo instances.
    \label{fig:twin}}
\end{figure}
\begin{figure}
    \centering
    \subfloat[\label{fig:reduction:floor}
    {\em floor}, thus {\em non-ceiling}]{
    \includegraphics[width = .45\textwidth, height = 6.375cm]
    {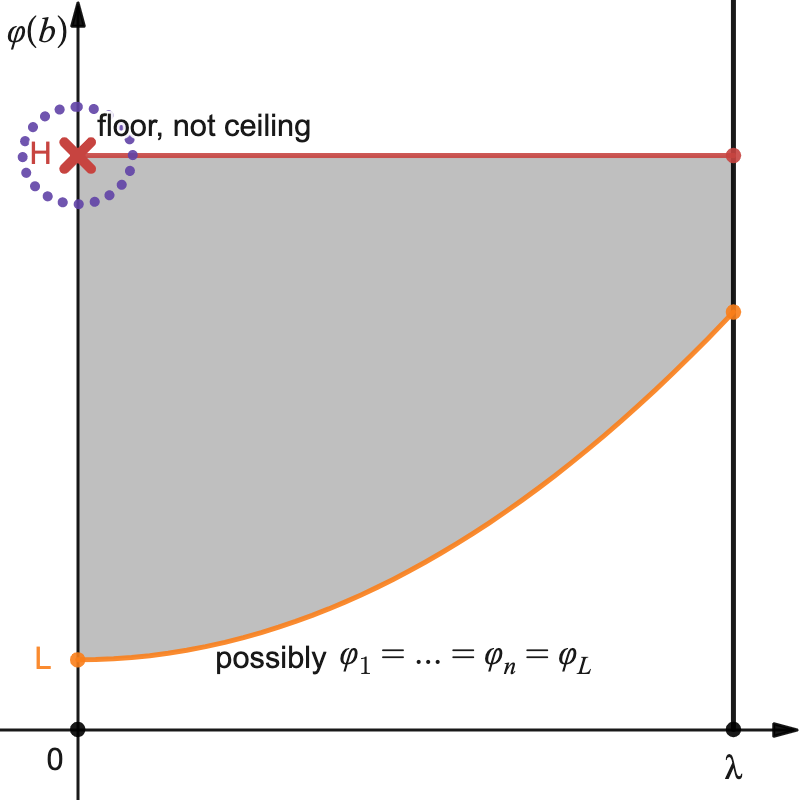}}
    \hfill
    \subfloat[\label{fig:reduction:non_strong}
    {\em ceiling} but {\em non strong ceiling}]{
    \includegraphics[width = .45\textwidth, height = 6.375cm]
    {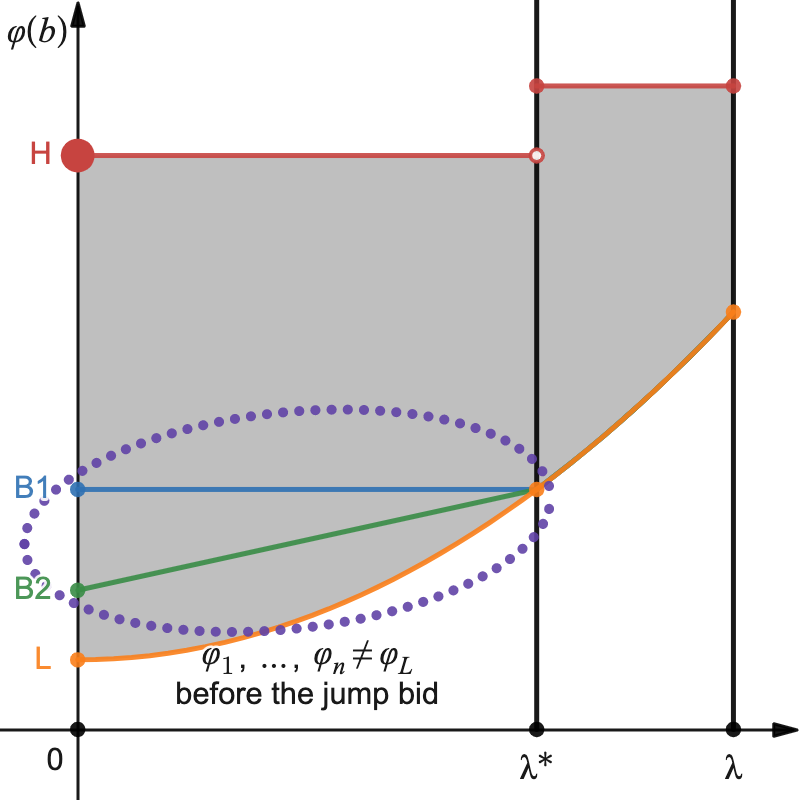}} \\
    \subfloat[\label{fig:reduction:pseudo}
    {\em strong} but {\em non twin ceiling}; {\bf pseudo jump}]{
    \includegraphics[width = .45\textwidth, height = 6.375cm]
    {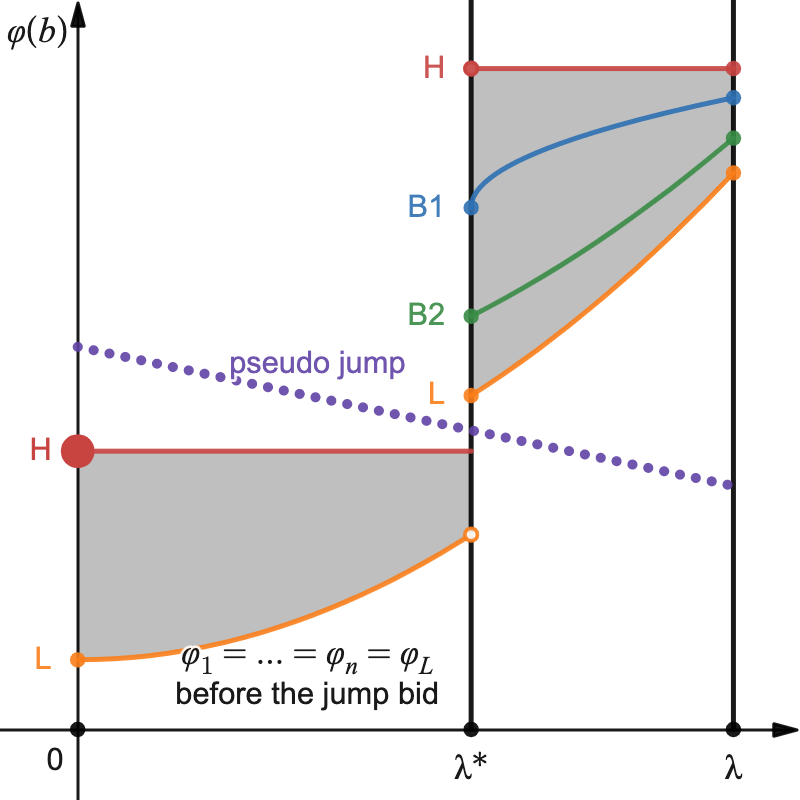}}
    \hfill
    \subfloat[\label{fig:reduction:real}
    {\em strong} but {\em non twin ceiling}; {\bf real jump}]{
    \includegraphics[width = .45\textwidth, height = 6.375cm]
    {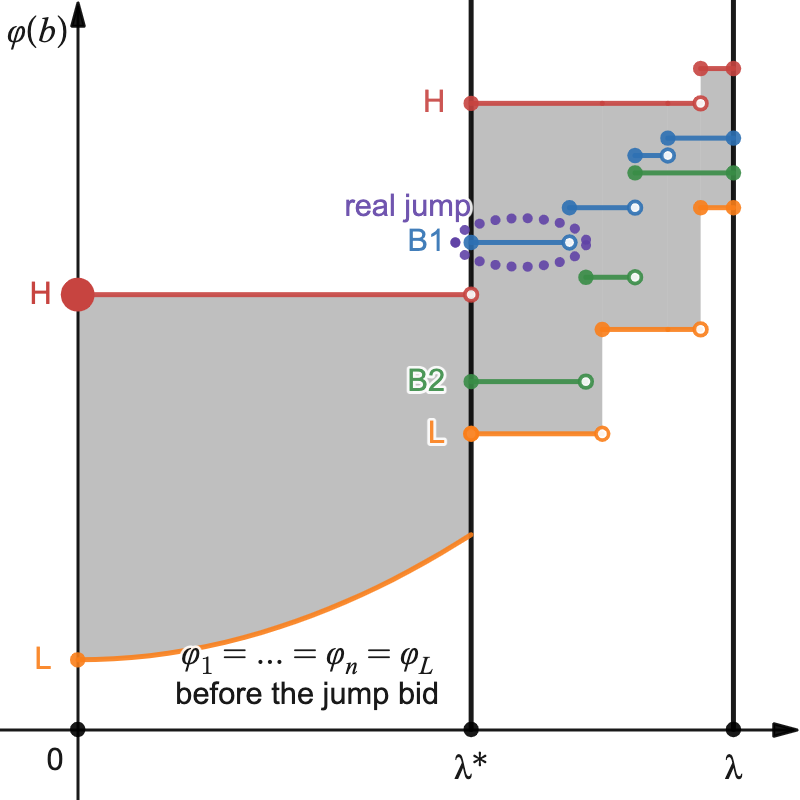}}
    \caption{Diagram of the case analysis in \Cref{sec:reduction}.
    \label{fig:reduction}}
\end{figure}
\clearpage}

\noindent
\term[\textbf{Sketch}]{ref:reduction:sketch}{\bf .}
Based on the above materials, now we are ready to describe how to transform {\em floor}/{\em ceiling} pseudo instances $(\Bfloor \cup \Bceiling)$ into {\em twin ceiling} pseudo instances $\Btwin$, as claimed. Indeed, for an {\em undesirable} pseudo instance $H \otimes \bB \otimes L \in (\Bfloor \cup (\Bceiling \setminus \Btwin))$, there are four types of possible modifications, which are one-to-one illustrated in \Cref{fig:reduction:floor,fig:reduction:non_strong,fig:reduction:pseudo,fig:reduction:real}. (Cf.\ \Cref{fig:twin} for comparison with {\em twin ceiling} pseudo instances $\in \Btwin$. Recall that $\Btwin \subsetneq \Bstrong \subsetneq \Bceiling$.)
\begin{flushleft}
\begin{itemize}
    \item {\bf $H \otimes \bB \otimes L \in \Bfloor$:}
    A {\em floor} pseudo instance (\Cref{fig:reduction:floor}){\bf .} \\
    We will deal with this case in \Cref{subsec:slice}, using the \blackref{alg:slice} reduction.
    
    \item {\bf $H \otimes \bB \otimes L \in (\Bceiling \setminus \Bstrong)$:}
    A {\em ceiling} but {\em non strong ceiling} pseudo instance (\Cref{fig:reduction:non_strong}){\bf .} \\
    We will deal with this case in \Cref{subsec:collapse}, using the \blackref{alg:collapse} reduction.
    
    \item {\bf $H \otimes \bB \otimes L \in (\Bstrong \setminus \Btwin)$; \blackref{pseudo_jump} $\sigma^{*} = L$:}
    A {\em strong ceiling} but {\em non twin ceiling} pseudo instance, and the jump bidder is the pseudo bidder $\sigma^{*} = L$ (\Cref{fig:reduction:pseudo}){\bf .} \\
    We will deal with this case in \Cref{subsec:halve}, using the \blackref{alg:halve} reduction.
    
    \item {\bf $H \otimes \bB \otimes L \in (\Bstrong \setminus \Btwin)$; \blackref{real_jump} $\sigma^{*} \neq L$:} 
    A {\em strong ceiling} but {\em non twin ceiling} pseudo instance, and the jump bidder is one of real bidders $\sigma^{*} \in \{H\} \cup [n]$ (\Cref{fig:reduction:real}){\bf .} \\
    We will deal with this case in \Cref{subsec:AD}, using the \blackref{alg:AD} reduction.
\end{itemize}
\end{flushleft}
As mentioned, in \Cref{subsec:main} we will leverage all the four reductions to build the \blackref{alg:main} procedure, and upper bound its running time through the {\em potential method}.

\subsection{{\slice}: From floor to ceiling}
\label{subsec:slice}

This subsection presents the {\blackref{alg:slice}} reduction (see \Cref{fig:alg:slice,fig:slice} for its description and a visual aid), which transforms (\Cref{def:ceiling_floor:restate}) a {\em floor} pseudo instance $H^{\downarrow} \otimes \bB \otimes L \in \Bfloor$ into a {\em ceiling} pseudo instance $\tilde{H}^{\uparrow} \otimes \tilde{\bB} \otimes \tilde{L} \in \Bceiling$.

% \begin{intuition*}
% \yj{to continue}
% \end{intuition*}

\Cref{lem:slice} summarizes performance guarantees of the \blackref{alg:slice} reduction. (The intuition of this reduction might be obscure. In brief, it generalizes the ideas behind the \blackref{alg:translate} reduction and the \blackref{alg:polarize} reduction from \Cref{sec:preprocessing}.)

\begin{lemma}[{\slice}; \Cref{fig:slice}]
\label{lem:slice}
Under reduction $\tilde{H}^{\uparrow} \otimes \tilde{\bB} \otimes \tilde{L} \gets \slice(H^{\downarrow} \otimes \bB \otimes L)$:
\begin{enumerate}[font = {\em\bfseries}]
    \item\label{lem:slice:property}
    The output is a ceiling pseudo instance $\tilde{H}^{\uparrow} \otimes \tilde{\bB} \otimes \tilde{L} \in \Bceiling$.
    
    \item\label{lem:slice:potential}
    The potential strictly decreases $\Psi(\tilde{H}^{\uparrow} \otimes \tilde{\bB} \otimes \tilde{L}) \leq \Psi(H^{\downarrow} \otimes \bB \otimes L) - 1$.
    
    \item\label{lem:slice:poa}
    A (weakly) worse bound is yielded $\PoA(\tilde{H}^{\uparrow} \otimes \tilde{\bB} \otimes \tilde{L}) \leq \PoA(H^{\downarrow} \otimes \bB \otimes L)$.
\end{enumerate}
\end{lemma}

\afterpage{
\begin{figure}[t]
    \centering
    \begin{mdframed}
    Reduction $\term[\slice]{alg:slice}(H^{\downarrow} \otimes \bB \otimes L)$
    
    \begin{flushleft}
    {\bf Input:}
    A {\em floor} pseudo instance $H^{\downarrow} \otimes \bB \otimes L \in \Bfloor$.\white{\term[\text{\em input}]{slice_input}}
    \hfill
    \Cref{def:ceiling_floor}
    
    \vspace{.05in}
    {\bf Output:}
    A {\em ceiling} pseudo instance $\tilde{H}^{\uparrow} \otimes \tilde{\bB} \otimes \tilde{L} \in \Bceiling$.\white{\term[\text{\em output}]{slice_output}}
    \hfill
    \Cref{def:ceiling_floor}
    
    \begin{enumerate}
        \item\label{alg:slice:spectrum}
        Define a spectrum of {\em floor} pseudo instances $H^{(t)\downarrow} \otimes \bB^{(t)} \otimes L^{(t)}$, with $0 \leq t < \lambda$, \\
        given by $B_{\sigma}^{(t)}(b) \equiv B_{\sigma}(b + t) \cdot \indicator(b \geq 0)$ for $\sigma \in \{L\} \cup [n] \cup \{L\}$.
        \white{\term[\text{\em interim}]{slice_interim}}
        
        \item\label{alg:slice:minimum}
        Define the \term[\text{\em minimizer}]{slice_minimizer} $t^{*} \eqdef \argmin \big\{\, \PoA(H^{(t)\downarrow} \otimes \bB^{(t)} \otimes L^{(t)}): 0 \leq t < \lambda \,\big\}$; \\
        breaking ties in favor of the smallest one among all alternatives.
        
        \item\label{alg:slice:output}
        {\bf Return} $H^{(t^{*})\uparrow} \otimes \bB^{(t^{*})} \otimes L^{(t^{*})}$, the particular {\em ceiling} pseudo instance given by the $t^{*}$. \\
        \OliveGreen{$\triangleright$ This is paired with the minimizer {\em floor} pseudo instance $H^{(t^{*})\downarrow} \otimes \bB^{(t^{*})} \otimes L^{(t^{*})}$.}
    \end{enumerate}
    \end{flushleft}
    \end{mdframed}
    \caption{The {\slice} reduction
    \label{fig:alg:slice}}
\end{figure}
\begin{figure}
    \centering
    \subfloat[\label{fig:slice:input}
    The {\em floor} input $H^{\downarrow} \otimes \bB \otimes L \in \Bfloor$]{
    \includegraphics[width = .49\textwidth]
    {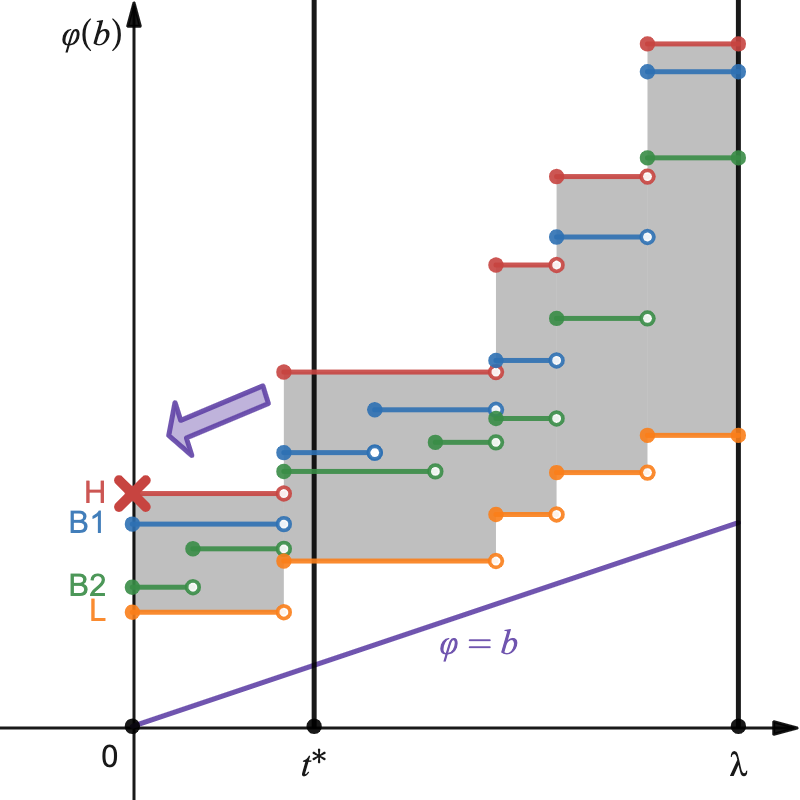}}
    \hfill
    \subfloat[\label{fig:slice:output}
    The {\em ceiling} output $H^{(t^{*})\uparrow} \otimes \bB^{(t^{*})} \otimes L^{(t^{*})} \in \Bceiling$]{
    \includegraphics[width = .49\textwidth]
    {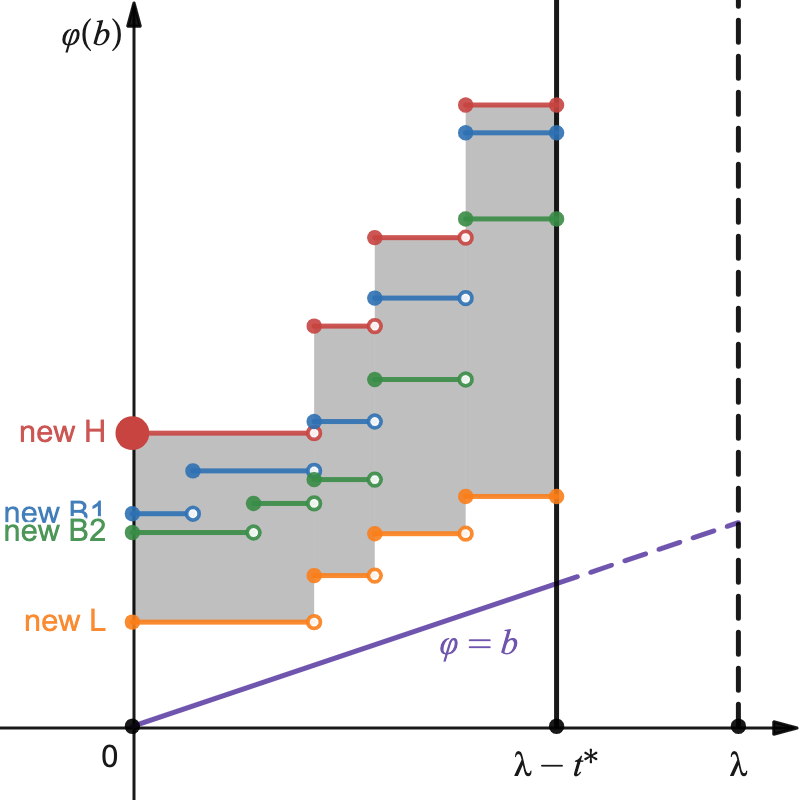}}
    \caption{Diagram of the {\slice} reduction (\Cref{fig:alg:slice}), which transforms (\Cref{def:ceiling_floor:restate}) a {\em floor} pseudo instance $H^{\downarrow} \otimes \bB \otimes L \in \Bfloor$ into a {\em ceiling} pseudo instance $H^{(t^{*})\uparrow} \otimes \bB^{(t^{*})} \otimes L^{(t^{*})} \in \Bceiling$.
    \label{fig:slice}}
\end{figure}
\clearpage}

\newcommand{\FF}{\mathfrak{F}}
\newcommand{\OO}{\mathfrak{O}}
\newcommand{\PP}{\mathfrak{P}}

\begin{proof}
See \Cref{fig:slice} for a visual aid.
Without loss of generality, we assume that the input bound is strictly less than one $\PoA(H^{\downarrow} \otimes \bB \otimes L) < 1$; otherwise, the input {\em floor} pseudo instance $H^{\downarrow} \otimes \bB \otimes L$ cannot be (one of) the worst cases.
For brevity, let $[N] \equiv \{H\} \cup [n] \cup \{L\}$.

The \blackref{alg:slice} reduction
(Line~\ref{alg:slice:spectrum}) considers a spectrum $\big\{\, H^{(t)\downarrow} \otimes \bB^{(t)} \otimes L^{(t)}: 0 \leq t < \lambda \,\big\}$ of \blackref{slice_interim} pseudo instances,
then (Line~\ref{alg:slice:minimum}) takes the {\PoA}-minimizer $t^{*} \in [0,\, \lambda)$,
and then (Line~\ref{alg:slice:output}) outputs the {\em ceiling} pseudo instance $H^{(t^{*})\uparrow} \otimes \bB^{(t^{*})} \otimes L^{(t^{*})}$ given by the $t^{*}$.

Our analysis relies on {\bf \Cref{fact:slice:minimizer}}. But we defer its proof to {\bf \Cref{lem:slice:poa}} to ease readability.

\setcounter{fact}{0}

\begin{fact}
\label{fact:slice:minimizer}
% \begin{flushleft}
The minimizer $t^{*} = \argmin \big\{\, \PoA(H^{(t)\downarrow} \otimes \bB^{(t)} \otimes L^{(t)}): 0 \leq t < \lambda \,\big\}$, breaking ties in favor of the smallest one among all the alternatives (Line~\ref{alg:slice:minimum}), is well defined.
% and bounded away from the supremum bid $\lambda$.
% \end{flushleft}
\end{fact}

\noindent
{\bf \Cref{lem:slice:property}.}
Every \blackref{slice_interim} $H^{(t)\downarrow} \otimes \bB^{(t)} \otimes L^{(t)}$ preserves \blackref{re:discretization}, \blackref{re:monotonicity} and \blackref{re:layeredness}:
This just shifts the input {\em floor} pseudo instance $H^{\downarrow} \otimes \bB \otimes L \equiv H^{(0)\downarrow} \otimes \bB^{(0)} \otimes L^{(0)} \in \Bfloor$ (\Cref{def:ceiling_floor:restate}) by a distance of $-t$, and restricts the support to the nonnegative bids $b \in [0,\, \lambda - t]$.
Every \blackref{slice_interim} $H^{(t)\downarrow} \otimes \bB^{(t)} \otimes L^{(t)}$ further promises \blackref{re:floorness}, thus being a {\em floor} pseudo instance,
including the \blackref{slice_minimizer} $H^{(t^{*})\downarrow} \otimes \bB^{(t^{*})} \otimes L^{(t^{*})} \in \Bfloor$ since the $t^{*} \in [0,\, \lambda)$ is well defined ({\bf \Cref{fact:slice:minimizer}}).
Hence, the paired \blackref{slice_output} $H^{(t^{*})\uparrow} \otimes \bB^{(t^{*})} \otimes L^{(t^{*})}$ satisfies all the required conditions, \blackref{re:ceilingness} etc, being a {\em ceiling} pseudo instance.
{\bf \Cref{lem:slice:property}} follows then.

\vspace{.1in}
\noindent
{\bf \Cref{lem:slice:potential}.}
The potential of a {\em ceiling} pseudo instance $\hat{H}^{\uparrow} \otimes \hat{\bB} \otimes \hat{L}$ counts all the ultra-ceiling entries $\Psi(\text{\em ceiling}) = \big|\big\{(\sigma,\, j) \in \hat{\bPhi}: \hat{\phi}_{\sigma,\, j} > \hat{\phi}_{H,\, 0}\big\}\big|$ in the bid-to-value table $\hat{\bPhi} = \big[\hat{\phi}_{\sigma,\, j}\big]$ (\Cref{def:potential}) yet the paired {\em floor} pseudo instance $\hat{H}^{\downarrow} \otimes \hat{\bB} \otimes \hat{L}$ counts another one $\Psi(\text{\em floor}) = \Psi(\text{\em ceiling}) + 1$.

Following {\bf \Cref{lem:slice:property}}, it suffices to prove that every \blackref{slice_interim} $H^{(t)\downarrow} \otimes \bB^{(t)} \otimes L^{(t)}$ for $t \in [0,\, \lambda)$ has a smaller or equal potential $\Psi(\blackref{slice_interim}) \leq \Psi(\blackref{slice_input})$ than the \blackref{slice_input} $H^{\downarrow} \otimes \bB \otimes L \equiv H^{(0)\downarrow} \otimes \bB^{(0)} \otimes L^{(0)}$, which accommodates the \blackref{slice_minimizer} $H^{(t^{*})\downarrow} \otimes \bB^{(t^{*})} \otimes L^{(t^{*})}$.

Without loss of generality, consider a specific $t \in [\lambda_{k},\, \lambda_{k + 1})$ that locates in the index $k \in [0: m]$ piece of the underlying partition $\bLambda = [0 \equiv \lambda_{0},\, \lambda_{1}) \cup
[\lambda_{1},\, \lambda_{2}) \cup \dots \cup
[\lambda_{m},\, \lambda_{m + 1} \equiv \lambda]$.
As \Cref{fig:slice:output} suggests, the \blackref{slice_interim} table $\bPhi^{(t)}$ entrywise shifts the \blackref{slice_input} table $\bPhi = \big[\phi_{\sigma,\, j}\big]$ for $(\sigma,\, j) \in [N] \times [0: m]$ by a distance of $-t$, and then discards the columns $j \in [0: k - 1]$. That is,
\begin{align*}
    \bPhi^{(t)} ~=~
    \Big[\, \text{$\phi_{\sigma,\, j}^{(t)} = \phi_{\sigma,\, j} - t$ for $(\sigma,\, j) \in [N] \times [k: m]$} \,\Big].
\end{align*}
In particular, the ceiling value changes into the SHIFTED row-$H$ column-$k$ entry $\phi_{H,\, k}^{(t)} = \phi_{H,\, k} - t$.
Since the \blackref{slice_input} $H^{\downarrow} \otimes \bB \otimes L \equiv H^{(0)\downarrow} \otimes \bB^{(0)} \otimes L^{(0)}$ and the \blackref{slice_interim} $H^{(t)\downarrow} \otimes \bB^{(t)} \otimes L^{(t)}$ are two {\em floor} pseudo instances, we can deduce that
\begin{align*}
    \Psi(\blackref{slice_interim})
    & ~=~ 1 + \big|\big\{(\sigma,\, j) \in \bPhi^{(t)}: \phi_{\sigma,\, j}^{(t)} > \phi_{H,\, k}^{(t)}\big\}\big| \phantom{\Big.} \\
    & ~=~ 1 + \big|\big\{(\sigma,\, j) \in \bPhi^{(t)}: \phi_{\sigma,\, j} > \phi_{H,\, k}\big\}\big|
    && \text{add back the shifts} \phantom{\Big.} \\
    & ~\leq~ 1 + \big|\big\{(\sigma,\, j) \in \bPhi^{(t)}: \phi_{\sigma,\, j} > \phi_{H,\, 0}\big\}\big|
    && \text{\blackref{re:monotonicity} $\phi_{H,\, 0} \leq \phi_{H,\, k}$} \phantom{\Big.} \\
    & ~\leq~ 1 + \big|\big\{(\sigma,\, j) \in \bPhi: \phi_{\sigma,\, j} > \phi_{H,\, 0}\big\}\big|
    ~=~ \Psi(\blackref{slice_input}).
    && \text{recount of the discarded entries} \phantom{\Big.}
\end{align*}
This equation holds for ({\bf \Cref{fact:slice:minimizer}}) the \blackref{slice_minimizer} $t^{*} \in [0,\, \lambda)$. The paired \blackref{slice_output} $H^{(t^{*})\uparrow} \otimes \bB^{(t^{*})} \otimes L^{(t^{*})}$ thus has a potential $\Psi(\blackref{slice_output}) = \Psi(\blackref{slice_minimizer}) - 1 \leq \Psi(\blackref{slice_input}) - 1$. {\bf \Cref{lem:slice:potential}} follows then.

\vspace{.1in}
\noindent
{\bf \Cref{lem:slice:poa}.}
Regarding every \blackref{slice_interim} $H^{(t)\downarrow} \otimes \bB^{(t)} \otimes L^{(t)}$,
let us consider the following formulas.
\begin{align*}
    & \FF(t)\,
    ~\eqdef~ \FPA(H^{(t)\downarrow} \otimes \bB^{(t)} \otimes L^{(t)}), \phantom{\big.} \\
    & \OO(t)
    ~\eqdef~ \OPT(H^{(t)\downarrow} \otimes \bB^{(t)} \otimes L^{(t)}), \phantom{\big.} \\
    & \PP(t)
    ~\eqdef~ \PoA(H^{(t)\downarrow} \otimes \bB^{(t)} \otimes L^{(t)}). \phantom{\big.}
\end{align*}
We decompose {\bf \Cref{lem:slice:poa}} into a sequence of facts.
%, each of which relies on the preceding ones.

\begin{restate}[{\Cref{fact:slice:minimizer}}]
\begin{flushleft}
The minimizer $t^{*} = \argmin \big\{\, \PP(t): 0 \leq t < \lambda \,\big\}$, breaking ties in favor of the smallest one among all the alternatives (Line~\ref{alg:slice:minimum}), is well defined.
\end{flushleft}
\end{restate}

\begin{proof}
As mentioned, we safely assume the \blackref{slice_input} {\PoA}-bound $\PP(0) \equiv \FF(0) \big/ \OO(0) \equiv \PoA(\blackref{slice_input}) < 1$.
Moreover, it is easy to verify that the formulas $\FF(t)$ and $\OO(t)$ are continuous functions;\footnote{This is because the \blackref{slice_interim} $H^{(t)\downarrow} \otimes \bB^{(t)} \otimes L^{(t)}$ changes ``continuously'' with respect to the $t \in [0,\, \lambda)$. Indeed, the formula $\FF(t)$ and $\OO(t)$ are even differentiable and we will show their derivatives $\FF'(t)$ and $\OO'(t)$ in {\bf \Cref{fact:slice:fpa,fact:slice:opt}}.}
then so is the {\PoA}-formula $\PP(t)$. To make the \blackref{slice_minimizer} $t^{*} = \argmin \big\{\, \PP(t): 0 \leq t < \lambda \,\big\}$ well defined, we just need to show that $\lim_{t \nearrow \lambda} \PP(t) = 1$.

% As \Cref{fig:slice:output} suggests, the \blackref{slice_interim} table $\bPhi^{(t)}$ entrywise shifts the \blackref{slice_input} table $\bPhi = \big[\phi_{\sigma,\, j}\big]$ for $(\sigma,\, j) \in [N] \times [0: m]$ by a distance of $-t$, and then discards the columns $j \in [0: k - 1]$. That is,

We consider a close enough $t \in [\lambda_{m},\, \lambda_{m + 1} \equiv \lambda)$, by which (\Cref{fig:slice}) the \blackref{slice_interim} table $\bPhi^{(t)}$ in {\bf \Cref{lem:slice:potential}} only contains the {\em rightmost} column $\phi_{H,\, m}^{(t)} \geq \phi_{1,\, m}^{(t)} \geq \dots \geq \phi_{n,\, m}^{(t)} \geq \phi_{L,\, m}^{(t)}$.
Thus the optimal {\SocialWelfare} (\Cref{lem:translate_welfare}; with \blackref{re:floorness} $P^{(t)\downarrow} \equiv 0$ and $B_{\sigma}^{(t)}(\lambda - t) = B_{\sigma}(\lambda) = 1$) is given by
\begin{align*}
    \OO(t)
    & = \int_{0}^{+\infty} \Big(1 - \prod_{\sigma \in [N]} \big(1 - \big(1 - B_{\sigma}^{(t)}(0)
    % \tfrac{}{B_{\sigma}^{(t)}(\lambda - t)}
    \big) \cdot \indicator(v < \phi_{\sigma,\, m}^{(t)})\big)\Big) \cdot \d v \\
    & \leq \int_{0}^{+\infty} \Big(\sum_{\sigma \in [N]} \big(1 - B_{\sigma}^{(t)}(0)
    % \tfrac{}{B_{\sigma}^{(t)}(\lambda - t)}
    \big) \cdot \indicator(v < \phi_{\sigma,\, m}^{(t)})\Big) \cdot \d v \\
    % && \mbox{$\prod (1 - x_{i}) \geq 1 - \sum x_{\sigma}$ for $|x_{\sigma}| \leq 1$} \\
    & = \sum_{\sigma \in [N]} \Big(\phi_{\sigma,\, m}^{(t)} \cdot \big(1 - B_{\sigma}^{(t)}(0)
    % \tfrac{}{B_{\sigma}^{(t)}(\lambda - t)}
    \big)\Big) \\
    & = \sum_{\sigma \in [N]} \Big(\int_{0}^{\lambda - t} \phi_{\sigma,\, m}^{(t)} \cdot {B_{\sigma}^{(t)}}'(b) \cdot \d b\Big) \\
    & ~\leq~ \tfrac{1}{\calB^{(t)}(0)} \cdot \sum_{\sigma \in [N]} \Big(\int_{0}^{\lambda - t} \phi_{\sigma,\, m}^{(t)} \cdot {B_{\sigma}^{(t)}}'(b) \cdot \tfrac{\calB^{(t)}(b)}{B_{\sigma}^{(t)}(b)} \cdot \d b\Big)
    && \text{$\calB^{(t)}(0) \leq \calB^{(t)}(b) \leq \tfrac{\calB^{(t)}(b)}{B_{\sigma}^{(t)}(b)}$} \\
    & = \tfrac{1}{\calB^{(t)}(0)} \cdot \FPA(H^{(t)\downarrow} \otimes \bB^{(t)} \otimes L^{(t)})
    = \tfrac{1}{\calB(t)} \cdot \FF(t).
    && \text{\Cref{lem:pseudo_welfare}} \phantom{\Big.}
\end{align*}
Thus, the {\PoA}-bound is at least $\PP(t) \equiv \FF(t) \big/ \OO(t) \geq \calB(t)$ for $t \in [\lambda_{m},\, \lambda)$.
The \blackref{slice_input} first-order bid distribution $\calB(b)$ is continuous on the bid support $b \in [0,\, \lambda]$ (???),
namely $\lim_{t \nearrow \lambda} \calB(t) = 1$.
By the squeeze theorem, we have $\lim_{t \nearrow \lambda} \PP(t) = 1$.
{\bf \Cref{fact:slice:minimizer}} follows then.
\end{proof}

Following {\bf \Cref{fact:slice:minimizer}}, without loss of generality, we can assume that the \blackref{slice_minimizer} is zero $t^{*} = 0$.\footnote{For any two $t,\, \tau \in [0,\, \lambda)$ that $t \leq \tau$, we can regard the $\tau$-\blackref{slice_interim} $H^{(\tau)\downarrow} \otimes \bB^{(\tau)} \otimes L^{(\tau)}$ as the result of (\`{a} la Line~\ref{alg:slice:spectrum}) shifting the $t$-\blackref{slice_interim} $H^{(t)\downarrow} \otimes \bB^{(t)} \otimes L^{(t)}$ by a distance $-(\tau - t)$.}
{\bf \Cref{fact:slice:fpa,fact:slice:opt}} help us to characterize the optimality condition at this \blackref{slice_minimizer} $t^{*} = 0$.

\begin{fact}
\label{fact:slice:fpa}
\begin{flushleft}
The derivative $\FF'(0) = -\big(1 - \calB(0)\big) - \sum_{\sigma \in [N]} \big(\frac{B'_{\sigma}(0)}{B_{\sigma}(0)} \cdot \phi_{\sigma,\, 0} \cdot \calB(0)\big) \leq 0$, where the first-order bid distribution (\Cref{lem:pseudo_welfare}) $\calB(b) = \prod_{\sigma \in [N]} B_{\sigma}(b)$.
\end{flushleft}
\end{fact}

\begin{proof}
We can write the auction \blackref{slice_interim} formula $\FF(t)$ as follows (\Cref{lem:pseudo_welfare}; \blackref{re:discretization} $\gamma = 0$ and \blackref{re:floorness} $P^{(t)\downarrow} \equiv 0$).
\begin{align*}
    \FF(t)
    & = \sum_{\sigma \in [N]} \Big(\int_{0}^{\lambda - t} \varphi_{\sigma}^{(t)}(b) \cdot \tfrac{{B_{\sigma}^{(t)}}'(b)}{B_{\sigma}^{(t)}(b)} \cdot \calB^{(t)}(b) \cdot \d b\Big) \\
    & = \sum_{\sigma \in [N]} \Big(\int_{t}^{\lambda} \big(\varphi_{\sigma}(b) - t\big) \cdot \tfrac{B_{\sigma}'(b)}{B_{\sigma}(b)} \cdot \calB(b) \cdot \d b\Big).
    \qquad\qquad \text{Line~\ref{alg:slice:spectrum}}
    \hspace{3.75cm}
\end{align*}
Therefore, the derivative $\FF'(0)$ is given by\footnote{\label{footnote:slice:fpa}Recall that a function $F(t) = \int_{t}^{\lambda} f(x,\, t) \cdot \d x$ has the derivative $F'(0) = -f(0,\, 0) + \int_{0}^{\lambda} \big(\frac{\partial f(x,\, t)}{\d t}\big)\bigmid_{t = 0} \cdot \d x$.}
\begin{align*}
    \FF'(0)
    & = -\sum_{\sigma \in [N]} \varphi_{\sigma}(0) \cdot \tfrac{B_{\sigma}'(0)}{B_{\sigma}(0)} \cdot \calB(0) - \sum_{\sigma \in [N]} \Big(\int_{0}^{\lambda} \tfrac{B_{\sigma}'(b)}{B_{\sigma}(b)} \cdot \calB(b) \cdot \d b\Big) \\
    & = -\sum_{\sigma \in [N]} \varphi_{\sigma}(0) \cdot \tfrac{B_{\sigma}'(0)}{B_{\sigma}(0)} \cdot \calB(0) - \int_{0}^{\lambda} \calB'(b) \cdot \d b
    \qquad\qquad \mbox{$\calB' = \big(\prod_{\sigma} B_{\sigma}\big)' = \calB \cdot \sum_{\sigma} \frac{B'_{\sigma}}{B_{\sigma}}$} \\
    & = -\sum_{\sigma \in [N]} \phi_{\sigma,\, 0} \cdot \tfrac{B_{\sigma}'(0)}{B_{\sigma}(0)} \cdot \calB(0) - \big(1 - \calB(0)\big).
    \qedhere
\end{align*}
\end{proof}

\begin{fact}
\label{fact:slice:opt}
\begin{flushleft}
The derivative $\OO'(0) = -\big(1 - \calB(0)\big) - \sum_{\sigma \in [N]} \big(\frac{B'_{\sigma}(0)}{B_{\sigma}(0)} \cdot \int_{0}^{\phi_{\sigma,\, 0}} \calV(v) \cdot \d v\big) \leq 0$, where the first-order value distribution (\Cref{lem:translate_welfare}) $\calV(v) = \prod_{(\sigma,\, j) \,\in\, \bPhi} \big(1 - \big(1 - \frac{B_{\sigma}(\lambda_{j})}{B_{\sigma}(\lambda_{j + 1})}\big) \cdot \indicator(v < \phi_{\sigma,\, j})\big)$.
\end{flushleft}
\end{fact}

% This just shifts the input {\em floor} pseudo instance $H^{\downarrow} \otimes \bB \otimes L$ (\Cref{def:ceiling_floor:restate}) by a distance of $-t$, and restricts the support to the nonnegative bids $b \in [0,\, \lambda - t]$.

% , by which (\Cref{fig:slice}) the \blackref{slice_interim} table $\bPhi^{(t)}$ studied in {\bf \Cref{lem:slice:potential}} is left with the {\em rightmost} column $\phi_{H,\, m}^{(t)} \geq \phi_{1,\, m}^{(t)} \geq \dots \geq \phi_{n,\, m}^{(t)} \geq \phi_{L,\, m}^{(t)}$.
% Following \Cref{lem:translate_welfare} (with $P^{(t)\downarrow} \equiv 0$), the optimal {\SocialWelfare} is given by

\begin{proof}
Given the \blackref{slice_input} partition $\bLambda = [0 \equiv \lambda_{0},\, \lambda_{1}) \cup [\lambda_{1},\, \lambda_{2}) \cup \dots \cup [\lambda_{m},\, \lambda_{m + 1} \equiv \lambda]$, let us consider a small enough $t \in [0,\, \lambda_{1})$.
Namely, the \blackref{slice_input} $H^{\downarrow} \otimes \bB \otimes L$ (Line~\ref{alg:slice:spectrum}) shifts by a distance $-t$ and only the shifted index-$0$ piece $[-t,\, \lambda_{1} - t)$ is restricted to the nonnegative bids $b \in [0,\, \lambda_{1} - t)$.

We consider the two-variable functions $f_{\sigma}(v,\, t) = 1 - \big(1 - \tfrac{B_{\sigma}(t)}{B_{\sigma}(\lambda_{1})}\big) \cdot \indicator(v < \phi_{\sigma,\, 0})$ for $\sigma \in [N]$ that account for the index-$0$ piece and the function $\calS(v) = \prod_{(\sigma,\, j) \,\in\, [N] \times [m]} \big(1 - \big(1 - \frac{B_{\sigma}(\lambda_{j})}{B_{\sigma}(\lambda_{j + 1})}\big) \cdot \indicator(v < \phi_{\sigma,\, j})\big)$ that accounts for the other pieces. The optimal \blackref{slice_interim} formula (\Cref{lem:translate_welfare}; \blackref{re:floorness} $P^{(t)\downarrow} \equiv 0$) can be written as $\OO(t) = \int_{t}^{+\infty} \big(1 - \calS(v) \cdot \prod_{\sigma \in [N]} f_{\sigma}(v,\, t)\big) \cdot \d v$.

% Consider the function $\calS(v) = \prod_{(\sigma,\, j) \,\in\, [N] \times [m]} \big(1 - \big(1 - \frac{B_{\sigma}(\lambda_{j})}{B_{\sigma}(\lambda_{j + 1})}\big) \cdot \indicator(v < \phi_{\sigma,\, j})\big)$ to account for the other pieces.

% We can write the optimal \blackref{slice_interim} formula as follows (\Cref{lem:translate_welfare}; \blackref{re:floorness} $P^{(t)\downarrow} \equiv 0$).
\begin{comment}
\begin{align*}
    \OO(t)
    & = \int_{0}^{+\infty} \Big(1 - \prod_{\sigma \in [N]} \big(1 - \big(1 - \tfrac{B_{\sigma}^{(t)}(0)}{B_{\sigma}^{(t)}(\lambda_{1} - t)}\big) \cdot \indicator(v < \phi_{\sigma,\, 0}^{(t)})\big) \cdot \calS(v + t)\Big) \cdot \d v
    \hspace{3.68cm} \\
    & = \int_{t}^{+\infty} \Big(1 - \prod_{\sigma \in [N]} \big(1 - \big(1 - \tfrac{B_{\sigma}(t)}{B_{\sigma}(\lambda_{1})}\big) \cdot \indicator(v < \phi_{\sigma,\, 0})\big) \cdot \calS(v)\Big) \cdot \d v.
    \qquad\qquad \text{Line~\ref{alg:slice:spectrum}}
\end{align*}
\end{comment}

Notice that $\big(\frac{\partial f_{\sigma}(v,\, t)}{\d t}\big)\bigmid_{t = 0} = \tfrac{B'_{\sigma}(t)}{B_{\sigma}(\lambda_{1})}$ for $v < \phi_{\sigma,\, 0}$, while $\big(\frac{\partial f_{\sigma}(v,\, t)}{\d t}\big)\bigmid_{t = 0} = 0$ for $v \geq \phi_{\sigma,\, 0}$. Further, we have $\calS(0) = \prod_{(\sigma,\, j) \,\in\, [N] \times [m]} \frac{B_{\sigma}(\lambda_{j})}{B_{\sigma}(\lambda_{j + 1})} = \prod_{\sigma \in [N]} \frac{B_{\sigma}(\lambda_{1})}{B_{\sigma}(\lambda_{m + 1})} = \prod_{\sigma \in [N]} B_{\sigma}(\lambda_{1}) = \calB(\lambda_{1})$.
Accordingly (cf.\ \Cref{footnote:slice:fpa}), the derivative $\OO'(0)$ is given by
\begin{align}
    \OO'(0)
    & = -\Big(1 - \calS(0) \cdot \prod_{\sigma \in [N]} f_{\sigma}(0,\, 0)\Big)
    -\sum_{\sigma \in [N]} \Big(\tfrac{B'_{\sigma}(0)}{B_{\sigma}(\lambda_{1})} \cdot \int_{0}^{\phi_{\sigma,\, 0}} \calS(v) \cdot \prod_{k \neq \sigma} f_{k}(v,\, 0) \cdot \d v\Big)
    \nonumber \\
    & = -\big(1 - \calB(0)\big) - \sum_{\sigma \in [N]} \Big(\tfrac{B'_{\sigma}(0)}{B_{\sigma}(0)} \cdot \int_{0}^{\phi_{\sigma,\, 0}} \calS(v) \cdot \prod_{k \in [N]} f_{k}(v,\, 0) \cdot \d v\Big)
    \label{eq:slice:S1}\tag{S1} \\
    & = -\big(1 - \calB(0)\big) - \sum_{\sigma \in [N]} \Big(\tfrac{B'_{\sigma}(0)}{B_{\sigma}(0)} \cdot \int_{0}^{\phi_{\sigma,\, 0}} \calV(v) \cdot \d v\Big).
    \nonumber
\end{align}
\eqref{eq:slice:S1}: The first term follows from $\calS(0) \cdot \prod_{\sigma} f_{\sigma}(0,\, 0) = \calS(0) \cdot \prod_{\sigma \in [N]} \big(\tfrac{B_{\sigma}(0)}{B_{\sigma}(\lambda_{1})}\big) = \calS(0) \cdot \frac{\calB(0)}{\calB(\lambda_{1})} = \calB(0)$. \\
And the second term follows from $f_{\sigma}(v,\, 0) = \tfrac{B_{\sigma}(0)}{B_{\sigma}(\lambda_{1})}$ when $v \in [0,\, \phi_{\sigma,\, 0})$.
% {\bf \Cref{fact:slice:opt}} follows then.
\end{proof}

Now instead of the {\em floor} \blackref{slice_minimizer} $H^{(t^{*})\downarrow} \otimes \bB^{(t^{*})} \otimes L^{(t^{*})} = H^{\downarrow} \otimes \bB \otimes L$ (with $P^{\downarrow} \equiv 0$), consider the paired {\em ceiling} \blackref{slice_output} $H^{(t^{*})\uparrow} \otimes \bB^{(t^{*})} \otimes L^{(t^{*})} = H^{\uparrow} \otimes \bB \otimes L$ (with $P^{\uparrow} \equiv \phi_{H,\, 0}$).
We can write its auction {\SocialWelfare} $\FPA(\blackref{slice_output}) = \FF(0) + \Delta_{\FPA}$ (\Cref{lem:pseudo_welfare}), using $\Delta_{\FPA} = \phi_{H,\, 0} \cdot \calB(0)$,
and its optimal {\SocialWelfare} $\OPT(\blackref{slice_output}) = \OO(0) + \Delta_{\OPT}$ (\Cref{lem:translate_welfare}), using $\Delta_{\OPT} = \int_{0}^{\phi_{H,\, 0}} \calV(v) \cdot \d v$.
Furthermore, the \blackref{slice_minimizer} $t^{*} = 0$ induces the optimality condition $\PP'(0) = \frac{\FF'(0) \cdot \OO(0) - \FF(0) \cdot \OO'(0)}{\OO(0)^{2}} \geq 0$, which after being rearranged gives $\frac{\FF(0)}{\OO(0)} \geq \frac{|\FF'(0)|}{|\OO'(0)|}$. We then deduce that
% \Cref{eq:slice:S2}.
\begin{align}
    \PP(0)
    ~=~ \frac{\FF(0)}{\OO(0)}
    ~\geq~ \frac{|\FF'(0)|}{|\OO'(0)|}
    & ~=~ \frac{\big(1 - \calB(0)\big) + \sum_{\sigma \in [N]} B'_{\sigma}(0) \big/ B_{\sigma}(0) \cdot \phi_{\sigma,\, 0} \cdot \calB(0) \hspace{0.665cm}}
    {\big(1 - \calB(0)\big) + \sum_{\sigma \in [N]} B'_{\sigma}(0) \big/ B_{\sigma}(0) \cdot \int_{0}^{\phi_{\sigma,\, 0}} \calV(v) \cdot \d v} \hspace{.7cm}
    \label{eq:slice:S2}\tag{S2} \\
    & ~\geq~ \frac{\sum_{\sigma \in [N]} B'_{\sigma}(0) \big/ B_{\sigma}(0) \cdot \phi_{\sigma,\, 0} \cdot \calB(0) \hspace{0.665cm}}
    {\sum_{\sigma \in [N]} B'_{\sigma}(0) \big/ B_{\sigma}(0) \cdot \int_{0}^{\phi_{\sigma,\, 0}} \calV(v) \cdot \d v}
    \label{eq:slice:S3}\tag{S3} \\
    & ~\geq~ \frac{\phi_{H,\, 0} \cdot \calB(0)}
    {\int_{0}^{\phi_{H,\, 0}} \calV(v) \cdot \d v}
    ~=~ \frac{\Delta_{\FPA}}{\Delta_{\OPT}}
    \label{eq:slice:S4}\tag{S4}
\end{align}
\eqref{eq:slice:S2}: Apply {\bf \Cref{fact:slice:fpa,fact:slice:opt}}. \\
\eqref{eq:slice:S3}: The dropped two terms are equal, while $\RHS \text{ of } \eqref{eq:slice:S2} \leq \PP(0) < 1$. \\
% follows for two reasons. The numerator in \Cref{eq:slice:S2} is less than or equal to the denominator $\big| \FF'(0) \big| \leq \big| \OO'(0) \big|$, due to the optimality condition~\eqref{eq:slice:S4}. \\
\eqref{eq:slice:S4}: Replace $\phi_{H,\, 0} \geq \phi_{1,\, 0} \geq \dots \geq \phi_{n,\, 0} \geq \phi_{L,\, 0}$ (\blackref{re:layeredness}) with the highest ceiling value $\phi_{H,\, 0}$.
Notice that $\big(\phi \cdot \calB(0)\big) \big/ \big(\int_{0}^{\phi} \calV(v) \cdot \d v\big)$ is decreasing in $\phi > 0$ because the CDF $\calV(v)$ is increasing.

We thus conclude that $\PoA(\blackref{slice_output})
= \frac{\FF(0) + \Delta_{\FPA}}{\OO(0) + \Delta_{\OPT}}
\leq \frac{\FF(0)}{\OO(0)}
= \PoA(\blackref{slice_minimizer})
\leq \PoA(\blackref{slice_input})$.
I.e., the {\em ceiling} \blackref{slice_output} yields a (weakly) worse bound. {\bf \Cref{lem:halve:poa}} follows then.
This finishes the proof.
\end{proof}

\subsection{{\collapse}: From ceiling to strong ceiling}
\label{subsec:collapse}

Following the previous subsection, we are able to handle {\em floor} pseudo instances $\in \Bfloor$.
Therefore,
% except the desirable {\em twin ceiling} pseudo instances,
we can restrict our attention to {\em ceiling} pseudo instances $H^{\uparrow} \otimes \bB \otimes L \in \Bceiling$.
Indeed, every {\em ceiling} pseudo instance can be transformed into a {\em strong ceiling} pseudo instance, through the {\blackref{alg:collapse}} reduction (see \Cref{fig:alg:collapse,fig:collapse} for its description and a visual aid). For ease of reference, let us rephrase the definition of {\em strong ceiling} pseudo instances.

\begin{restate}[\Cref{def:strong}]
{\em A {\em ceiling} pseudo instance $H^{\uparrow} \otimes \bB \otimes L \in \Bceiling$ from \Cref{def:ceiling_floor:restate} is further called {\em strong ceiling} when (\term[\textbf{non-monopoly collapse}]{restate:pro:collapse}) each non-monopoly bidder $i \in [n]$ (if existential) exactly takes the pseudo bid-to-value mapping $\varphi_{i}(b) \equiv \varphi_{L}(b)$ before the jump bid $b \in [0,\, \lambda^{*})$; therefore in each before-jump column $j \in [0: j^{*} - 1]$ of the bid-to-value table $\bPhi = [\phi_{\sigma,\, j}]$, all of the non-monopoly entries are the same as the pseudo entry $\phi_{1,\, j} = \dots = \phi_{n,\, j} = \phi_{L,\, j}$.
That is, before the jump bid $b \in [0,\, \lambda^{*})$, each non-monopoly bidder $i \in [n]$ has a {\em constant} bid distribution $B_{i}(b) = B_{i}(\lambda^{*})$ and thus has no effect.}
\end{restate}

\begin{intuition*}
Regarding a ceiling pseudo instance $H^{\uparrow} \otimes \bB \otimes L \in \Bceiling$, in the sense of \Cref{lem:ceiling_welfare}, the optimal {\SocialWelfare} formula $\OPT(H^{\uparrow} \otimes \bB \otimes L)$ turns out to be irrelevant to the non-ultra-ceiling entries $\phi_{\sigma,\, j} \leq \phi_{H,\, 0}$ in the bid-to-value table $\bPhi = \big[\phi_{\sigma,\, j}\big]$ (except the ceiling value $\phi_{H,\, 0}$ itself).
Conceivably, we would modify those non-ultra-ceiling entries $\phi_{\sigma,\, j} \leq \phi_{H,\, 0}$ to minimize the auction {\SocialWelfare} $\FPA(H^{\uparrow} \otimes \bB \otimes L)$, while preserving \blackref{re:monotonicity} etc.
% \pl{Not quite understand this. }
\end{intuition*}

\begin{lemma}[{\SocialWelfares}]
\label{lem:ceiling_welfare}
Following \Cref{lem:translate_welfare}, for a ceiling pseudo instance $H^{\uparrow} \otimes \bB \otimes L$, the expected optimal {\SocialWelfare} can be written as
\begin{align*}
    \OPT(H^{\uparrow} \otimes \bB \otimes L)
    ~=~ \int_{0}^{+\infty} \bigg(1 - \indicator(v \geq \phi_{H,\, 0}) \cdot \prod_{(\sigma,\, j) \,\in\, \bPhi^{*}} \Big(1 - \omega_{\sigma,\, j} \cdot \indicator(v < \phi_{\sigma,\, j})\Big)\bigg) \cdot \d v.
\end{align*}
where the subtable $\bPhi^{*} \subseteq \bPhi$ includes all of the ultra-ceiling entries $\bPhi^{*} \supseteq \big\{(\sigma,\, j) \in \bPhi: \phi_{\sigma,\, j} > \phi_{H,\, 0}\big\}$ but otherwise is arbitrary.
\end{lemma}

\begin{proof}
Following \Cref{lem:translate_welfare} (with \blackref{re:ceilingness} $P^{\uparrow} \equiv \phi_{H,\, 0}$), the optimal {\SocialWelfare} is given by $\OPT(H^{\uparrow} \otimes \bB \otimes L)
= \int_{0}^{+\infty} \big(1 - \indicator(v \geq \phi_{H,\, 0}) \cdot \prod_{(\sigma,\, j) \,\in\, \bPhi} \big(1 - \omega_{\sigma,\, j} \cdot \indicator(v < \phi_{\sigma,\, j})\big)\big) \cdot \d v$.
Each non-ultra-ceiling entry $\phi_{\sigma,\, j} \leq \phi_{H,\, 0}$ has no effect because of the leading {\em indicator} term $\indicator(v \geq \phi_{H,\, 0})$.
\end{proof}

The \blackref{alg:collapse} reduction implements the above intuition, and its performance guarantees are summarized in \Cref{lem:collapse}.

\begin{lemma}[{\collapse}; \Cref{fig:alg:collapse}]
\label{lem:collapse}
Under reduction $H^{\uparrow} \otimes \tilde{\bB} \otimes \tilde{L} \gets \collapse(H^{\uparrow} \otimes \bB \otimes L)$:
\begin{enumerate}[font = {\em\bfseries}]
    \item\label{lem:collapse:property}
    The output is a strong ceiling pseudo instance $H^{\uparrow} \otimes \tilde{\bB} \otimes \tilde{L} \in \Bstrong \subsetneq \Bceiling$; the monopolist $H^{\uparrow}$ is unmodified.
    
    \item\label{lem:collapse:potential}
    The potential keeps the same $\Psi(H^{\uparrow} \otimes \tilde{\bB} \otimes \tilde{L}) = \Psi(H^{\uparrow} \otimes \bB \otimes L)$.
    
    \item\label{lem:collapse:poa}
    A (weakly) worse bound is yielded $\PoA(H^{\uparrow} \otimes \tilde{\bB} \otimes \tilde{L}) \leq \PoA(H^{\uparrow} \otimes \bB \otimes L)$.
\end{enumerate}
\end{lemma}

\afterpage{
\begin{figure}[t]
    \centering
    \begin{mdframed}
    Reduction $\term[\collapse]{alg:collapse}(H \otimes \bB \otimes L)$
    
    \begin{flushleft}
    {\bf Input:} A {\em ceiling} pseudo instance $H \otimes \bB \otimes L \in \Bceiling$
    \white{\term[\text{\em input}]{collapse_input}}
    \hfill
    \Cref{def:ceiling_floor:restate}
    
    \vspace{.05in}
    {\bf Output:} A {\em strong ceiling} pseudo instance $H \otimes \tilde{\bB} \otimes \tilde{L} \in \Bstrong \subsetneq \Bceiling$.
    \white{\term[\text{\em output}]{collapse_output}}
    \hfill
    \Cref{def:strong}
    
    \begin{enumerate}
        \item\label{alg:collapse:distribution}
        Define $\tilde{B}_{i}(b) \equiv \max\big(B_{i}(b),\, B_{i}(\lambda^{*})\big)$ for $i \in [n]$ and $\tilde{L}(b) \equiv L(b) \cdot \prod_{i \in [n]} \big(B_{i}(b) \big/ \tilde{B}_{i}(b)\big)$.
        
        \item[] \OliveGreen{$\triangleright$ The jump bid is one of the partition bids $\lambda^{*} \in \{\lambda_{1} < \dots < \lambda_{m + 1} \equiv \lambda\} \subseteq (0,\, \lambda]$ \\
        except the nil bid $\lambda_{0} \equiv 0$ (\Cref{def:jump,lem:jump,rem:jump}).}
        
        \item\label{alg:collapse:output}
        {\bf Return} $H \otimes \tilde{\bB} \otimes \tilde{L}$.
    \end{enumerate}
    \end{flushleft}
    \end{mdframed}
    \caption{The {\collapse} reduction
    \label{fig:alg:collapse}}
\end{figure}
\begin{figure}
    \centering
    \subfloat[\label{fig:collapse:input}
    {The {\em ceiling} input $H \otimes \bB \otimes L \in \Bceiling$}]{
    \includegraphics[width = .49\textwidth]
    {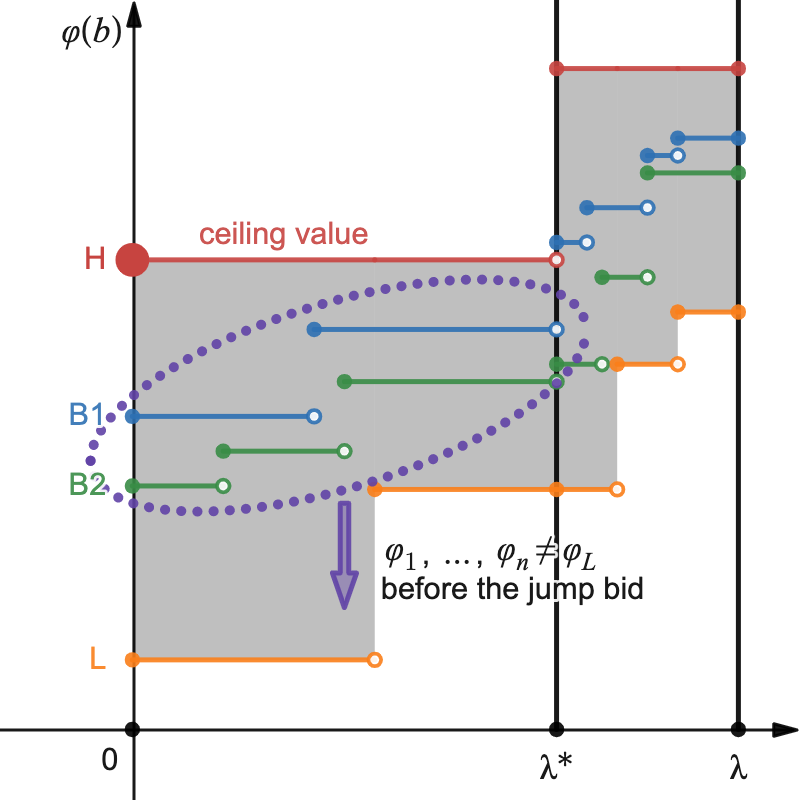}}
    \hfill
    \subfloat[\label{fig:collapse:output}
    {The {\em strong ceiling} output $H \otimes \tilde{\bB} \otimes \tilde{L} \in \Bstrong$}]{
    \includegraphics[width = .49\textwidth]
    {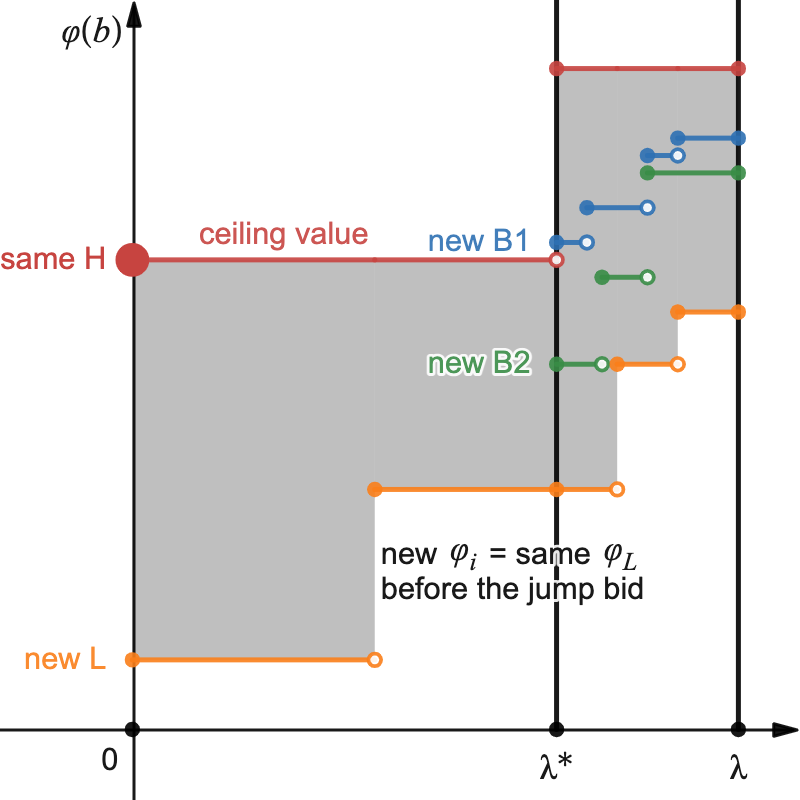}}
    \caption{{Diagram of the {\collapse} reduction (\Cref{fig:alg:collapse}), which transforms (\Cref{def:ceiling_floor:restate}) a {\em ceiling} pseudo instance $H^{\uparrow} \otimes \bB \otimes L \in \Bceiling$ into (\Cref{def:strong}) a {\em strong ceiling} pseudo instance $H^{\uparrow} \otimes \tilde{\bB} \otimes \tilde{L} \in \Bstrong$.}
    % Indeed, this reduction works regardless of (\Cref{def:ceiling_floor:restate,def:jump}) {\bf discretization} before the jump bid $b \in [0,\, \lambda^{*})$.
    \label{fig:collapse}}
\end{figure}
\clearpage}

\begin{proof}
See \Cref{fig:collapse} for a visual aid.
For brevity, let $[N] \equiv \{H\} \cup [n] \cup \{L\}$.
% Let us verify {\bf \Cref{lem:collapse:property,lem:collapse:poa}} one by one. 

\vspace{.1in}
\noindent
{\bf \Cref{lem:collapse:property}.}
By construction, the \blackref{collapse_output} non-monopoly bidders $\tilde{B}_{i}(b) \equiv \max\big(B_{i}(b),\, B_{i}(\lambda^{*})\big)$ for $i \in [n]$ and the \blackref{collapse_output} pseudo bidder $\tilde{L}(b) \equiv L(b) \cdot \prod_{i \in [n]} \big(B_{i}(b) \big/ \tilde{B}_{i}(b)\big)$ are well defined (Line~\ref{alg:collapse:distribution}). Moreover, the monopolist $H^{\uparrow}$ is invariant.

After the jump bid $b \in (\lambda^{*},\, \lambda]$, everything keeps the same.
And before the jump bid $b \in [0,\, \lambda^{*}]$, the first-order bid distribution $\tilde{\calB}(b) = \prod_{\sigma \in [N]} \tilde{B}_{\sigma}(b) = \prod_{\sigma \in [N]} B_{\sigma}(b) = \calB(b)$ keeps the same;
so the pseudo mapping $\varphi_{L}(b) = b + \frac{\calB(b)}{\calB'(b)}$ and the monopolist $H^{\uparrow}$'s mapping $\varphi_{H}(b) = b + (\frac{\calB'(b)}{\calB(b)} - \frac{H'(b)}{H(b)})^{-1}$ are invariant (cf.\ \Cref{fig:collapse}).
Moreover, each non-monopoly bidder $i \in [n]$ {\em does} have a constant bid distributions $\tilde{B}_{i}(b) = B_{i}(\lambda^{*})$ for $b \in [0,\, \lambda^{*}]$ (\Cref{def:strong}).
Hence, the \blackref{collapse_output} $H^{\uparrow} \otimes \tilde{\bB} \otimes \tilde{L}$ satisfies \blackref{re:discretization}, \blackref{re:monotonicity}, \blackref{re:layeredness}, and \blackref{re:ceilingness}
\`{a} la the {\em ceiling} \blackref{collapse_input} $H^{\uparrow} \otimes \bB \otimes L$ (\Cref{def:ceiling_floor:restate}) and additionally satisfies \blackref{restate:pro:collapse} (\Cref{def:strong}), being a {\em strong ceiling} pseudo instance $H^{\uparrow} \otimes \tilde{\bB} \otimes \tilde{L} \in \Bstrong$.
{\bf \Cref{lem:collapse:property}} follows then.

\vspace{.1in}
\noindent
{\bf \Cref{lem:collapse:potential}.}
The potential of a {\em ceiling} pseudo instance $\hat{H}^{\uparrow} \otimes \hat{\bB} \otimes \hat{L}$ counts all the ultra-ceiling entries $\Psi(\text{\em ceiling}) = \big|\big\{(\sigma,\, j) \in \hat{\bPhi}: \hat{\phi}_{\sigma,\, j} > \hat{\phi}_{H,\, 0}\big\}\big|$ in the bid-to-value table $\hat{\bPhi} = \big[\hat{\phi}_{\sigma,\, j}\big]$ (\Cref{def:potential}).
We do consider two {\em ceiling} pseudo instances ({\bf \Cref{lem:collapse:property}}) and need to verify that the \blackref{collapse_output} table $\tilde{\bPhi}$ has the same amount of ultra-ceiling entries as the \blackref{collapse_input} table $\bPhi$.

The \blackref{collapse_output} table $\tilde{\bPhi}$ (Line~\ref{alg:collapse:distribution} and \Cref{fig:collapse}) descends the non-monopoly entries $\phi_{1,\, j} \geq \dots \geq \phi_{n,\, j}$ to the {\em lower} pseudo entries $\phi_{L,\, j}$ (\blackref{restate:pro:collapse}) in before-jump columns $j \in [0: j^{*} - 1]$.
All of these entries $(i,\, j) \in [n] \times [0: j^{*} - 1]$ (\blackref{re:layeredness} and \Cref{def:jump}) are non-ultra-ceiling regardless of the reduction: $[\tilde{\phi}_{i,\, j} = \phi_{i,\, L}] \leq [\phi_{i,\, j}] \leq [\phi_{H,\, j} = \phi_{H,\, 0}]$.
So, the \blackref{collapse_output} $H^{\uparrow} \otimes \tilde{\bB} \otimes \tilde{L}$ keeps the same potential $\Psi(\blackref{collapse_output}) = \Psi(\blackref{collapse_input})$.

\vspace{.1in}
\noindent
{\bf \Cref{lem:collapse:poa}.}
We would show that compared to the \blackref{collapse_input} $H^{\uparrow} \otimes \bB \otimes L$, the \blackref{collapse_output} $H^{\uparrow} \otimes \tilde{\bB} \otimes \tilde{L}$ (weakly) decreases the auction {\SocialWelfare} $\FPA(\blackref{collapse_output}) \leq \FPA(\blackref{collapse_input})$ but keeps the same optimal {\SocialWelfare} $\OPT(\blackref{collapse_output}) = \OPT(\blackref{collapse_input})$.

\vspace{.1in}
\noindent
{\bf Auction {\SocialWelfares}.}
Following \Cref{lem:pseudo_welfare} (with $\gamma \equiv 0$ and $P^{\uparrow} \equiv \phi_{H,\, 0}$ by \Cref{def:ceiling_floor:restate}), the \blackref{collapse_input} has an auction {\SocialWelfare}
$\FPA(\blackref{collapse_input}) = \phi_{H,\, 0} \cdot \calB(0) + \sum_{\sigma \in [N]} \big(\int_{0}^{\lambda} \varphi_{\sigma}(b) \cdot \frac{B'_{\sigma}(b)}{B_{\sigma}(b)} \cdot \calB(b) \cdot \d b\big)$.
As mentioned ({\bf \Cref{lem:collapse:property}}), the following keep the same:
(i)~Everything after the jump bid $b \in (\lambda^{*},\, \lambda]$.
(ii)~Everything about the monopolist $H^{\uparrow}$. (iii)~The first-order value distribution $\tilde{\calB}(b) \equiv \calB(b)$.
We can formulate the output auction {\SocialWelfare} $\FPA(\blackref{collapse_output})$ likewise and, by concentrating on the changed parts (namely $b \in [0,\, \lambda^{*}]$ and $\sigma \neq H$), deduce that
\begin{align}
    \FPA(\blackref{collapse_output})
    & ~=~ \sum_{\sigma \neq H} \bigg(\int_{0}^{\lambda^{*}} \tilde{\varphi}_{\sigma}(b) \cdot \frac{\tilde{B}'_{\sigma}(b)}{\tilde{B}_{\sigma}(b)} \cdot \calB(b) \cdot \d b\bigg) \hspace{0.02cm} ~+~ \mathrm{const}
    \nonumber \\
    & ~=~ \int_{0}^{\lambda^{*}} \tilde{\varphi}_{L}(b) \cdot \frac{\tilde{L}'(b)}{\tilde{L}(b)} \cdot \calB(b) \cdot \d b \hspace{1.35cm} ~+~ \mathrm{const}
    \label{eq:collapse:auction:1}\tag{C1} \\
    & ~=~ \int_{0}^{\lambda^{*}} \varphi_{L}(b) \cdot \bigg(\sum_{\sigma \neq H} \frac{B'_{\sigma}(b)}{B_{\sigma}(b)}\bigg) \cdot \calB(b) \cdot \d b ~+~ \mathrm{const}
    \label{eq:collapse:auction:2}\tag{C2} \\
    & ~\leq~ \sum_{\sigma \neq H} \bigg(\int_{0}^{\lambda^{*}} \varphi_{\sigma}(b) \cdot \frac{B'_{\sigma}(b)}{B_{\sigma}(b)} \cdot \calB(b) \cdot \d b\bigg) ~+~ \mathrm{const}
    ~=~ \FPA(\blackref{collapse_input}).
    \label{eq:collapse:auction:3}\tag{C3}
\end{align}
\eqref{eq:collapse:auction:1}:
$\tilde{B}'_{i} \big/ \tilde{B}_{i} = 0$ for each $i \in [n]$ before the jump bid $b \in [0,\, \lambda^{*}]$ (\blackref{restate:pro:collapse}). \\
\eqref{eq:collapse:auction:2}:
$\tilde{\varphi}_{L} = \varphi_{L}$ ({\bf \Cref{lem:collapse:property}}) and $\tilde{L}' \big/ \tilde{L} = \sum_{\sigma \neq H} B'_{\sigma} \big/ B_{\sigma}$ (Line~\ref{alg:collapse:distribution})
before the jump bid $b \in [0,\, \lambda^{*}]$. \\
% $\tilde{L}(b) = L(b) \cdot \prod_{i \in [n]} \big(B_{i}(b) \big/ B_{i}(\lambda^{*})\big)$. \\
\eqref{eq:collapse:auction:3}:
The pseudo mapping is the dominated mapping $\varphi_{L}(b) \equiv \min(\bvarphi(b^{\otimes n + 1}))$ (\Cref{lem:pseudo_mapping}).

\vspace{.1in}
\noindent
{\bf Optimal {\SocialWelfares}.}
The optimal {\SocialWelfare} keeps the same $\OPT(\blackref{collapse_output}) = \OPT(\blackref{collapse_input})$, as an implication of \Cref{lem:ceiling_welfare}. Basically, the optimal {\SocialWelfare} formula is {\em irrelevant} to the non-ultra-ceiling entries $\phi_{\sigma,\, j} < \phi_{H,\, 0}$, including all the modified (non-monopoly before-jump) entries $\tilde{\phi}_{i,\, j} = \phi_{L,\, j}$ for $(i,\, j) \in [n] \times [0: j^{*} - 1]$.
{\bf \Cref{lem:collapse:poa}} follows then. This finishes the proof.
\end{proof}

\subsection{{\halve}: From strong ceiling towards twin ceiling, under pseudo jumps}
\label{subsec:halve}

Following the previous two subsections, we are able to deal with {\em floor}/{\em ceiling} but {\em non strong ceiling} pseudo instances $(\Bfloor \cup (\Bceiling \setminus \Bstrong))$.
As a consequence, it remains to deal with {\em strong ceiling} but {\em non twin ceiling} pseudo instances $H^{\uparrow} \otimes \bB \otimes L \in (\Bstrong \setminus \Btwin)$.
For ease of reference, below we rephrase the observations from \Cref{subsec:tech_prelim} on such pseudo instances.

\begin{restate}[\Cref{lem:potential,lem:jump,lem:strong}]
For a strong ceiling but non twin ceiling pseudo instance $H^{\uparrow} \otimes \bB \otimes L \in (\Bstrong \setminus \Btwin)$:
\begin{enumerate}[font = {\em\bfseries}]
    \item It violates \blackref{re:twin_ceiling}, i.e., its potential is nonzero $\Psi(H^{\uparrow} \otimes \bB \otimes L) > 0$.
    
    \item The \blackref{jump_entry} $(\sigma^{*},\, j^{*}) \in (\{H\} \cup [n] \cup \{L\}) \times [m]$ exists and cannot be in column $0$.
    
    % \item The jump bidder is either the pseudo bidder $\sigma^{*} = L$ or one of the real bidders $\sigma^{*} \in \{H\} \cup [n]$.
    
    \item The jump bid as the index $j^{*} \in [m]$ partition bid $\lambda^{*} \equiv \lambda_{j^{*}} \in (0,\, \lambda)$ neither can be the nil bid $\lambda_{0} \equiv 0$ nor can be the supremum bid $\lambda_{m + 1} \equiv \lambda$.
\end{enumerate}
\end{restate}

In particular, the current subsection deals with the case of a \term[\textbf{pseudo jump}]{halve_jump} $\sigma^{*} = L$, using the {\blackref{alg:halve}} reduction (see \Cref{fig:alg:halve,fig:halve} for its description and a visual aid).

\begin{figure}[t]
    \centering
    \begin{mdframed}
    Reduction $\term[\halve]{alg:halve}(H^{\uparrow} \otimes \bB \otimes L)$
    
    \begin{flushleft}
    {\bf Input:} A {\em strong ceiling} but {\em non twin ceiling} pseudo instance $H^{\uparrow} \otimes \bB \otimes L \in (\Bstrong \setminus \Btwin)$ \\
    \white{\bf Input:} that has a \blackref{halve_jump} $(\sigma^{*},\, j^{*}) \in \{L\} \times [m]$.
    \white{\term[\text{\em input}]{halve_input}}
    % \hfill
    % \Cref{def:strong,def:twin,def:jump}
    
    \vspace{.05in}
    {\bf Output:} A {\em floor}/{\em ceiling} pseudo instance $\tilde{H} \otimes \tilde{\bB} \otimes \tilde{L} \in (\Bfloor \cup \Bceiling)$.
    \white{\term[\text{\em output}]{halve_output}}
    
    {\begin{enumerate}
        \item\label{alg:halve:left}
        Define the \term[\text{\em left}]{halve_left} candidate $H^{\sf L\uparrow} \otimes \bB^{\sf L} \otimes L^{\sf L}$ as the {\em ceiling} pseudo instance \\
        given by $B_{\sigma}^{\sf L}(b) \equiv \min\big(B_{\sigma}(b) \big/ B_{\sigma}(\lambda^{*}),\, 1\big)$ for $\sigma \in \{H\} \cup [n] \cup \{L\}$.
        
        \item\label{alg:halve:right}
        Define the \term[\text{\em right}]{halve_right} candidate $H^{\sf R\downarrow} \otimes \bB^{\sf R} \otimes L^{\sf R}$ as the {\em floor} pseudo instance \\
        given by $B_{\sigma}^{\sf R}(b) \equiv B_{\sigma}(b + \lambda^{*}) \cdot \indicator(b \geq 0)$ for $\sigma \in \{H\} \cup [n] \cup \{L\}$.
        
        \item\label{alg:halve:output}
        {\bf Return} the {\PoA}-worse candidate $\argmin \big\{\, \PoA(\text{\em left}),~ \PoA(\text{\em right}) \,\big\}$; \\
        \white{\bf Return} breaking ties in favor of the {\em left} candidate $H^{\sf L \uparrow} \otimes \bB^{\sf L} \otimes L^{\sf L}$.
    \end{enumerate}}
    \end{flushleft}
    \end{mdframed}
    \caption{The {\halve} reduction
    \label{fig:alg:halve}}
\end{figure}

\begin{intuition*}
As \Cref{fig:halve:input} suggests, a \blackref{halve_jump} $(\sigma^{*},\, j^{*}) \in \{L\} \times [m]$ divides the {\em strong ceiling} but {\em non twin ceiling} pseudo instance $H^{\uparrow} \otimes \bB \otimes L \in (\Bstrong \setminus \Btwin)$ into the {\em left-lower} part and the {\em right-upper} part. It is conceivable that we can accordingly construct two {\em sub pseudo instances}, such that at least one of them yields a (weakly) worse {\PoA} bound.
\end{intuition*}

The \blackref{alg:halve} reduction implements the above intuition, and its performance guarantees are summarized in \Cref{lem:halve}.

\begin{lemma}[{\halve}; \Cref{fig:alg:halve}]
\label{lem:halve}
Under reduction $(\tilde{H} \otimes \tilde{\bB} \otimes \tilde{L}) \gets \halve(H^{\uparrow} \otimes \bB \otimes L)$:
% \footnote{Indeed, an {\em ceiling} output is further a {\em twin ceiling} output $(\tilde{H} = \tilde{H}^{\uparrow}) \otimes \tilde{\bB} \otimes \tilde{L} \cong \tilde{H}^{\uparrow} \otimes \tilde{L} \in \Btwin \subsetneq \Bstrong \subsetneq \Bceiling$, as \Cref{fig:halve:left} suggests. But the current (weaker) version of \Cref{lem:halve} is sufficient for our purpose.}
{\begin{enumerate}[font = {\em\bfseries}]
    \item\label{lem:halve:property}
    The output is a floor/ceiling pseudo instance $\tilde{H} \otimes \tilde{\bB} \otimes \tilde{L} \in (\Bfloor \cup \Bceiling)$.
    
    \item\label{lem:halve:potential}
    The potential strictly decreases $\Psi(\tilde{H} \otimes \tilde{\bB} \otimes \tilde{L}) \leq \Psi(H^{\uparrow} \otimes \bB \otimes L) - 1$.
    
    % The potential decreases or keeps the same $\Psi(\tilde{H} \otimes \tilde{\bB} \otimes \tilde{L}) \leq \Psi(H^{\uparrow} \otimes \bB \otimes L)$ in the case of a floor output $(\tilde{H} = \tilde{H}^{\downarrow}) \otimes \tilde{\bB} \otimes \tilde{L} \in \Bfloor$,
    % OR strictly decreases $\Psi(\tilde{H} \otimes \tilde{\bB} \otimes \tilde{L}) \leq \Psi(H^{\uparrow} \otimes \bB \otimes L) - 1$ in the case of a ceiling output $(\tilde{H} = \tilde{H}^{\uparrow}) \otimes \tilde{\bB} \otimes \tilde{L} \in \Bceiling$.
    
    \item\label{lem:halve:poa}
    A (weakly) worse bound is yielded $\PoA(\tilde{H} \otimes \tilde{\bB} \otimes \tilde{L}) \leq \PoA(H^{\uparrow} \otimes \bB \otimes L)$.
\end{enumerate}}
\end{lemma}

\afterpage{
\begin{figure}[t]
    \centering
    \subfloat[{The {\em strong ceiling} but {\em non twin ceiling} input $H^{\uparrow} \otimes \bB \otimes L \in (\Bstrong \setminus \Btwin)$ with a \textbf{pseudo jump} $\sigma^{*} = L$}
    \label{fig:halve:input}]{
    \parbox[c]{.99\textwidth}{
    {\centering
    \includegraphics[width = .49\textwidth]{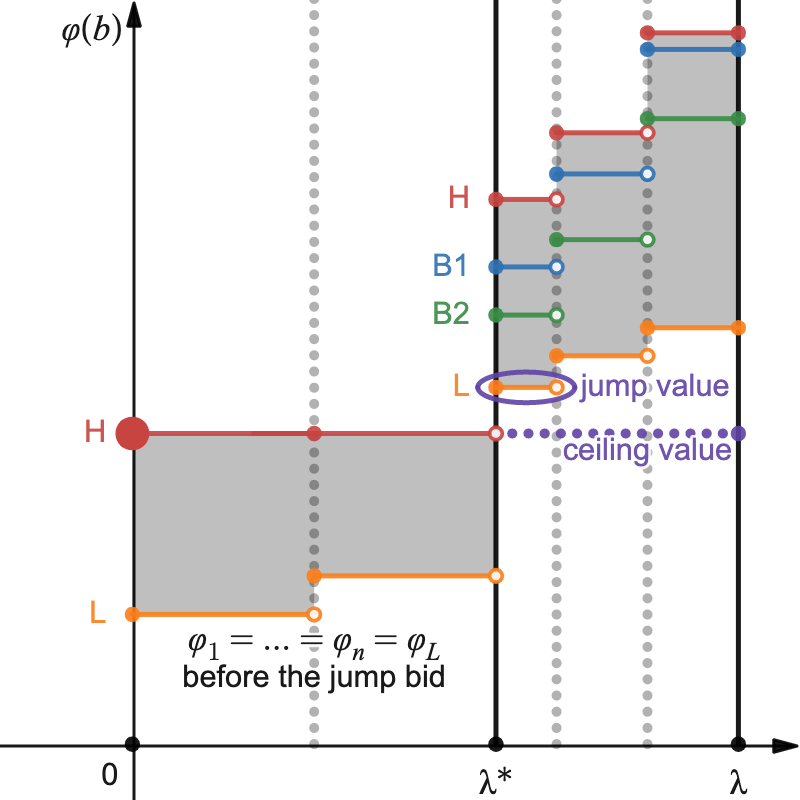}
    \par}}} \\
    \vspace{1cm}
    \subfloat[{The {\em ceiling} left candidate $H^{\sf L\uparrow} \otimes \bB^{\sf L} \otimes L^{\sf L} \in \Bceiling$}
    \label{fig:halve:left}]{
    \includegraphics[width = .49\textwidth]{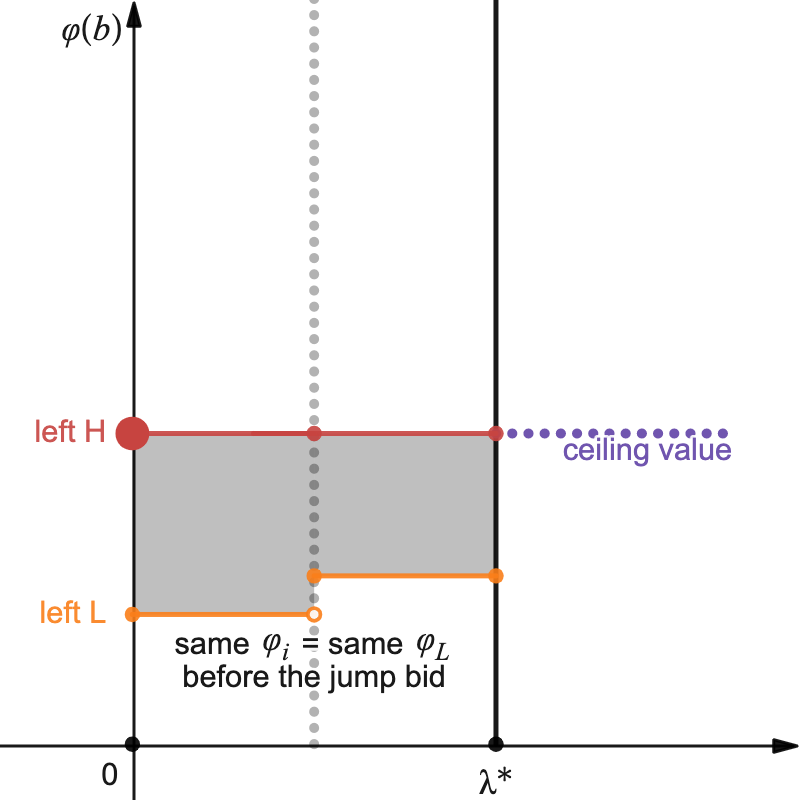}}
    \hfill
    \subfloat[{The {\em floor} right candidate $H^{\sf R\downarrow} \otimes \bB^{\sf R} \otimes L^{\sf R} \in \Bfloor$}
    \label{fig:halve:right}]{
    \includegraphics[width = .49\textwidth]{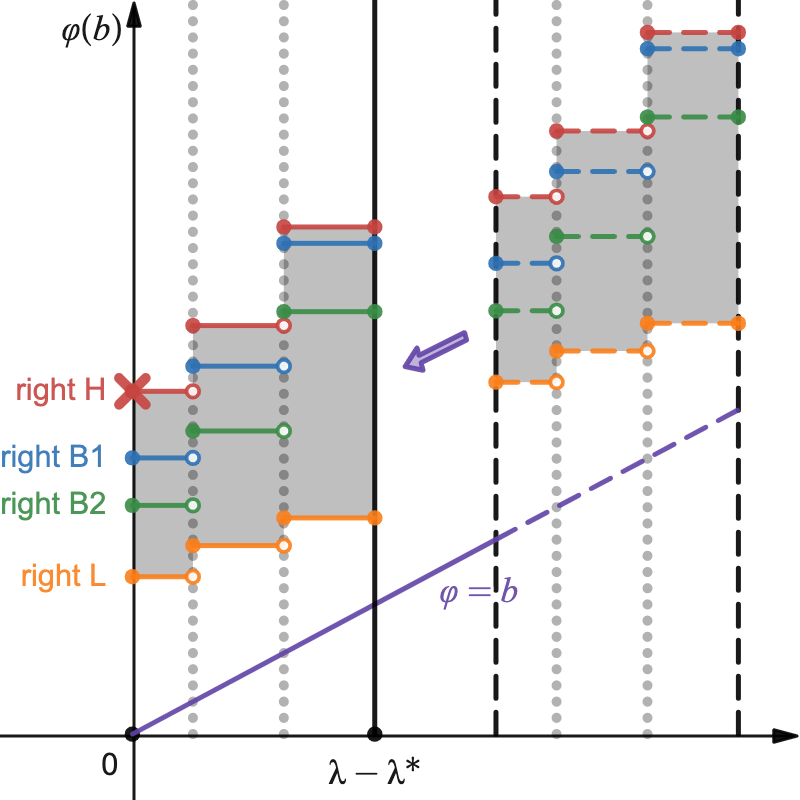}}
    \caption{Diagram of the {\halve} reduction (\Cref{fig:alg:halve}), which transforms a {\em strong ceiling} but {\em non twin ceiling} pseudo instance $H^{\uparrow} \otimes \bB \otimes L \in (\Bstrong \setminus \Btwin)$ that has a \textbf{pseudo jump} $\sigma = L$ (\Cref{def:twin,def:strong,def:jump}) into EITHER the {\em ceiling} left candidate $H^{\sf L\uparrow} \otimes \bB^{\sf L} \otimes L^{\sf L} \in \Bceiling$ OR the {\em floor} right candidate $H^{\sf R\downarrow} \otimes \bB^{\sf R} \otimes L^{\sf R} \in \Bfloor$.
    \label{fig:halve}}
\end{figure}
\clearpage}

\begin{proof}
See \Cref{fig:halve} for a visual aid. The \blackref{alg:halve} reduction chooses the {\PoA}-worse pseudo instance (Line~\ref{alg:halve:output}) between the \blackref{halve_left}/\blackref{halve_right} candidates $H^{\sf L\uparrow} \otimes \bB^{\sf L} \otimes L^{\sf L}$ and $H^{\sf R\downarrow} \otimes \bB^{\sf R} \otimes L^{\sf R}$, breaking ties in favor of the \blackref{halve_left} candidate. For brevity, let $[N] \equiv \{H\} \cup [n] \cup \{L\}$.

\vspace{.1in}
\noindent
{\bf \Cref{lem:halve:property}.}
We shall prove that the \blackref{halve_left} candidate $H^{\sf L\uparrow} \otimes \bB^{\sf L} \otimes L^{\sf L}$ is a {\em ceiling} pseudo instance while the \blackref{halve_right} candidate $H^{\sf R\downarrow} \otimes \bB^{\sf R} \otimes L^{\sf R}$ is a {\em floor} pseudo instance.

The \blackref{halve_left} candidate $H^{\sf L\uparrow} \otimes \bB^{\sf L} \otimes L^{\sf L}$ (Line~\ref{alg:halve:left} and \Cref{fig:halve:left}) shrinks the support to before-jump bids $b \in [0,\, \lambda^{*}]$, taking the form $B_{\sigma}^{\sf L}(b) = B_{\sigma}(b) \big/ B_{\sigma}(\lambda^{*})$ for $\sigma \in [N]$. We can easily check that for $b \in [0,\, \lambda^{*}]$, each \blackref{halve_left} bidder $B_{\sigma}^{\sf L}$ keeps the same mapping $\varphi_{\sigma}^{\sf L}(b) = \varphi_{\sigma}(b)$ as the counterpart \blackref{halve_input} bidder $B_{\sigma}$ (\Cref{def:ceiling_floor:restate}), therefore preserving \blackref{re:discretization}, \blackref{re:monotonicity}, and \blackref{re:layeredness}.
The construction further promises \blackref{re:ceilingness}, making the \blackref{halve_left} candidate $H^{\sf L\uparrow} \otimes \bB^{\sf L} \otimes L^{\sf L} \in \Bceiling$ a {\em ceiling} pseudo instance.

The \blackref{halve_right} candidate $H^{\sf R\downarrow} \otimes \bB^{\sf R} \otimes L^{\sf R}$ (Line~\ref{alg:halve:right} and \Cref{fig:halve:right}) shrinks the support to SHIFTED after-jump bids $b \in [0,\, \lambda - \lambda^{*}]$, taking the form $B_{\sigma}^{\sf R}(b) = B_{\sigma}(b + \lambda^{*})$ for $\sigma \in [N]$.
The exactly same construction has been used in the \blackref{alg:slice} reduction to yield a spectrum of {\em floor} pseudo instances $\big\{\, H^{(t)\downarrow} \otimes \bB^{(t)} \otimes L^{(t)}: 0 \leq t < \lambda \,\big\}$ (see the proof of \Cref{lem:slice:property} of \Cref{lem:slice}); the one given by the jump bid $\lambda^{*} \in (0,\, \lambda)$ is precisely the \blackref{halve_right} candidate $H^{\sf R\downarrow} \otimes \bB^{\sf R} \otimes L^{\sf R} \equiv H^{(\lambda^{*})\downarrow} \otimes \bB^{(\lambda^{*})} \otimes L^{(\lambda^{*})} \in \Bfloor$.

Putting both parts together gives {\bf \Cref{lem:halve:property}}.

\vspace{.1in}
\noindent
{\bf \Cref{lem:halve:potential}.}
The potential of a {\em ceiling} pseudo instance $\hat{H}^{\uparrow} \otimes \hat{\bB} \otimes \hat{L}$ counts all the ultra-ceiling entries $\Psi(\text{\em ceiling}) = \big|\big\{(\sigma,\, j) \in \hat{\bPhi}: \hat{\phi}_{\sigma,\, j} > \hat{\phi}_{H,\, 0}\big\}\big|$ in the bid-to-value table $\hat{\bPhi} = \big[\hat{\phi}_{\sigma,\, j}\big]$ (\Cref{def:potential}) yet the paired {\em floor} pseudo instance $\hat{H}^{\downarrow} \otimes \hat{\bB} \otimes \hat{L}$ counts another one $\Psi(\text{\em floor}) = \Psi(\text{\em ceiling}) + 1$.
% Hence, for the three pseudo instances studied here, it suffices to show that the \blackref{halve_left}/\blackref{halve_right} tables $\bPhi^{\sf L}$ and $\bPhi^{\sf R}$ each have $\geq 2$ FEWER ultra-ceiling entries than the \blackref{halve_input} table $\bPhi$.

The ultra-ceiling entries in the \blackref{halve_input} table $\bPhi$ (\Cref{def:jump,fig:halve:input}) are precisely those in the after-jump columns $\bPhi([N] \times [j^{*}: m])$, because the \blackref{halve_jump} entry $(\sigma^{*},\, j^{*}) \in \{L\} \times [m]$ is the {\em leftmost-then-lowest} ultra-ceiling entry. 

The \blackref{halve_left} table $\bPhi^{\sf L} = \bPhi([N] \times [0: j^{*} - 1])$ (Line~\ref{alg:halve:left} and \Cref{fig:halve:left}) discards all those ultra-ceiling entries in the after-jump columns $j \in [j^{*}: m]$. The ceiling value $\phi_{H,\, 0}$ keeps the same given $j^{*} \neq 0$.
Therefore, the \blackref{halve_left} table $\bPhi^{\sf L}$ has NO ultra-ceiling entry, while the \blackref{halve_input} table $\bPhi$ has at least one ultra-ceiling entry, i.e., the \blackref{halve_jump} entry $(\sigma^{*},\, j^{*})$.
% $|[N] \times [j^{*}: m]| \geq 2$ ultra-ceiling entries ($[N] = \{H\} \cup [n] \cup \{L\}$, where $n \geq 0$).
The first part of {\bf \Cref{lem:halve:potential}} follows.

The \blackref{halve_right} table $\bPhi^{\sf R} = \bPhi([N] \times [j^{*}: m]) - \blambda^{*}$ (Line~\ref{alg:halve:right} and \Cref{fig:halve:right}) discards all the non-ultra-ceiling entries in the before-jump columns $j \in [0; j^{*} - 1]$ and shifts the remaining ones (namely the ultra-ceiling entries in the \blackref{halve_input} table $\bPhi$) each by a distance $-\lambda^{*}$.
The ceiling value changes to the shifted row-$H$ jump-column entry $\phi_{H,\, j^{*}}^{\sf R} \equiv \phi_{H,\, j^{*}} - \lambda^{*}$.
Thus, the set of ultra-ceiling entries shrinks by removing all the \blackref{halve_right} entries $(\sigma,\, j) \in \bPhi^{\sf R}$ for which $\phi_{\sigma,\, j}^{\sf R} \leq \phi_{H,\, j^{*}}^{\sf R}$,
namely all the \blackref{halve_input} entries $\phi_{\sigma,\, j} \leq \phi_{H,\, j^{*}}$ for $(\sigma,\, j) \in [N] \times [j^{*}: m]$.
In particular, all the entries $(\sigma,\, j) \in [N] \times \{j^{*}\}$ in the jump column (\blackref{re:layeredness}) are removed from the set of ultra-ceiling entries; there are $|[N] \times \{j^{*}\}| \geq 2$ such entries ($[N] = \{H\} \cup [n] \cup \{L\}$, where $n \geq 0$).
The second part of {\bf \Cref{lem:halve:potential}} follows.

\vspace{.1in}
\noindent
{\bf \Cref{lem:halve:poa}.}
We claim that either the \blackref{halve_left} candidate
$H^{\sf L\uparrow} \otimes \bB^{\sf L} \otimes L^{\sf L}$ or the \blackref{halve_right} candidate $H^{\sf R\downarrow} \otimes \bB^{\sf R} \otimes L^{\sf R}$ yields a (weakly) worse {\PoA} bound than the \blackref{halve_input} $H^{\uparrow} \otimes \bB \otimes L$.
% , through a coupling argument.

% \vspace{1cm}
% Following \Cref{lem:pseudo_welfare} (given \blackref{re:ceilingness} $P^{\uparrow} \equiv \phi_{H,\, 0}$), the \blackref{collapse_input} auction {\SocialWelfare} is given by
% $\FPA(\blackref{collapse_input}) = \phi_{H,\, 0} \cdot \calB(0) + \sum_{\sigma \in [N]} \big(\int_{0}^{\lambda} \varphi_{\sigma}(b) \cdot \frac{B'_{\sigma}(b)}{B_{\sigma}(b)} \cdot \calB(b) \cdot \d b\big)$.
% Likewise, the \blackref{halve_left} counterpart (given \blackref{re:ceilingness} $P^{{\sf L}\uparrow} \equiv \phi_{H,\, 0}^{\sf L}$) is given by

\vspace{.1in}
\noindent
{\bf Auction {\SocialWelfares}.}
The \blackref{halve_left} candidate $H^{{\sf L}\uparrow} \otimes \bB^{\sf L} \otimes L$ yields (\Cref{lem:pseudo_welfare}; with $P^{{\sf L}\uparrow} \equiv \phi_{H,\, 0}^{\sf L}$) the auction {\SocialWelfare}
\begin{align*}
    \FPA(\blackref{halve_left})\hspace{.34cm}
    & ~=~ \phi_{H,\, 0}^{\sf L} \cdot \calB^{\sf L}(0)
    + \sum_{\sigma \in [N]} \int_{0}^{\lambda^{*}} \varphi_{\sigma}^{\sf L}(b) \cdot \frac{{B_{\sigma}^{\sf L}}'(b)}{B_{\sigma}^{\sf L}(b)} \cdot \calB^{\sf L}(b) \cdot \d b \\
    & ~=~ \Big(\phi_{H,\, 0} \cdot \calB(0)
    + \sum_{\sigma \in [N]} \int_{0}^{\lambda^{*}} \varphi_{\sigma}(b) \cdot \frac{B'_{\sigma}(b)}{B_{\sigma}(b)} \cdot \calB(b) \cdot \d b\Big) \cdot \frac{1}{\calB(\lambda^{*})}.
    && \parbox{2cm}{Line~\ref{alg:halve:left}}
\end{align*}
The \blackref{halve_right} candidate $H^{{\sf R}\downarrow} \otimes \bB^{\sf R} \otimes L^{\sf R}$ yields (\Cref{lem:pseudo_welfare}; with $P^{{\sf R}\downarrow} \equiv 0$) the auction {\SocialWelfare}
\begin{align*}
    \FPA(\blackref{halve_right})\hspace{.08cm}
    & ~=~ \sum_{\sigma \in [N]} \int_{0}^{\lambda - \lambda^{*}} \varphi_{\sigma}^{\sf R}(b) \cdot \frac{{B_{\sigma}^{\sf R}}'(b)}{B_{\sigma}^{\sf R}(b)} \cdot \calB^{\sf R}(b) \cdot \d b \\
    & ~=~ \sum_{\sigma \in [N]} \int_{\lambda^{*}}^{\lambda} \big(\varphi_{\sigma}(b) - \lambda^{*}\big) \cdot \frac{B'_{\sigma}(b)}{B_{\sigma}(b)} \cdot \calB(b) \cdot \d b
    && \text{Line~\ref{alg:halve:right}} \\
    & ~=~ \sum_{\sigma \in [N]} \int_{\lambda^{*}}^{\lambda} \varphi_{\sigma}(b) \cdot \frac{B'_{\sigma}(b)}{B_{\sigma}(b)} \cdot \calB(b) \cdot \d b
    - \int_{\lambda^{*}}^{\lambda} \lambda^{*} \cdot \calB'(b) \cdot \d b
    \hspace{1.cm}
    && \parbox{2cm}{$\calB = \prod_{\sigma} B_{\sigma}$} \\
    & ~=~ \sum_{\sigma \in [N]} \int_{\lambda^{*}}^{\lambda} \varphi_{\sigma}(b) \cdot \frac{B'_{\sigma}(b)}{B_{\sigma}(b)} \cdot \calB(b) \cdot \d b
    - \lambda^{*} \cdot \big(1 - \calB(\lambda^{*})\big).
    && \text{$\calB(\lambda) = 1$}
\end{align*}
Further, the \blackref{halve_input} $H^{\uparrow} \otimes \bB \otimes L$ yields (\Cref{lem:pseudo_welfare}; with $P^{\uparrow} \equiv \phi_{H,\, 0}$) the auction {\SocialWelfare}
\begin{align*}
    \FPA(\blackref{collapse_input})
    & ~=~ \phi_{H,\, 0} \cdot \calB(0) + \sum_{\sigma \in [N]} \Big(\int_{0}^{\lambda} \varphi_{\sigma}(b) \cdot \frac{B'_{\sigma}(b)}{B_{\sigma}(b)} \cdot \calB(b) \cdot \d b\Big) \\
    & ~=~ \FPA(\blackref{halve_left}) \cdot \calB(\lambda^{*}) + \FPA(\blackref{halve_right}) + \lambda^{*} \cdot \big(1 - \calB(\lambda^{*})\big)\hspace{1.55cm}
    && \parbox{2cm}{substitution} \phantom{\Big.}
\end{align*}

\vspace{.1in}
\noindent
{\bf Optimal {\SocialWelfares}.}
Let $\bPhi^{*} = [N] \times [j^{*}: m]$ be the after-jump subtable.
Following {\bf \Cref{lem:halve:potential}} and \Cref{fig:halve}:
(i)~The \blackref{halve_left} table $\bPhi^{\sf L}$ has NO ultra-ceiling entry $\emptyset$.
(ii)~The \blackref{halve_right} table $\bPhi^{\sf R}$'s ultra-ceiling entries $\phi_{\sigma,\, j}^{\sf R} > \phi_{H,\, j^{*}}^{\sf R}$ are all included in the subtable $\bPhi^{*}$.
(iii)~The \blackref{halve_input} table $\bPhi$' ultra-ceiling entries $\phi_{\sigma,\, j} > \phi_{H,\, 0}$ are precisely the subtable $\bPhi^{*}$.
Hence, the \blackref{halve_left} candidate $H^{{\sf L}\uparrow} \otimes \bB^{\sf L} \otimes L$ yields (\Cref{lem:ceiling_welfare}; with $P^{{\sf L}\uparrow} \equiv \phi_{H,\, 0}^{\sf L}$) the optimal {\SocialWelfare}
\begin{align*}
    \OPT(\blackref{halve_left})\hspace{.34cm}
    & ~=~ \int_{0}^{+\infty} \big(1 - \indicator(v \geq \phi_{H,\, 0}^{\sf L})\big) \cdot \d v
    ~=~ \phi_{H,\, 0}^{\sf L}
    ~=~ \phi_{H,\, 0}.\hspace{.99cm}
    \text{Line~\ref{alg:halve:left}}
    \hspace{3.18cm}
\end{align*}
The \blackref{halve_right} candidate $H^{{\sf R}\downarrow} \otimes \bB^{\sf R} \otimes L^{\sf R}$ yields (\Cref{lem:ceiling_welfare}; with $P^{{\sf R}\downarrow} \equiv 0$) the optimal {\SocialWelfare}\footnote{The \blackref{halve_right} candidate (Line~\ref{alg:halve:right}) preserves the probabilities $\omega_{\sigma,\, j}^{\sf R} = 1 - \frac{B_{\sigma}^{\sf R}(\lambda_{j})}{B_{\sigma}^{\sf R}(\lambda_{j + 1})} = 1 - \frac{B_{\sigma}(\lambda_{j})}{B_{\sigma}(\lambda_{j + 1})} = \omega_{\sigma,\, j}$ for $(\sigma,\, j) \in \bPhi^{*}$.}
\begin{align}
    \OPT(\blackref{halve_right})\hspace{.08cm}
    & ~=~ \int_{0}^{+\infty} \Big(1 - \prod_{(\sigma,\, j) \,\in\, \bPhi^{*}} \big(1 - \omega_{\sigma,\, j} \cdot \indicator(v < \phi_{\sigma,\, j}^{\sf R})\big)\Big) \cdot \d v
    \nonumber \\
    & ~=~ \int_{\lambda^{*}}^{+\infty} \Big(1 - \prod_{(\sigma,\, j) \,\in\, \bPhi^{*}} \big(1 - \omega_{\sigma,\, j} \cdot \indicator(v < \phi_{\sigma,\, j})\big)\Big) \cdot \d v
    \qquad \text{Line~\ref{alg:halve:right}}
    \nonumber \\
    & ~=~ \int_{0}^{+\infty} \Big(1 - \prod_{(\sigma,\, j) \,\in\, \bPhi^{*}} \big(1 - \omega_{\sigma,\, j} \cdot \indicator(v < \phi_{\sigma,\, j})\big)\Big) \cdot \d v
    ~-~ \lambda^{*} \cdot \big(1 - \calB(\lambda^{*})\big).
    \hspace{.75cm}
    \label{eq:halve:H1}\tag{H1}
\end{align}
\eqref{eq:halve:H1}: When $v \leq \lambda^{*}$, the integrand $\overset{(\dagger)}{=} 1 - \prod_{\bPhi^{*}} (1 - \omega_{\sigma,\, j})
= 1 - \prod_{\bPhi^{*}} \frac{B_{\sigma}(\lambda_{j})}{B_{\sigma}(\lambda_{j + 1})}
\overset{(\ddagger)}{=} 1 - \prod_{[N]} B_{\sigma}(\lambda^{*})
= 1 - \calB(\lambda^{*})$; \\
($\dagger$) $\phi_{\sigma,\, j} > \phi_{H,\, 0} = \lim_{b \nearrow \lambda^{*}} \varphi_{H}(b) > \lambda^{*}$ given a \blackref{halve_jump} $(\sigma^{*},\, j^{*}) \in \{L\} \times [m]$ and \Cref{lem:pseudo_mapping}; \\
($\ddagger$) $\lambda_{j^{*}} \equiv \lambda^{*}$ and $\calB(\lambda_{m + 1} \equiv \lambda) = 1$.

\vspace{.1in}
\noindent
Further, the \blackref{halve_input} $H^{\uparrow} \otimes \bB \otimes L$ yields (\Cref{lem:ceiling_welfare}; with $P^{\uparrow} \equiv \phi_{H,\, 0}$) the optimal {\SocialWelfare}
\begin{align*}
    \OPT(\blackref{halve_input})
    & ~=~ \int_{0}^{+\infty} \Big(1 - \indicator(v \geq \phi_{H,\, 0}) \phantom{\Big)} \cdot \prod_{(\sigma,\, j) \,\in\, \bPhi^{*}} \big(1 - \omega_{\sigma,\, j} \cdot \indicator(v < \phi_{\sigma,\, j})\big)\Big) \cdot \d v \\
    & ~=~ \int_{0}^{\phi_{H,\, 0}}
    % \Big(1 - \indicator(v \geq \phi_{H,\, 0})\Big) \cdot 
    \prod_{(\sigma,\, j) \,\in\, \bPhi^{*}} \big(1 - \omega_{\sigma,\, j} \cdot \indicator(v < \phi_{\sigma,\, j})\big) \cdot \d v \\
    & \phantom{~=~} + \OPT(\blackref{halve_right}) + \lambda^{*} \cdot \big(1 - \calB(\lambda^{*})\big)
    \qquad\qquad\qquad~~~~~~~~~\, \text{substitute \eqref{eq:halve:H1}}
    \phantom{\Big.} \\
    & ~=~ \phi_{H,\, 0} \cdot \calB(\lambda^{*}) + \OPT(\blackref{halve_right}) + \lambda^{*} \cdot \big(1 - \calB(\lambda^{*})\big)
    \qquad\;~~~~~ \text{reuse arguments for \eqref{eq:halve:H1}}
    \phantom{\Big.} \\
    & ~=~ \OPT(\blackref{halve_input}) \cdot \calB(\lambda^{*}) + \OPT(\blackref{halve_right}) + \lambda^{*} \cdot \big(1 - \calB(\lambda^{*})\big).
    \phantom{\Big.}
\end{align*}
Therefore, either or both of the \blackref{halve_left}/\blackref{halve_right} candidates yields a (weakly) worse {\PoA} bound, as follows. (Recall that a {\PoA} bound is always at most $1$.)
\begin{align*}
    \PoA(\blackref{halve_input})
    % & \min \bigg\{\, \PoA(\blackref{halve_left}) \equiv \frac{\FPA(\blackref{halve_left})}{\OPT(\blackref{halve_left})},~
    % \PoA(\blackref{halve_right}) = \frac{\FPA(\blackref{halve_right})}{\OPT(\blackref{halve_right})} \,\bigg\} \\
    % & ~\leq~ \frac{\FPA^{\sf L} \cdot q + \FPA^{\sf right}}{\OPT^{\sf L} \cdot q + \OPT^{\sf right}}
    \,=\, \frac{\FPA(\blackref{halve_left}) \cdot \calB(\lambda^{*}) + \FPA(\blackref{halve_right}) + \lambda^{*} \cdot (1 - \calB(\lambda^{*}))}
    {\OPT(\blackref{halve_left}) \cdot \calB(\lambda^{*}) + \OPT(\blackref{halve_right}) + \lambda^{*} \cdot (1 - \calB(\lambda^{*}))}
    \,\geq\, \min \big\{\, \PoA(\blackref{halve_left}),~ \PoA(\blackref{halve_right}) \,\big\}.
\end{align*}
% The \blackref{alg:halve} reduction (Line~\ref{alg:halve:output}) chooses the {\PoA}-worse one between the two candidates, breaking ties in favor of the {\em ceiling} ``left'' candidate. This together with the above inequality give
{\bf \Cref{lem:halve:poa}} follows then.
This finishes the proof.
\end{proof}

\subsection{{\AD}: From strong ceiling towards twin ceiling, under real jumps}
\label{subsec:AD}

This subsection also considers {\em strong ceiling} but {\em non twin ceiling} pseudo instances $H^{\uparrow} \otimes \bB \otimes L \in (\Bstrong \setminus \Btwin)$. In particular, we deal with the case of a \term[\textbf{real jump}]{AD_jump} $\sigma^{*} \in \{H\} \cup [n]$ through the {\blackref{alg:AD}} reduction (see \Cref{fig:alg:AD,fig:AD} for its description and a visual aid).
For ease of reference, let us rephrase the observations from \Cref{subsec:tech_prelim}.

\begin{restate}[\Cref{lem:potential,lem:jump,lem:strong}]
For a strong ceiling but non twin ceiling pseudo instance $H^{\uparrow} \otimes \bB \otimes L \in (\Bstrong \setminus \Btwin)$:
\begin{enumerate}[font = {\em\bfseries}]
    \item It violates \blackref{re:twin_ceiling}, i.e., its potential is nonzero $\Psi(H^{\uparrow} \otimes \bB \otimes L) > 0$.
    
    \item The \blackref{jump_entry} $(\sigma^{*},\, j^{*}) \in (\{H\} \cup [n] \cup \{L\}) \times [m]$ exists and cannot be in column $0$.
    
    % \item The jump bidder is either the pseudo bidder $\sigma^{*} = L$ or one of the real bidders $\sigma^{*} \in \{H\} \cup [n]$.
    
    \item The jump bid as the index $j^{*} \in [m]$ partition bid $\lambda^{*} \equiv \lambda_{j^{*}} \in (0,\, \lambda)$ neither can be the nil bid $\lambda_{0} \equiv 0$ nor can be the supremum bid $\lambda_{m + 1} \equiv \lambda$.
\end{enumerate}
\end{restate}

\Cref{lem:AD} summarizes performance guarantees of the \blackref{alg:AD} reduction.
(The intuition of this reduction might be obscure. Roughly speaking, it adopts a {\em win-win} approach between two relatively natural modifications that both {\em decrease the potential}.)

% leverages and generalizes the ideas behind the \blackref{alg:translate} reduction and \blackref{alg:polarize} reduction from \Cref{sec:preprocessing}.)

\begin{lemma}[{\AD}]
\label{lem:AD}
Under reduction $\tilde{H}^{\uparrow} \otimes \tilde{\bB} \otimes \tilde{L} \gets \AD(H^{\uparrow} \otimes \bB \otimes L)$:
\begin{enumerate}[font = {\em\bfseries}]
    \item\label{lem:AD:property}
    The output is a ceiling pseudo instance $\tilde{H}^{\uparrow} \otimes \tilde{\bB} \otimes \tilde{L} \in \Bceiling$.
    
    \item\label{lem:AD:potential}
    The potential strictly decreases $\Psi(\tilde{H}^{\uparrow} \otimes \tilde{\bB} \otimes \tilde{L}) \leq \Psi(H^{\uparrow} \otimes \bB \otimes L) - 1$.
    
    \item\label{lem:AD:poa}
    A (weakly) worse bound is yielded $\PoA(\tilde{H}^{\uparrow} \otimes \tilde{\bB} \otimes \tilde{L}) \leq \PoA(H^{\uparrow} \otimes \bB \otimes L)$.
\end{enumerate}
\end{lemma}

\begin{proof}
See \Cref{fig:AD} for a visual aid. The \blackref{alg:AD} reduction chooses the {\PoA}-worse pseudo instance (Line~\ref{alg:AD:output}) between the \blackref{AD_ascend}/\blackref{AD_descend} candidates $\bar{H}^{\uparrow} \otimes \bar{\bB} \otimes \bar{L}$ and $\underline{H}^{\downarrow} \otimes \underline{\bB} \otimes \underline{L}$, breaking ties in favor of the \blackref{AD_ascend} candidate. For brevity, let $[N] \equiv \{H\} \cup [n] \cup \{L\}$.

\vspace{.1in}
\noindent
{\bf \Cref{lem:AD:property}.}
We shall prove that both candidates $\bar{H}^{\uparrow} \otimes \bar{\bB} \otimes \bar{L}$ and $\underline{H}^{\uparrow} \otimes \underline{\bB} \otimes \underline{L}$ are {\em ceiling} pseudo instances, which are reconstructed (Lines~\ref{alg:AD:ascend} and \ref{alg:AD:descend}) essentially from the \blackref{AD_ascend}/\blackref{AD_descend} tables $\bar{\bPhi} = \big[\bar{\phi}_{\sigma,\, j}\big]$ and $\underline{\bPhi} = \big[\underline{\phi}_{\sigma,\, j}\big]$ according to \Cref{lem:pseudo_distribution}.

\afterpage{
\begin{figure}[t]
    \centering
    \begin{mdframed}
    Reduction $\term[\AD]{alg:AD}(H^{\uparrow} \otimes \bB \otimes L)$
    
    \begin{flushleft}
    {\bf Input:} A {\em strong ceiling} but {\em non twin ceiling} pseudo instance $H^{\uparrow} \otimes \bB \otimes L \in (\Bstrong \setminus \Btwin)$ \\
    \white{\bf Input:} that has a \blackref{AD_jump} $(\sigma^{*},\, j^{*}) \in (\{H\} \cup [n]) \times [m]$.
    \white{\term[\text{\em input}]{AD_input}}
    % \hfill
    % \Cref{def:strong,def:twin,def:jump}
    
    \vspace{.05in}
    {\bf Output:} A {\em ceiling} pseudo instance $\tilde{H}^{\uparrow} \otimes \tilde{\bB} \otimes \tilde{L} \in \Bceiling$.
    \white{\term[\text{\em output}]{AD_output}}
    \begin{enumerate}
    \setcounter{enumi}{-1}
        \item Consider the underlying partition $\bLambda = [0 \equiv \lambda_{0},\, \lambda_{1}) \cup [\lambda_{1},\, \lambda_{2}) \cup \dots \cup [\lambda_{m},\, \lambda_{m + 1} \equiv \lambda]$ \\
        and the {\em input} bid-to-value table $\bPhi = \big[\phi_{\sigma,\, j}\big]$ for $(\sigma,\, j) \in (\{H\} \cup [n] \cup \{L\}) \times [0: m]$.
        
        \item\label{alg:AD:ascend}
        Define the {\em ascended} table $\bar{\bPhi} = \big[\bar{\phi}_{\sigma,\, j}\big]$ by ONLY {\em ascending} the row-$H$ before-jump entries, from the ceiling value $\phi_{H,\, 0} = \dots = \phi_{H,\, j^{*} - 1}$ to the jump value $\phi^{*} \equiv \phi_{\sigma^{*},\, j^{*}}$.
        
        Reconstruct the ({\em ceiling}) \term[\text{\em ascended}]{AD_ascend} candidate $\bar{H}^{\uparrow} \otimes \bar{\bB} \otimes \bar{L} \in \Bceiling$ from the piecewise constant mappings $\bar{\bvarphi} = \{\bar{\varphi}_{\sigma}\}_{\sigma \in \{H\} \cup [n] \cup \{L\}}$ given by $(\bLambda,\, \bar{\bPhi})$ according to \Cref{lem:pseudo_distribution}.
        
        % \item[]
        % \OliveGreen{$\triangleright$ Changed: The monopolist $H^{\uparrow}$: bid distribution $H(b)$ and bid-to-value mapping $\varphi_{H}(b)$ before the jump bid $b \in [0, \phi^{*} \equiv \phi_{\sigma^{*}, j^{*}})$. \\
        % $\triangleright$ Unchanged: The pseudo row $L \neq H$ of table $\bPhi$ keeps the same; so do the pseudo mapping $\varphi_{L}(b)$ and the first-order bid distribution $\calB(b) = \exp\big(-\int_{b}^{\lambda} (\varphi_{L}(b) - b)^{-1} \cdot \d b\big)$ (\Cref{lem:pseudo_distribution}).}
        
        \item\label{alg:AD:descend}
        Define the {\em descended} table $\underline{\bPhi} = \big[\underline{\phi}_{\sigma,\, j}\big]$ by ONLY {\em descending} the jump entry $(\sigma^{*},\, j^{*})$, from the jump value $\phi_{\sigma^{*},\, j^{*}} \equiv \phi^{*}$ to the ceiling value $\phi_{H,\, 0}$.
        
        Reconstruct the ({\em ceiling}) \term[\text{\em descended}]{AD_descend} candidate $\underline{H}^{\uparrow} \otimes \underline{\bB} \otimes \underline{L} \in \Bceiling$ from the piecewise constant mappings $\underline{\bvarphi} = \{\underline{\varphi}_{\sigma}\}_{\sigma \in \{H\} \cup [n] \cup \{L\}}$ given by $(\bLambda,\, \underline{\bPhi})$ according to \Cref{lem:pseudo_distribution}.
        
        % \item[]
        % \OliveGreen{$\triangleright$ The pseudo row $L \neq \sigma^{*}$ of table $\bPhi$ keeps the same; so do the pseudo mapping $\varphi_{L}(b)$ and the first-order bid distribution $\calB(b) = \exp\big(-\int_{b}^{\lambda} (\varphi_{L}(b) - b)^{-1} \cdot \d b\big)$ (\Cref{lem:pseudo_distribution}).}
        
        \item\label{alg:AD:output}
        {\bf Return} the {\PoA}-worse candidate $\argmin \big\{\, \PoA(\text{\em ascended}),~ \PoA(\text{\em descended}) \,\big\}$; \\
        \white{\bf Return} breaking ties in favor of the {\em ascended} candidate $\bar{H}^{\uparrow} \otimes \bar{\bB} \otimes \bar{L}$.
    \end{enumerate}
%     \begin{enumerate}
%         \item\label{alg:ascend:reduction}
%         \red{Let $\bar{\bPhi} = \big[\bar{\phi}_{\sigma,\, j}\big]$ be modified entry-wise from $\bPhi = \big[\phi_{\sigma,\, j}\big]$, such that \\
%         (i)~$\bar{\phi}_{1,\, j} = \phi^{*}$ for each $j \in [0:\, j^{*} - 1]$; and (ii)~all other entries keep the same.}
%         
%         \item \red{Let $\underline{\bPhi} = \big[\, \underline{\phi}_{\sigma,\, j} \,\big]$ be modified entry-wise from $\bPhi = \big[\phi_{\sigma,\, j}\big]$, such that \\
%         (i)~$P^{\uparrow}$, note that $P^{\uparrow} = \varphi_{1}(0) = \phi_{H,\, 0}$; and (ii)~all other entries keep the same.}
%     \end{enumerate}
    \end{flushleft}
    \end{mdframed}
    {\small\begin{tabular}{|c|c|>{\centering\arraybackslash}p{4.55cm}|>{\centering\arraybackslash}p{5.35cm}|}
        \hline
        \multicolumn{2}{|c|}{\rule{0pt}{11pt}} & bid-to-value mapping $\varphi_{\sigma}(b)$ & bid distribution $B_{\sigma}(b)$ \\ [1pt]
        \hline
        \hline
        \rule{0pt}{11pt} & monopolist $\sigma = H^{\uparrow}$ & ascended on $b \in [0,\, \lambda_{j^{*}})$ & modified on $b \in [0,\, \lambda_{j^{*}})$ \\ [1pt]
        \cline{2-4}
        \rule{0pt}{11pt}{\em ascended} & pseudo bidder $\sigma = L$ & -- & ``adapted'' to preserve $H^{\uparrow}(b) \cdot L(b)$ \\ [1pt]
        \cline{2-4}
        \rule{0pt}{11pt} & $\sigma \notin \{H^{\uparrow}, L\}$ & -- & -- \\ [1pt]
        \hline
        \hline
        \rule{0pt}{11pt} & jump bidder $\sigma = j^{*}$ & descended on $b \in [\lambda_{j^{*}},\, \lambda_{j^{*} + 1})$ & modified on $b \in [0,\, \lambda_{j^{*} + 1})$ \\ [1pt]
        \cline{2-4}
        \rule{0pt}{11pt}{\em descended} & pseudo bidder $\sigma = L$ & -- & ``adapted'' to preserve $B_{j^{*}}(b) \cdot L(b)$ \\ [1pt]
        \cline{2-4}
        \rule{0pt}{11pt} & $\sigma \notin \{j^{*}, L\}$ & -- & -- \\ [1pt]
        \hline
    \end{tabular}}
    \caption{\label{fig:alg:AD}
    The {\AD} reduction, where the table summarizes whether the bid-to-value mappings $\bvarphi$ and the bid distributions $\bB$ change or not (--). \\
    Both candidates {\em ascended} and {\em descended} preserve the pseudo row $L \notin \{H^{\uparrow},\, j^{*}\}$ of the table $\bPhi$. Likewise, the pseudo mapping $\varphi_{b}(b)$ and the first-order bid distribution $\calB(b) = \exp\big(-\int_{b}^{\lambda} (\varphi_{L}(b) - b)^{-1} \cdot \d b\big)$ (\Cref{lem:pseudo_distribution}) are invariant over $b \in [\gamma,\, \lambda]$.}
\end{figure}
\clearpage}

The {\em table-based} construction (Lines~\ref{alg:AD:ascend} and \ref{alg:AD:descend}) intrinsically ensures \blackref{re:discretization} and \blackref{re:ceilingness}; we just need to check
\blackref{re:monotonicity},
\blackref{re:layeredness}, and
{\bf the \Cref{lem:pseudo_mapping} conditions}.
The \blackref{AD_input} $H^{\uparrow} \otimes \bB \otimes L \in (\Bstrong \setminus \Btwin)$ as a {\em strong ceiling} but {\em non twin ceiling} not only satisfies these, but also  (\Cref{lem:jump,lem:strong}) has a \blackref{AD_jump} and satisfies \blackref{re:collapse}.
\begin{flushleft}
\begin{itemize}
    \item {\bf \blackref{re:monotonicity} and \blackref{re:layeredness}:}
    The \blackref{AD_input} table $\bPhi = [\phi_{\sigma,\, j}]$ for $(\sigma,\, j) \in [N] \times [0: m]$
    is increasing $\phi_{\sigma,\, 0} \leq \dots \leq \phi_{\sigma,\, j} \leq \dots \leq \phi_{\sigma,\, m}$ in each row $\sigma \in [N] = \{H\} \cup [n] \cup \{L\}$ and
    is decreasing $\phi_{H,\, j} \geq \phi_{1,\, j} \geq \dots \geq \phi_{n,\, j} \geq \phi_{L,\, j}$ in each column $j \in [0: m]$.
    
    \item {\bf the \Cref{lem:pseudo_mapping} conditions:}\footnote{Only {\bf \Cref{lem:pseudo_mapping:inequality}} is presented. {\bf \Cref{lem:pseudo_mapping:dominance}} that $\min(\bvarphi(b^{\otimes n + 1})) = \varphi_{L}(b) > b$ is ignored; the equality restates \blackref{re:layeredness} and the inequality is vacuous because the pseudo mapping $\varphi_{L}(b)$ is invariant (Lines~\ref{alg:AD:ascend} and Line~\ref{alg:AD:descend}; $L \notin \{H,\, \sigma^{*}\}$).}
    $\sum_{i \in \{H\} \cup [n]} (\phi_{i,\, j} - b)^{-1} \geq n \cdot (\phi_{L,\, j} - b)^{-1}$ over the index-$j$ piece
    $b \in [\lambda_{j},\, \lambda_{j + 1})$, for each column $j \in [0: m]$.
    
    \item The \blackref{AD_jump} entry $(\sigma^{*},\, j^{*}) \in (\{H\} \cup [n]) \times [m]$ as the {\em leftmost-then-lowest}
    ultra-ceiling entry $\phi_{\sigma^{*},\, j^{*}} > \phi_{H,\, 0}$, cannot be in the row $L$ and cannot be in the column $0$.
    In each before-jump column $j \in [0: j^{*} - 1]$, all the non-monopoly entries (\blackref{re:collapse}) collapse to the pseudo entry $\phi_{1,\, j} = \dots = \phi_{n,\, j} = \phi_{L,\, j}$.
\end{itemize}
\end{flushleft}

The \blackref{AD_ascend} table $\bar{\bPhi}$ (Line~\ref{alg:AD:ascend} and \Cref{fig:AD:ascend}) ascends the row-$H$ before-jump entries, from the ceiling value $\phi_{H,\, 0} = \dots = \phi_{H,\, j^{*} - 1}$ to the {\em higher} jump value $\phi_{\sigma^{*},\, j^{*}}$.
The only modified row $H$ keeps increasing
$[\bar{\phi}_{H,\, 0} = \dots = \bar{\phi}_{H,\, j^{*} - 1} = \phi_{\sigma^{*},\, j^{*}}]
\leq [\phi_{H,\, j^{*}}]
\leq \dots \leq
[\phi_{H,\, m}]$
(\blackref{re:monotonicity}),
where the only nontrivial inequality $\phi_{\sigma^{*},\, j^{*}} \leq \phi_{H,\, j^{*}}$ stems from \blackref{re:layeredness} of the \blackref{AD_input} table $\bPhi$.
Each modified/before-jump column $j \in [0: j^{*} - 1]$ keeps decreasing
$[\bar{\phi}_{H,\, j} = \phi_{\sigma^{*},\, j^{*}} > \phi_{H,\, j}]
\geq [\phi_{1,\, j}]
= \dots = [\phi_{n,\, j}]
= [\phi_{L,\, j}]$ (\blackref{re:layeredness})
and satisfies that \\
$(\bar{\phi}_{H,\, j} - b)^{-1} + \sum_{i \neq H} (\phi_{i,\, j} - b)^{-1}
~=~ (\phi_{\sigma^{*},\, j^{*}} - b)^{-1} + n \cdot (\phi_{L,\, j} - b)^{-1}
~\geq~ n \cdot (\phi_{L,\, j} - b)^{-1}$
({\bf \Cref{lem:pseudo_mapping}}).

\afterpage{
\begin{figure}[t]
    \centering
    \subfloat[{The {\em strong ceiling} but {\em non twin ceiling} input $H^{\uparrow} \otimes \bB \otimes L \in (\Bstrong \setminus \Btwin)$ with a \textbf{real jump} $\sigma^{*} \in \{H\} \cup [n]$}
    \label{fig:AD:input}]{
    \parbox[c]{.99\textwidth}{
    {\centering
    \includegraphics[width = .49\textwidth]{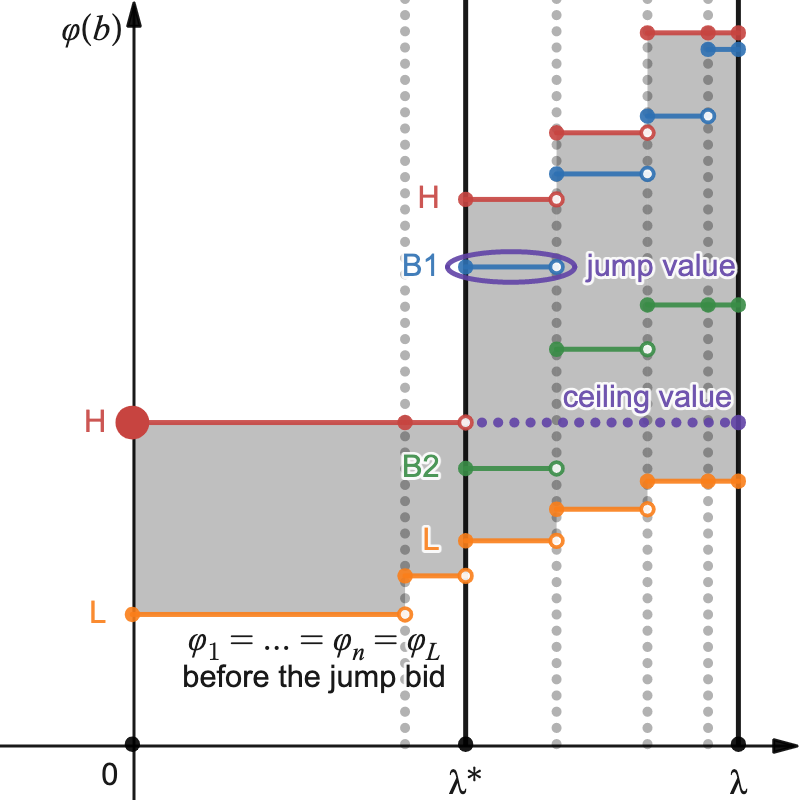}
    \par}}} \\
    \vspace{1cm}
    \subfloat[{The {\em ceiling} ascended candidate $\bar{H}^{\uparrow} \otimes \bar{\bB} \otimes \bar{L} \in \Bceiling$}
    \label{fig:AD:ascend}]{
    \includegraphics[width = .49\textwidth]{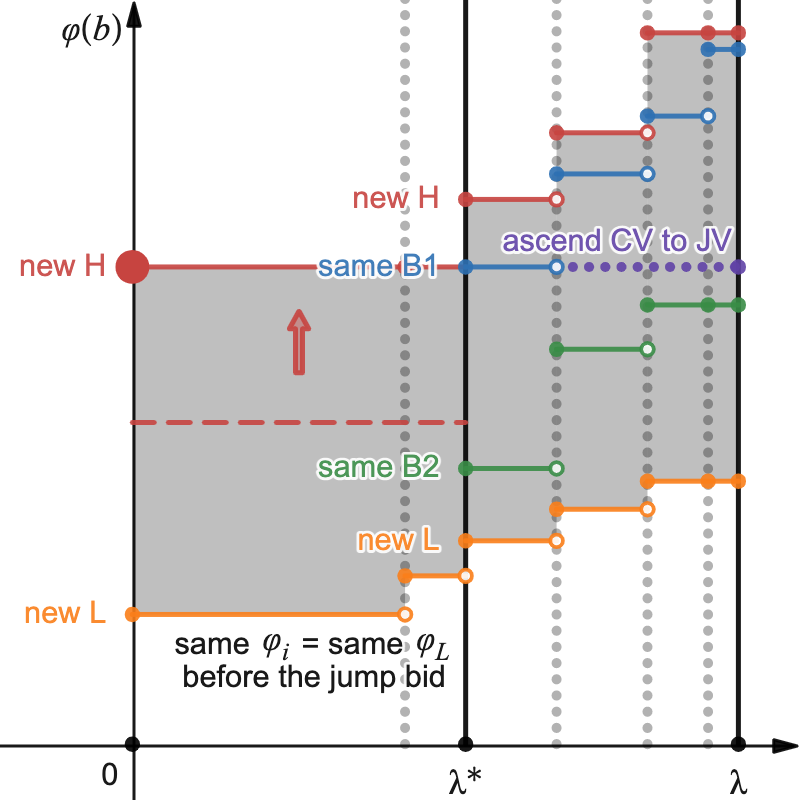}}
    \hfill
    \subfloat[{The {\em ceiling} descended candidate $\underline{H}^{\uparrow} \otimes \underline{\bB} \otimes \underline{L} \in \Bceiling$}
    \label{fig:AD:descend}]{
    \includegraphics[width = .49\textwidth]{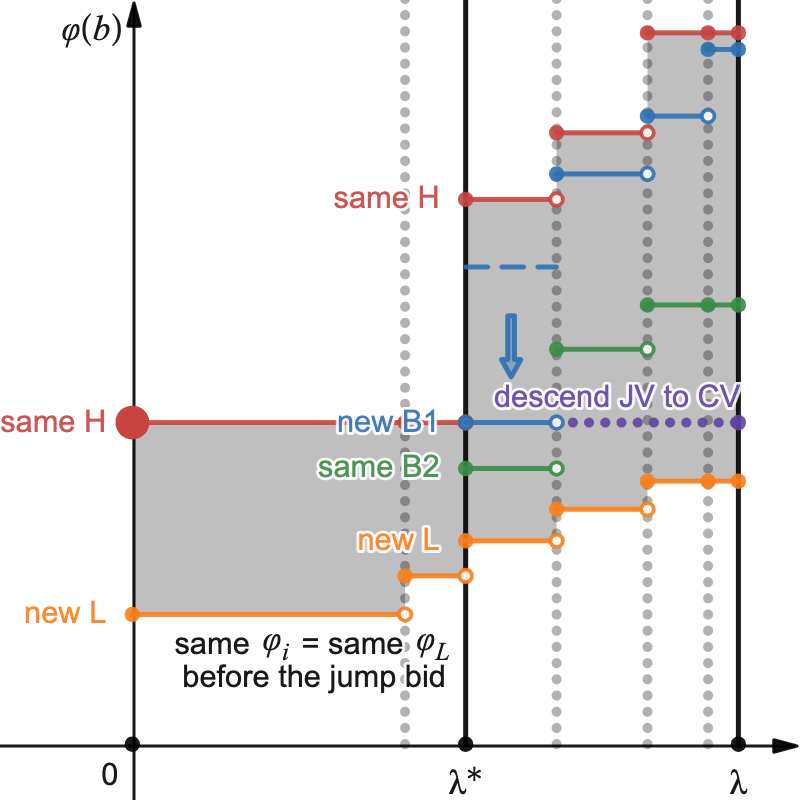}}
    \caption{Diagram of the {\AD} reduction (\Cref{fig:alg:AD}), which transforms (\Cref{def:twin,def:strong}) a {\em strong ceiling} but {\em non twin ceiling} pseudo instance $H^{\uparrow} \otimes \bB \otimes L \in (\Bstrong \setminus \Btwin)$ that has a \textbf{real jump} $\sigma \in \{H\} \cup [n]$ to EITHER the {\em ceiling} ascended candidate $\bar{H}^{\uparrow} \otimes \bar{\bB} \otimes \bar{L} \in \Bceiling$ OR the {\em ceiling} descended candidate $\underline{H}^{\uparrow} \otimes \underline{\bB} \otimes \underline{L} \in \Bceiling$. \\
    Shorthand: ceiling value (CV) and jump value (JV). \\
    \Cref{fig:AD:ascend}: invariant $\bar{H}(b) \cdot \bar{L}(b) \equiv H(b) \cdot L(b)$ and $\bar{\varphi}_{L}(b) \equiv \varphi_{L}(b)$, but modified $\bar{H}(b)$ and $\bar{L}(b)$. \\
    \Cref{fig:AD:descend}: invariant $\underline{H}(b) \cdot \underline{L}(b) \equiv H(b) \cdot L(b)$ and $\underline{\varphi}_{L}(b) \equiv \varphi_{L}(b)$, but modified $\underline{H}(b)$ and $\underline{L}(b)$.
    \label{fig:AD}}
\end{figure}
\clearpage}

\vspace{.1in}
The \blackref{AD_descend} table $\underline{\bPhi}$ (Line~\ref{alg:AD:descend} and \Cref{fig:AD:descend}) descends the \blackref{AD_jump} entry $(\sigma^{*},\, j^{*})$, from the jump value $\phi_{\sigma^{*},\, j^{*}}$ to the {\em lower} ceiling value $\phi_{H,\, 0}$.
This is the {\em leftmost-then-lowest} ultra-ceiling entry $(\sigma^{*},\, j^{*}) \in (\{H\} \cup [n]) \times [m]$ and cannot be in row $L$ and cannot be in column $0$; thus the {\em lefter}/{\em lower} entries
% $(\sigma^{*},\, j^{*} - 1) \in (\{H\} \cup [n]) \times [0: m - 1]$ and $(\sigma^{*} + 1,\, j^{*}) \in ([n] \cup \{L\}) \times [m]$
both are non-ultra-ceiling $[\phi_{\sigma^{*},\, j^{*} - 1}],\, [\phi_{\sigma^{*} + 1,\, j^{*}}] \leq \phi_{H,\, 0}$.
% , whereas (ii)~both of its {\em right-adjacent}/{\em up-adjacent} entries $(\sigma^{*},\, j^{*} + 1)$ and $(\sigma^{*} - 1,\, j^{*})$ (if existential) must be ultra-ceiling  $\phi_{\sigma^{*},\, j^{*} + 1},\, \phi_{\sigma^{*} - 1,\, j^{*}} > \phi_{H,\, 0}$.
So, the only modified/jump row $\sigma^{*}$ keeps increasing
$[\phi_{\sigma^{*},\, 0}]
\leq \dots
\leq [\phi_{\sigma^{*},\, j^{*} - 1}]
\leq [\underline{\phi}_{\sigma^{*},\, j^{*}} = \phi_{H,\, 0} < \phi_{\sigma^{*},\, j^{*}}]
\leq [\phi_{\sigma^{*},\, j^{*} + 1}]
\leq \dots \leq
[\phi_{\sigma^{*},\, m}]$
(\blackref{re:monotonicity}).
% $[\underline{\phi}_{\sigma^{*},\, 0} \equiv \phi_{\sigma^{*},\, 0}]
% \leq \dots
% \leq [\underline{\phi}_{\sigma^{*},\, j^{*} - 1} \equiv \phi_{\sigma^{*},\, j^{*} - 1}]
% \leq [\underline{\phi}_{\sigma^{*},\, j^{*}} = \phi_{H,\, 0} < \phi_{\sigma^{*},\, j^{*}}]
% \leq [\underline{\phi}_{\sigma^{*},\, j^{*} + 1} \equiv \phi_{\sigma^{*},\, j^{*} + 1}]
% \leq \dots \leq
% [\underline{\phi}_{\sigma^{*},\, m} \equiv \phi_{\sigma^{*},\, m}]$
% (\blackref{re:monotonicity}).
The only modified/jump column $j^{*}$ keeps decreasing
$[\phi_{H,\, j^{*}}]
\geq \dots
\geq [\phi_{\sigma^{*} - 1,\, j^{*}}] \geq [\phi_{\sigma^{*},\, j^{*}} > \underline{\phi}_{\sigma^{*},\, j^{*}} = \phi_{H,\, 0}]
\geq [\phi_{\sigma^{*} + 1,\, j^{*}}]
\geq \dots
\geq [\phi_{L,\, j^{*}}]$ (\blackref{re:layeredness})
% $[\phi_{H,\, j^{*}} \equiv \underline{\phi}_{H,\, j^{*}}]
% \geq \dots \geq [\phi_{\sigma^{*} - 1,\, j^{*}} \equiv \underline{\phi}_{\sigma^{*} - 1,\, j^{*}}] \geq [\phi_{\sigma^{*},\, j^{*}} > \underline{\phi}_{\sigma^{*},\, j^{*}} = \phi_{H,\, 0}] \geq [\phi_{\sigma^{*} + 1,\, j^{*}} \equiv \underline{\phi}_{\sigma^{*} + 1,\, j^{*}}] \geq \dots \geq [\phi_{L,\, j^{*}} \equiv \underline{\phi}_{L,\, j^{*}}]$.
and satisfies that \\
$(\underline{\phi}_{\sigma^{*},\, j^{*}} - b)^{-1} + \sum_{i \neq \sigma^{*}} (\phi_{i,\, j^{*}} - b)^{-1}
~\geq~ \sum_{i} (\phi_{i,\, j^{*}} - b)^{-1}
~\geq~ n \cdot (\phi_{L,\, j^{*}} - b)^{-1}$
({\bf \Cref{lem:pseudo_mapping}}).

\vspace{.1in}
To conclude, both candidates are {\em ceiling} pseudo instance. {\bf \Cref{lem:AD:property}} follows then.

\vspace{.1in}
\noindent
{\bf \Cref{lem:AD:potential}.}
The potential of a {\em ceiling} pseudo instance $\hat{H}^{\uparrow} \otimes \hat{\bB} \otimes \hat{L}$ counts all the ultra-ceiling entries $\Psi(\text{\em ceiling}) = \big|\big\{(\sigma,\, j) \in \hat{\bPhi}: \hat{\phi}_{\sigma,\, j} > \hat{\phi}_{H,\, 0}\big\}\big|$ in the bid-to-value table $\hat{\bPhi} = \big[\hat{\phi}_{\sigma,\, j}\big]$ (\Cref{def:potential}).
We do consider three {\em ceiling} pseudo instances ({\bf \Cref{lem:AD:property}}) and need to show that the \blackref{AD_ascend}/\blackref{AD_descend} tables $\bar{\bPhi}$ and $\underline{\bPhi}$ EACH have strictly fewer ultra-ceiling entries than the \blackref{AD_input} table $\bPhi$.

The \blackref{AD_ascend} table $\bar{\bPhi}$ (Line~\ref{alg:AD:ascend} and \Cref{fig:AD:ascend}) ascends the row-$H$ before-jump entries, from the ceiling value $\phi_{H,\, 0} = \dots = \phi_{H,\, j^{*} - 1}$ to the {\em higher} jump value $\phi_{\sigma^{*},\, j^{*}}$.
% ; these entries $\{H\} \times [0: j^{*} - 1]$ are non-ultra-ceiling in both of the \blackref{AD_input}/\blackref{AD_ascend} tables $\bPhi$ and $\bar{\bPhi}$.
Therefore, the \blackref{AD_ascend} table $\bar{\bPhi}$ shrinks the set of ultra-ceiling entries by removing all the unmodified entries $\phi_{\sigma,\, j} \in (\phi_{H,\, 0},\, \phi_{\sigma^{*},\, j^{*}}]$ between the \blackref{AD_input} ceiling value $\phi_{H,\, 0}$ (excluded) and the \blackref{AD_input} jump value $\phi_{\sigma^{*},\, j^{*}}$ (included); specifically, the \blackref{AD_jump} entry $(\sigma^{*},\, j^{*})$ itself gets removed.
This precisely means the \blackref{AD_ascend} table $\bar{\bPhi}$ has strictly fewer ultra-ceiling entries.

The \blackref{AD_descend} table $\underline{\bPhi}$ (Line~\ref{alg:AD:descend} and \Cref{fig:AD:descend}) descends the \blackref{AD_jump} entry $(\sigma^{*},\, j^{*})$, from the jump value $\phi_{\sigma^{*},\, j^{*}}$ to the {\em lower} ceiling value $\phi_{H,\, 0}$.
All the other entries are unmodified, including the ceiling value $\phi_{H,\, 0}$ (given that $j^{*} \neq 0$).
So, the \blackref{AD_descend} table $\underline{\bPhi}$ shrinks the set of ultra-ceiling entries by just removing the \blackref{AD_jump} entry $(\sigma^{*},\, j^{*})$ itself.
Namely, the \blackref{AD_descend} table $\underline{\bPhi}$ has strictly fewer ultra-ceiling entries.

To conclude, both candidates EACH strictly decrease the potential. {\bf \Cref{lem:AD:potential}} follows then.

\vspace{.1in}
\noindent
{\bf \Cref{lem:AD:poa}.}
We investigate the auction/optimal {\SocialWelfares} separately.

\vspace{.1in}
\noindent
{\bf Auction {\SocialWelfares}.}
Using \Cref{lem:translate_welfare} (with \blackref{re:discretization} and \blackref{re:ceilingness} $P^{\uparrow} \equiv \phi_{H,\, 0}$), the \blackref{AD_input} auction {\SocialWelfare} $\FPA(\blackref{AD_input}) \equiv \FPA(H^{\uparrow} \otimes \bB \otimes L)$ is given by
\begin{align*}
    \FPA(\blackref{AD_input})
    & ~=~ \phi_{H,\, 0} \cdot \calB(0) ~+~ \sum_{i \in \{H\} \cup [n]} \sum_{j \in [0: m]} \int_{\lambda_{j}}^{\lambda_{j + 1}} \Big(\frac{\phi_{i,\, j}}{\varphi_{L}(b) - b} - \frac{\phi_{i,\, j} - \varphi_{L}(b)}{\phi_{i,\, j} - b}\Big) \cdot \calB(b) \cdot \d b \\
    & \phantom{~=~ \phi_{H,\, 0} \cdot \calB(0)} ~-~ \int_{0}^{\lambda} (n - 1) \cdot \frac{\varphi_{L}(b)}{\varphi_{L}(b) - b} \cdot \calB(b) \cdot \d b.
\end{align*}

The \blackref{AD_ascend}/\blackref{AD_descend} candidates modify either the row $H$ or the row $\sigma^{*}$ of table $\bPhi$ (Lines~\ref{alg:AD:ascend} and \ref{alg:AD:descend}), while the pseudo row $L \notin \{H, \sigma^{*}\}$ and the pseudo mapping $\varphi_{L}(b)$ keep the same. Likewise, the first-order bid distribution $\calB(b) = \exp\big(-\int_{b}^{\lambda} (\varphi_{L}(b) - b)^{-1} \cdot \d b\big)$ (\Cref{lem:pseudo_distribution}), which is determined by the pseudo mapping, also keeps the same.

The \blackref{AD_ascend} table $\bar{\bPhi}$ (Line~\ref{alg:AD:ascend} and \Cref{fig:AD:ascend}) ascends the row-$H$ before-jump entries, from the ceiling value $\phi_{H,\, 0} = \dots = \phi_{H,\, j^{*} - 1}$ to the {\em higher} jump value $\phi_{\sigma^{*},\, j^{*}}$.
Considering the modified terms, we can formulate the \blackref{AD_ascend} counterpart as $\FPA(\blackref{AD_ascend}) = \FPA(\blackref{AD_input}) + \bar{\Delta}_{\FPA}$, using
\begin{align*}
    \bar{\Delta}_{\FPA}
    & ~=~ (\phi_{\sigma^{*},\, j^{*}} - \phi_{H,\, 0}) \cdot \calB(0)
    ~+~ \sum_{j = 0}^{j^{*} - 1} \int_{\lambda_{j}}^{\lambda_{j + 1}} \Big(\Big(\frac{\phi_{\sigma^{*},\, j^{*}}}{\varphi_{L}(b) - b} - \frac{\phi_{\sigma^{*},\, j^{*}} - \varphi_{L}(b)}{\phi_{\sigma^{*},\, j^{*}} - b}\Big) \\
    & \phantom{~=~ (\phi_{\sigma^{*},\, j^{*}} - \phi_{H,\, 0}) \cdot \calB(0) ~~ \sum_{j = 0}^{j^{*} - 1} \int_{\lambda_{j}}^{\lambda_{j + 1}}\Big(}
    - \Big(\frac{\phi_{H,\, 0}}{\varphi_{L}(b) - b} - \frac{\phi_{H,\, 0} - \varphi_{L}(b)}{\phi_{H,\, 0} - b}\Big)\Big) \cdot \calB(b) \cdot \d b.
    \hspace{1.1cm}
\end{align*}

The \blackref{AD_descend} table $\underline{\bPhi}$ (Line~\ref{alg:AD:descend} and \Cref{fig:AD:descend}) descends the \blackref{AD_jump} entry $(\sigma^{*},\, j^{*})$, from the jump value $\phi_{\sigma^{*},\, j^{*}}$ to the {\em lower} ceiling value $\phi_{H,\, 0}$.
Considering the only modified term, we can formulate the \blackref{AD_descend} counterpart as $\FPA(\blackref{AD_descend}) = \FPA(\blackref{AD_input}) - \underline{\Delta}_{\FPA}$, using
\begin{align*}
    \underline{\Delta}_{\FPA}
    & = \int_{\lambda_{j^{*}}}^{\lambda_{j^{*} + 1}} \Big(\Big(\frac{\phi_{\sigma^{*},\, j^{*}}}{\varphi_{L}(b) - b} - \frac{\phi_{\sigma^{*},\, j^{*}} - \varphi_{L}(b)}{\phi_{\sigma^{*},\, j^{*}} - b}\Big) - \Big(\frac{\phi_{H,\, 0}}{\varphi_{L}(b) - b} - \frac{\phi_{H,\, 0} - \varphi_{L}(b)}{\phi_{H,\, 0} - b}\Big)\Big) \cdot \calB(b) \cdot \d b.
    \hspace{.75cm}
\end{align*}

\noindent
{\bf Optimal {\SocialWelfares}.}
Consider the \blackref{AD_input} subtable $\bPhi^{*} \eqdef \bPhi \setminus (\{H\} \times [0: j^{*} - 1] \cup \{(\sigma^{*},\, j^{*})\})$ that is invariant in both $\bar{\bPhi}$ and $\underline{\bPhi}$. Following {\bf \Cref{lem:AD:potential}} and \Cref{fig:AD}, (i)~the \blackref{AD_input} table $\bPhi$ has the ceiling value $\phi_{H,\, 0}$ and its ultra-ceiling entries are all included in $\bPhi^{*} \cup \{(\sigma^{*},\, j^{*})\}$; (ii)~the \blackref{AD_ascend} table $\bar{\bPhi}$ has the higher ceiling value $= \phi_{\sigma^{*},\, j^{*}}$ and its ultra-ceiling entries are all included in $\bPhi^{*}$; and (iii)~the \blackref{AD_descend} table $\underline{\bPhi}$ has the {\em same} ceiling value $= \phi_{H,\, 0}$ and its ultra-ceiling entries are all included in $\bPhi^{*}$.
Therefore, using the invariant function $\calI(v) \eqdef \prod_{(\sigma,\, j) \in \bPhi^{*}} \big(1 - \omega_{\sigma,\, j} \cdot \indicator(v < \phi_{\sigma,\, j})\big)$,\footnote{Precisely, each probability $\omega_{\sigma,\, j} = 1 - \frac{B_{\sigma}(\lambda_{j})}{B_{\sigma}(\lambda_{j + 1})}$ for $(\sigma,\, j) \in \bPhi^{*}$ is also invariant $\omega_{\sigma,\, j} = \bar{\omega}_{\sigma,\, j} = \underline{\omega}_{\sigma,\, j}$ (cf.\ \Cref{lem:pseudo_distribution} and Lines~\ref{alg:AD:ascend} and \ref{alg:AD:descend}); we omit the detailed verification for brevity.}
we deduce from \Cref{lem:ceiling_welfare} that
\begin{align*}
    \OPT(\blackref{AD_input})~~~~~~
    & ~=~ \int_{0}^{+\infty} \Big(1 - \indicator(v \geq \phi_{H,\, 0}) \cdot \big(1 - \omega_{\sigma^{*},\, j^{*}} \cdot \indicator(v < \phi_{\sigma^{*},\, j^{*}})\big) \cdot \calI(v)\Big) \cdot \d v,
    \hspace{1.75cm} \\
    \OPT(\blackref{AD_ascend})\;\,
    & ~=~ \int_{0}^{+\infty} \Big(1 - \indicator(v \geq \phi_{\sigma^{*},\, j^{*}}) \cdot \calI(v)\Big) \cdot \d v, \\
    \OPT(\blackref{AD_descend})
    & ~=~ \int_{0}^{+\infty} \Big(1 - \indicator(v \geq \phi_{H,\, 0}) \cdot \calI(v)\Big) \cdot \d v.
\end{align*}
We have $\OPT(\blackref{AD_ascend}) \geq \OPT(\blackref{AD_input}) \geq \OPT(\blackref{AD_descend})$ given that $\phi_{H,\, 0} < \phi_{\sigma^{*},\, j^{*}}$.
As before, let $\bar{\Delta}_{\OPT} \equiv \OPT(\blackref{AD_ascend}) - \OPT(\blackref{AD_input})$ and $\underline{\Delta}_{\OPT} \equiv \OPT(\blackref{AD_input}) - \OPT(\blackref{AD_descend})$ be the absolute changes.
The remaining proof relies on {\bf \Cref{fact:AD:poa}}.
(For brevity, we ignore the ``$0 / 0$'' issue.)

\setcounter{fact}{0}

\begin{fact}
\label{fact:AD:poa}
$\bar{\Delta}_{\FPA} \big/ \bar{\Delta}_{\OPT} \,\leq\, \underline{\Delta}_{\FPA} \big/ \underline{\Delta}_{\OPT}$.
% and the denominators are positive $\bar{\Delta}_{\OPT},\, \underline{\Delta}_{\OPT} > 0$.
\end{fact}

\begin{proof}
The \blackref{AD_ascend} auction {\SocialWelfare} change $\bar{\Delta}_{\FPA} \equiv \FPA(\blackref{AD_ascend}) - \FPA(\blackref{AD_input})$ is at most
\begin{align}
    \bar{\Delta}_{\FPA}
    & ~=~ (\phi^{*} - \phi_{H,\, 0}) \cdot \calB(0)
    + \sum_{j = 0}^{j^{*} - 1} \int_{\lambda_{j}}^{\lambda_{j + 1}} \Big(\big(\tfrac{\phi^{*}}{\varphi_{L}(b) - b} - \tfrac{\phi^{*} - \varphi_{L}(b)}{\phi^{*} - b}\big)
    \nonumber \\
    & \phantom{~=~ (\phi^{*} - \phi_{H,\, 0}) \cdot \calB(0) \sum_{j = 0}^{j^{*} - 1} \int_{\lambda_{j}}^{\lambda_{j + 1}} \Big(}
    - \big(\tfrac{\phi_{H,\, 0}}{\varphi_{L}(b) - b} - \tfrac{\phi_{H,\, 0} - \varphi_{L}(b)}{\phi_{H,\, 0} - b}\big)\Big) \cdot \calB(b) \cdot \d b
    \tag{restate} \\
    & ~=~ (\phi^{*} - \phi_{H,\, 0}) \cdot \calB(0)
    + \int_{0}^{\lambda_{j^{*}}} \Big(\tfrac{\phi^{*} - \phi_{H,\, 0}}{\varphi_{L}(b) - b} + \big(\tfrac{\phi_{H,\, 0} - \varphi_{L}(b)}{\phi_{H,\, 0} - b} - \tfrac{\phi^{*} - \varphi_{L}(b)}{\phi^{*} - b}\big)\Big) \cdot \calB(b) \cdot \d b
    \hspace{1.95cm}~~\;\; \nonumber \\
    % \label{eq:win_win:FPA:A2}\tag{A1} \\
    & ~\leq~ (\phi^{*} - \phi_{H,\, 0}) \cdot \calB(0)
    + \int_{0}^{\lambda_{j^{*}}} \phantom{\Big(}\tfrac{\phi^{*} - \phi_{H,\, 0}}{\varphi_{L}(b) - b} \cdot \calB(b) \cdot \d b
    \label{eq:win_win:FPA:A3}\tag{A1} \\
    & ~=~ (\phi^{*} - \phi_{H,\, 0}) \cdot \calB(0)
    + \int_{0}^{\lambda_{j^{*}}} (\phi^{*} - \phi_{H,\, 0}) \cdot \calB'(b) \cdot \d b \phantom{\bigg.}
    \label{eq:win_win:FPA:A4}\tag{A2} \\
    & ~=~ (\phi^{*} - \phi_{H,\, 0}) \cdot \calB(\lambda_{j^{*}}). \phantom{\bigg.}
    \nonumber
    % \label{eq:win_win:FPA:A5}\tag{A3}
\end{align}
% \eqref{eq:win_win:FPA:A2}: Additivity of integrals, given $\lambda_{0} \equiv 0$. Rearrange the integrand. \\
\eqref{eq:win_win:FPA:A3}: The dropped term $\leq 0$ (\Cref{def:jump,lem:pseudo_mapping}; $\phi^{*} \equiv \phi_{\sigma^{*},\, j^{*}} > \phi_{H,\, 0} \geq \varphi_{L}(b) > b$). \\
\eqref{eq:win_win:FPA:A4}: The pseudo mapping $\varphi_{L} = b + \calB \big/ \calB'$ (\Cref{def:pseudo}).

\vspace{.1in}
\noindent
The \blackref{AD_descend} auction {\SocialWelfare} change $\underline{\Delta}_{\FPA} \equiv \FPA(\blackref{AD_input}) - \FPA(\blackref{AD_descend})$ is at least
\begin{align}
    \underline{\Delta}_{\FPA}
    & ~=~ \int_{\lambda_{j^{*}}}^{\lambda_{j^{*} + 1}} \Big(\big(\tfrac{\phi^{*}}{\varphi_{L}(b) - b} - \tfrac{\phi^{*} - \varphi_{L}(b)}{\phi^{*} - b}\big) - \big(\tfrac{\phi_{H,\, 0}}{\varphi_{L}(b) - b} - \tfrac{\phi_{H,\, 0} - \varphi_{L}(b)}{\phi_{H,\, 0} - b}\big)\Big) \cdot \calB(b) \cdot \d b
    \tag{restate} \\
    & ~=~ \int_{\lambda_{j^{*}}}^{\lambda_{j^{*} + 1}} \Big(\big(\tfrac{\phi^{*} - \phi_{H,\, 0}}{\varphi_{L}(b) - b} - \tfrac{\phi^{*} - \phi_{H,\, 0}}{\phi^{*} - b}\big) + \big(\tfrac{\phi_{H,\, 0} - \varphi_{L}(b)}{\phi_{H,\, 0} - b} - \tfrac{\phi_{H,\, 0} - \varphi_{L}(b)}{\phi^{*} - b}\big)\Big) \cdot \calB(b) \cdot \d b
    \hspace{2.76cm}~~ \nonumber \\
    % \label{eq:win_win:FPA:D2}\tag{D1} \\
    & ~\geq~ \int_{\lambda_{j^{*}}}^{\lambda_{j^{*} + 1}} \phantom{\Big(} \big(\tfrac{\phi^{*} - \phi_{H,\, 0}}{\varphi_{L}(b) - b} - \tfrac{\phi^{*} - \phi_{H,\, 0}}{\phi^{*} - b}\big) \cdot \calB(b) \cdot \d b
    \label{eq:win_win:FPA:D3}\tag{D1} \\
    & ~=~ \int_{\lambda_{j^{*}}}^{\lambda_{j^{*} + 1}}~~\; (\phi^{*} - \phi_{H,\, 0}) \cdot B'_{\sigma^{*}}(b) \cdot \tfrac{\calB(b)}{B_{\sigma^{*}}(b)} \cdot \d b
    \label{eq:win_win:FPA:D4}\tag{D2} \\
    & ~\geq~ \int_{\lambda_{j^{*}}}^{\lambda_{j^{*} + 1}}~~\; (\phi^{*} - \phi_{H,\, 0}) \cdot B'_{\sigma^{*}}(b) \cdot \tfrac{\calB(\lambda_{j^{*}})}{B_{\sigma^{*}}(\lambda_{j^{*}})} \cdot \d b
    \label{eq:win_win:FPA:D5}\tag{D3} \\
    & ~=~ (\phi^{*} - \phi_{H,\, 0}) \cdot \big(\tfrac{B_{\sigma^{*}}(\lambda_{j^{*} + 1})}{B_{\sigma^{*}}(\lambda_{j^{*}})} - 1\big) \cdot \calB(\lambda_{j^{*}})
    ~\geq~ 0.
    \nonumber
    % \label{eq:win_win:FPA:D6}\tag{D4}
\end{align}
% \eqref{eq:win_win:FPA:D2}: Merge the first/third terms into $\frac{\phi^{*} - \phi_{H,\, 0}}{\varphi_{L}(b) - b}$. Split the second term into $-\frac{\phi^{*} - \phi_{H,\, 0}}{\phi^{*} - b} - \frac{\phi_{H,\, 0} - \varphi_{L}(b)}{\phi^{*} - b}$. \\
\eqref{eq:win_win:FPA:D3}: The dropped term $\geq 0$ (\Cref{def:jump,lem:pseudo_mapping}; $\phi^{*} \equiv \phi_{\sigma^{*},\, j^{*}} > \phi_{H,\, 0} \geq \varphi_{L}(b) > b$). \\
\eqref{eq:win_win:FPA:D4}: $\phi^{*} \equiv \phi_{\sigma^{*},\, j^{*}} = b + \big(\calB' \big/ \calB - B'_{\sigma^{*}} \big/ B_{\sigma^{*}}\big)^{-1}$ (given a \blackref{AD_jump} $\sigma^{*} \neq L$) and $\varphi_{L} = b + \calB \big/ \calB'$. \\
\eqref{eq:win_win:FPA:D5}: $\calB \big/ B_{\sigma^{*}} = \prod_{\sigma \in [N] \setminus \{\sigma^{*}\}} B_{\sigma}$ is an increasing CDF.
% \eqref{eq:win_win:FPA:D6}: Resolve the integral.

Combining the above two equations and applying the substitution $\omega^{*} \equiv \omega_{\sigma^{*},\, j^{*}} = 1 - \frac{B_{\sigma^{*}}(\lambda_{j^{*}})}{B_{\sigma^{*}}(\lambda_{j^{*} + 1})}$ (\Cref{lem:translate_welfare}), we deduce that $\bar{\Delta}_{\FPA} \big/ \underline{\Delta}_{\FPA} \leq (1 - \omega^{*}) / \omega^{*}$.

% the only difference is a scaling factor $\big(\frac{B_{\sigma^{*}}(\lambda_{j^{*} + 1})}{B_{\sigma^{*}}(\lambda^{*})} - 1\big) > 0$.

\vspace{.1in}
\noindent
The \blackref{AD_ascend} optimum change $\bar{\Delta}_{\OPT} \equiv \OPT(\blackref{AD_ascend}) - \OPT(\blackref{AD_input})$ is equal to
\begin{align*}
    \bar{\Delta}_{\OPT}
    & ~=~ \int_{0}^{+\infty} \Big(\indicator(v \geq \phi_{H,\, 0}) \cdot
    \big(1 - \omega^{*} \cdot \indicator(v < \phi^{*})\big) - \indicator(v \geq \phi^{*})\big) \cdot \calI(v) \cdot \d v
    \tag{restate} \\
    & ~=~ \int_{0}^{+\infty} \Big(\big(1 - \omega^{*} \cdot \indicator(v < \phi^{*})\big) - \indicator(v \geq \phi^{*})\Big) \cdot \indicator(v \geq \phi_{H,\, 0}) \cdot \calI(v) \cdot \d v
    \hspace{3.35cm}~~ \\
    & ~=~ \big(1 - \omega^{*}\big) \cdot \int_{0}^{+\infty} \indicator\big(v \geq \phi_{H,\, 0}\big) \cdot
    \indicator\big(v < \phi^{*}\big) \cdot \calI(v) \cdot \d v.
\end{align*}
The second step: $\indicator(v \geq \phi^{*}) \equiv \indicator(v \geq \phi_{H,\, 0}) \cdot \indicator(v \geq \phi^{*})$ because the jump value $\phi^{*} \equiv \phi_{\sigma^{*},\, j^{*}} > \phi_{H,\, 0}$. \\
The last step: $\indicator(v \geq \phi^{*}) \equiv 1 - \indicator(v < \phi^{*})$.

\vspace{.1in}
\noindent
The \blackref{AD_descend} optimum change $\underline{\Delta}_{\OPT} \equiv \OPT(\blackref{AD_input}) - \OPT(\blackref{AD_descend})$ is given by
\begin{align*}
    \underline{\Delta}_{\OPT}
    & ~=~ \int_{0}^{+\infty} \Big(\Big(1 - \indicator(v \geq \phi_{H,\, 0}) \cdot \big(1 - \omega^{*} \cdot \indicator(v < \phi^{*})\big) \cdot \calI(v)\Big) - \Big(1 - \indicator(v \geq \phi_{H,\, 0}) \cdot \calI(v)\Big)\Big) \cdot \d v \\
    & ~=~ \omega^{*} \cdot \int_{0}^{+\infty} \indicator(v \geq \phi_{H,\, 0}) \cdot \indicator(v < \phi^{*}) \cdot \calI(v) \cdot \d v.
\end{align*}
We thus conclude that $\bar{\Delta}_{\OPT} \big/ \underline{\Delta}_{\OPT} = (1 - \omega^{*}) / \omega^{*} \geq \bar{\Delta}_{\FPA} \big/ \underline{\Delta}_{\FPA}$, which implies {\bf \Cref{fact:AD:poa}}.
\end{proof}

Assume the opposite to {\bf \Cref{lem:AD:poa}}: Both of the \blackref{AD_ascend}/\blackref{AD_descend} candidates $\bar{H}^{\uparrow} \otimes \bar{\bB} \otimes \bar{L}$ and $\underline{H}^{\uparrow} \otimes \underline{\bB} \otimes \underline{L}$ yield strictly larger bounds than the \blackref{AD_input} $H^{\uparrow} \otimes \bB \otimes L$. That is,
\begin{align*}
    & \PoA(\blackref{AD_ascend}) \;\;\;>~ \PoA(\blackref{AD_input})
    && \iff &&
    \frac{\FPA(\blackref{AD_input}) + \bar{\Delta}_{\FPA}}{\OPT(\blackref{AD_input}) + \bar{\Delta}_{\OPT}}
    ~>~ \frac{\FPA(\blackref{AD_input})}{\OPT(\blackref{AD_input})}, \\
    & \PoA(\blackref{AD_descend}) ~>~ \PoA(\blackref{AD_input})
    && \iff &&
    \frac{\FPA(\blackref{AD_input}) - \underline{\Delta}_{\FPA}}{\OPT(\blackref{AD_input}) - \underline{\Delta}_{\OPT}}
    ~>~ \frac{\FPA(\blackref{AD_input})}{\OPT(\blackref{AD_input})}.
\end{align*}
Rearranging both equations gives $\bar{\Delta}_{\FPA} \big/ \bar{\Delta}_{\OPT} \,>\, \PoA(\blackref{AD_input}) \,>\, \underline{\Delta}_{\FPA} \big/ \underline{\Delta}_{\OPT}$ (notice that the denominators $\bar{\Delta}_{\OPT},\, \underline{\Delta}_{\OPT} > 0$), which contradicts {\bf \Cref{fact:AD:poa}}.
Refute our assumption: At least one candidate between $\bar{H}^{\uparrow} \otimes \bar{\bB} \otimes \bar{L}$ and $\underline{H}^{\uparrow} \otimes \underline{\bB} \otimes \underline{L}$ yields a (weakly) worse bound. {\bf \Cref{lem:AD:poa}} follows.

This finishes the proof of \Cref{lem:AD}.
\end{proof}

\subsection{The main algorithm: From floor/ceiling to twin ceiling}
\label{subsec:main}

This subsection presents the \blackref{alg:main} procedure (see \Cref{fig:alg:main} for its description), which transforms a {\em floor}/{\em ceiling} \blackref{main_input} pseudo instance $H \otimes \bB \otimes L \in (\Bfloor \cup \Bceiling)$ into a {\em twin ceiling} output pseudo instance $\tilde{H} \otimes \tilde{L} \in \Btwin$ (\Cref{def:twin}).
This \blackref{alg:main} procedure runs iteratively and is built from all the given reductions:
\blackref{alg:slice}, \blackref{alg:collapse}, \blackref{alg:halve}, and \blackref{alg:AD} (\Cref{subsec:slice,subsec:collapse,subsec:halve,subsec:AD}).
Overall, through the potential method (\Cref{def:potential}), we will see that this procedure terminates in {\em finite} iterations, returning the {\em twin ceiling} output $\tilde{H} \otimes \tilde{L} \in \Btwin$ afterward.

For ease of reference, we rephrase \Cref{lem:slice,lem:collapse,lem:halve,lem:AD}.

\begin{restate}[{\Cref{lem:slice}}]
\begin{flushleft}
\blackref{alg:slice} transforms a Floor $H^{\downarrow} \otimes \bB \otimes L \in \Bfloor$ to
a \textsf{PoA}-worse Ceiling $\tilde{H}^{\uparrow} \otimes \tilde{\bB} \otimes \tilde{L} \in \Bceiling$ such that the potential decreases $\tilde{\Psi} \leq \Psi - 1$.
\end{flushleft}
\end{restate}

\begin{restate}[{\Cref{lem:collapse}}]
\begin{flushleft}
\blackref{alg:collapse} transforms a Ceiling $H^{\uparrow} \otimes \bB \otimes L \in \Bceiling$ to a \textsf{PoA}-worse Strong Ceiling $\tilde{H}^{\uparrow} \otimes \tilde{\bB} \otimes \tilde{L} \in \Bstrong \subsetneq \Bceiling$
such that the potential keeps the same $\tilde{\Psi} = \Psi$.
\end{flushleft}
\end{restate}

\begin{restate}[{\Cref{lem:halve}}]
\begin{flushleft}
\blackref{alg:halve} transforms a Strong Ceiling $H^{\uparrow} \otimes \bB \otimes L \in \Bstrong \subsetneq \Bceiling$ to
a \textsf{PoA}-worse Floor/Ceiling $\tilde{H} \otimes \tilde{\bB} \otimes \tilde{L} \in (\Bfloor \cup \Bceiling)$
such that the potential decreases $\tilde{\Psi} \leq \Psi - 1$.
\end{flushleft}
\end{restate}

\begin{restate}[{\Cref{lem:AD}}]
% \begin{flushleft}
\blackref{alg:AD} transforms a Strong Ceiling $H^{\uparrow} \otimes \bB \otimes L \in \Bstrong$ to a \textsf{PoA}-worse Ceiling $\tilde{H}^{\uparrow} \otimes \tilde{\bB} \otimes \tilde{L} \in \Bceiling$ such that the potential decreases $\tilde{\Psi} \leq \Psi - 1$.
% \end{flushleft}
\end{restate}

\Cref{lem:main} summarizes performance guarantees of the \blackref{alg:main} procedure.

\begin{lemma}[{\main}; \Cref{fig:alg:main}]
\label{lem:main}
Under procedure $\tilde{H}^{\uparrow} \otimes \tilde{L} \gets \main(H \otimes \bB \otimes L)$:
\begin{enumerate}[font = {\em\bfseries}]
    \item\label{lem:main:time}
    It terminates in at most $(1 + 2\Psi^{*})$ iterations, where $\Psi^{*} = \Psi(H \otimes \bB \otimes L) < +\infty$ (\Cref{lem:potential:bound} of \Cref{lem:potential}) is the finite potential of the input floor/ceiling pseudo instance.
    
    \item\label{lem:main:property}
    The output is a twin ceiling pseudo instance $\tilde{H}^{\uparrow} \otimes \tilde{L} \in \Btwin$.
    
    \item\label{lem:main:poa}
    A (weakly) worse bound is yielded $\PoA(\tilde{H}^{\uparrow} \otimes \tilde{L}) \leq \PoA(H \otimes \bB \otimes L)$.
\end{enumerate}
\end{lemma}

\begin{figure}[t]
    \centering
    \begin{mdframed}
    Procedure $\term[\main]{alg:main}(H \otimes \bB \otimes L)$
    
    \begin{flushleft}
    {\bf Input:}
    A {\em floor}/{\em ceiling} pseudo instance $H \otimes \bB \otimes L \in (\Bfloor \cup \Bceiling)$.
    \white{\term[\text{\em input}]{main_input}}
    \hfill
    \Cref{def:ceiling_floor:restate}
    
    \vspace{.05in}
    {\bf Output:}
    A {\em twin ceiling} pseudo instance $\tilde{H}^{\uparrow} \otimes \tilde{L} \in \Btwin$.
    \white{\term[\text{\em input}]{main_output}}
    \hfill
    \Cref{def:twin}
    
    \vspace{.05in}
    {\bf Remark:}
    ``{\rename}'' denotes the reassignment ``$H \otimes \bB \otimes L \,\gets\, \tilde{H} \otimes \tilde{\bB} \otimes \tilde{L}$''.
    % A superscript up-arrow may be added to the conditional value distributions ($P^{\uparrow}$ and $\tilde{P}^{\uparrow}$) \\
    % to emphasize the {\em ceiling} pseudo instances (\Cref{def:ceiling_floor}).
    
    \begin{enumerate}
        \item\label{alg:main:while_begin}
        \term[{\bf While}]{main_while} $\big\{ H \otimes \bB \otimes L \notin \Btwin \big\}${\bf :}
        % \subsetneq \Bstrong \subsetneq \Bceiling
        % \hfill
        % \OliveGreen{$\triangleright$ $\Btwin \subsetneq \Bstrong \subsetneq \Bceiling$}
        
        \item\label{alg:main:slice}
        \qquad
        {\bf If} $\big\{ H \otimes \bB \otimes L \notin \Bceiling \big\}${\bf :}
        \hfill
        \OliveGreen{$\triangleright$ $\neg$(\colorref{OliveGreen}{re:ceilingness})}
        
        \item
        \qquad\qquad
        $\tilde{H}^{\uparrow} \otimes \tilde{\bB} \otimes \tilde{L} \,\gets\, \slice(H^{\downarrow} \otimes \bB \otimes L)$ and {\rename}.
        
        \item\label{alg:main:collapse}
        \qquad
        {\bf Else If} $\big\{ H \otimes \bB \otimes L \in (\Bceiling \setminus \Bstrong) \big\}${\bf :}
        \hfill
        \OliveGreen{$\triangleright$ $\neg$(\colorref{OliveGreen}{re:collapse})}
        
        \item
        \qquad\qquad
        $\tilde{H}^{\uparrow} \otimes \tilde{\bB} \otimes \tilde{L} \,\gets\, \collapse(H^{\uparrow} \otimes \bB \otimes L)$ and {\rename}.
        
        \item\label{alg:main:halve}
        \qquad
        {\bf Else If} $\big\{ H \otimes \bB \otimes L \in (\Bstrong \setminus \Btwin) \big\}${\bf :}
        \hfill
        \OliveGreen{$\triangleright$ $\neg$(\colorref{OliveGreen}{re:twin_ceiling})}
        
        \item[]
        \qquad
        \OliveGreen{$\triangleright$ The jump entry $(\sigma^{*},\, j^{*}) \in (\{H\} \cup [n] \cup \{L\}) \times [m]$ (\Cref{lem:jump}).}
        
        \item
        \qquad\qquad
        {\bf Case $\{\sigma^{*} = L \}$:}
        % \hfill
        % \OliveGreen{$\triangleright$ $\neg$(\colorref{OliveGreen}{re:connectivity})}
        $\tilde{H} \otimes \tilde{\bB} \otimes \tilde{L} \,\gets\, \halve(H^{\uparrow} \otimes \bB \otimes L)$ and {\rename}.

        \item\label{alg:main:AD}
        \qquad\qquad
        {\bf Case $\{\sigma^{*} \neq L \}$:}
        % \OliveGreen{$\triangleright$ the jump bidder is one real bidder $\sigma^{*} \in \{H\} \cup [n]$}
        % \item
        % \qquad\qquad\qquad
        $\tilde{H}^{\uparrow} \otimes \tilde{\bB} \otimes \tilde{L} \,\gets\, \AD(H^{\uparrow} \otimes \bB \otimes L)$ and {\rename}.
        
        \item {\bf Return}
        $H^{\uparrow} \otimes \bB \otimes L \cong \tilde{H}^{\uparrow} \otimes \tilde{L} \in \Btwin$.
    \end{enumerate}
    \end{flushleft}
    \end{mdframed}
    \caption{The {\main} procedure
    \label{fig:alg:main}}
\end{figure}

\begin{proof}
It suffices to show {\bf \Cref{lem:main:time}}. Suppose so, {\bf \Cref{lem:main:property,lem:main:poa}} follow directly from the termination condition for the \blackref{main_while} loop (Line~\ref{alg:main:while_begin}) and a combination of \Cref{lem:slice,lem:collapse,lem:halve,lem:AD}.

\vspace{.1in}
\noindent
{\bf \Cref{lem:main:time}.}
We can divide all of the four reductions into two types.
\begin{itemize}
    \item \textbf{Type-1: \blackref{alg:slice}, \blackref{alg:halve}, and \blackref{alg:AD}.}
    
    Such a reduction decreases the potential $\tilde{\Psi} \leq \Psi - 1$. Therefore, such reductions in total can be invoked at most ``the \blackref{main_input} potential'' many times. (Recall \Cref{lem:potential} that a {\em twin ceiling} pseudo instance $\hat{H}^{\uparrow} \otimes \hat{\bB} \otimes \hat{L} \in \Btwin$, which triggers the termination condition for the \blackref{main_while} loop, has a zero potential.) Formally, we have $\#\big[\textbf{Type-1}\big] \,\leq\, \Psi^{*}$.
    
    \item \textbf{Type-2: \blackref{alg:collapse}.}
    
    This reduction transforms a {\em ceiling} but {\em non strong ceiling} pseudo instance $\in (\Bceiling \setminus \Bstrong)$ that violates \blackref{re:collapse} into a {\em strong ceiling} pseudo instance $\in \Bstrong$ that satisfies \blackref{re:collapse}. Therefore, between any two invocations of this reduction, (Line~\ref{alg:main:collapse}) at least one invocation of another \text{Type-1} reduction is required. Formally, we have $\#\big[\textbf{Type-2}\big] \,\leq\, \#\big[\textbf{Type-1}\big] + 1 \,\leq\, \Psi^{*} + 1$.
\end{itemize}
Hence, the \blackref{alg:main} procedure terminates in at most $\#\big[\textbf{Type-1}\big] + \#\big[\textbf{Type-2}\big] \leq 2\Psi^{*} + 1$ iterations.
{\bf \Cref{lem:main:time}} follows then.
This finishes the proof.
\end{proof}

\begin{wrapfigure}{r}{0.3\textwidth}
    \centering
    \includegraphics[width = \linewidth]
    {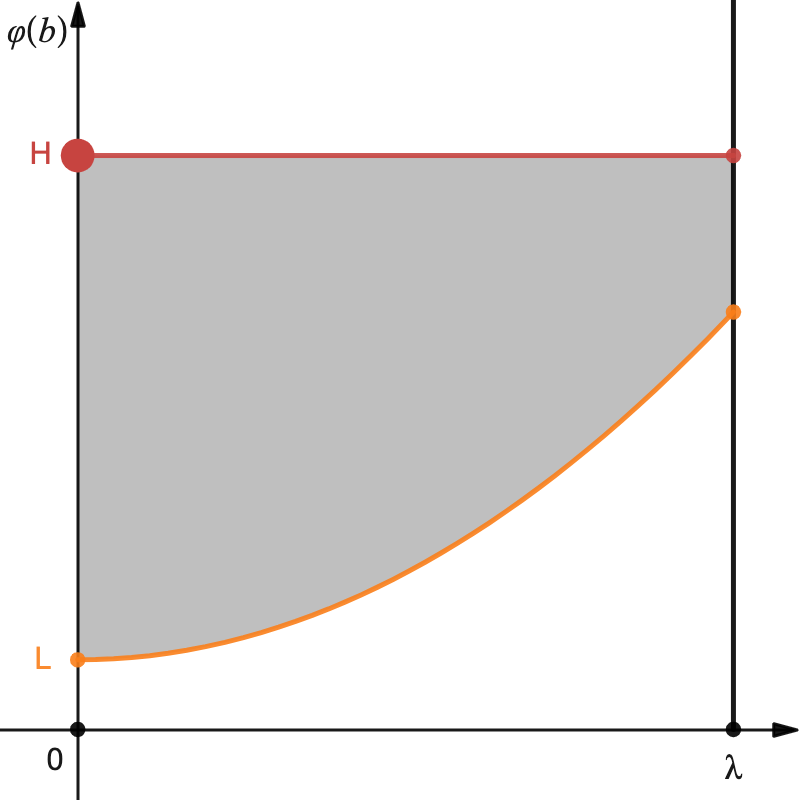}
    \caption{$H \otimes L \in \Btwin$
    %{\em Twin ceiling} pseudo instances.
    \label{fig:twin:restate}}
\end{wrapfigure}

%\begin{figure}[t]
%    \centering
%    \includegraphics[width = .45\textwidth, height = 6.375cm]
%    {twin_restate.png}
%    \caption{Diagram of {\em twin ceiling} pseudo instances.
%    \label{fig:twin:restate}}
%\end{figure}

Put \Cref{lem:main,cor:preprocess} together: Towards a lower bound on the {\PriceofAnarchy}, we can focus on {\em twin ceiling} pseudo instances $H \otimes L \in \Btwin$ (see \Cref{fig:twin:restate} for a visual aid).
For such a pseudo instance, the monopolist $H$ has a constant bid-to-value mapping $\varphi_{H}(b) = \phi_{H,\, 0}$ over the bid support $b \in [0,\, \lambda]$ and the optimal {\SocialWelfare} is exactly the ceiling value $\OPT(H \otimes L) = \phi_{H,\, 0}$ (\Cref{lem:ceiling_welfare}).
Obviously, scaling this pseudo instance to normalize the ceiling value $\phi_{H,\, 0} = 1$ does not change the {\PoA} bound.
% , for a {\em twin ceiling} pseudo instance $H \otimes L \in \Btwin$

In addition, the \blackref{re:layeredness} condition becomes vacuously true (cf.\ \Cref{lem:pseudo_mapping}), while removing the \blackref{re:discretization} condition will expand the search space, which is fine for our purpose towards a lower bound on the {\PriceofAnarchy}.

Taking these into account, we would redefine the (normalized) {\em twin ceiling} pseudo instances and consider a modified functional optimization (\Cref{cor:reduction}).

\begin{definition}[Twin ceiling pseudo instances]
\label{def:twin:restate}
For a {\em twin ceiling} pseudo instance $H \otimes L$:
\begin{itemize}
    \item The (real) {\em monopolist} $H$ competes with the pseudo bidder $L$, having a {\em constant} conditional value $P^{\uparrow} \equiv 1$ and a {\em constant} bid-to-value mapping $\varphi_{H}(b) = b + \frac{L(b)}{L'(b)} = 1$ on $b \in [0,\, \lambda]$. Hence, the supremum bid is bounded between $\lambda \in (0,\, 1)$ and the pseudo bidder $L$'s bid distribution is given by $L(b) = L_{\lambda}(b) \eqdef \frac{1 - \lambda}{1 - b}$.
    
    \item The pseudo bidder $L$ competes with both bidders $H \otimes L$, having an {\em increasing} bid-to-value mapping $\varphi_{L}(b) = b + \big(\frac{H'(b)}{H(b)} + \frac{L'(b)}{L(b)}\big)^{-1} = b + \big(\frac{H'(b)}{H(b)} + \frac{1}{1 - b}\big)^{-1} = 1 - \frac{(1 - b)^{2}}{H(b) / H'(b) + 1 - b}$ on $b \in [0,\, \lambda]$.
\end{itemize}
Given a supremum bid $\lambda \in (0,\, 1)$, the pseudo bidder $L = L_{\lambda}$ is determined, so the search space is simply $\mathbb{H}_{\lambda} \eqdef \big\{H \bigmid \text{the pseudo mapping $\varphi_{L}(b)$ is increasing on $b \in [0,\, \lambda]$} \big\}$.
\end{definition}

\begin{corollary}[Lower bound]
\label{cor:reduction}
Regarding {\FirstPriceAuctions}, the {\PriceofAnarchy} is at least
\begin{align}
\label{eq:reduction}\tag{$\clubsuit$}
    \PoA ~\geq~ \inf \Big\{\, \FPA(H \otimes L_{\lambda}) \,\Bigmid\, \lambda \in (0,\, 1) ~\text{\em and}~ H \in \mathbb{H}_{\lambda} \,\Big\}.
\end{align}
\end{corollary}

\newpage

\section{Lower Bound Analysis}
\label{sec:UB}

Following \Cref{cor:reduction} and Optimization~\eqref{eq:reduction}, we study the space of twin ceiling pseudo instances $H \otimes L_{\lambda}$ under all possible supremum bids $\lambda \in (0,\, 1)$ and aim to find the worst case $H^{*} \otimes L^{*}$.  This task is accomplished in three steps:
\begin{itemize}
    \item \Cref{subsec:UB_welfare} presents an explicit formula for the expected auction {\SocialWelfare} $\FPA(H \otimes L_{\lambda})$. Also, we prove that imposing certain properties (such as {\em differentiability}) on the monopolist's bid distribution $H$ does not change the solution to Optimization~\eqref{eq:reduction}.
    
    \item \Cref{subsec:UB_ODE} resolves Optimization~\eqref{eq:reduction} by leveraging tools from Calculus of Variations (i.e., the imposed properties enable these tools). Roughly speaking, the worst case $H^{*} \otimes L^{*}$ turns out to be specified by the supremum bid $\lambda \in (0,\, 1)$ and one more parameter $\mu > 0$, thus taking the form $H^{*} \otimes L^{*} = H_{\mu^{*}} \otimes L_{\lambda^{*}}$.
    
    \item \Cref{subsec:UB_lambda_mu} captures the worst case $H_{\mu^{*}} \otimes L_{\lambda^{*}}$ by optimizing the parameters $\lambda$ and $\mu$, which immediately implies the lower-bound part of \Cref{thm:main} that the $\PoA \geq 1 - 1 / e^{2}$.
\end{itemize}

\subsection{The expected auction {\SocialWelfare}}
\label{subsec:UB_welfare}

We introduce the space of {\em twice continuously differentiable} CDF's $\mathbb{C}_{\lambda}^{2}$ right below (\Cref{def:differentiability}). Afterward, we (i)~get an explicit reformulation for Optimization~\eqref{eq:reduction}, then (ii)~prove that the same solution can be achieved from a smaller search space $\mathbb{H}_{\lambda} \cap \mathbb{C}_{\lambda}^{2} \subsetneq \mathbb{H}_{\lambda}$, and then (iii)~relax/expand the search space to $\mathbb{C}_{\lambda}^{2}$.

\begin{definition}[Differentiability]
\label{def:differentiability}
Denote by $\mathbb{C}_{\lambda}^{0}$ the space of all {\em continuous} CDF's that are supported on $[0,\, \lambda]$. The smaller space $\mathbb{C}_{\lambda}^{1}$ (resp.\ the even smaller space $\mathbb{C}_{\lambda}^{2}$) further requires {\em continuously differentiable} CDF's (resp.\ {\em twice continuously differentiable} CDF's), i.e., the derivatives (resp.\ the second derivatives) of those CDF's are continuous functions.
\end{definition}

\begin{lemma}[{\PriceofAnarchy}]
\label{lem:reformulation}
Regarding {\FirstPriceAuctions}, the {\PriceofAnarchy} satisfies that
\begin{align}
    \PoA
    & ~\geq~ 1 - \sup \Bigg\{\, (1 - \lambda) \cdot \int_{0}^{\lambda} \frac{H(x) \cdot H'(x) \cdot \d x}{H(x) + (1 - x) \cdot H'(x)} \,\Biggmid\, \text{\em $\lambda \in (0,\, 1)$ and $H \in \mathbb{H}_{\lambda}$} \,\Bigg\}
    \label{eq:reformulation:1}\tag{I} \\
    & ~=~ 1 - \sup \Bigg\{\, (1 - \lambda) \cdot \int_{0}^{\lambda} \frac{H(x) \cdot H'(x) \cdot \d x}{H(x) + (1 - x) \cdot H'(x)} \,\Biggmid\, \text{\em $\lambda \in (0,\, 1)$ and $H \in (\mathbb{H}_{\lambda} \cap \mathbb{C}_{\lambda}^{2})$} \,\Bigg\}
    \label{eq:reformulation:2}\tag{II} \\
    & ~\geq~ 1 - \sup \Bigg\{\, (1 - \lambda) \cdot \int_{0}^{\lambda} \frac{H(x) \cdot H'(x) \cdot \d x}{H(x) + (1 - x) \cdot H'(x)} \,\Biggmid\, \text{\em $\lambda \in (0,\, 1)$ and $H \in \mathbb{C}_{\lambda}^{2}$} \,\Bigg\}.
    \label{eq:reformulation:3}\tag{III}
\end{align}
\end{lemma}

Let us prove the three steps of \Cref{lem:reformulation} one by one.

\begin{proof}[Proof of Optimization~\eqref{eq:reformulation:1}]
This is a reformulation of Optimization~\eqref{eq:reduction}. The monopolist $H$ has a constant bid-to-value mapping $\varphi_{H}(b) = v_{H} = 1$ for $b \in [0,\, \lambda]$ (\Cref{def:twin:restate}). Following the formula in \Cref{lem:pseudo_welfare}, the expected auction {\SocialWelfare} $\FPA(H \otimes L_{\lambda})$ is given by
\begin{align*}
    \FPA(H \otimes L_{\lambda})
    & ~=~ v_{H} \cdot H(0) \cdot L_{\lambda}(0) + \int_{0}^{\lambda} \bigg(v_{H} \cdot H'(b) \cdot L_{\lambda}(b) + \varphi_{L}(b) \cdot L'_{\lambda}(b) \cdot H(b)\bigg) \cdot \d b \\
    & ~=~ v_{H} \cdot H(\lambda) \cdot L_{\lambda}(\lambda) - \int_{0}^{\lambda} \big(v_{H} - \varphi_{L}(b)\big) \cdot L'_{\lambda}(b) \cdot H(b) \cdot \d b \\
    & ~=~ 1 - \int_{0}^{\lambda} \frac{(1 - b)^{2}}{H(b) / H'(b) + 1 - b} \cdot L'_{\lambda}(b) \cdot H(b) \cdot \d b \\
    & ~=~ 1 - (1 - \lambda) \cdot \int_{0}^{\lambda} \frac{H(b) \cdot \d b}{H(b) / H'(b) + 1 - b}
\end{align*}
Here the second step uses integration by parts $\int_{0}^{\lambda} \big(H'(b) \cdot L_{\lambda}(b) + L'_{\lambda}(b) \cdot H(b)\big) \cdot \d b = H(b) \cdot L_{\lambda}(b) \bigmid_{b = 0}^{\lambda}$. The third step uses $H(\lambda) = L_{\lambda}(\lambda) = 1$ (because $\lambda$ is the supremum bid) and the $\varphi_{L}(b)$ formula in \Cref{def:twin:restate}. And the last step applies $L'_{\lambda}(b) = \frac{1 - \lambda}{(1 - b)^{2}}$ (\Cref{def:twin:restate}).

The above formula $\FPA(H \otimes L_{\lambda})$ together with \Cref{cor:reduction} implies Optimization~\eqref{eq:reformulation:1}.
\end{proof}

\begin{proof}[Proof of Optimization~\eqref{eq:reformulation:2}]
Given a supremum bid $\lambda \in (0,\, 1)$, we consider a specific CDF $H \in \mathbb{H}_{\lambda}$. The integral $\int_{0}^{\lambda} \frac{H(b) \cdot H'(b) \cdot \d b}{H(b) + (1 - b) \cdot H'(b)}$ must be {\em Lebesgue integrable}, since the integrand $\frac{H(b) \cdot H'(b)}{H(b) + (1 - b) \cdot H'(b)}$ is nonnegative and the integral is upper bounded by $\int_{0}^{\lambda} \frac{H(b) \cdot H'(b) \cdot \d b}{H(b)} = H(\lambda) - H(0) \leq 1$.
Further, the integrand, if being regarded as a function of three variables $b \in (0,\, \lambda)$, $H \in [0,\, 1]$ and $H' \geq 0$, is {\em twice continuously differentiable}. For these reasons, the CDF $H \in \mathbb{H}_{\lambda}$ can be replaced
% \footnote{\blue{By definition, it can be easily checked that both $\mathbb{H}_{\lambda}$ and $\mathbb{C}_{\lambda}^{2}$ are closed sets, i.e., the limit of a {\em converged} sequence of CDF's $\in \mathbb{H}_{\lambda}$ (resp.\ $\mathbb{C}_{\lambda}^{2}$) is still a CDF $\in \mathbb{H}_{\lambda}$ (resp.\ $\mathbb{C}_{\lambda}^{2})$. Thus the intersection $(\mathbb{H}_{\lambda} \cap \mathbb{C}_{\lambda}^{2})$ is also a closed set, which guarantees the existence of a substitute CDF $\tilde{H}$ as required.}}
with another {\em twice continuously differentiable} CDF $\tilde{H} \in (\mathbb{H}_{\lambda} \cap \mathbb{C}_{\lambda}^{2})$ such that the integral {\em keeps the same}, namely the replacement incurs an infinitesimal error whose magnitude is irrelevant to the $\lambda \in (0,\, 1)$.

To conclude, Optimization~\eqref{eq:reformulation:2} has the same solution as Optimization~\eqref{eq:reformulation:1}.
\end{proof}

\begin{proof}[Proof of Optimization~\eqref{eq:reformulation:3}]
This step trivially holds, since it relaxes/enlarges the search space.
% from $(\mathbb{H}_{\lambda} \cap \mathbb{C}_{\lambda}^{2})$ to $\mathbb{C}_{\lambda}^{2}$.
\end{proof}

\subsection{The optima are solutions to an ODE}
\label{subsec:UB_ODE}

Following Optimization~\eqref{eq:reformulation:3}, we consider a {\em constant} supremum bid $\lambda \in (0,\, 1)$ for the moment, and study the next optimization (i.e., a functional of all twice differentiable CDF's $H \in \mathbb{C}_{\lambda}^{2}$):
\begin{align}
\label{eq:functional}\tag{IV}
    \sup \Bigg\{\, \int_{0}^{\lambda} \frac{H(x) \cdot H'(x) \cdot \d x}{H(x) + (1 - x) \cdot H'(x)} \,\Biggmid\, \text{$H \in \mathbb{C}_{\lambda}^{2}$} \,\Bigg\}.
\end{align}
Below we first (\Cref{lem:ODE}) leverage tools from Calculus of Variations to get a {\em necessary condition}, i.e., an {\em ordinary differential equation} (ODE), for any worst-case CDF $H \in \mathbb{C}_{\lambda}^{2}$ of Optimization~\eqref{eq:functional}, and then (\Cref{lem:implicit_equation}) explicitly resolve this ODE. In this way, any candidate worst-case CDF $H \in \mathbb{C}_{\lambda}^{2}$ will be controlled by just one parameter $\mu > 0$, namely $H = H_{\mu}$ (see \Cref{fig:function_H_mu} for a visual aid). The following two lemmas formalize our proof plan.

\begin{figure}[t]
    \centering
    \includegraphics[width = .7\textwidth]{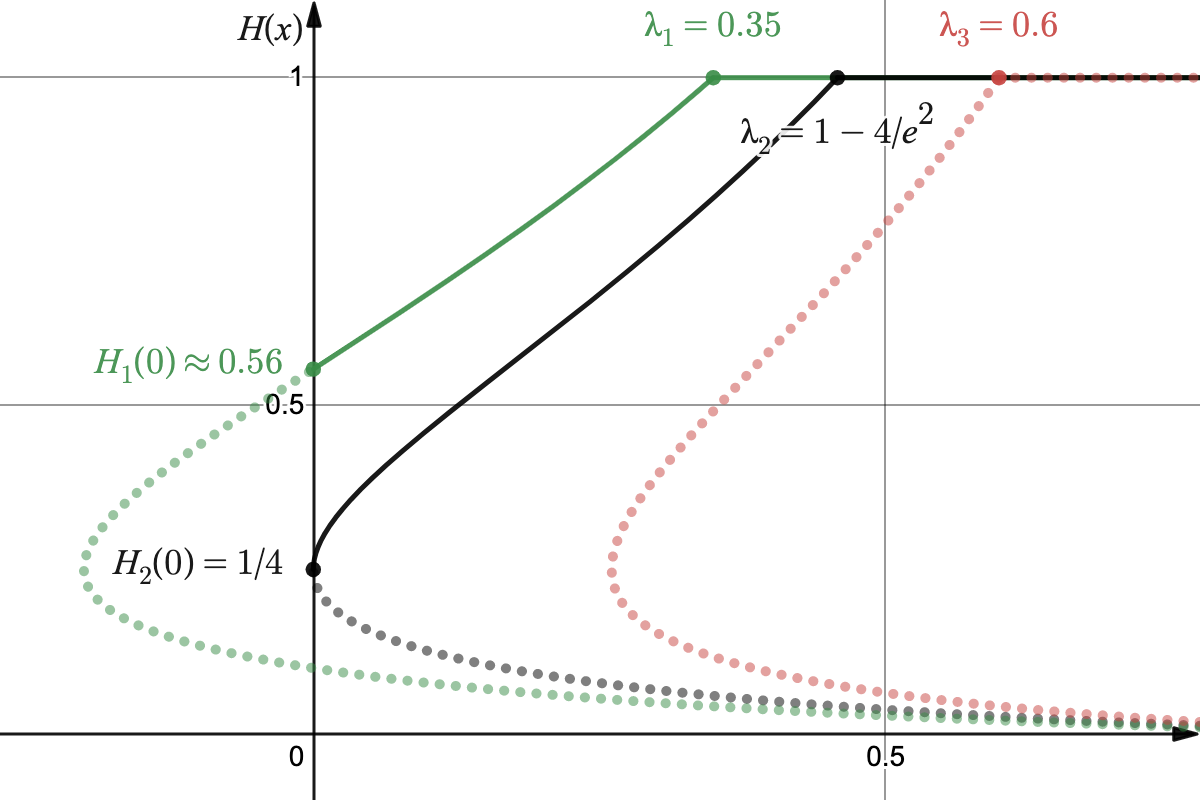}
    \caption{Diagram of the CDF $H_{\mu}$ given by Implicit Equation~\eqref{eq:implicit_equation}, where a {\em full curve} is a well-defined CDF restricted to $\{(x,\, H_{\mu}) \mid 0 \leq x \leq \lambda ~\text{and}~ (1 + 1 / \mu)^{-2} \leq H_{\mu} \leq 1\}$. \\ [.05in]
    In contrast, a {\em dotted curve} is just an {\em analytic continuation}, through the same implicit equation, to all $H_{\mu} \in [0,\, 1]$. Such an analytic continuation violates the above restriction and/or Condition~\eqref{eq:condition}, so it is {\em not} a well-defined CDF supported on $[0,\, \lambda]$. \\ [.05in]
    All the three CDF's in the figure choose the same parameter $\mu = 1$, so Condition~\eqref{eq:condition} is equivalent to $\lambda \leq 1 - (1 + 1 / \mu)^{2} \cdot e^{-2 / \mu} = 1 - 4 / e^{2} \approx 0.4587$.
    (i)~The green {\em full/dotted} curves choose a supremum bid $\lambda_{1} = 0.35$ that lies {\em within} the feasible space, by which the $y$-axis crosses the implicit equation twice, hence an {\em intersecting line}.
    (ii)~The black {\em full/dotted} curves choose the {\em boundary-case} supremum bid $\lambda_{2} = 1 - (1 + 1 / \mu)^{2} \cdot e^{-2 / \mu}$, by which the $y$-axis just touches the implicit equation, hence a {\em tangent line}.
    (iii)~The red {\em dotted} curve chooses an {\em infeasible} supremum bid $\lambda_{3} = 0.6$, by which the $y$-axis {\em never} crosses the implicit equation, hence a {\em disjoint line}.}
    \label{fig:function_H_mu}
\end{figure}

\begin{lemma}[The ODE; a necessary condition]
\label{lem:ODE}
Given a supremum bid $\lambda \in (0,\, 1)$, any extremum CDF $H \in \mathbb{C}_{\lambda}^{2}$ of Optimization~\eqref{eq:functional} satisfies the following ODE for any $x \in (0,\, \lambda)$:
\begin{align}
\label{eq:ODE}
    \frac{(1 - x) \cdot (H'(x))^{2}}{(H(x) + (1 - x) \cdot H'(x))^{2}}
    ~=~
    \frac{2 \cdot H(x) \cdot H'(x)}{(H(x) + (1 - x) \cdot H'(x))^{2}} ~-~ \frac{2 \cdot (1 - x) \cdot (H(x))^{2} \cdot H''(x)}{(H(x) + (1 - x) \cdot H'(x))^{3}}.
\end{align}
\end{lemma}

\begin{lemma}[The ODE solutions]
\label{lem:implicit_equation}
Given a supremum bid $\lambda \in (0,\, 1)$, any extremum CDF $H \in \mathbb{C}_{\lambda}^{2}$ of Optimization~\eqref{eq:functional} is determined by one parameter $\mu > 0$ such that
\begin{align}
\label{eq:condition}
    1 - (1 + 1 / \mu)^{2} \cdot e^{-2 / \mu} ~\geq~ \lambda,
\end{align}
namely $H = H_{\mu}$. Moreover, this extremum CDF $H_{\mu}$ is defined by the following implicit equation:
\begin{align}
    \label{eq:implicit_equation}
    \left\{\, (x,\, H_{\mu}) \,\middlemid\,
    \begin{aligned}
        & x ~=~ 1 - (1 - \lambda) \cdot H_{\mu} \cdot \exp\Big((1 + 1 / \mu) \cdot \big(2 - 2\sqrt{H_{\mu}}\big)\Big) \\
        & \text{\em such that $0 \leq x \leq \lambda$ and $(1 + 1 / \mu)^{-2} \leq H_{\mu} \leq 1$} \phantom{\bigg.}
    \end{aligned}\,\right\}.
\end{align}
\end{lemma}

From subsequent proof details, it will be clear that whenever the parameter $\mu > 0$ satisfies Condition~\eqref{eq:condition}, Implicit Equation~\eqref{eq:implicit_equation} does specify a well-defined CDF $H_{\mu}$.

We first establish \Cref{lem:ODE}. The proof exploits the {\em Euler-Lagrange equation} (stated below as \Cref{thm:Euler-Lagrange}), which is a standard tool from Calculus of Variations.

\begin{theorem}[The Euler-Lagrange equation; Calculus of Variations {\cite{GS20}}]
\label{thm:Euler-Lagrange}
Assume the following four premises for the functional $\calJ[f] \eqdef \int_{x_{1}}^{x_{2}} \calL\big(x, f(x), f'(x)\big) \cdot \d x$.
\begin{itemize}
    \item $x_{1}$ and $x_{2}$ are constants.
    
    \item $f(x)$ is twice continuously differentiable.
    
    \item $f'(x) = \frac{\d}{\d x} f(x)$.
    
    \item $\calL\big(x, f(x), f'(x)\big)$ is twice continuously differentiable with respect to variables $x$, $f$ and $f'$.
\end{itemize}
Then any extremum $\calJ[f]$ satisfies (a necessary condition) the Euler-Lagrange equation
\begin{align*}
    \frac{\partial \calL}{\partial f} ~=~ \frac{\d}{\d x} \bigg(\frac{\partial \calL}{\partial f'}\bigg)
    \quad \text{for any} \quad x \in (x_{1},\, x_{2}).
\end{align*}
\end{theorem}

\begin{remark}[The Euler-Lagrange equation]
We emphasize that (i)~the {\LHS} of this equation takes the partial derivative of $\calL(x, f, f')$ with respect to $f$; and (ii)~the {\RHS} first regards the partial derivative $\frac{\partial \calL}{\partial f'}$ as a function of {\em one} variable $x$, and then takes the derivative (with respect to $x$).
\end{remark}

Now we use the Euler-Lagrange equation to obtain the desired necessary condition. Regarding Optimization~\eqref{eq:functional}, we will consider the following functional.
\begin{align*}
    \calJ[H] ~\eqdef~ \int_{0}^{\lambda} \calL\big(x, H(x), H'(x)\big) \cdot \d x,
    \quad \text{where} \quad
    \calL\big(x, H(x), H'(x)\big) ~\eqdef~ \frac{H(x) \cdot H'(x)}{H(x) + (1 - x) \cdot H'(x)}.
\end{align*}
Notice that all the four premises of the Euler-Lagrange equation are satisfied:
\begin{itemize}
    \item $x_{1} = 0$ and $x_{2} = \lambda$ are constants.
    
    \item We only consider those twice continuously differentiable CDF's $H \in \mathbb{C}_{\lambda}^{2}$.
    
    \item Obviously $H'(x) = \frac{\d}{\d x} H(x)$.
    
    \item It is easy to verify that, with respect to variables $x \in (0,\, \lambda) \subseteq (0,\, 1)$, $H \in [0,\, 1]$ and $H' \geq 0$, the integrand $\calL\big(x, H(x), H'(x)\big)$ is twice continuously differentiable.
\end{itemize}

\begin{proof}[Proof of \Cref{lem:ODE}]
For ease of notation, we would use the shorthands $H = H(x)$, $h = H'(x)$ and $h' = H''(x)$, thus writing the integrand as $\calL(x, H, h) = \frac{H \cdot h}{H + (1 - x) \cdot h}$. Then regarding our functional, the {\LHS} of the Euler-Lagrange equation is given by
\[
    \frac{\partial \calL}{\partial H}
    ~=~ \frac{h \cdot (H + (1 - x) \cdot h) - H \cdot h}{(H + (1 - x) \cdot h)^{2}}
    ~=~ \frac{(1 - x) \cdot (H'(x))^{2}}{(H(x) + (1 - x) \cdot H'(x))^{2}}
    ~=~ \LHS \text{ of ODE~\eqref{eq:ODE}}
\]
where the second step combines the like terms and then substitutes back $H = H(x)$ and $h = H'(x)$.

Moreover, the partial derivative of $\calL(x, H, h)$ with respect to variable $h$ is given by
\[
    \frac{\partial \calL}{\partial h}
    ~=~ \frac{H \cdot (H + (1 - x) \cdot h) - (H \cdot h) \cdot (1 - x)}{(H + (1 - x) \cdot h)^{2}}
    ~=~ \frac{H^{2}}{(H + (1 - x) \cdot h)^{2}},
\]
where the last step combines the like terms. When being regarded as a function of {\em one} variable $x$,
% namely $\frac{\partial \calL}{\partial h} = \frac{(H(x))^{2}}{(H(x) + (1 - x) \cdot H'(x))^{2}}$,
the derivative of function $\frac{\partial \calL}{\partial h}$ (i.e., the {\RHS} of the Euler-Lagrange equation) is given by
\begin{align*}
    \frac{\d}{\d x} \bigg(\frac{\partial \calL}{\partial h}\bigg)
    & ~=~ \frac{\d}{\d x} \big(H^{2}\big) ~\cdot~ \frac{1}{(H + (1 - x) \cdot h)^{2}}
    ~+~ H^{2} ~\cdot~ \frac{\d}{\d x} \bigg(\frac{1}{(H + (1 - x) \cdot h)^{2}}\bigg) \\
    & ~=~ \frac{2 \cdot H \cdot h}{(H + (1 - x) \cdot h)^{2}} \hspace{2.06cm} +~ H^{2} \cdot \frac{-2 \cdot (1 - x) \cdot h'}{(H + (1 - x) \cdot h)^{3}} \phantom{\bigg.} \\
    & ~=~ \frac{2 \cdot H(x) \cdot H'(x)}{(H(x) + (1 - x) \cdot H'(x))^{2}}
    \hspace{0.78cm} -~ (H(x))^{2} \cdot \frac{2 \cdot (1 - x) \cdot H''(x)}{(H(x) + (1 - x) \cdot H'(x))^{3}} \phantom{\bigg.} \\
    & ~=~ \RHS \text{ of ODE~\eqref{eq:ODE}}, \phantom{\bigg.}
\end{align*}
where the third step substitutes back $H = H(x)$, $h = H'(x)$ and $h' = H''(x)$.

Therefore, ODE~\eqref{eq:ODE} is just a reformulation of the Euler-Lagrange equation and we immediately conclude the lemma from \Cref{thm:Euler-Lagrange}.
\end{proof}

In the rest of this subsection, we would resolve ODE~\eqref{eq:ODE}. For ease of notation, we often write $H = H(x)$, $H' = H'(x)$ and $H'' = H''(x)$. Following \Cref{lem:ODE}, we can deduce that
\begin{align}
    \text{ODE~\eqref{eq:ODE}} \quad \iff \quad \parbox{3.35cm}{\hfill $\dfrac{(1 - x) \cdot (H')^{2}}{(H + (1 - x) \cdot H')^{2}}$}
    & ~=~ \frac{2 \cdot H \cdot H'}{(H + (1 - x) \cdot H')^{2}} ~-~ \frac{2 \cdot (1 - x) \cdot H^{2} \cdot H''}{(H + (1 - x) \cdot H')^{3}} \phantom{\bigg.}
    \nonumber \\
    \iff \quad \parbox{3.35cm}{\hfill $\dfrac{2 \cdot (1 - x) \cdot H^{2} \cdot H''}{(H + (1 - x) \cdot H')^{3}}$}
    & ~=~ \frac{2H \cdot H'}{(H + (1 - x) \cdot H')^{2}} ~-~ \frac{(1 - x) \cdot (H')^{2}}{(H + (1 - x) \cdot H')^{2}} \phantom{\bigg.}
    \nonumber \\
    \iff \quad \parbox{3.35cm}{\hfill $2 \cdot (1 - x) \cdot \dfrac{H''}{H'}$}
    & ~=~ \bigg(2H \cdot H' - (1 - x) \cdot (H')^{2}\bigg) \cdot \frac{H + (1 - x) \cdot H'}{H^{2} \cdot H'}
    \nonumber \\
    \iff \quad \parbox{3.35cm}{\hfill $2 \cdot (1 - x) \cdot \dfrac{H''}{H'}$}
    & ~=~ \bigg(2 - (1 - x) \cdot \frac{H'}{H}\bigg) \cdot \bigg(1 + (1 - x) \cdot \frac{H'}{H}\bigg).
    \label{eq:solve_ODE:1}
\end{align}
Here the first step restates ODE~\eqref{eq:ODE}. The second step moves the first term to the {\RHS} and moves the third term to the {\LHS}. The third step first multiplies\footnote{Because $x \in (0,\, \lambda)$ and $\lambda \in (0,\, 1)$, this term $\geq \frac{1}{H^{2} \cdot H'} \cdot \binom{3}{2} \cdot H^{2} \cdot (1 - x) \cdot H' = 3 \cdot (1 - x)$ is {\em strictly positive}.} the both hand sides by the same term $\frac{1}{H^{2} \cdot H'} \cdot (H + (1 - x) \cdot H')^{3}$ and then simplifies the {\LHS}. And the last step rearranges the {\RHS}.

We define the function $z(x) \eqdef (1 - x) \cdot \frac{H'(x)}{H(x)}$ and, for brevity, may write $z = z(x)$ and $z' = z'(x)$. We have $z(x) \geq 0$ for any $x \in (0,\, \lambda)$, since the supremum bid $\lambda \in (0,\, 1)$ and the $H$ is a {\em nonnegative increasing} CDF. Therefore, there exists a {\em nonnegative} subset $\Omega \subseteq [0,\, +\infty]$ such that, over the whole support $x \in [0,\, \lambda]$, the CDF $H$ can be formulated in term of a {\em parametric equation}
\[
    \Big\{\, (x,\, H) = \big(\calX(z),\, \calH(z)\big) \,\Bigmid\, z \in \Omega \,\Big\}.
\]
Below we would capture the analytic formulas $x = \calX(z)$ and $H = \calH(z)$ and then (in the proof of \Cref{lem:implicit_equation}) study the condition for which the above parametric equation does yield a CDF $H$ that is well defined on the support $x \in [0,\, \lambda]$.

\begin{proof}[Proof of the $x = \calX(z)$ formula]
We can deduce that
\begin{align*}
    z \cdot H
    & ~=~ (1 - x) \cdot H'
    && \Longleftarrow && \mbox{definition of $z(x)$} \phantom{\dfrac{H''}{H'}} \\
    z' \cdot H ~+~ z \cdot H'
    & ~=~ -H' ~+~ (1 - x) \cdot H''
    && \Longleftarrow && \mbox{take derivative} \phantom{\dfrac{H''}{H'}} \\
    (1 - x) \cdot \frac{H''}{H'}
    & ~=~ 1 ~+~ z ~+~ z' \cdot \frac{H}{H'}
    && \Longleftarrow && \mbox{multiply $1 \big/ H'$ and rearrange} \phantom{\dfrac{H''}{H'}} \\
    (1 - x) \cdot \frac{H''}{H'}
    & ~=~ 1 ~+~ z ~+~ (1 - x) \cdot \frac{z'}{z}
    && \Longleftarrow && \mbox{substitute $H \big/ H' = \frac{1 - x}{z}$ for the {\RHS}} \phantom{\dfrac{H''}{H'}}
\end{align*}
Apply this identity (resp.\ the definition $z = (1 - x) \cdot \frac{H'}{H}$) to the {\LHS} (resp.\ {\RHS}) of ODE~\eqref{eq:solve_ODE:1}:
\begin{align}
    \text{ODE~\eqref{eq:solve_ODE:1}} \quad \iff \quad \parbox{4.35cm}{\hfill $2 \cdot \bigg(1 + z + (1 - x) \cdot \dfrac{z'}{z}\bigg)$}
    & ~=~ (2 - z) \cdot (1 + z)
    \nonumber \\
    \iff \quad \parbox{4.35cm}{\hfill $2 \cdot (1 - x) \cdot \dfrac{z'}{z}$}
    & ~=~ -z \cdot (1 + z) \phantom{\bigg.}
    \nonumber \\
    \iff \quad \parbox{4.35cm}{\hfill $\dfrac{2}{z^{2} \cdot (1 + z)} \cdot z'$}
    & ~=~ \frac{-1}{1 - x}.
    \label{eq:solve_ODE:2}
\end{align}
Here the second step cancels the common term $(2 + 2z)$ on the both hand sides. And the last step multiplies the both hand sides by the same positive term $\frac{1}{(1 - x) \cdot z \cdot (1 + z)}$.

We denote $\mu = z(\lambda)$ for some parameter $\mu > 0$, namely the boundary condition for \eqref{eq:solve_ODE:2} as a {\em partial differential equation} (PDE). For this PDE, the integration of the {\LHS} on any interval between $z = z(x)$ and $\mu = z(\lambda)$, is equal to that of the {\RHS} on the corresponding interval between $x$ and $\lambda$. Namely, we have
\begin{align}
    \text{PDE~\eqref{eq:solve_ODE:2}} \quad
    \Longrightarrow \quad \parbox{4.75cm}{\hfill $\displaystyle{\int_{\mu}^{z} \frac{2}{y^{2} \cdot (1 + y)} \cdot \d y}$}
    & ~=~ \int_{\lambda}^{x} \frac{-1}{1 - y} \cdot \d y
    \nonumber \\
    \Longrightarrow \quad \parbox{4.75cm}{\hfill $\Big(2\ln(1 + 1 / y) - 2 / y\Big)\Bigmid_{y \,=\, \mu}^{z}$}
    & ~=~ \ln(1 - y)\Bigmid_{y \,=\, \lambda}^{x} \phantom{\bigg.}
    \nonumber \\
    \Longrightarrow \quad \parbox{4.75cm}{\hfill $\displaystyle{2\ln\bigg(\frac{1 + 1 / z}{1 + 1 / \mu}\bigg) - 2 / z + 2 / \mu}$}
    & ~=~ \ln(1 - x) - \ln(1 - \lambda)
    \nonumber \\
    \Longrightarrow \quad \parbox{4.75cm}{\phantom{}}
    & \hspace{-4.75cm} x ~=~ \calX(z) ~=~ 1 - \frac{1 - \lambda}{(1 + 1 / \mu)^{2} \cdot e^{-2 / \mu}} \cdot (1 + 1 / z)^{2} \cdot e^{-2 / z}, \phantom{\bigg.}
    \label{eq:solve_ODE:3}
\end{align}
where the last step rearranges the equation. Notably, at $z = \mu$ we achieve the supremum bid $\calX(\mu) = \lambda$.

Formula~\eqref{eq:solve_ODE:3} is precisely the first part $x = \calX(z)$ of our parametric equation.
\end{proof}

\begin{proof}[Proof of the $H = \calH(z)$ formula]
To further obtain this formula, we observe that
\begin{align}
    \label{eq:solve_ODE:4}
    \frac{1}{H} \cdot \frac{\d H}{\d x} ~=~ \frac{H'}{H}
    ~=~ \frac{z}{1 - x}
    ~=~ \frac{(1 + 1 / \mu)^{2} \cdot e^{-2 / \mu}}{1 - \lambda} \cdot \frac{z}{(1 + 1 / z)^{2} \cdot e^{-2 / z}},
\end{align}
where the second step applies the definition $z = (1 - x) \cdot \frac{H'}{H}$ and the last step applies Formula~\eqref{eq:solve_ODE:3}. Moreover, the derivative of function $\calX(z)$ is given by
\begin{align}
    \calX'(z)
    ~=~ \frac{\d x}{\d z}
    & ~=~ -\frac{1 - \lambda}{(1 + 1 / \mu)^{2} \cdot e^{-2 / \mu}}
    \cdot \frac{\d}{\d z}\bigg((1 + 1 / z)^{2} \cdot e^{-2 / z}\bigg)
    \nonumber \\
    & ~=~ -\frac{1 - \lambda}{(1 + 1 / \mu)^{2} \cdot e^{-2 / \mu}}
    \cdot \bigg(2 \cdot (1 + 1 / z) \cdot \frac{-1}{z^{2}} \cdot e^{-2 / z} + (1 + 1 / z)^{2} \cdot e^{-2 / z} \cdot \frac{2}{z^{2}}\bigg)
    \nonumber \\
    & ~=~ -\frac{1 - \lambda}{(1 + 1 / \mu)^{2} \cdot e^{-2 / \mu}}
    \cdot \frac{2 \cdot (1 + 1 / z) \cdot e^{-2 / z}}{z^{3}},
    \label{eq:solve_ODE:5}
\end{align}
where the first step applies Formula~\eqref{eq:solve_ODE:3} and the last step combines the like terms. Notice that this derivative $\calX'(z)$ is {\em strictly negative}, so $\calX(z)$ is a {\em strictly decreasing} function.

Combining \Cref{eq:solve_ODE:4,eq:solve_ODE:5} together gives
\begin{align}
    \big(\ln \calH(z)\big)'
    & ~=~ \frac{\d}{\d z}\big(\ln H\big)
    ~=~ \frac{1}{H} \cdot \frac{\d H}{\d z}
    ~=~ \bigg(\frac{1}{H} \cdot \frac{\d H}{\d x}\bigg) \cdot \bigg(\frac{\d x}{\d z}\bigg)
    \nonumber \\
    & ~=~ \bigg(\frac{z}{(1 + 1 / z)^{2} \cdot e^{-2 / z}}\bigg) \cdot \bigg(-\frac{2 \cdot (1 + 1 / z) \cdot e^{-2 / z}}{z^{3}}\bigg)
    \nonumber \\
    & ~=~ \frac{-2}{z \cdot (z + 1)}.
    \label{eq:solve_ODE:6}
\end{align}
This ODE~\eqref{eq:solve_ODE:6} must admit the boundary condition $\calH(\mu) = 1$, given that (as mentioned) at $z = \mu$ we achieve the supremum bid $\calX(\mu) = \lambda$. For this reason, taking the integration of the both hand sides on any interval $[\mu,\, z]$ gives
\begin{align}
    \text{ODE~\eqref{eq:solve_ODE:6}} \quad
    \Longrightarrow \quad \parbox{3.05cm}{\hfill $\displaystyle{\int_{\mu}^{z} \big(\ln \calH(y)\big)' \cdot \d y}$}
    & ~=~ \int_{\mu}^{z} \frac{-2}{y \cdot (y + 1)} \cdot \d y
    \nonumber \\
    \Longrightarrow \quad \parbox{3.05cm}{\hfill $\ln \calH(z) - \ln \calH(\mu)$}
    & ~=~ \Big(2\ln(1 + 1 / y)\Big)\Bigmid_{y \,=\, \mu}^{z} \phantom{\bigg.}
    \nonumber \\
    \Longrightarrow \quad \parbox{3.05cm}{\hfill $\ln \calH(z)$}
    & ~=~ 2\ln\bigg(\frac{1 + 1 / z}{1 + 1 / \mu}\bigg)
    \nonumber \\
    \Longrightarrow \quad \parbox{3.05cm}{\hfill $H ~=~ \calH(z)$}
    & ~=~ \frac{(1 + 1 / z)^{2}}{(1 + 1 / \mu)^{2}}.
    \label{eq:solve_ODE:7}
\end{align}
where the third step applies the boundary condition $\calH(\mu) = 1$ to the {\LHS}. Notice that $\calH(z)$ is a {\em strictly decreasing} function.

Formula~\eqref{eq:solve_ODE:7} is precisely the second part $H = \calH(z)$ of our parametric equation.
\end{proof}

Based on the above discussions, we are able to prove \Cref{lem:implicit_equation}.

\begin{proof}[Proof of \Cref{lem:implicit_equation}]
We first check Condition~\eqref{eq:condition} and then capture Implicit Equation~\eqref{eq:implicit_equation}.

\vspace{.1in}
\noindent
{\bf Condition~\eqref{eq:condition}.}
As mentioned, both of $\calX(z)$ and $\calH(z)$ are {\em strictly decreasing} functions (\Cref{eq:solve_ODE:5,eq:solve_ODE:7}). Moreover, at $z = \mu$ we achieve the supremum bid $\calX(\mu) = \lambda$ and the boundary condition $\calH(\mu) = 1$. By considering all possible parameters $z \in [\mu,\, +\infty]$, our parametric equation $\big(\calX(z),\, \calH(z)\big)$ yields a {\em strictly increasing} function $\bar{H}$ with the domain $x \in [\calX(+\infty),\, \lambda]$. In this domain, function $\bar{H}$ is bounded between $\calH(+\infty) = (1 + 1 / \mu)^{-2} \geq 0$ and $\calH(\mu) = 1$.

This function $\bar{H}$ must be {\em consistent} with the considered CDF $H$. In this way, we require that: (i)~the support $x \in [0,\, \lambda]$ of CDF $H$ is a {\em subset} of the domain $x \in [\calX(+\infty),\, \lambda]$ of function $\bar{H}$; and, if so, (ii)~function $\bar{H}$ can be a {\em well-defined} CDF on the support $x \in [0,\, \lambda]$.

Clearly, the first requirement is equivalent to
\[
    0 ~\geq~ \calX(+\infty) ~=~ 1 - (1 - \lambda) \cdot (1 + 1 / \mu)^{-2} \cdot e^{2 / \mu},
\]
which after being rearranged gives Condition~\eqref{eq:condition}. Indeed, the first requirement ensures the second one -- function $\bar{H}$ is {\em strictly increasing} and is bounded between $\calH(+\infty) = (1 + 1 / \mu)^{-2} \geq 0$ and $\calH(\mu) = 1$ in the whole domain $x \in [\calX(+\infty),\, \lambda]$, thus (once the first requirement is satisfied) being a well-defined CDF on the support $x \in [0,\, \lambda] \subseteq [\calX(+\infty),\, \lambda]$.

\vspace{.1in}
\noindent
{\bf Implicit Equation~\eqref{eq:implicit_equation}.}
We reformulate our {\em parametric equation} $(x,\, H) = \big(\calX(z),\, \calH(z)\big)$ below, which precisely gives the claimed implicit equation.
\begin{align*}
    \text{Formula~\eqref{eq:solve_ODE:3}} ~~ \iff ~~
    x & ~=~ 1 - \frac{1 - \lambda}{(1 + 1 / \mu)^{2} \cdot e^{-2 / \mu}} \cdot (1 + 1 / z)^{2} \cdot e^{-2 / z} \phantom{\bigg.} \\
    & ~=~ 1 - \frac{1 - \lambda}{(1 + 1 / \mu)^{2} \cdot e^{-2 / \mu}} \cdot \Big((1 + 1 / \mu)^{2} \cdot H\Big) \cdot \exp\Big(2 - (1 + 1 / \mu) \cdot 2\sqrt{H}\Big) \phantom{\bigg.} \\
    % & ~=~ 1 - \parbox{3.02cm}{\centering $\Big((1 - \lambda) \cdot e^{2 / \mu}\Big)$} \cdot \parbox{2.83cm}{\centering $H$} \cdot \exp\Big(2 - (1 + 1 / \mu) \cdot 2\sqrt{H}\Big) \phantom{\bigg.} \\
    & ~=~ 1 - (1 - \lambda) \cdot H \cdot \exp\Big((1 + 1 / \mu) \cdot \big(2 - 2\sqrt{H}\big)\Big). \phantom{\bigg.}
\end{align*}
Here the first step restates Formula~\eqref{eq:solve_ODE:3}. The second step (Formula~\eqref{eq:solve_ODE:7}) substitutes $(1 + 1 / z)^{2} = (1 + 1 / \mu)^{2} \cdot H$ and $-1 / z = 1 - (1 + 1 / \mu) \cdot \sqrt{H}$.\footnote{Since the variable $z \in [\mu,\, +\infty]$ is positive, we safely ignore the other (impossible) case $-1 / z = 1 + (1 + 1 / \mu) \cdot \sqrt{H}$.}
% The third step cancels the common term $(1 + 1 / \mu)^{2}$.
And the last step rearranges the equation.

\vspace{.1in}
Combining Condition~\eqref{eq:condition} and Implicit Equation~\eqref{eq:implicit_equation} together finishes the proof.
\end{proof}

As an implication of \Cref{lem:implicit_equation}, by setting $x = 0$ in Implicit Equation~\eqref{eq:implicit_equation} we can characterize the pointmass $h_{\mu} \eqdef H_{\mu}(0)$ of a candidate worst-case CDF. This is formalized into \Cref{cor:pointmass}. Also, because \Cref{lem:implicit_equation} holds for any given supremum bid $\lambda \in (0,\, 1)$, combining it with \Cref{lem:reformulation} gives \Cref{cor:reformulation}. These two corollaries are more convenient for our later use.

\begin{corollary}[The ODE solutions]
\label{cor:pointmass}
Given a supremum bid $\lambda \in (0,\, 1)$ and a parameter $\mu > 0$ that meets Condition~\eqref{eq:condition}, the CDF $H_{\mu}$ defined by Implicit Equation~\eqref{eq:implicit_equation} has a pointmass $h_{\mu} \eqdef H_{\mu}(0)$ at $x = 0$, which is the unique solution to the following equation:
\begin{align*}
    \left\{\,~
    \begin{aligned}
        & 1 ~=~ (1 - \lambda) \cdot h_{\mu} \cdot \exp\Big((1 + 1 / \mu) \cdot \big(2 - 2\sqrt{h_{\mu}}\big)\Big) \\
        & \text{\em such that $(1 + 1 / \mu)^{-2} \leq h_{\mu} < 1$} \phantom{\bigg.}
    \end{aligned}
    ~\,\right\}.
\end{align*}
\end{corollary}

\begin{corollary}[{\PriceofAnarchy}]
\label{cor:reformulation}
Regarding {\FirstPriceAuctions}, the {\PriceofAnarchy} satisfies that
\begin{align*}
    \PoA ~\geq~
    1 - \sup \left\{\, (1 - \lambda) \cdot \int_{0}^{\lambda} \frac{H_{\mu}(x) \cdot H'_{\mu}(x) \cdot \d x}{H_{\mu}(x) + (1 - x) \cdot H'_{\mu}(x)} \,\middlemid\,
    \begin{aligned}
        & 0 < \mu < +\infty \\
        & 0 < \lambda \leq 1 - (1 + 1 / \mu)^{2} \cdot e^{-2 / \mu}
    \end{aligned}\,\right\}.
\end{align*}
\end{corollary}

\begin{proof}
Applying \Cref{lem:implicit_equation} to Optimization~\eqref{eq:reformulation:3} gives an interim optimization that takes the supremum operation {\em twice}: (i)~the inner operation, according to Condition~\eqref{eq:condition}, works on feasible space $\big\{\, \mu > 0 \,\bigmid\, 1 - (1 + 1 / \mu)^{2} \cdot e^{-2 / \mu} \geq \lambda \,\big\}$ given by the $\lambda \in (0,\, 1)$; and (ii)~the outer operation works on feasible space $\lambda \in (0,\, 1)$. Switching the order of both operations gives the corollary.
\end{proof}

\subsection{Optimizing the parameters \texorpdfstring{$(\lambda,\, \mu)$}{}}
\label{subsec:UB_lambda_mu}

We solve the optimization in \Cref{cor:reformulation} in two steps: (i)~\Cref{lem:UB_lambda_mu} transforms the considered {\em integral formula} into an equivalent {\em analytic formula}, Optimization~\eqref{eq:UB_lambda_mu}; and (ii)~\Cref{lem:worst_case} further derives the analytic solution, i.e., finding the (unique) worst case $(\lambda^{*},\, \mu^{*})$ of Optimization~\eqref{eq:UB_lambda_mu}.

\begin{lemma}[{\PriceofAnarchy}]
\label{lem:UB_lambda_mu}
Regarding {\FirstPriceAuctions}, the {\PriceofAnarchy} satisfies that
\begin{align}
\label{eq:UB_lambda_mu}\tag{V}
    \PoA ~\geq~
    1 - \sup \left\{\, (1 - \lambda) \cdot \bigg((1 - h_{\mu}) - \frac{2 - 2\sqrt{h_{\mu}}}{1 + 1 / \mu}\bigg) \,\middlemid\,
    \begin{aligned}
        & 0 < \mu < +\infty \\
        & 0 < \lambda \leq 1 - (1 + 1 / \mu)^{2} \cdot e^{-2 / \mu} \\
        & \text{\em $h_{\mu}$ defined by \Cref{cor:pointmass}}
    \end{aligned}\,\right\}.
\end{align}
\end{lemma}

\begin{lemma}[Worst case]
\label{lem:worst_case}
The solution to Optimization~\eqref{eq:UB_lambda_mu} is $1 - 1 / e^{2} \approx 0.8647$, which can be achieved by the worst case $(\lambda^{*},\, \mu^{*}) = (1 - 4 / e^{2},\, 1)$ and the resulting pointmass $h_{\mu^{*}} = 1 / 4$.
\end{lemma}

\begin{proof}[Proof of \Cref{lem:UB_lambda_mu}]
Following \Cref{cor:reformulation}, given any feasible pair $(\lambda,\, \mu)$, we study the integral
\begin{align}
    (1 - \lambda) \cdot \int_{0}^{\lambda} \frac{H_{\mu}(x) \cdot H'_{\mu}(x) \cdot \d x}{H_{\mu}(x) + (1 - x) \cdot H'_{\mu}(x)}
    & ~=~ (1 - \lambda) \cdot \int_{x \,\in\, [0,\, \lambda]} \frac{H_{\mu} \cdot \big(\d H_{\mu} \big/ \d x\big) \cdot \d x}{H_{\mu} + (1 - x) \cdot \big(\d H_{\mu} \big/ \d x\big)}
    \nonumber \\
    & ~=~ (1 - \lambda) \cdot \int_{H_{\mu} \,\in\, [h_{\mu},\, 1]} \frac{H_{\mu} \cdot \d H_{\mu}}{H_{\mu} + (1 - x) \big/ (\d x / \d H_{\mu})},
    \label{eq:PoA_formula:1}
\end{align}
where the last step changes the variables. We aim to give an analytic formula for the term $\frac{1 - x}{\d x / \d H_{\mu}}$. To this end, let us regard Implicit Equation~\eqref{eq:implicit_equation} as a function $x(H_{\mu})$ of variable $H_{\mu} \in [h_{\mu},\, 1]$ (i.e., the inverse function of the considered CDF), then the derivative is given by
\begin{align*}
    \frac{\d x}{\d H_{\mu}}
    & ~=~ -(1 - \lambda) \cdot \frac{\d}{\d H_{\mu}} \bigg(H_{\mu} \cdot \exp\Big((1 + 1 / \mu) \cdot \big(2 - 2\sqrt{H_{\mu}}\big)\Big)\bigg) \\
    & ~=~ -(1 - \lambda) \cdot \bigg(\exp\Big((1 + 1 / \mu) \cdot \big(2 - 2\sqrt{H_{\mu}}\big)\Big)\Bigg. \\
    & \phantom{~=~} \hspace{2.5cm} ~+~ \bigg.H_{\mu} \cdot \exp\Big((1 + 1 / \mu) \cdot \big(2 - 2\sqrt{H_{\mu}}\big)\Big) \cdot (1 + 1 / \mu) \cdot \frac{-1}{\sqrt{H_{\mu}}}\bigg) \\
    & ~=~ -(1 - \lambda) \cdot \Big(1 - (1 + 1 / \mu) \cdot \sqrt{H_{\mu}}\Big) \cdot \exp\Big((1 + 1 / \mu) \cdot \big(2 - 2\sqrt{H_{\mu}}\big)\Big) \phantom{\bigg.}
\end{align*}
Combining this formula with Implicit Equation~\eqref{eq:implicit_equation} gives 
\begin{align*}
    \frac{1 - x}{\d x \big/ \d H_{\mu}}
    & ~=~ \frac{\phantom{-}(1 - \lambda) \cdot \parbox{3.88cm}{\centering $H_{\mu}$} \cdot \exp\Big((1 + 1 / \mu) \cdot \big(2 - 2\sqrt{H_{\mu}}\big)\Big)}{-(1 - \lambda) \cdot \Big(1 - (1 + 1 / \mu) \cdot \sqrt{H_{\mu}}\Big) \cdot \exp\Big((1 + 1 / \mu) \cdot \big(2 - 2\sqrt{H_{\mu}}\big)\Big)} \\
    & ~=~ \frac{H_{\mu}}{(1 + 1 / \mu) \cdot \sqrt{H_{\mu}} - 1}.
\end{align*}

Plugging this analytic formula for the term $\frac{1 - x}{\d x / \d H_{\mu}}$ back to Integral~\eqref{eq:PoA_formula:1} gives
\begin{align*}
    \text{Integral~\eqref{eq:PoA_formula:1}}
    & ~=~ (1 - \lambda) \cdot \int_{H_{\mu} \,\in\, [h_{\mu},\, 1]} \frac{H_{\mu} \cdot \d H_{\mu}}{H_{\mu} + \frac{H_{\mu}}{(1 + 1 / \mu) \cdot \sqrt{H_{\mu}} - 1}} \\
    & ~=~ (1 - \lambda) \cdot \int_{H_{\mu} \in [h_{\mu},\, 1]} \bigg(1 - \frac{1}{(1 + 1 / \mu) \cdot \sqrt{H_{\mu}}}\bigg) \cdot \d H_{\mu} \\
    & ~=~ (1 - \lambda) \cdot \bigg(H_{\mu} - \frac{2\sqrt{H_{\mu}}}{1 + 1 / \mu}\bigg) \biggmid_{H_{\mu} \,=\, h_{\mu}}^{1} \\
    & ~=~ (1 - \lambda) \cdot \bigg((1 - h_{\mu}) - \frac{2 - 2\sqrt{h_{\mu}}}{1 + 1 / \mu}\bigg),
\end{align*}
which is precisely the formula in Optimization~\eqref{eq:UB_lambda_mu}. This finishes the proof of the lemma.
\end{proof}

Optimization~\eqref{eq:UB_lambda_mu} involves just three variables and, indeed, one of them $h_{\mu}$ is even determined by the other two $(\lambda,\, \mu)$. \Cref{lem:worst_case} gives the analytic solution to Optimization~\eqref{eq:UB_lambda_mu}.

\begin{proof}[Proof of \Cref{lem:worst_case}]
For the moment, we regard the $\mu > 0$ as a given parameter. It follows from \Cref{cor:pointmass} that $\lambda = G(h_{\mu})$, where the function $G$ is defined in the domain $h_{\mu} \in \big[(1 + 1 / \mu)^{-2},\, 1\big)$:
\begin{align*}
    G(h_{\mu}) ~\eqdef~ 1 - (1 / h_{\mu}) \cdot \exp\Big(-\big(2 - 2\sqrt{h_{\mu}}\big) \cdot (1 + 1 / \mu)\Big).
\end{align*}
This $G(h_{\mu})$ is a continuous function. And for any $h_{\mu} > (1 + 1 / \mu)^{-2}$, its derivative
\begin{align*}
    G'(h_{\mu})
    & ~=~ (1 / h_{\mu}^{2}) \cdot \exp\Big(-\big(2 - 2\sqrt{h_{\mu}}\big) \cdot (1 + 1 / \mu)\Big) \\
    & \hspace{1cm} ~-~ (1 / h_{\mu}) \cdot \exp\Big(-\big(2 - 2\sqrt{h_{\mu}}\big) \cdot (1 + 1 / \mu)\Big) \cdot \Big((1 + 1 / \mu) \big/ \sqrt{h_{\mu}}\Big) \\
    & ~=~ (1 / h_{\mu}^{2}) \cdot \exp\Big(-\big(2 - 2\sqrt{h_{\mu}}\big) \cdot (1 + 1 / \mu)\Big) \cdot \Big(1 - (1 + 1 / \mu) \cdot \sqrt{h_{\mu}}\Big)
\end{align*}
is {\em strictly positive}. Thus, this function $G(h_{\mu})$ is {\em strictly increasing} in the whole {\em left-closed right-open} domain $h_{\mu} \in \big[(1 + 1 / \mu)^{-2},\, 1\big)$. Moreover, the image of function $G(h_{\mu})$,
\[
    \Big\{\, G(h_{\mu}) \,\Bigmid\, (1 + 1 / \mu)^{-2} \leq h_{\mu} < 1 \,\Big\} ~=~ \big(0,\, 1 - (1 + 1 / \mu)^{2} \cdot e^{-2 / \mu}\big]
\]
is precisely the feasible space of $\lambda$ in Optimization~\eqref{eq:UB_lambda_mu}, which means the mapping between all feasible $\lambda$ and all $h_{\mu} \in \big[(1 + 1 / \mu)^{-2},\, 1\big)$ is {\em one-to-one}. For these reasons, we can deduce that
\begin{align*}
    \text{Optimization~\eqref{eq:UB_lambda_mu}}
    & ~=~ 1 - \sup \left\{\, \big(1 - G(h_{\mu})\big) \cdot \bigg((1 - h_{\mu}) - \frac{2 - 2\sqrt{h_{\mu}}}{1 + 1 / \mu}\bigg) \phantom{\frac{}{\Big.}}\middlemid\,
    \begin{aligned}
        & 0 < \mu < +\infty \\
        & (1 + 1 / \mu)^{-2} \leq h_{\mu} < 1
    \end{aligned}\,\right\} \\
    & ~=~ 1 - \sup \left\{\, \frac{(1 - h_{\mu}) - \big(2 - 2\sqrt{h_{\mu}}\big) \cdot (1 + 1 / \mu)^{-1}}{h_{\mu} \cdot \exp\Big(\big(2 - 2\sqrt{h_{\mu}}\big) \cdot (1 + 1 / \mu)\Big)} \,\middlemid\,
    \begin{aligned}
        & 0 < \mu < +\infty \\
        & (1 + 1 / \mu)^{-2} \leq h_{\mu} < 1
    \end{aligned}\,\right\} \\
    & ~=~ 1 - \sup \left\{\, \calG(\beta,\, \gamma) \eqdef \frac{(1 - \beta^{2}) - (2 - 2\beta) \cdot \gamma}{\beta^{2} \cdot \exp\big((2 - 2\beta) \cdot \gamma^{-1}\big)}\phantom{\frac{}{\Big.}} \middlemid\, \text{$0 < \gamma \leq \beta < 1$} \,\right\},
\end{align*}
where the first step switches the optimization variable from $\lambda$ to $h_{\mu} \in \big[(1 + 1 / \mu)^{-2},\, 1\big)$, the second step plugs the $G(h_{\mu})$ formula, and the last step substitutes $\beta \eqdef \sqrt{h_{\mu}}$ and $\gamma \eqdef (1 + 1 / \mu)^{-1}$.

Below we would reason about the function $\calG(\beta,\, \gamma)$ defined above. The partial derivative of this function with respect to variable $\gamma$ is given by
\begin{align*}
    \frac{\partial \calG}{\partial \gamma}
    & ~=~ \frac{-(2 - 2\beta)}{\beta^{2} \cdot \exp\big((2 - 2\beta) \cdot \gamma^{-1}\big)}
    ~+~ \frac{(1 - \beta^{2}) - (2 - 2\beta) \cdot \gamma}{\beta^{2} \cdot \exp\big((2 - 2\beta) \cdot \gamma^{-1}\big)} \cdot (2 - 2\beta) \cdot \gamma^{-2} \\
    & ~=~ \frac{(1 - \beta^{2}) - (2 - 2\beta) \cdot \gamma - \gamma^{2}}{\beta^{2} \cdot \exp\big((2 - 2\beta) \cdot \gamma^{-1}\big)} \cdot (2 - 2\beta) \cdot \gamma^{-2},
\end{align*}
where the last step rearranges the formula. Clearly, in the feasible space $\{0 < \gamma \leq \beta < 1\}$, the sign of this partial derivative $\partial \calG \big/ \partial \gamma$ depends on the numerator $(1 - \beta^{2}) - 2\gamma \cdot (1 - \beta) - \gamma^{2}$.

\begin{figure}[t]
    \centering
    \includegraphics[width = .9\textwidth]{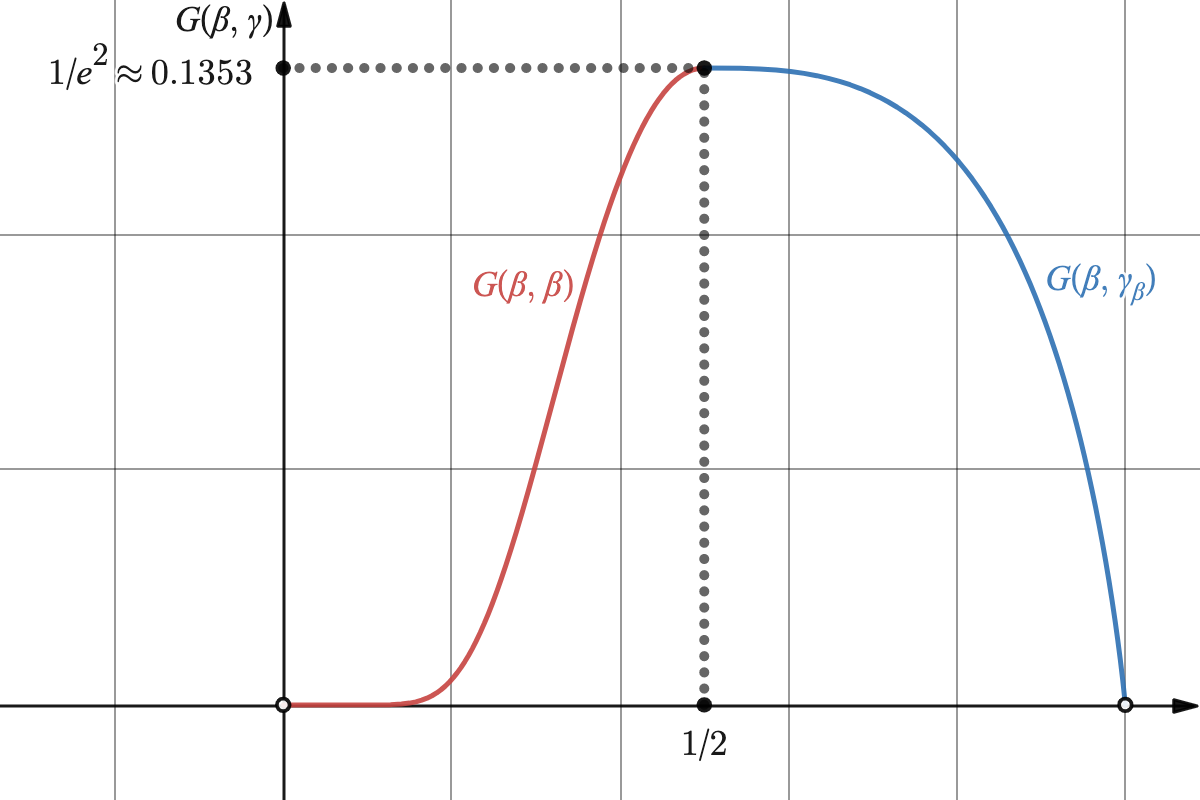}
    \caption{Diagram of (the red curve) the $\calG(\beta,\, \beta)$ as an increasing function of $\beta \in (0,\, 1 / 2]$ and (the blue curve) the $\calG(\beta,\, \gamma_{\beta})$ as a decreasing function of $\beta \in (1 / 2,\, 1)$.
    \label{fig:functionG}}
\end{figure}

For the moment, we regard $\beta \in (0,\, 1)$ as a given parameter and consider the parabola
\[
    P_{\beta}(\gamma) ~\eqdef~ (1 - \beta^{2}) - 2\gamma \cdot (1 - \beta) - \gamma^{2}
\]
in the domain $\gamma \in [0,\, \beta]$. This parabola opens {\em downwards} and has a {\em nonpositive} axis of symmetry $\gamma = -(1 - \beta)$, thus being a {\em decreasing} function in the whole domain $\gamma \in [0,\, \beta]$. Also, the left/right ends are respectively $P_{\beta}(0) = 1 - \beta^{2} > 0$ and $P_{\beta}(\beta) = 1 - 2\beta$. Let us do case analysis.

\vspace{.1in}
\noindent
{\bf Case~I that $0 < \beta \leq 1 / 2$.}
Now the both ends $P_{\beta}(0)$ and $P_{\beta}(\beta)$, as well as the whole parabola, are {\em nonnegative}, which implies that $\partial \calG \big/ \partial \gamma \geq 0$ for any $\gamma \in (0,\, \beta]$. Accordingly, we have
\begin{align*}
    \sup \Big\{\, \calG(\beta,\, \gamma) \,\Bigmid\, \text{\bf Case~I} \,\Big\}
    & ~=~ \sup \Big\{\, \calG(\beta,\, \beta) \,\Bigmid\, \text{$0 < \beta \leq 1 / 2$} \,\Big\} \\
    & ~=~ \sup \Big\{\, (1 / \beta - 1)^{2} \cdot e^{-(2 / \beta - 2)} \,\Bigmid\, \text{$0 < \beta \leq 1 / 2$} \,\Big\} \\
    & ~=~ 1 / e^{2}, \phantom{\Big.}
\end{align*}
where the last two steps can be checked through elementary algebra. In particular, the $\calG(\beta,\, \beta)$ as a function of $\beta \in (0,\, 1 / 2]$ is increasing (see \Cref{fig:functionG} for a visual aid), and the supremum $1 / e^{2} \approx 0.1353$ is achieved when $\beta = \gamma = 1 / 2$.

\vspace{.1in}
\noindent
{\bf Case~II that $1 / 2 < \beta < 1$.}
Now the left end $P_{\beta}(0) = 1 - \beta^{2}$ is {\em positive}, the right end $P_{\beta}(\beta) = 1 - 2\beta$ is {\em negative}, and (elementary algebra) the parabola has exactly one root $\gamma_{\beta} \eqdef \sqrt{(2 - 2\beta)} - (1 - \beta)$ in its domain. This implies that $\partial \calG \big/ \partial \gamma \geq 0$ for any $\gamma \in [0,\, \gamma_{\beta}]$ and $\partial \calG \big/ \partial \gamma \leq 0$ for any $\gamma \in [\gamma_{\beta},\, \beta]$. As a consequence, we have
\begin{align*}
    \sup \Big\{\, \calG(\beta,\, \gamma) \,\Bigmid\, \text{\bf Case~II} \,\Big\}
    & ~=~ \sup \Big\{\, \calG(\beta,\, \gamma_{\beta}) \,\Bigmid\, \text{$1 / 2 < \beta < 1$} \,\Big\} \\
    & ~=~ \sup \Big\{\, \Big(\tfrac{\sqrt{2} - \sqrt{1 - \beta}}{\beta / \sqrt{1 - \beta}}\Big)^{2} \cdot \exp\Big(\tfrac{-2\sqrt{1 - \beta}}{\sqrt{2} - \sqrt{1 - \beta}}\Big) \,\Bigmid\, \text{$1 / 2 < \beta < 1$} \,\Big\} \\
    & ~=~ 1 / e^{2}, \phantom{\Big.}
\end{align*}
where the last two steps can be checked through elementary algebra. In particular, the $\calG(\beta,\, \gamma_{\beta})$ as a function of $\beta \in (1 / 2,\, 1)$ is decreasing (see \Cref{fig:functionG} for a visual aid), and the supremum $1 / e^{2} \approx 0.1353$ is approached (but cannot be achieved) when $\gamma = \gamma_{\beta}$ and $\beta \searrow 1 / 2$.

\vspace{.1in}
Combining both cases together gives $\sup \{\calG(\beta,\, \gamma) \mid 0 < \gamma \leq \beta < 1\} = 1 / e^{2}$ and thus
\begin{align*}
    \text{Optimization~\eqref{eq:UB_lambda_mu}}
    ~=~ 1 - \sup \Big\{\, \calG(\beta,\, \gamma) \,\Bigmid\, 0 < \gamma \leq \beta < 1 \,\Big\}
    ~=~ 1 - 1 / e^{2}.
\end{align*}
In particular, the worst case can be achieved when $\beta = \gamma = 1 / 2$, or equivalently, when the supremum bid $\lambda = 1 - 4 / e^{2} \approx 0.4587$, the parameter $\mu = 1$ and the pointmass $h_{\mu} = 1 / 4$.

This finishes the proof of \Cref{lem:worst_case}.
\end{proof}

\Cref{lem:UB_lambda_mu,lem:worst_case} together imply the lower-bound part of \Cref{thm:main}, which is restated below.

\begin{restate}[{\Cref{thm:main}}]
The {\PriceofAnarchy} in {\FirstPriceAuctions} is $\geq 1 - 1 / e^{2} \approx 0.8647$.
% Selling an indivisible item to $n \geq 1$ bidders whose values are distributed independently, via {\FirstPriceAuction}, the {\sf Bayes Nash \PriceofAnarchy} $\geq 1 - 1 / e^{2} \approx 0.8647$.
\end{restate}

\newpage

\section{Upper Bound Analysis}
\label{sec:LB}

In this section, we first (\Cref{exp:LB}) present our worst case instances in terms of {\em bid distributions}; then (\Cref{lem:LB_validity}) check the validity, i.e., the corresponding {\em value distributions} are well defined; and then (\Cref{lem:LB_poa}) evaluate the auction/optimal {\SocialWelfares} from this instance. The upper-bound part of \Cref{thm:main} will be a direct consequence of \Cref{lem:LB_validity,lem:LB_poa}.

The idea is to simulate the worst-case pseudo instance $H^{*} \otimes L^{*}$ from \Cref{sec:UB}, {\em using a number of i.i.d.\ low-impact bidders $\{L\}^{\otimes n}$ to approximate the pseudo bidder $L^{*}$} (\Cref{rem:pseudo_instance}): Each individual low-impact bidder $L$ is likely to have a small bid $\approx 0$, but the highest bid from $\{L\}^{\otimes n}$ is (almost) identically distributed as the worst-case pseudo bidder $L^{*}$.

\begin{example}[Worst-case instances]
\label{exp:LB}
Let $\lambda^{*}=1 - 4 / e^{2}$ and $L^{*}(b) = \frac{1 - \lambda^{*}}{1 - b}$ for any $b \in [0,\, \lambda^{*}]$ be a CDF function. Another CDF function $H^{*}(b)$ is  given by the implicit equation
\begin{align}
\label{eq:LB:H}
    \left\{\, (b,\, H^{*}) \,\middlemid\,
    \begin{aligned}
        & b ~=~ 1 - 4H^{*} \cdot \exp\Big(2 - 4\sqrt{H^{*}}\Big) \\
        & \text{such that $0 \leq b \leq \lambda^{*}$ and $1 / 4 \leq H^{*} \leq 1$} \phantom{\bigg.}
    \end{aligned} \,\right\}.
\end{align} 
For an arbitrarily small constant $\epsilon \in (0,\, 1)$, consider the following $(n + 1)$-bidder instance $H \otimes \{L\}^{\otimes n}$ for $n \eqdef \lceil 3 / \epsilon \rceil \geq 4$:
\begin{itemize}
    \item $H(b) \eqdef H^{*}(b)$ for $b \in [0,\, \lambda^{*}]$ denotes the monopolist.

    \item $L(b) \eqdef \big(L^{*}(b)\big)^{\frac{1}{n - 1}}$ for $b \in [0,\, \lambda^{*}]$ denotes the {\em common} bid distribution of $n \geq 4$ many i.i.d.\ {\em low-impact} bidders.
\end{itemize}
\end{example}

Here we choose $L(b) \eqdef \big(L^{*}(b)\big)^{\frac{1}{n - 1}}$ rather than $ \big(L^{*}(b)\big)^{\frac{1}{n}}$ such that each lower-impact bidder $L$ has the competing bid distribution $H(b)\cdot \big(L(b)\big)^{n - 1} = H^{*}(b) \cdot L^{*}(b)$, exactly the same as the pseudo bidder $L^{*}$. But in this way, the monopolist $H$'s competing bid $\big(L(b)\big)^{n}=\big(L^{*}(b)\big)^{\frac{n}{n - 1}}$ slightly differs from the $L^{*}(b)$,\footnote{The other choice $L(b) \eqdef\big(L^{*}(b)\big)^{\frac{1}{n}}$ is also fine, but the analysis will be more complicated.}
so his/her value is not a constant $\equiv 1$ but a random variable $\approx 1$ spanning in a tiny range. 
These are all formalized in \Cref{lem:LB_validity}, which checks the validity of this $(n + 1)$-bidder instance; see \Cref{fig:LB} for a visual aid.

\begin{lemma}[Validity]
\label{lem:LB_validity}
Given an arbitrarily small constant $\epsilon \in (0,\, 1)$, the following hold for the $(n + 1)$-bidder instance $H \otimes \{L\}^{\otimes n}$ in \Cref{exp:LB}:
\begin{enumerate}[font = {\em\bfseries}]
    \item\label{lem:LB_validity:1}
    The monopolist $H$ has an increasing bid-to-value mapping $\varphi_{H}(b) \eqdef 1 - \frac{1 - b}{n}$ for $b \in [0,\, \lambda^{*}]$ and a well-defined value distribution $V_{H}(v)$ given by the implicit equation
    \begin{align*}
        \left\{\, (v,\, V_{H}) \,\middlemid\,
        \begin{aligned}
            & v ~=~ 1 - \tfrac{1}{n} \cdot 4V_{H} \cdot \exp\Big(2 - 4\sqrt{V_{H}}\Big) \\
            & \text{\em such that $1 - \tfrac{1}{n} \leq v \leq 1 - \tfrac{4 / e^{2}}{n}$ and $1 / 4 \leq V_{H} \leq 1$} \phantom{\bigg.}
        \end{aligned} \,\right\}.
    \end{align*}
    The strategy $s_{H}$ as the inverse of $\varphi_{H}(b)$ is given by $s_H(v) = 1 - n \cdot (1 - v)$.

    \item\label{lem:LB_validity:2}
    Each low-impact bidder $L$ has an increasing bid-to-value mapping $\varphi_{L}(b)$ given by the parametric equation
    \begin{align*}
        \Big\{\, (b,\, \varphi_{L}) = \big(1 - t^{2} \cdot e^{2 - 2t},\quad 1 - t \cdot e^{2 - 2t}\big) \,\Bigmid\,
        1 \leq t \leq 2 \,\Big\}
    \end{align*}
    and a well-defined value distribution $V_{L}(v)$ given by the parametric equation
    \begin{align*}
        \Big\{\, (v,\, V_{L}) = \big(1 - t \cdot e^{2 - 2t},\quad \big(4 / t^{2} \cdot e^{2t - 4}\big)^{\frac{1}{n - 1}}\big) \,\Bigmid\,
        1 \leq t \leq 2 \,\Big\}.
    \end{align*}
    The strategy $s_{L}$ as the inverse of $\varphi_{L}(b)$ is given by the parametric equation $s_L(1 - t \cdot e^{2 - 2t})=1 - t^{2} \cdot e^{2 - 2t}$ for $1 \leq t \leq 2$.
    
    \item\label{lem:LB_validity:3}
    The strategy profile $s_{H} \otimes \{s_{L}\}^{\otimes n}$ forms an equilibrium $s_{H} \otimes \{s_{L}\}^{\otimes n} \in \bbBNE(V_{H} \otimes \{V_{L}\}^{\otimes n})$.
\end{enumerate}
\end{lemma}

\begin{figure}[t]
    \centering
    \includegraphics[width = .9\textwidth]{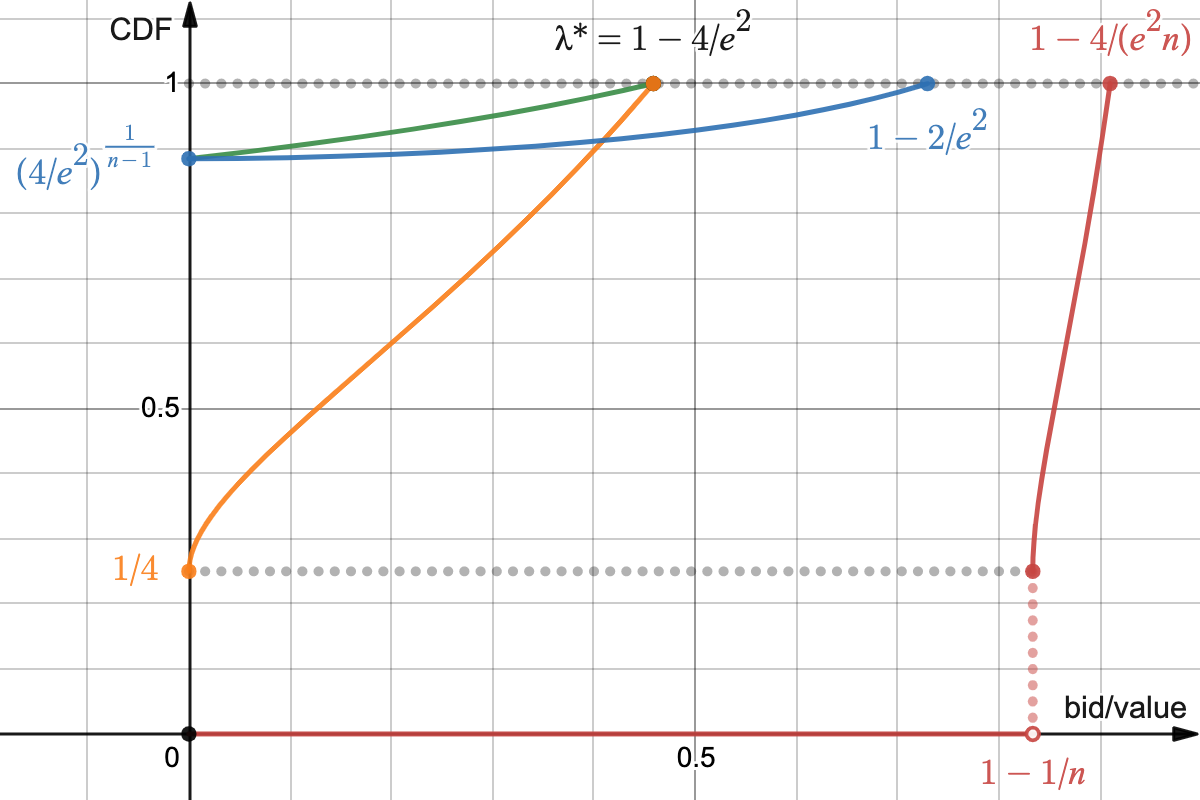}
    \caption{Diagram of the $(n + 1)$-bidder instance $H \otimes \{L\}^{\otimes n}$ in \Cref{exp:LB}. \\
    (i)~The monopolist $H$ has a bid CDF $H$ (orange) given by Implicit Equation~\eqref{eq:LB:H} and a value CDF $V_{H}$ (red) given by \Cref{lem:LB_validity:1} of \Cref{lem:LB_validity}. \\
    (ii)~Each of the i.i.d.\ low-impact bidders $\{L\}^{\otimes n}$ has a bid CDF $L$ (green) given by $L(b) = (\frac{1 - \lambda^{*}}{1 - b})^{\frac{1}{n - 1}}$ for any $b \in [0,\, \lambda^{*}] = [0,\, 1 - 4 / e^{2}]$ and a value CDF $V_{L}$ (blue) given by \Cref{lem:LB_validity:2} of \Cref{lem:LB_validity}.}
    \label{fig:LB}
\end{figure}

\begin{proof}
Let us prove {\bf \Cref{lem:LB_validity:1,lem:LB_validity:2}} one by one. Then {\bf \Cref{lem:LB_validity:3}} will be a direct consequence.

\vspace{.1in}
\noindent
{\bf \Cref{lem:LB_validity:1}.}
The monopolist $H$ has a competing bid distribution $L(b)^{n} = L^{*}(b)^{\frac{n}{n - 1}} = (\frac{1 - \lambda^{*}}{1 - b})^{\frac{n}{n - 1}}$ and thus (by elementary algebra) a bid-to-value mapping
\begin{align*}
    \varphi_{H}(b) ~=~ b + L(b)^{n} \big/ \big(L(b)^{n}\big)' ~=~ b + \tfrac{n - 1}{n} \cdot (1 - b) ~=~ 1 - \tfrac{1 - b}{n}.
\end{align*}
This mapping $\varphi_{H}(b)$ is increasing on $b \in [0,\, \lambda^{*}] = [0,\, 1 - 4 / e^{2}]$, with the minimum $\varphi_{H}(0) = 1 - \frac{1}{n}$ and the maximum $\varphi_{H}(\lambda^{*}) = 1 - \frac{1 - \lambda^{*}}{n} = 1 - \frac{4 / e^{2}}{n}$.

According to \Cref{lem:value_dist}, the value distribution $V_{H}(v)$ follows the parametric equation
\begin{align*}
    \Big\{\, (v,\, V_{H}) = \big(\varphi_{H}(b),\, H(b)\big) \,\Bigmid\, b \in [0,\, \lambda^{*}] \,\Big\},
\end{align*}
where the formula $V_{H} = H(b) = H^{*}(b)$ is defined by Implicit Equation~\eqref{eq:LB:H}. Based on elementary algebra (i.e., using the substitution $b = 1 - n \cdot (1 - v)$ due to the formula $v = \varphi_{H}(b) = 1 - \frac{1 - b}{n}$), we obtain the claimed implicit equation for $V_{H}(v)$. {\bf \Cref{lem:LB_validity:1}} follows then.

\vspace{.1in}
\noindent
{\bf \Cref{lem:LB_validity:2}.}
For notational brevity, below we may simply write $H^{*} = H^{*}(b)$ and $L^{*} = L^{*}(b)$ etc.
We let $t \eqdef 2\sqrt{H^{*}} \in [1,\, 2]$  as a function of $b$. Then
$b = 1 - t^2 \cdot \exp(2 - 2t)$ by \Cref{eq:LB:H} and
\begin{equation}
  \tfrac{\d b}{\d t} ~=~ (2t^{2} - 2t) \cdot \exp(2 - 2t).
\end{equation}
The low-impact bidders $L^{\otimes n}$ have the same competing bid distribution $B(b) = H(b) \cdot \big(L(b)\big)^{n - 1}= H^{*}(b) \cdot L^{*}(b)$. This distribution can be written as
\[
    B ~=~ H^{*} \cdot L^{*} ~=~ H^{*} \cdot \tfrac{4 / e^{2}}{1 - b} ~=~ \exp\big(-4 + 2t\big),
\]
where
% the second step applies $L^{*}(b) = \frac{1 - \lambda^{*}}{1 - b} = \frac{4 / e^{2}}{1 - b}$, and
the last step applies  $H^{*}=t^2/4$ and $b = 1 - t^2 \cdot \exp(2 - 2t)$. The above formula switches the range $t\in [1,\, 2]$ into the range $B \in [1 / e^{2},\, 1]$.
Thus
\begin{align*}
    \varphi_{L}
    & ~=~ b + B \big/ (\tfrac{\d B}{\d b}) \\
    & ~=~ b + \exp\big(-4 + 2t \big) \cdot (\tfrac{\d b}{\d t}) \big/ (\tfrac{\d B}{\d t}) \\
    & ~=~ 1 - t^2 \cdot \exp(2 - 2t) + (t^{2} - t) \cdot \exp(2 - 2t) \\
    & ~=~ 1 - t \cdot e^{2 - 2t}.
\end{align*}

We conclude with the parametric equation for $ \varphi_{L}(b)$ claimed in the statement of \Cref{lem:LB_validity:2}:
\begin{align*}
    \Big\{\, (b,\, \varphi_{L}) = \big(1 - t^{2} \cdot e^{2 - 2t},\quad 1 - t \cdot e^{2 - 2t}\big) \,\Bigmid\,
    1 \leq t \leq 2 \,\Big\}
\end{align*}
With respect to $t \in [1,\, 2]$, the both formulas $b = 1 - t^{2} \cdot e^{2 - 2t}$ and $\varphi_{L} = 1 - t \cdot e^{2 - 2t}$ are increasing functions. Thus, the bid-to-value mapping $\varphi_{L}(b)$ is increasing for $b \in [0,\, 1 - 4 / e^{2}] = [0,\, \lambda^{*}]$, with the minimum $\varphi_{L}(0) = 0$ and the maximum $\varphi_{L}(\lambda^{*}) = 1 - 2 / e^{2}$.

According to \Cref{lem:value_dist}, the value distribution $V_{L}(v)$ follows the parametric equation
\begin{align*}
    \Big\{\, (v,\, V_{L}) = \big(\varphi_{L}(b),\, L(b)\big) \,\Bigmid\, b \in [0,\, \lambda^{*}] \,\Big\},
\end{align*}
where the formula $V_{L} = L(b) = (\frac{1 - \lambda^{*}}{1 - b})^{\frac{1}{n - 1}} = (\frac{4 / e^{2}}{1 - b})^{\frac{1}{n - 1}}$. By elementary algebra (i.e., employing the formula $b = 1 - t^{2} \cdot e^{2 - 2t}$ in the defining parametric equation for $\varphi_{L}(b) $), we derive the claimed parametric equation for $V_{H}(v)$. {\bf \Cref{lem:LB_validity:2}} follows then. This finishes the proof.
\end{proof}

Now we study the auction/optimal {\SocialWelfares} from our $(n + 1)$-bidder instance $H \otimes \{L\}^{\otimes n}$.

\begin{lemma}[{\PriceofAnarchy}]
\label{lem:LB_poa}
Given an arbitrarily small constant $\epsilon \in (0,\, 1)$, the following hold for the $(n + 1)$-bidder instance $H \otimes \{L\}^{\otimes n}$ in \Cref{exp:LB}:
\begin{enumerate}[font = {\em\bfseries}]
    \item\label{lem:LB_poa:1}
    The expected optimal {\SocialWelfare} $\OPT(H \otimes \{L\}^{\otimes n}) \geq 1 - \epsilon$; and

    \item\label{lem:LB_poa:2}
    The expected auction {\SocialWelfare} $\FPA(H \otimes \{L\}^{\otimes n}) \leq 1 - 1 / e^{2} $.
\end{enumerate}
\end{lemma}

\begin{proof}
Let us prove {\bf \Cref{lem:LB_poa:1,lem:LB_poa:2}} one by one.

\vspace{.1in}
\noindent
{\bf \Cref{lem:LB_poa:1}.}
The value $v_{H} \sim V_{H}$ of the monopolist $H$, supported on $\supp(V_{H}) = \big[1 - \frac{1}{n},\, 1 - \frac{4 / e^{2}}{n}\big]$, is higher than the values of the lower-impact bidders $\{L\}^{\otimes n}$, which are supported on 
$[0, 1-2/e^2]$. Thus the optimal {\SocialWelfare} $\OPT(H \otimes \{L\}^{\otimes n})$ always stems from the monopolist $H$:
\[
    \OPT(H \otimes \{L\}^{\otimes n}) ~=~ \E_{v_{H}}\big[v_{H}\big] ~\geq~ 1 - 1 \big/ n ~\geq~ 1 - \epsilon,
    \hspace{5.95cm}
\]
where the last step uses $n = \lceil 3 / \epsilon \rceil \geq 1 / \epsilon$ (\Cref{exp:LB}). {\bf \Cref{lem:LB_poa:1}} follows then.

\vspace{.1in}
\noindent
{\bf \Cref{lem:LB_poa:2}.}
Recall that the low-impact bidders $L^{\otimes n}$ have the common bid-to-value mapping $\varphi_{L}(b)$.
Following \Cref{lem:auction_welfare} (with $\calB(b) = H(b) \cdot L(b)^{n}$),
the auction {\SocialWelfare} $\FPA(H \otimes \{L\}^{\otimes n})$ from our $(n + 1)$-bidder instance is given by
\begin{align*}
    \FPA(H \otimes \{L\}^{\otimes n})
    & ~=~ \varphi_{H}(0) \cdot \calB(0) + \int_{0}^{\lambda^{*}} \Big(\varphi_{H}(b) \cdot \tfrac{H'(b)}{H(b)} \cdot \calB(0) + n \cdot \varphi_{L}(b) \cdot \tfrac{L'(b)}{L(b)} \cdot \calB(b)\Big) \cdot \d b \\
    & ~\leq~ \calB(0) + \int_{0}^{\lambda^{*}} \Big(\tfrac{H'(b)}{H(b)} \cdot \calB(0) + n \cdot \varphi_{L}(b) \cdot \tfrac{L'(b)}{L(b)} \cdot \calB(b)\Big) \cdot \d b \\
    & ~=~ \calB(0) + \int_{0}^{\lambda^{*}} \calB'(b) \cdot \d b - \int_{0}^{\lambda^{*}} n \cdot (1 - \varphi_{L}(b)) \cdot \tfrac{L'(b)}{L(b)} \cdot \calB(b) \cdot \d b \\
    & ~=~ 1 - \int_{0}^{\lambda^{*}} n \cdot (1-\varphi_{L}(b)) \cdot \tfrac{1}{(n - 1) \cdot (1 - b)} \cdot H(b) \cdot \big(\tfrac{1 - \lambda^{*}}{1 - b}\big)^{\frac{n}{n - 1}} \cdot \d b \\
    & ~\leq~ 1 - \tfrac{n}{n - 1} \cdot (1 - \lambda^{*})^{\frac{n}{n - 1}} \cdot \int_{0}^{\lambda^{*}} \frac{(1 - \varphi_{L}(b)) \cdot H(b)}{(1 - b)^2} \cdot \d b \\
    % & \qquad ~=~ 1 - \tfrac{n}{n - 1} \cdot (1 - \lambda^{*})^{\frac{n}{n-1}} \cdot \int_{0}^{\lambda^{*}} \Big( t e^{2 - 2t} \cdot \frac{t^2}{4} \cdot  t^{-4} e^{-4 + 4t} \Big) \cdot \d b \\
    & ~=~ 1 - \tfrac{n}{n - 1} \cdot (1 - \lambda^{*})^{\frac{n}{n - 1}} \cdot \int_{1}^{2} \frac{(t \cdot e^{2 - 2t}) \cdot (t^{2} / 4)}{(t^{2} \cdot e^{2 - 2t})^{2}} \cdot \Big(\frac{\d b}{\d t}\Big) \cdot \d t \\
    % & \qquad ~=~ 1 - \tfrac{n}{n - 1} \cdot (1 - \lambda^{*})^{\frac{n}{n-1}} \cdot \int_{1}^{2} \tfrac{1}{4t \cdot e^{2 - 2t}} \cdot (2t^{2} - 2t) \cdot \exp(2 - 2t) \cdot \d t \\
    & ~=~ 1 - \tfrac{n}{n - 1} \cdot (1 - \lambda^{*})^{\frac{n}{n - 1}} \cdot \int_{1}^{2} \frac{t - 1}{2} \cdot \d t \\
    & ~=~ 1 - \tfrac{n}{n - 1} \cdot (4 / e^2)^{\frac{1}{n - 1}} \cdot e^{-2} \phantom{\bigg.} \\
    & ~\leq~ 1 - e^{-2},
\end{align*}
where the last inequality uses the fact $\frac{n}{n - 1} \cdot (4/e^2)^{\frac{1}{n - 1}} = \big((1 + \frac{1}{n - 1})^{n - 1}\cdot 4 / e^2\big)^{\frac{1}{n - 1}} \geq (8 / e^2)^{\frac{1}{n - 1}}> 1$. The definition of $t \in [1,\, 2]$ and the expressions of $b$, $H$, $\varphi_{L}$ and $\frac{\d b}{\d t}$ in terms of $t$ are given in the proof of \Cref{lem:LB_poa}.
{\bf \Cref{lem:LB_poa:2}} follows then.
This finishes the proof.
\end{proof}

We directly infer the upper-bound part of \Cref{thm:main} (restated below) from \Cref{lem:LB_poa}, since the constant $\epsilon \in (0,\, 1)$ chosen in \Cref{exp:LB} can be arbitrarily small.

\begin{restate}[{\Cref{thm:main}}]
The {\PriceofAnarchy} in {\FirstPriceAuctions} is $\leq 1 - 1 / e^{2} \approx 0.8647$.
\end{restate}

\section*{Acknowledgements}
We are grateful to Xi Chen for invaluable discussions in multiple stages throughout this work and would like to thank Xiaohui Bei, Hu Fu, Tim Roughgarden, Rocco Servedio, Zhihao Gavin Tang, Zihe Wang, and anonymous reviewers for helpful comments.

% Y.J.\ is supported by NSF grants IIS-1838154, CCF-1563155, CCF-1703925, CCF-1814873, CCF-2106429, and CCF-2107187.
P.L.\ is supported by Science and Technology Innovation 2030 – ``New Generation of Artificial Intelligence'' Major Project No.(2018AAA0100903), NSFC grant 61922052 and 61932002, Innovation Program of Shanghai Municipal Education Commission, Program for Innovative Research Team of Shanghai University of Finance and Economics, and Fundamental Research Funds for Central Universities.

\newpage

\begin{flushleft}
\bibliographystyle{alpha}
\bibliography{main}
\end{flushleft}

\end{document}